\documentclass[11pt,canadian,english,phd,oneside,norunningheaders]{ubcthesis}
\usepackage{newcent}

\usepackage[T1]{fontenc}
\usepackage[latin9]{inputenc}
\usepackage{geometry}
\usepackage{array}
\usepackage{float}
\usepackage{textcomp}
\usepackage{amsmath}
\usepackage{amsthm}
\usepackage{amssymb}
\usepackage{graphicx}
\usepackage{setspace}
\usepackage[numbers]{natbib}
\usepackage{feyn}
\PassOptionsToPackage{normalem}{ulem}
\usepackage{ulem}
\makeatletter

\newcommand{\lyxmathsym}[1]{\ifmmode\begingroup\def\b@ld{bold}
  \text{\ifx\math@version\b@ld\bfseries\fi#1}\endgroup\else#1\fi}

\providecommand{\tabularnewline}{\\}
\floatstyle{ruled}
\newfloat{algorithm}{tbp}{loa}[chapter]
\providecommand{\algorithmname}{Algorithm}
\floatname{algorithm}{\protect\algorithmname}

\theoremstyle{plain}
\ifx\thechapter\undefined
	\newtheorem{thm}{\protect\theoremname}
\else
	\newtheorem{thm}{\protect\theoremname}[chapter]
\fi
\theoremstyle{definition}
\newtheorem{defn}[thm]{\protect\definitionname}
\theoremstyle{plain}
\newtheorem{prop}[thm]{\protect\propositionname}
\ifx\proof\undefined
\newenvironment{proof}[1][\protect\proofname]{\par
	\normalfont\topsep6\p@\@plus6\p@\relax
	\trivlist
	\itemindent\parindent
	\item[\hskip\labelsep\scshape #1]\ignorespaces
}{%
	\endtrivlist\@endpefalse
}
\providecommand{\proofname}{Proof}
\fi
\theoremstyle{plain}
\newtheorem{lem}[thm]{\protect\lemmaname}
\theoremstyle{plain}
\newtheorem{cor}[thm]{\protect\corollaryname}

\usepackage{float}

\usepackage[unicode=true,
  hyperfootnotes=false,
  linktocpage,
  linkbordercolor={0.5 0.5 1},
  citebordercolor={0.5 1 0.5},
  linkcolor=blue]{hyperref}

\usepackage{algorithm}
\usepackage{algpseudocode}

\institution{The University Of British Columbia}

\faculty{The Faculty of Graduate and Postdoctoral Studies}
\institutionaddress{Vancouver}

\previousdegree{B.Sc., Technion - Israel Institute of Technology, 2011}
\previousdegree{M.Math., University of Waterloo, 2014}

\renewcommand\titlefont{\Huge \bfseries}
\renewcommand\subtitlefont{\LARGE \bfseries}
\title{Reductions in finite-dimensional quantum mechanics:} 
\subtitle{from symmetries to operator algebras and beyond}
\author{Oleg Kabernik}
\copyrightyear{2021}
\submitdate{\monthname\ \number\year} 
\program{Physics}

\renewcommand{\thepart}{\Roman{part}}
\renewcommand{\thechapter}{\arabic{chapter}}
\renewcommand{\thesection}{\thechapter.\arabic{section}}
\renewcommand{\thesubsection}{\thesection.\arabic{subsection}}
\renewcommand{\thesubsubsection}{\thesubsection.\arabic{subsubsection}}
\renewcommand{\theparagraph}{\thesubsubsection.\arabic{paragraph}}

\floatstyle{ruled}
\newfloat{Program}{htbp}{lop}[chapter]

\makeatother

\usepackage{babel}
\addto\captionscanadian{\renewcommand{\algorithmname}{Algorithm}}
\addto\captionscanadian{\renewcommand{\corollaryname}{Corollary}}
\addto\captionscanadian{\renewcommand{\definitionname}{Definition}}
\addto\captionscanadian{\renewcommand{\lemmaname}{Lemma}}
\addto\captionscanadian{\renewcommand{\propositionname}{Proposition}}
\addto\captionscanadian{\renewcommand{\theoremname}{Theorem}}
\addto\captionsenglish{\renewcommand{\corollaryname}{Corollary}}
\addto\captionsenglish{\renewcommand{\definitionname}{Definition}}
\addto\captionsenglish{\renewcommand{\lemmaname}{Lemma}}
\addto\captionsenglish{\renewcommand{\propositionname}{Proposition}}
\addto\captionsenglish{\renewcommand{\theoremname}{Theorem}}
\providecommand{\corollaryname}{Corollary}
\providecommand{\definitionname}{Definition}
\providecommand{\lemmaname}{Lemma}
\providecommand{\propositionname}{Proposition}
\providecommand{\theoremname}{Theorem}

\begin{document}
\global\long\def\ket#1{\left|#1\right\rangle }%

\global\long\def\bra#1{\left\langle #1\right|}%

\global\long\def\braket#1#2{\left.\left\langle #1\right.\,\right|\left.#2\right\rangle }%

\global\long\def\ketbra#1#2{\left|#1\right\rangle \left\langle #2\right|}%

\global\long\def\tr{\mathbf{tr}}%

\global\long\def\spn{\mathbf{span}}%

\global\long\def\dim{\mathbf{dim}}%

\global\long\def\rnk{\mathbf{rank}}%

\global\long\def\min{\mathbf{min}}%

\maketitle

\frontmatter~

\noindent The following individuals certify that they have read, and
recommend to the Faculty of Graduate and Postdoctoral Studies for
acceptance, the dissertation entitled:~\\
\\
\uline{Reductions in finite-dimensional quantum mechanics: from
symmetries to operator algebras and beyond \hspace*{\fill}}

\noindent ~\\
submitted by\uline{ Oleg Kabernik }in partial fulfillment of the
requirements for \\
\\
the degree of \uline{Doctor of Philosophy\hspace*{\fill}}\\
\\
in \uline{Physics\hspace*{\fill}}\\
\\
\\
\textbf{Examining Committee:}\\
\textbf{}\\
\uline{Robert Raussendorf, Associate Professor, Department of Physics
and Astronomy, UBC\hspace*{\fill}~}\\
Supervisor\\
\textbf{}\\
\uline{Ian Affleck, Professor, Department of Physics and Astronomy,
UBC~\hspace*{\fill}~}\\
Supervisory Committee Member \\
\textbf{}\\
\uline{Gordon W. Semenoff, Professor, Department of Physics and
Astronomy, UBC\hspace*{\fill}~}\\
University Examiner\\
\textbf{}\\
\uline{Sven Bachmann, Associate Professor, Department of Mathematics,
UBC ~\hspace*{\fill}~}\\
University Examiner\\
\\
\\
\textbf{Additional Supervisory Committee Members:}\\
\\
\uline{Mark Van Raamsdonk, Professor, Department of Physics and
Astronomy, UBC\hspace*{\fill}~}\\
Supervisory Committee Member\\
\\
\uline{Joshua Folk, Associate Professor, Department of Physics
and Astronomy, UBC\hspace*{\fill}~}\\
Supervisory Committee Member
\begin{abstract}
The idea that symmetries simplify or reduce the complexity of a system
has been remarkably fruitful in physics, and especially in quantum
mechanics. On a mathematical level, symmetry groups single out a certain
structure in the Hilbert space that leads to a reduction. This structure
is given by the irreducible representations of the group, and in general
it can be identified with an operator algebra (a.k.a. $C^{*}$-algebra
or von Neumann algebra). The primary focus of this thesis is the extension
of the framework of reductions from symmetries to operator algebras,
and its applications in finite-dimensional quantum mechanics.

Finding the irreducible representations structure is the principal
problem when working with operator algebras. We will therefore review
the representation theory of finite-dimensional operator algebras
and elucidate this problem with the help of two novel concepts: minimal
isometries and bipartition tables. One of the main technical results
that we present is the Scattering Algorithm for analytical derivations
of the irreducible representations structure of operator algebras.

For applications, we will introduce a symmetry-agnostic approach to
the reduction of dynamics where we circumvent the non-trivial task
of identifying symmetries, and directly reduce the dynamics generated
by a Hamiltonian. We will also consider quantum state reductions that
arise from operational constraints, such as the partial trace or the
twirl map, and study how operational constraints lead to decoherence.
Apart from our primary focus we will extend the idea of reduction
beyond operator algebras to operator systems, and formulate a quantum
notion of coarse-graining that so far only existed in classical probability
theory. In addition, we will characterize how the uncertainty principle
transitions to the classical regime under coarse-grained measurements
and discuss the implications in a finite-dimensional setting.
\end{abstract}

\chapter{Lay Summary}

In modern physics, and especially in quantum mechanics, symmetry has
been recognized as a powerful concept that explains many aspects of
the physical world around us. On a mathematical level, symmetries
identify a certain structure that reduces the complexity of a physical
system. It turns out that such reductions are not primarily identified
by symmetries, but by a rather more general mathematical concept of
an operator algebra. The primary focus of this thesis is the extension
of the framework of reductions from symmetries to operator algebras,
and its applications in quantum mechanics. We will present an algorithm
for deriving the complexity reducing structures directly from operator
algebras and demonstrate its applications with problems from quantum
information and quantum computing.

\chapter{Preface}

All the work presented in this thesis was conducted by the author
as a member of the Quantum Information group lead by Robert Raussendorf
at the University of British Columbia, Point Grey campus.\\
\\
Some of the ideas presented in Chapters \ref{chap:Operator-algebras-and},
\ref{chap:Operational-reductions-of-dyn} and \ref{chap:Quantum-coarse-graining}
have been published {[}Kabernik O., \textquotedbl Quantum coarse
graining, symmetries, and reducibility of dynamics\textquotedbl ,
Phys. Rev. A 97 (2018){]}. These include the concepts of bipartition
tables, quantum coarse-graining and the reduction of Hamiltonians
with symmetries. I am the sole author of that work and I was responsible
for all aspects of its development.\\
\\
The Scattering Algorithm presented in Chapter \ref{chap:Identifying-the-irreps}
and some of its applications described in Chapter \ref{chap:Operational-reductions-of-st}
have been published {[}Kabernik O., Pollack J., and Singh A., \textquotedbl Quantum
state reduction: Generalized bipartitions from algebras of observables\textquotedbl ,
Phys. Rev. A 101 (2020){]}. I was the lead investigator responsible
for concept formation, analysis and manuscript composition. Pollack
J. and Singh A. were involved in the initial formulation of these
ideas and have contributed to the manuscript composition.\\
\\
Chapter \ref{chap:The-uncertainty-principle} is a modified version
of the preprint {[}Kabernik O., \textquotedbl Quantifying The Uncertainty
Principle and The Effects of Minimal Length From a Finite-Dimensional
Perspective\textquotedbl , arXiv:2002.01564{]}. I am the sole author
of that work and I was responsible for all aspects of its development.\\
\\
All other ideas presented in this thesis, in particular the contents
of Chapter \ref{chap:Operational-reductions-of-dyn}, were developed
solely by the author and were not previously published.

\tableofcontents{}

 \listoffigures

\chapter{Acknowledgments}

I acknowledge the material support of the government of Canada through
the Natural Sciences and Engineering Research Council (NSERC). \\
\\
I would like to express gratitude to all the people in my personal
and professional lives whose support and encouragement enabled this
journey. I am grateful to my research supervisor Robert Raussendorf
for the stress-free environment, patient guidance, and unwavering
confidence in my work. I thank my colleagues who had endured my ramblings
about coarse-grainings and operator algebras, and in return offered
feedback, encouragement and advice. In particular, I would like to
thank Dongsheng Wang, Fumika Suzuki, Jason Pollack, Ashmeet Singh,
Pedro Lopes, Michael Zurel and the members of my supervisory committee.
I want to thank my parents, Felix and Svetlana Kabernik, for being
so thoughtful and supportive. Most of all, I am grateful to the one
person who chooses to be by my side through thick and thin; Rita,
without you this journey would never begin. Finally, I must acknowledge
the calming influence of a particularly talented cat who shall remain
unnamed.\\

\begin{center}
\includegraphics[width=1\columnwidth]{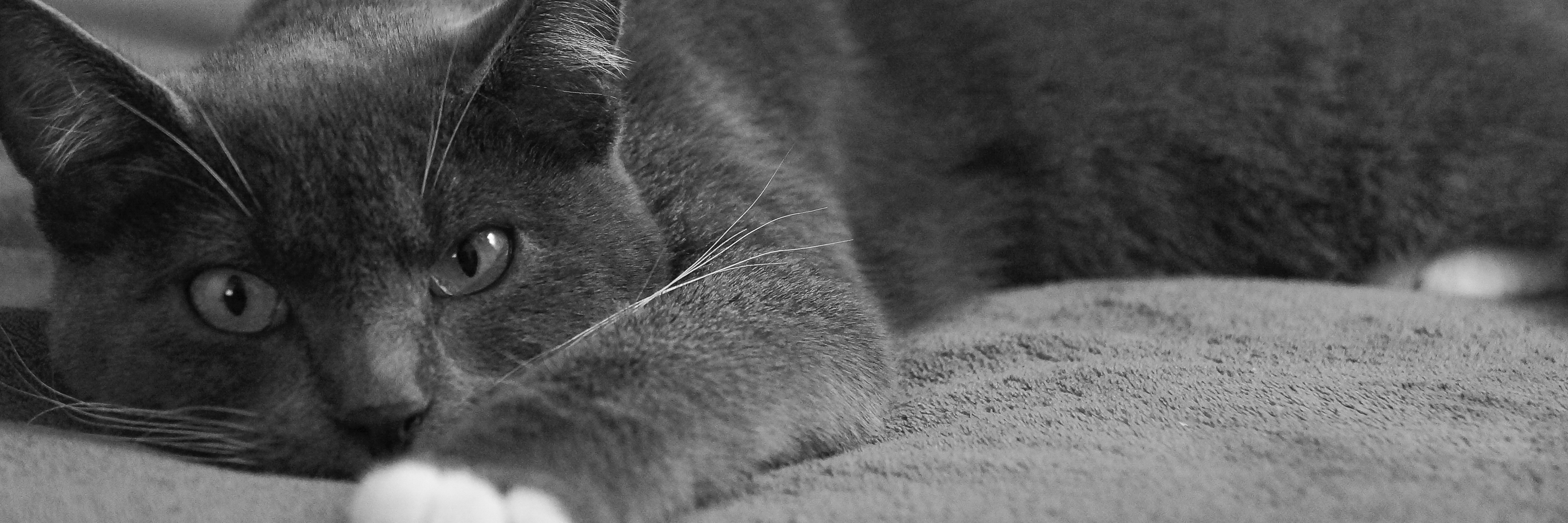}
\par\end{center}

\begin{flushleft}
\par\end{flushleft}

\mainmatter

\chapter{Introduction}

The main task of a theoretical physicist is to translate physical
reality into mathematical models, analyze the models, and then translate
the results back into predictions and explanations. The translation
from physics\emph{ }to math, however, is far from a rigorous process
and the only guiding principle, besides consistency with phenomenology,
is that ``Everything should be made as simple as possible, but no
simpler''.\footnote{This quote is commonly attributed to Albert Einstein.}

Any simplification in the description of a physical system can be
called a \emph{reduction.} In this thesis we will study the methods
and structures associated with reductions in finite-dimensional quantum
systems. Before we get into the details of what that means, we will
approach the idea of reduction from a broader perspective.

The essence of reduction is to single out some significant information
about the system and disregard everything else. For example, in classical
mechanics we reduce the state of a rigid body that consists of some
$10^{23}$ individual particles, to the position, momentum and angular
momentum of their center of mass. Similarly, in statistical mechanics
we reduce the intractable state of a many-body system to a handful
of variables such as the number of particles, energy, temperature,
chemical potential, volume, and pressure. The very existence of such
classical descriptions of physical systems is predicated on the idea
of reduction.

A very general perspective on reduction is to think of it as the result
of some coarse-graining where we choose not to distinguish between
every possible state of the system. One of the earliest formulations
of this perspective appears in the work of Paul and Tatiana Ehrenfest
\citep{ehrenfest1990conceptual}, where they elucidate the ideas of
Boltzmann, Gibbs and Einstein on statistical mechanics. Although coarse-graining
is a very simple and powerful perspective, it does not translate naturally
into quantum theory where the states are not fully distinguishable
to begin with.

Many modern analytical methods in physics can be viewed as reductions.
These include the mean-field approximation \citep{kadanoff2009more,kadanoff2000statistical},
separation of scales and renormalization group methods \citep{Wilson75renormalization,kadanoff2000statistical,mccomb2004renormalization},
matrix product states and tensor networks \citep{Verstraete2006Matrix,perez2007matrix,Vidal2007Entanglement}.
Furthermore, the idea of \emph{model reduction} in dynamical systems
is studied on a more general level as a subject of applied mathematics
\citep{schilders2008model,gorban2006model}. The landscape of all
the approaches to reduction is beyond the scope of this thesis, but
the notion of reduction that we will explore here is still based on
the same general principle: single out some significant information
about the system and disregard everything else.

Which information is significant and which is not, is the principal
problem of reduction. A remarkably fruitful approach to this problem
traces back to the seminal work by Emmy Noether, known as Noether's
theorem \citep{Noether1918}. This theorem states that when the dynamics
of a system have symmetries, there are conserved quantities that do
not evolve with time. When considering time evolutions we can therefore
focus on the non-conserved quantities, and disregard the conserved
ones; this results in a reduction. A particle in central potential
in three-dimensional space is thus reduced from having three dynamical
variables of position, to one.

Today, symmetry methods are a well established staple in physics with
far reaching implications and many dedicated textbooks (see for example
\citep{cornwell1997group,tung1985group,Georgi99}). The integration
of symmetry methods in quantum theory traces back to a result by Eugene
Wigner, known as Wigner's theorem \citep{wigner1959group,weinberg1995quantum}.
It states that in quantum mechanics symmetries are represented by
a group of unitary or anti-unitary operators acting on a Hilbert space.
Since anti-unitary representations are rare, especially in finite-dimensional
quantum mechanics, the mathematical formalism of symmetry methods
that we will focus on is that of unitary group representations.

With the advancement of finite-dimensional quantum theory driven by
the development of quantum information and quantum computing, the
mathematical structure given by the irreducible representations, or
\emph{irreps, }of groups has been brought into sharper focus. The
irreps structure has been recognized as the principal structure in
applications such as quantum error correction and fault tolerance
\citep{Zanardi97Noiseless,Knill2000,kribs2005operator,blume2008characterizing},
quantum reference frames and superselection rules \citep{Bartlett07,Kitaev2004Superselection},
and commodification of asymmetry as a resource \citep{marvian2013theory,Marvian2014Modes,marvian2014extending}.
Along the way it was realized that the irreps structure is primarily
defined not by the representations of groups, but by a rather more
general set of transformations known as \emph{operator algebra}s.

The use of operator algebras in quantum theory was originally pioneered
by John von Neumann and Francis Murray \citep{neumann1930algebra,murray1936rings},
focusing on infinite dimensional Hilbert spaces. The use of operator
algebras in finite-dimensional quantum mechanics was realized much
later; see the reviews in \citep{beny2015algebraic} or \citep{harlow2017ryu},
for example. These algebras are often distinguished as $C^{*}$-algebras
or von Neumann algebras but in finite dimensions these distinctions
are inconsequential, so we will keep calling them operator algebras.
We will later formally define operator algebras; for now let us just
say that these are sets of transformations, like groups, but they
are not represented by unitaries and they can be composed into both
products and sums.

The primary focus of this thesis is the generalization of reduction
methods from symmetries to operator algebras in finite-dimensional
quantum mechanics. One motivation for this generalization follows
from the fact that there is no systematic way to identify symmetries,
and we are mostly restricted to intuitively recognizable symmetries
such as rotations in space. By shifting the focus from symmetries
to operator algebras we will develop a symmetry-agnostic approach
to the reduction of dynamics.

Let us first briefly review what a reduction due to symmetries looks
like in finite-dimensional quantum mechanics. We consider a group
$\mathcal{G}$ represented by the unitaries $U\left(\mathcal{G}\right)$
acting on the Hilbert space $\mathcal{H}$ of our system. We call
the group $\mathcal{G}$ a symmetry if it commutes with the Hamiltonian
\[
\left[H,U\left(g\right)\right]=0\hspace{1cm}\forall g\in\mathcal{G}.
\]
Without going into the details of group representation theory, the
standard textbook procedure \citep{cornwell1997group,tung1985group,Georgi99}
for the reduction of dynamics due to symmetries can be summarized
in the following steps:
\begin{enumerate}
\item Identify the smallest subspaces $\mathcal{H}_{q,i}\subset\mathcal{H}$
that are closed under the symmetry transformations $U\left(\mathcal{G}\right)$.
These subspaces are called \emph{irreducible} because the action of
symmetry transformations cannot be further restricted to smaller subspaces.
Both indices $q$ and $i$ enumerate the irreducible subspaces $\mathcal{H}_{q,i}$,
but with $q$ we distinguish the subspaces with distinct representations
of $\mathcal{G}$, and with $i$ we distinguish the subspaces with
identical representations.
\item Change to the new basis $\ket{e_{ik}^{q}}$ in $\mathcal{H}$, where
$q$ and $i$ identify the irreducible subspace $\mathcal{H}_{q,i}$,
and $k$ enumerates the basis elements inside the subspace. Such change
of basis is explicitly specified by what is known as the Clebsch-Gordan
coefficients.
\item Assuming that $U\left(\mathcal{G}\right)$ commutes with $H$, use
the results of Schur's lemmas to conclude that 
\begin{equation}
\bra{e_{jl}^{p}}H\,\Bigl|e_{ik}^{q}\Bigr\rangle=\delta_{pq}\delta_{lk}\bra{e_{j}^{q}}H\,\Bigl|e_{i}^{q}\Bigr\rangle,\label{eq:H reduction from schur lemma}
\end{equation}
where we suppress the index $k$ on the right because the matrix elements
do not depend on it.
\end{enumerate}
Thus, the Hamiltonian reduces to the block-diagonal form where $q$
and $k$ enumerate the blocks, and $i$,$j$ refer to the matrix elements
inside each block. The key structure here is given by the basis $\ket{e_{ik}^{q}}$,
that in the language of group representation theory identify the irreducible
representations of $U\left(\mathcal{G}\right)$.

The fact that $U\left(\mathcal{G}\right)$ commutes with $H$ is the
defining property that makes $U\left(\mathcal{G}\right)$ not just
a group of transformations but a ``symmetry''. The above procedure
suggests that we can only rely on the irreps structure of symmetries
in order to reduce the dynamics, otherwise Schur's lemmas cannot be
invoked. One of the results that we will later show is that this is
not quite the case. That is, even groups that are not symmetries can
lead to a reduction of dynamics under a relaxed condition on the commutators
$\left[H,U\left(g\right)\right]$.

As an alternative to groups, we will introduce the symmetry-agnostic
approach where instead of asking ``Which group commutes with the
Hamiltonian?'' we ask ``Which operator algebra contains the Hamiltonian?''.
Although this may seem like two different questions, the later formulation
is a generalization of the former. The idea behind the symmetry-agnostic
approach is to identify an operator algebra $\mathcal{A}$ that contains
$H$, find the irreps structure of $\mathcal{A}$ given by the basis
$\ket{e_{ik}^{q}}$, and then reduce $H$ using the irrep basis as
we did in Eq. \eqref{eq:H reduction from schur lemma}.

For example, consider the Hilbert space of three qubits $\mathcal{H}_{qubit}^{\otimes3}$
and the Hamiltonian
\[
H\left(\epsilon\right)=H_{int}+\epsilon\sigma_{z}\otimes I\otimes I,
\]
where $\sigma_{z}$ is a Pauli matrix acting on the first qubit and
$\epsilon$ is a real parameter. The interaction term $H_{int}$ is
such that its first and only excited states are $\ket{++0}$ and $\ket{+11}$
(we use the notation $\ket{\pm}\propto\ket 0\pm\ket 1$), so if we
normalize its energy gap to $1$ it is just 
\[
H_{int}=\ketbra{++0}{++0}+\ketbra{+11}{+11}.
\]
The terms $H_{int}$ and $\sigma_{z}\otimes I\otimes I$ do not commute
so we cannot simultaneously diagonalize both terms to find the spectrum
of $H\left(\epsilon\right)$ as a function of $\epsilon$. We can,
however, reduce $H\left(\epsilon\right)$ and find its spectrum as
a function of $\epsilon$ from the reduced Hamiltonian blocks.

In principle, there is a symmetry group that reduces $H\left(\epsilon\right)$
but it can be difficult to identify, and even then, one has to find
the irrep basis that lead to the reduction. The symmetry-agnostic
approach offers an alternative where we directly derive the irrep
basis that lead to the reduction. The idea is to observe that $H\left(\epsilon\right)$
is a linear combination of two terms, so it is an element of the algebra
\[
\mathcal{A}=\left\langle H_{int},\,\sigma_{z}\otimes I\otimes I\right\rangle 
\]
generated by these terms. If we find the irrep basis of the operator
algebra $\mathcal{A}$, we can reduce $H\left(\epsilon\right)$.

It turns out, as we will see later, that the irreps structure of $\mathcal{A}$
is given by the basis 
\[
\ket{0-0},\ket{001},\ket{1-0},\ket{101},\ket{0+0},\ket{1+0},\ket{011},\ket{111}
\]
and when we present the Hamiltonian in these basis (in the above order)
we get
\[
H\left(\epsilon\right)=\begin{pmatrix}\epsilon\\
 & \epsilon\\
 &  & -\epsilon\\
 &  &  & -\epsilon\\
 &  &  &  & \frac{1}{2}+\epsilon & \frac{1}{2}\\
 &  &  &  & \frac{1}{2} & \frac{1}{2}-\epsilon\\
 &  &  &  &  &  & \frac{1}{2}+\epsilon & \frac{1}{2}\\
 &  &  &  &  &  & \frac{1}{2} & \frac{1}{2}-\epsilon
\end{pmatrix}.
\]
Thus, the spectrum of $H\left(\epsilon\right)$ consists of $\pm\epsilon$,
and the rest is given by the eigenvalues of the $2\times2$ block
\[
\begin{pmatrix}\frac{1}{2}+\epsilon & \frac{1}{2}\\
\frac{1}{2} & \frac{1}{2}-\epsilon
\end{pmatrix}.
\]

The advantage of this approach is that identifying the operator algebra
that contains the Hamiltonian is trivial compared to identifying symmetries.
The real challenge is in finding the irreps structure and the associated
basis. 

We are aware of two approaches in the literature to the problem of
finding the irreps structure. First, Murota \emph{et al}. \citep{Murota2010}
has proposed a numerical algorithm based on random sampling and motivated
by problems in semidefinite programming (it was adapted in \citep{wang2013numerical}
for physical applications). Second, Holbrook \emph{et al}. \citep{holbrook2003noiseless}
have proposed an algorithm without sampling, but it requires the ability
to find spans of sets of operators.

Ideally, just as we have a symbolic (not inherently numeric) algorithm
for diagonalizing a matrix using pen and paper, we would have a symbolic
algorithm for finding the irreps structure of a set of matrices. One
of the main technical contributions of this thesis is the derivation
of such algorithm.

The proposed algorithm is called the \emph{Scattering Algorithm}.
It is constructed around the basic operation called\emph{ }``scattering''
that acts on pairs of projections and is symbolically represented
as follows (the input is on the left, and the output is on the right):
\[
\begin{array}{c}
\Pi_{1}\\
\\
\Pi_{2}
\end{array}\Diagram{fdA &  & fuA\\
 & f\\
fuA &  & fdA
}
\begin{array}{c}
\Pi_{1}^{\left(\lambda_{1}\right)},\,\Pi_{1}^{\left(\lambda_{2}\right)},...\,,\Pi_{1}^{\left(0\right)}\\
\\
\Pi_{2}^{\left(\lambda_{1}\right)},\,\Pi_{2}^{\left(\lambda_{2}\right)},...\,,\Pi_{2}^{\left(0\right)}.
\end{array}
\]
The output projections $\Pi_{i}^{\left(\lambda\right)}$ are defined
as the elements of the spectral decompositions 
\[
\Pi_{1}\Pi_{2}\Pi_{1}=\sum_{\lambda\neq0}\lambda\Pi_{1}^{\left(\lambda\right)}\textrm{\hspace{1cm}and\hspace{1cm}}\Pi_{2}\Pi_{1}\Pi_{2}=\sum_{\lambda\neq0}\lambda\Pi_{2}^{\left(\lambda\right)}.
\]
From the sums of output projections we can recover the input projections
as $\Pi_{i}=\sum_{\lambda}\Pi_{i}^{\left(\lambda\right)}$ (the element
$\Pi_{i}^{\left(0\right)}$ ensures that). In this sense the scattering
operation ``breaks'' the input projections into lower rank constituents.
The main idea of the Scattering Algorithm is to start with the spectral
projections of the generators of the operator algebra, break them
into the minimal possible constituents, and then construct the irreps
structure from these minimal projections.

One application for the Scattering Algorithm that we will demonstrate
is in finding the possible qubit encodings for a given control Hamiltonian.
Since different encodings have different physical characteristics,
it is desirable to exhaust the possibilities of an encoding for a
given Hamiltonian. Finding such encodings is not a trivial task and
the symmetry of the Hamiltonian is often used to point to the possible
solution. With the Scattering Algorithm we can approach this task
in a more systematic manner and find solutions associated with less
obvious symmetries.

In quantum dot arrays, for example, it is possible to implement the
nearest neighbor Heisenberg interaction with tunable terms \citep{Loss98Quantum}
\[
H=\epsilon_{12}\vec{S}_{1}\cdot\vec{S}_{2}+\epsilon_{23}\vec{S}_{2}\cdot\vec{S}_{3}+\epsilon_{34}\vec{S}_{3}\cdot\vec{S}_{4}.
\]
DiVincenzo \emph{et al. }\citep{divincenzo2000universal} and Bacon
\emph{et al. }\citep{Bacon2000Universal} have proposed qubit encoding
in such arrays based on the reduction due to the $SU\left(2\right)$
symmetry of $H$. By adopting the symmetry-agnostic approach and using
the Scattering Algorithm we will find additional qubit encodings that
cannot be revealed by the $SU\left(2\right)$ symmetry alone.

Thus, a physical application of the Scattering Algorithm is in characterizing
the dynamics of Hamiltonians beyond its obvious symmetries. In particular,
for the purposes of quantum information processing, the Scattering
Algorithm identifies the possible qubit encodings in a systematic
manner without relying on the intuition of symmetries.

So far, we have focused on reductions that follow from the dynamics
of the system. There is, however, another kind of reductions that
arises when we have inaccessible degrees of freedom, such as the degrees
of freedom of the ``environment''. In order to distinguish such
reductions from the reductions of dynamics, we will refer to them
as the \emph{reductions of states}.

The prototypical state reduction is the partial trace map that reduces
the state $\rho_{AB}$ of a composite system $AB$ into the state
$\rho_{B}$ of subsystem $B$ alone. The operational meaning of the
reduced state $\rho_{B}$ is that it contains only the information
accessible with measurements on subsystem $B$. In other words, the
partial trace map is a map that accounts for the operational constraint
that only allows measurements on subsystem $B$.

There are more sophisticated operational constraints that cannot be
associated with a physical subsystem. A well known example of that
is the operational constraint that arises from a lack of common reference
frame \citep{Bartlett07}. That is, when two parties (Alice and Bob)
do not share a common reference frame, any information about the quantum
state that relies on this frame of reference is inaccessible to the
other party. The resulting operational constraint is a restriction
to observables that are symmetric under transformations of this reference
frame. The state reduction map that accounts for this constraint is
called the twirl \citep{Bartlett07} and it is rooted in the irreps
structure of the reference frame transformations group.

We will introduce and study the idea of state reductions due to operational
constraints in the common mathematical framework of operator algebras.
Within this framework state reductions are constructed directly from
operational constraints, and it subsumes the specialized state reduction
maps such as the partial trace or the twirl. Ideas such as noiseless
subsystems \citep{lidar2003decoherence} can also be incorporated
into this framework by observing that the requirement for logical
operations to commute with the operations of noise is an operational
constraint. In that case, the state reduction map is a map that decodes
the logical information from the physical state. 

It is also interesting to consider the dynamics of reduced states.
It is well known that the reduced state of a subsystem $B$ may undergo
decoherence if the composite system $AB$ evolves in a certain way.
What ``certain way'' means is that the Hamiltonian $H$ of $AB$
has an interaction term $H_{int}$ that couples the two subsystems
\citep{breuer2002theory}. That is, 
\[
H=H_{self}+H_{int},
\]
where $H_{self}=I_{A}\otimes H_{B}+H_{A}\otimes I_{B}$ and $H_{int}$
cannot be expressed in this way. Identifying the interaction term
that is responsible for decoherence is simple when we talk about subsystems.
However, when considering state reduction maps associated with different
operational constraints, the distinction between the ``self'' and
``interaction'' terms of the Hamiltonian is not as clear.

For example, we will consider the composite system $\underline{l}\otimes\underline{\frac{1}{2}}\otimes\underline{\frac{1}{2}}$
of two spin-$\frac{1}{2}$'s and an integer angular momentum $l$,
such as the Hydrogen atom, but with the simple Hamiltonian of uniform
magnetic field along the $\hat{y}$ axis
\[
H=\epsilon L_{y}+S_{1;y}+S_{2;y}.
\]
This Hamiltonian has no interaction terms so we do not expect it to
induce decoherence. This is true if we consider the reduced states
of the individual spins or angular momentum, but it is not the case
if different operational constraints are imposed. In particular, for
the reduced states that arise due to the lack of common reference
frame of directions in space, this Hamiltonian induces decoherence.
We will see that under this operational constraint the ``self''
and the ``interaction'' terms of the Hamiltonian are
\[
H=\epsilon\underset{H_{self}}{\underbrace{\left(L_{y}+S_{1;y}+S_{2;y}\right)}}+\left(1-\epsilon\right)\underset{H_{int}}{\underbrace{\left(S_{1;y}+S_{2;y}\right)}}.
\]

Taking a step back, let us return to the simple classical notion of
reduction associated with coarse-graining. Such notion of reduction
is implicit in statistical mechanics where we choose to distinguish
only between states that have different macroscopic properties. The
partitioning of the micro state space into macroscopic classes of
states, or macro states, is exactly what we mean by coarse-graining.
The idea of coarse-graining naturally extends into probability theory
where the probability for a macro state to occur is given by the probability
for any micro state in the class to occur. This kind of reasoning,
however, does not seem to extend naturally into quantum theory.

In the following, we will attempt to bridge this conceptual gap. In
order to do that we will have to extend the mathematical framework
beyond operator algebras and into \emph{operator systems}. In the
process we will end up generalizing the usual notion of a subsystem
(or a \emph{virtual} subsystem \citep{Zanardi2001Virtual,Zanardi:2004zz}),
to what we call a \emph{partial }subsystem that is no longer defined
by a tensor product bipartition of the Hilbert space. The main result
is the definition of a state reduction map called \emph{quantum coarse-graining,
}and the derivation of its operational meaning. These ideas will be
illustrated with a simple classical coarse-graining of a probability
distribution, and its quantum analogue. As a motivating example we
will consider the encoding (compression) of a three level system (qutrit)
into a two level system (qubit) using quantum coarse-graining.

Finally, we will carry out a case study of the uncertainty principle
on a lattice. Unlike previous topics, where the emphasis was on the
methods, here we will focus on specific physical questions. Because
a much simpler notion of coarse-graining will be used here, these
analysis will be presented in a self contained manner without relying
on the previously discussed mathematical framework.

It is well known that due to the uncertainty principle, the Planck
constant sets a resolution boundary in phase space (see the original
paper by Heisenberg \citep{heisenberg1927} or \citep{busch2006complementarity,busch2007heisenberg}
for a modern review). It is also known that in the classical regime
the outcomes of sufficiently coarse measurements of position and momentum
can simultaneously be determined. If we then continuously vary the
resolution of measurements, the uncertainty principle should transition
between the quantum and classical regimes, but the picture of how
this transition unfolds is not so clear.

In the following we will clarify this picture by studying a characteristic
function that quantifies the mutual disturbance effects responsible
for the uncertainty principle. Since it is also expected that the
uncertainty principle is modified by the existence of minimal length
in space (this is known as the generalized uncertainty principle \citep{ali2009discreteness}),
we will conduct our investigation on a lattice.

We will see how the discontinuity of the lattice perturbes the uncertainty
principle and its transition to the classical regime. We will also
see that in terms of lattice units, the uncertainty principle imposes
a resolution boundary given by the square root of the length of the
lattice, and the Planck constant is derived from it. We will discuss
the implications of these results for the existence of minimal length
in space.

As a guide to the reader, we summarize the contributions of this thesis
by chapters as follows:
\begin{description}
\item [{Chapter~2}] We introduce finite-dimensional operator algebras
and the structure of irreducible representations in a pedagogical,
self-contained manner. Here we will mostly derive previously known
results but with the help of two novel concepts: minimal isometries
and bipartition tables.
\item [{Chapter~3}] We introduce the Scattering Algorithm for finding
irreducible representations of operator algebras.
\item [{Chapter~4}] We introduce the framework of state reductions due
to operational constraints and integrate it with the study of decoherence.
\item [{Chapter~5}] We introduce the framework of symmetry-agnostic reduction
of dynamics and relax the condition for reduction with symmetries.
\item [{Chapter~6}] We introduce the notions of partial subsystems and
quantum coarse-graining along with their operational meaning.
\item [{Chapter~7}] We study a characteristic function that quantifies
how the uncertainty principle transitions between the quantum and
classical regimes on a lattice.
\end{description}

\chapter{Operator algebras and the structure of irreducible representations
\label{chap:Operator-algebras-and}}

The irreducible representations (irreps) structure is at the core
of all forms of reductions that rely on symmetries or operator algebras.
In fact, identifying the irreps structure is usually one of the main
technical challenges in the analysis that involve symmetries or operator
algebras. Before we can begin to address this challenge we need to
understand what operator algebras are, and what is the irreps structure.
This is the subject of this chapter.

The abstract mathematical notion of an \emph{algebra }and the more
concrete notion of an \emph{operator algebra} are well established
fields of study in the mathematical literature. The study of operator
algebras have been introduced and developed in the context of mathematical
physics by John von Neumann and Francis Murray \citep{neumann1930algebra,murray1936rings}
which is known today as the study of von Neumann algebras. In the
modern physics literature another name that is commonly used is a
$C^{*}$-algebra which is a slight generalization of the von Neumann
algebra\emph{. }Many\emph{ }of the subtleties in the study of operator
algebras (and thus the proliferation of different names) arise from
the issues associated with the infinite dimensionality of Hilbert
spaces of continuous functions. Since we are only concerned with finite-dimensional
Hilbert spaces, we can avoid the full mathematical treatment of this
subject and restrict our attention to finite-dimensional operator
algebras.

In the following we will introduce the main ideas behind finite-dimensional
operator algebras in a pedagogical manner focusing on the structural
aspects that are suitable for our purposes. Similar accounts of finite-dimensional
operator algebras in the physics literature can be found in \citep{beny2015algebraic}
or in the appendix of \citep{harlow2017ryu}. The abstract mathematical
treatment of this subject appears in many textbooks (mostly focusing
on the subtleties of infinite dimensional spaces), see for example
\citep{farenick2012algebras} or the notes \citep{Jones2015kbb}.

The central result of the representation theory of finite-dimensional
operator algebras is known as the Wedderburn Decomposition which we
will derive in Theorem \ref{thm: Wedderburn Decomp}. Even though
this result is far from novel, the path that we will take there, including
most of the proofs, will not follow any of the standard references.
In particular, we will introduce the notions of minimal isometries
and bipartition tables that anticipate the ideas behind the Scattering
Algorithm presented in Chapter \ref{chap:Identifying-the-irreps}.

In Section \ref{sec:Notation-and-some} we will begin by setting up
the notation and stating some basic mathematical facts. In Section
\ref{sec:Finite-dimensional-operator-alge} we will introduce the
finite-dimensional operator algebras and identify the key structural
elements. The general irreps structure will be identified in Section
\ref{sec:Irrep-structures-and} along with some implications and examples.
The irreps structure of groups will be treated as a special case.

\section{Notation and some mathematical facts\label{sec:Notation-and-some}}

Unless stated otherwise, we will assume $\hbar\equiv1$.

All Hilbert spaces are assumed to be complex and finite-dimensional.
We will denote with $\mathcal{H}_{1}\cong\mathcal{H}_{2}$ the Hilbert
spaces that are isometric to each other, which means that $\mathcal{H}_{1}$
and $\mathcal{H}_{2}$ are only different in how we label their basis.
All Hilbert spaces are therefore $\mathcal{H}\cong\mathbb{C}^{d}$
for some integer $d$.

We will denote with $\mathcal{L}\left(\mathcal{H}\right)$ the space
of linear operators on the Hilbert space $\mathcal{H}$. In order
to avoid unnecessary notation we will not distinguish between the
notions of linear operators\emph{ }(or simply operators) and their
representations as matrices. Assuming the dimension of $\mathcal{H}$
is $d$, the space $\mathcal{L}\left(\mathcal{H}\right)$ is also
a finite-dimensional Hilbert space of dimension $d^{2}$.

The symbol $\dagger$ will denote the conjugate-transpose for matrices
and Hermitian adjoint for operators. We will use $0$ to denote both
the scalar $0\in\mathbb{C}$ and the null operator $0\in\mathcal{L}\left(\mathcal{H}\right)$.
We will often invoke the fact that $A=0\in\mathcal{L}\left(\mathcal{H}\right)$
if and only if $AA^{\dagger}=0$.

We will denote sets of element such as $\left\{ e_{jk}^{i}\right\} $
with the convention that, unless explicitly stated otherwise, the
set consists of all the elements obtained by varying the free indices
($i,j,k$ in this case). When we explicitly state the free indices,
such as $\left\{ e_{jk}^{i}\right\} _{i=1,...}$ , it will mean that
only those indices are indeed free and the rest are a fixed constant
for all elements in the set.

Projection operators are define as follows.
\begin{defn}
An operator $\Pi\in\mathcal{L}\left(\mathcal{H}\right)$ is a \emph{projection
}(a.k.a. orthogonal projection) if $\Pi\neq0$, $\Pi=\Pi^{\dagger}$
and $\Pi=\Pi^{2}$.
\end{defn}
The following facts about projections will be used implicitly throughout
this thesis.
\begin{prop}
If $\Pi,\Pi'\in\mathcal{L}\left(\mathcal{H}\right)$ are projections
such that $\Pi=c\Pi'$ for some $c\in\mathbb{C}$, then $c=1$.
\end{prop}
\begin{proof}
Since $\Pi$ and $\Pi'$ are projections we have
\[
c\Pi'=\Pi=\Pi^{2}=\left(c\Pi'\right){}^{2}=c^{2}\Pi'.
\]
Since $\Pi'\neq0$ we must have $c=1$.
\end{proof}
When referring to a set of projections $\left\{ \Pi_{k}\right\} $
as orthogonal\emph{ }we will always mean that in the sense of pairwise
orthogonal: $\Pi_{k}\Pi_{k'}=\delta_{kk'}\Pi_{k}$ for all $k,k'$.
The eigenspace\emph{ }of a projection $\Pi$ is the subspace of $\mathcal{H}$
on which $\Pi$ projects all element of $\mathcal{H}$ ($\Pi$ acts
as the identity on its own eigenspace). The\emph{ }rank\emph{ }of
a projection\emph{ }$\Pi$ is the number of its non-zero eigenvalues
which is also its trace and it is also the dimension of its eigenspace
\[
\rnk\left[\Pi\right]=\tr\left[\Pi\right]=\dim\left[\mathbf{eigenspace}\left[\Pi\right]\right].
\]

Another special type of operators that we will work with are partial
isometries.
\begin{defn}
An operator $S\in\mathcal{L}\left(\mathcal{H}\right)$ is a \emph{partial
isometry} if $SS^{\dagger}=\Pi_{fin}$ for some projection $\Pi_{fin}$.
\end{defn}
The following facts about partial isometries will often be used implicitly.
\begin{prop}
If $S$ is a partial isometry then $S^{\dagger}S=\Pi_{in}$ is another
projection with the same rank as $\Pi_{fin}$.
\end{prop}
\begin{proof}
Clearly $\Pi_{in}^{\dagger}=\left(S^{\dagger}S\right)^{\dagger}=\Pi_{in}$.
Note that $\Pi_{in}^{2}$ is a projection because $\left(\Pi_{in}^{2}\right)^{\dagger}=\Pi_{in}^{2}$
and
\[
\left(\Pi_{in}^{2}\right)^{2}=\left(S^{\dagger}S\right)^{4}=S^{\dagger}\Pi_{fin}^{3}S=S^{\dagger}\Pi_{fin}S=\Pi_{in}^{2}.
\]
Since $\Pi_{in}^{2}$ is a projection then $\sqrt{\Pi_{in}^{2}}=\Pi_{in}$
is a projection. The ranks of $\Pi_{fin}$ and $\Pi_{in}$ are the
same because
\[
\tr\left[\Pi_{fin}\right]=\tr\left[SS^{\dagger}\right]=\tr\left[S^{\dagger}S\right]=\tr\left[\Pi_{in}\right].
\]
\end{proof}
The eigenspace of $\Pi_{in}$, that is the initial space, is isometrically
mapped by $S$ to the eigenspace of $\Pi_{fin}$, that is the final
space. The partial isometry $S$ is only supported on the eigenspace
of $\Pi_{in}$ and all vectors that are orthogonal to it are annihilated.
Every projection $\Pi$ is also a partial isometry ($\Pi_{in}=\Pi_{fin}=\Pi$),
therefore we will say that $S$ is a proper\emph{ }partial isometry
if it is a partial isometry but it is not a projection.
\begin{prop}
If $S\in\mathcal{L}\left(\mathcal{H}\right)$ is a partial isometry
with the projections $\Pi_{in}$ and $\Pi_{fin}$ on its initial and
final spaces, then
\begin{eqnarray*}
 &  & S\Pi_{in}=\Pi_{fin}S=S\\
 &  & \Pi_{in}S^{\dagger}=S^{\dagger}\Pi_{fin}=S^{\dagger}
\end{eqnarray*}
\end{prop}
\begin{proof}
Note that
\begin{eqnarray}
 & \left(S\Pi_{in}-S\right)\left(S\Pi_{in}-S\right)^{\dagger} & =S\Pi_{in}\Pi_{in}S^{\dagger}-S\Pi_{in}S^{\dagger}-S\Pi_{in}S^{\dagger}+SS^{\dagger}\\
 &  & =-SS^{\dagger}SS^{\dagger}+SS^{\dagger}=0
\end{eqnarray}
Therefore, $S\Pi_{in}-S=0$ and so $S\Pi_{in}=S$. The rest follows
because $S\Pi_{in}=SS^{\dagger}S=\Pi_{fin}S$ and $\Pi_{in}S^{\dagger}=\left(S\Pi_{in}\right)^{\dagger}$.
\end{proof}
\begin{prop}
If $S,S'\in\mathcal{L}\left(\mathcal{H}\right)$ are partial isometries
such that $S'=cS$ for some $c\in\mathbb{C}$, then $c=e^{i\varphi}$
with some real phase $\varphi$.
\end{prop}
\begin{proof}
Since both $S'S'^{\dagger}$ and $SS^{\dagger}$ are projections such
that $S'S'^{\dagger}=cc^{*}SS^{\dagger}$, we must have $cc^{*}=\left|c\right|^{2}=1$.
\end{proof}

\section{Finite-dimensional operator algebras\label{sec:Finite-dimensional-operator-alge}}

Since we are only concerned with finite-dimensional physical applications,
our operators are always naturally represented as matrices with respect
to some basis in the Hilbert space. In order to avoid an unnecessary
level of abstraction and notation, we will not distinguish between
operators and their defining representations as complex matrices.
Therefore, the algebras that we will be dealing with are the algebras
of complex matrices with the regular matrix multiplication and summation
rules. From here on, by \emph{operator algebra} we will always mean
the finite-dimensional algebra of complex matrices representing physical
operators (which is a special case of $C^{*}$-algebra and Von Neumann
algebra).
\begin{defn}
\label{def: operator-algebra}An \emph{operator algebra} is a subset
of operators $\mathcal{A}\subseteq\mathcal{L}\left(\mathcal{H}\right)$
such that:

(1) For all $A_{1},A_{2}\in\mathcal{A}$ and $c_{1},c_{2}\in\mathbb{C}$
we have $c_{1}A_{1}+c_{2}A_{2}\in\mathcal{A}$.

(2) For all $A_{1},A_{2}\in\mathcal{A}$ we have $A_{1}A_{2}\in\mathcal{A}.$

(3) For all $A\in\mathcal{A}$ we have $A^{\dagger}\in\mathcal{A}$.

\noindent We will say that an algebra $\mathcal{A}$ is a \emph{subalgebra}
of $\tilde{\mathcal{A}}$ and denote it as $\mathcal{A}\subseteq\tilde{\mathcal{A}}$,
if both $\mathcal{A}$ and $\tilde{\mathcal{A}}$ are algebras and
$\mathcal{A}$ is a subset of $\tilde{\mathcal{A}}$.
\end{defn}
\noindent Condition (1) of the definition \ref{def: operator-algebra}
(with the regular matrix summation and scalar multiplication rules)
implies that $\mathcal{A}$ is a vector space. If the dimension of
$\mathcal{H}$ is $d$ then $\mathcal{L}\left(\mathcal{H}\right)$
is a $d^{2}$ dimensional vector space and $\mathcal{A}\subseteq\mathcal{L}\left(\mathcal{H}\right)$
is a $D\leq d^{2}$ dimensional vector space. There is always a finite
subset of elements $\left\{ A_{1},A_{2},...,A_{D}\right\} \subset\mathcal{A}$
that spans the whole $\mathcal{A}$
\[
\mathcal{A}=\spn\left\{ A_{1},A_{2},...,A_{D}\right\} =\left\{ \sum_{n=1}^{D}c_{n}A_{n}\,|\,c_{n}\in\mathbb{C}\right\} .
\]

Conditions (2) and (3) imply that $\mathcal{A}$ is also equipped
with the non-vector-space operations of matrix product and Hermitian
adjoint that leave the vector space $\mathcal{A}$ closed. It is worth
noting that elements of unitary groups comply with conditions (2)
and (3) but not with (1). In this sense, operator algebras generalize
unitary groups by allowing linear combinations of elements in addition
to products and adjoints.

Some important canonical examples of operator algebras are:
\begin{enumerate}
\item A trivial example is the set of all linear operators $\mathcal{L}\left(\mathcal{H}\right)$
which we will refer to as the \emph{full }or \emph{trivial} algebra
of operators acting on $\mathcal{H}$. All algebras that consist of
operators acting on $\mathcal{H}$ are subalgebras of $\mathcal{L}\left(\mathcal{H}\right)$.
\item If the Hilbert space is composed of two (or more) subsystems $\mathcal{H}=\mathcal{H}_{L}\otimes\mathcal{H}_{R}$
then all operators acting only on one subsystem form an algebra
\[
\mathcal{A}=\left\{ A\otimes I_{R}\,|\,A\in\mathcal{L}\left(\mathcal{H}_{L}\right)\right\} .
\]
\item If the Hilbert space is composed of two (or more) sectors $\mathcal{H}=\mathcal{H}_{0}\oplus\mathcal{H}_{1}$
then all operators acting only on one sector form an algebra
\[
\mathcal{A}=\left\{ A_{0}\oplus0_{1}\,|\,A_{0}\in\mathcal{L}\left(\mathcal{H}_{0}\right)\right\} .
\]
(The operator $0_{1}$ is the null operator on $\mathcal{H}_{1}$.)
\item All operators that are proportional to some projection $\Pi$ (in
particular $\Pi\equiv I$), form an one-dimensional algebra
\[
\mathcal{A}=\left\{ c\Pi\,|\,c\in\mathbb{C}\right\} .
\]
\item All operators in the span of projections $\left\{ \Pi_{k}\right\} _{k=0}^{m}$
that are all orthogonal to each other $\Pi_{k}\Pi_{k'}=\delta_{kk'}$,
form an $m$-dimensional algebra 
\[
\mathcal{A}=\spn\left\{ \Pi_{k}\right\} _{k=0}^{m}.
\]
\item All operators in the span of some group $\mathcal{G}$ represented
by the unitary operators $\left\{ U_{g}\right\} _{g\in\mathcal{G}}$
form an algebra called the \emph{group algebra} 
\[
\mathcal{A}=\spn\left\{ U_{g}\right\} _{g\in\mathcal{G}}.
\]
\item All operators that commute with all operators in some subset $\mathcal{B}\subset\mathcal{L}\left(\mathcal{H}\right)$,
form an algebra called the \emph{commutant }of $\mathcal{B}$
\[
\mathcal{A}=\left\{ A\,|\,\left[A,B\right]=0,\,\,\forall B\in\mathcal{B}\right\} .
\]
In particular, for every algebra $\mathcal{A}\subseteq\mathcal{L}\left(\mathcal{H}\right)$
we have the commutant algebra that will be denoted by $\mathcal{A}'$.
\end{enumerate}
Note that in examples 4 and 5 all elements of the algebra commute
with each other. Such algebras are called \emph{commutative }or \emph{abelian}.
Also note that in general, the identity operator $I\in\mathcal{L}\left(\mathcal{H}\right)$
does not have to be an element of the algebra. This is clearly the
case in examples 3 and 4 (with $\Pi\neq I$). Operator algebras that
include an element that \emph{acts} as the identity on all other elements
in the algebra are called \emph{unital}. Finite-dimensional operator
algebras are always unital and there is always a projection $\Pi\in\mathcal{A}$
(possibly $\Pi\equiv I$) that acts as the identity on all elements
of $\mathcal{A}$. We will not prove this fact because in our applications
we can always have the full identity operator $I$ included in $\mathcal{A}$.

We will now define a common way to specify operator algebras via a
finite set of generators.
\begin{defn}
\label{def: generated OA}The operator algebra $\mathcal{A}=\left\langle M_{1},M_{2},...,M_{m}\right\rangle $
is said to be\emph{ generated} by the operators $M_{i}$ if it is
the closure of the subset $\left\{ M_{i}\right\} _{i=1,...,m}\subset\mathcal{L}\left(\mathcal{H}\right)$
with respect to the conditions (1) - (3) of the Definition \ref{def: operator-algebra}.
\end{defn}
In applications, operator algebras are often specified this way. Moreover,
all operator algebras can be specified via a finite set of generators.
This trivially follows from the observation that the spanning set
of an algebra is in particular its generating set, and all subalgebras
of $\mathcal{L}\left(\mathcal{H}\right)$ have finite spanning sets.

We can always assume that the generators $M_{i}$ are self-adjoint
operators because every non-self-adjoint operator $M$ can be expressed
as a linear combination of two self-adjoint operators 
\[
M=\frac{M+M^{\dagger}}{2}+i\frac{M-M^{\dagger}}{2i}=M_{+}+iM_{-}
\]
so we can always use $M_{+}$ and $M_{-}$ as generators instead of
$M$ and $M^{\dagger}$. The closure with respect to the conditions
(1) - (3) then means that the elements of $\mathcal{A}$ are all the
possible products of the generators and the linear combinations of
these products
\[
\left\langle M_{1},...,M_{m}\right\rangle =\left\{ \sum_{n=1}^{N}c_{n}M_{i_{n,1}}M_{i_{n,2}}\cdots M_{i_{n,K}}\,|\,N,K\in\mathbb{N};\,\,c_{n}\in\mathbb{C};\,\,i_{n,k}\in\left\{ 1,...,m\right\} \right\} .
\]
In principle, as with any subalgebra of $\mathcal{L}\left(\mathcal{H}\right)$,
the algebra $\left\langle M_{1},...,M_{m}\right\rangle $ is of finite
dimension $D\leq d^{2}$ so there is a finite spanning set $\left\{ A_{1},A_{2},...,A_{D}\right\} $
such that
\[
\left\langle M_{1},...,M_{m}\right\rangle =\left\{ \sum_{n=1}^{D}c_{n}A_{n}\,|\,c_{n}\in\mathbb{C}\right\} .
\]
In practice, given only the generators $M_{i}$, it is not a trivial
task to tell the dimension of $\left\langle M_{1},...,M_{m}\right\rangle $
and find a spanning set.

An important special case is the algebra $\left\langle M,I\right\rangle $
generated by a single self-adjoint operator $M=M^{\dagger}$ and the
identity $I$ (the identity is not really necessary here but we will
include it to avoid finding the operator that acts as the identity).
By definition, this algebra is the set
\begin{equation}
\left\langle M,I\right\rangle =\left\{ \sum_{n=0}^{N}c_{n}M^{n}\,|\,N\in\mathbb{N},\,c_{n}\in\mathbb{C}\right\} .\label{eq: def of <H,I>}
\end{equation}
 The key fact about this algebra is that it is spanned by the spectral
projections of $M$.
\begin{prop}
\label{prop: single generator irreps }Let $M$ be a self-adjoint
operator with the spectral decomposition 
\[
M=\sum_{k=1}^{m}\lambda_{k}\Pi_{k}+\lambda_{0}\Pi_{0}
\]
where $\lambda_{k=1...m}$ are the distinct non-zero eigenvalues,
$\Pi_{k=1...m}$ are the projections on the corresponding eigenspaces,
and $\Pi_{0}$ is the projection on the kernel of $M$ ($\lambda_{0}\equiv0$).
Then 
\[
\left\langle M,I\right\rangle =\spn\left\{ \Pi_{k}\right\} _{k=0}^{m}.
\]
\end{prop}
\begin{proof}
For every $k\geq1$ the spectral projection $\Pi_{k}$ can be expressed
as 
\begin{equation}
\Pi_{k}=\prod_{l\neq k}\frac{M-\lambda_{l}I}{\lambda_{k}-\lambda_{l}}.\label{eq:spec proj from operator}
\end{equation}
This means that $\Pi_{k}\in\left\langle M,I\right\rangle $ for all
$k\geq1$. This also means that $\Pi_{0}\in\left\langle M,I\right\rangle $
since 
\[
\Pi_{0}=I-\sum_{k=1}^{m}\Pi_{k}.
\]
Thus, $\spn\left\{ \Pi_{k}\right\} _{k=0}^{m}\subseteq\left\langle M,I\right\rangle $.

According to Eq. \eqref{eq: def of <H,I>}, every $A\in\left\langle M,I\right\rangle $
is of the form 
\[
A=\sum_{n=0}^{N}c_{n}M^{n}=c_{0}I+\sum_{n=1}^{N}c_{n}\sum_{k=1}^{m}\lambda_{k}^{n}\Pi_{k}=c_{0}\sum_{k=0}^{m}\Pi_{k}+\sum_{n=1}^{N}c_{n}\sum_{k=1}^{m}\lambda_{k}^{n}\Pi_{k}
\]
so $A\in\spn\left\{ \Pi_{k}\right\} _{k=0}^{m}$. Thus, $\left\langle M,I\right\rangle \subseteq\spn\left\{ \Pi_{k}\right\} _{k=0}^{m}$
and so 
\[
\left\langle M,I\right\rangle =\spn\left\{ \Pi_{k}\right\} _{k=0}^{m}.
\]
\end{proof}
We can always say that $\left\langle M_{i},I\right\rangle $ is a
subalgebra of $\mathcal{A}=\left\langle M_{1},M_{2},...,M_{m},I\right\rangle $
for each $i=1,..,m$, so the spectral projections $\left\{ \Pi_{i;k}\right\} $
of $M_{i}$ are elements of $\mathcal{A}$. Therefore, since the spectral
projections $\left\{ \Pi_{i;k}\right\} $ span the generators $M_{i}$,
we can use the projections as generators instead of $M_{i}$
\[
\mathcal{A}=\left\langle M_{1},M_{2},...,M_{m},I\right\rangle =\left\langle \left\{ \Pi_{1;k}\right\} ,\left\{ \Pi_{2;k}\right\} ,...,\left\{ \Pi_{m;k}\right\} ,I\right\rangle .
\]
This means that we can always use projections instead of self-adjoint
operators to generate algebras.

We will now begin introducing the concepts that characterizes the
structure of general operator algebras, starting with the simplest
building blocks defined as follows.
\begin{defn}
The projection $\Pi\in\mathcal{A}$ is called \emph{minimal projection}
if for every $A\in\mathcal{A}$ we have $\Pi A\Pi=c\Pi$ for some
$c\in\mathbb{C}$.
\end{defn}
Note that rank $1$ projections are of the form $\Pi=\ketbra{\psi}{\psi}$
for some $\ket{\psi}\in\mathcal{H}$ so they are always minimal 
\[
\ketbra{\psi}{\psi}A\ketbra{\psi}{\psi}=\left\langle A\right\rangle _{\psi}\ketbra{\psi}{\psi}\propto\ketbra{\psi}{\psi}.
\]
The name \emph{minimal }is chosen because of the following property.
\begin{prop}
Let $\Pi\in\mathcal{A}$ be a minimal projection and let $\Pi'\in\mathcal{A}$
be another projection such that $\Pi\Pi'=\Pi'$, then $\Pi=\Pi'$.
\end{prop}
\begin{proof}
If $\Pi'=\Pi\Pi'$ then $\Pi'=\left(\Pi\Pi'\right)^{\dagger}=\Pi'\Pi$
and so $\Pi'=\Pi\Pi'\Pi$. Therefore, if $\Pi'\neq\Pi$ then $\Pi'\not\propto\Pi$
and so $\Pi$ is not minimal.
\end{proof}
The next step in the characterization of the structure behind operator
algebras is the definition of the following sets.
\begin{defn}
The set of projections $\left\{ \Pi_{k}\right\} \subset\mathcal{A}$
is called a \emph{maximal set of minimal projections} if all $\Pi_{k}$
are minimal, pairwise orthogonal $\Pi_{k}\Pi_{k'}=\delta_{kk'}\Pi_{k}$,
and sum to the identity $I=\sum_{k}\Pi_{k}$.
\end{defn}
Note that a maximal set of minimal projections does not mean that
these are all the minimal projections in the algebra, it just means
that these are minimal projections that resolve the identity. Every
algebra has at least one maximal set of minimal projections.
\begin{lem}
\label{lem:exsistance of max set of min proj}Let $\mathcal{A}\in\mathcal{L}\left(\mathcal{H}\right)$
be an operator algebra, then there is at least one maximal set of
minimal projections $\left\{ \Pi_{k}\right\} \subset\mathcal{A}$.
\end{lem}
\begin{proof}
This can be shown recursively by starting with the set of just the
identity $\left\{ I\right\} \subset\mathcal{A}$. If $I$ is a minimal
projection in $\mathcal{A}$ then we are done. If not then there is
a minimal projection $\Pi_{1}\in\mathcal{A}$ such that $\Pi_{1}I=\Pi_{1}\neq I$and
there is the compliment projection $\Pi'_{1}=I-\Pi_{1}\in\mathcal{A}$.
If $\Pi'_{1}$ is also minimal then $\left\{ \Pi_{1},\Pi'_{1}\right\} $
is a maximal set of minimal projections and we are done. If not, then
there is a minimal projection $\Pi_{2}$ such that $\Pi_{2}\Pi'_{1}=\Pi_{2}\neq\Pi'_{1}$
and there is the compliment $\Pi'_{2}=\Pi_{1}'-\Pi_{2}\in\mathcal{A}$.
After $k$ iterations we get the set $\left\{ \Pi_{1},\Pi_{2},...,\Pi_{k},\Pi'_{k}\right\} $
of pairwise orthogonal projections that sums to the identity $I$.
The first $k$ elements in the set are minimal projections and we
are done when the last element $\Pi'_{k}$ is also minimal. The recursion
will terminate after a finite number of steps because $\mathbf{rank}\left[\Pi'_{k}\right]<\mathbf{rank}\left[\Pi'_{k-1}\right]$
and projections of rank $1$ are always minimal.
\end{proof}
We can partition the maximal set of minimal projections into subsets
that will identify a block-diagonal form of the elements of the algebra
using the following equivalence relation.
\begin{prop}
\label{prop: equiv classes of minimal proj}Let $\left\{ \Pi_{k}\right\} $
be a maximal set of minimal projections in the algebra $\mathcal{A}$.
Then, the relation ``$\sim$'' where $\Pi_{k}\sim\Pi_{l}$ if and
only if there is an $A\in\mathcal{A}$ such that $\Pi_{k}A\Pi_{l}\neq0$,
is an equivalence relation.
\end{prop}
\begin{proof}
This relation is reflexive ($\Pi_{k}\sim\Pi_{k}$) since $\Pi_{k}\Pi_{k}\Pi_{k}\neq0$,
it is symmetric ($\Pi_{k}\sim\Pi_{l}$ implies $\Pi_{l}\sim\Pi_{k}$)
since $\Pi_{k}A\Pi_{l}\neq0$ implies $\Pi_{l}A^{\dagger}\Pi_{k}\neq0$,
and it is transitive ($\Pi_{k}\sim\Pi_{l}$ and $\Pi_{l}\sim\Pi_{k'}$
implies $\Pi_{k}\sim\Pi_{k'}$) since $\Pi_{k}A\Pi_{l}\neq0$ and
$\Pi_{l}A\Pi_{k'}\neq0$ implies $\Pi_{k}A\Pi_{l}A\Pi_{k'}\neq0$.
\end{proof}
Using this equivalence relation we partition the maximal set of minimal
projections into equivalence classes labeled by $q$ such that $\Pi_{k}^{q}\sim\Pi_{l}^{q'}$
if and only if $q=q'$ (the indices $k$ and $l$ refer now to the
distinct elements inside the equivalence classes). Thus, for every
$A\in\mathcal{A}$ and $q\neq q'$ we have $\Pi_{k}^{q}A\Pi_{l}^{q'}=0$
which prescribes a block-diagonal form for all the elements in the
algebra. We will therefore refer to the equivalence classes $q$ as
\emph{blocks}.

We saw that every element $A\in\left\langle M,I\right\rangle $ is
spanned by the spectral projections of $M$
\[
A=\sum_{n=0}^{m}c_{n}\Pi_{k}.
\]
Since the spectral projections $\left\{ \Pi_{k}\right\} _{k=0}^{m}$
are all orthogonal to each other we must have $\Pi_{k}A\Pi_{l}\propto\delta_{kl}\Pi_{k}$.
Therefore, the set $\left\{ \Pi_{k}\right\} _{k=0}^{m}$ is not only
the spanning set of $\left\langle M,I\right\rangle $ but it is a
maximal set of minimal projections in this algebra (in fact, this
is the only such set). Furthermore, each spectral projection is in
its own equivalence class ($\Pi_{k}\nsim\Pi_{l}$ when $k\neq l$)
which identifies the block diagonal structure of $\left\langle M,I\right\rangle $.

The case of $\left\langle M,I\right\rangle $ is too special to draw
any general conclusions about minimal projections. In general, the
set of \emph{all} the minimal projections in an algebra $\left\{ \Pi_{\alpha}\right\} \subset\mathcal{A}$
does not consists of pairwise orthogonal projections $\Pi_{\alpha}\Pi_{\alpha'}\neq0$,
and it has more than one maximal set of minimal projections. Furthermore,
when dealing with algebras generated by multiple operators $\mathcal{A}=\left\langle M_{1},M_{2},...,M_{m},I\right\rangle $,
the spectral projections $\left\{ \Pi_{i;k}\right\} $ of each generator
$M_{i}$ are not necessarily minimal projections in $\mathcal{A}$.
The distillation of minimal projections from the spectral projections
of the generators of the algebra is at the heart of the algorithm
that we will present in Chapter \ref{chap:Identifying-the-irreps}.

In order to fully capture the structure of an operator algebra (at
least in the way that is suitable for our purposes) we will need a
slight generalization of the notion of minimal projections.
\begin{defn}
The partial isometry $S\in\mathcal{A}$ is called a \emph{minimal
isometry} if the projections on its initial $S^{\dagger}S=\Pi_{in}$
and final $SS^{\dagger}=\Pi_{fin}$ spaces are minimal.
\end{defn}
Since every projection is a partial isometry $\Pi\Pi^{\dagger}=\Pi^{\dagger}\Pi=\Pi$,
minimal projections are in particular minimal isometries (the converse
is of course not true).

The most important property of minimal isometries is that given the
initial and final spaces they are unique up to a phase factor.
\begin{lem}
\label{lem:uniqueness of min isometr }Let $S,S'\in\mathcal{A}$ be
minimal isometries such that $S^{\dagger}S=S'^{\dagger}S'=\Pi_{in}$
and $SS^{\dagger}=S'S'^{\dagger}=\Pi_{fin}$, then $S'=e^{i\varphi}S$
for some real phase factor $\varphi$.
\end{lem}
\begin{proof}
Since $\Pi_{in}$ is the initial space of both $S$ and $S'$ and
it is minimal we have 
\[
S^{\dagger}S'=\Pi_{in}S^{\dagger}S'\Pi_{in}=c\Pi_{in}=cS^{\dagger}S
\]
 for some $c\in\mathbb{C}$. Multiplying both sides of $S^{\dagger}S'=cS^{\dagger}S$
by $S$ from the left and using $SS^{\dagger}=\Pi_{fin}$ we get 
\[
S'=\Pi_{fin}S'=c\Pi_{fin}S=cS.
\]
Since both $S$ and $S'$ are partial isometries, we must have $c=e^{i\varphi}$.
\end{proof}
Although minimal isometries generalize minimal projections, we can
construct the former from the latter using the following Lemma.
\begin{lem}
\label{lem:construct of minimal isometr from minimal proj }Let $\Pi_{1}$
and $\Pi_{2}$ be minimal projections in $\mathcal{A}$, then for
any $A\in\mathcal{A}$ the operator $\tilde{S}=\Pi_{1}A\Pi_{2}$ is
proportional to a minimal isometry $S\in\mathcal{A}$. In particular,
when $\tilde{S}\neq0$, we have $\tr\left[\Pi_{1}\right]=\tr\left[\Pi_{2}\right]$
and the minimal isometry is given by $S=c\tilde{S}$ where 
\[
c=\sqrt{\frac{\tr\left[\Pi_{1}\right]}{\tr\left[\tilde{S}\tilde{S}^{\dagger}\right]}}\in\mathbb{R}.
\]
\end{lem}
\begin{proof}
If $\tilde{S}=0$ then $\tilde{S}$ is trivially proportional to all
operators. If $\tilde{S}\neq0$ then by the definition of minimal
projections there is a proportionality factor $a$ such that 
\[
\tilde{S}\tilde{S}^{\dagger}=\Pi_{1}A\Pi_{2}\Pi_{2}A^{\dagger}\Pi_{1}=a\Pi_{1}
\]
We know that $a\neq0$ (otherwise $\tilde{S}\tilde{S}^{\dagger}=0$
and so $\tilde{S}=0$), therefore 
\[
a=\frac{\tr\left[\tilde{S}\tilde{S}^{\dagger}\right]}{\tr\left[\Pi_{1}\right]}.
\]
Since both $\tilde{S}\tilde{S}^{\dagger}\neq0$ and $\Pi_{1}\neq0$
are non-negative self-adjoint operators, $a$ must be a positive real.
The operator $S=c\tilde{S}$ where $c=\frac{1}{\sqrt{a}}$ is then
a partial isometry because $SS^{\dagger}=\Pi_{1}$ is a projection.

Similarly, we know that there is a proportionality factor $a'\neq0$
such that
\[
S^{\dagger}S=\frac{1}{a}\tilde{S}^{\dagger}\tilde{S}=\frac{1}{a}\Pi_{2}A^{\dagger}\Pi_{1}\Pi_{1}A\Pi_{2}=\frac{a'}{a}\Pi_{2}.
\]
Since $S^{\dagger}S$ and $\Pi_{2}$ are projections, we must have
$\frac{a'}{a}=1$. Therefore, both $SS^{\dagger}=\Pi_{1}$ and $S^{\dagger}S=\Pi_{2}$
are minimal projections, $\tr\left[SS^{\dagger}\right]=\tr\left[S^{\dagger}S\right]$,
and $S$ is a minimal isometry.
\end{proof}
We can now easily prove the fact that minimal projections of different
ranks are orthogonal.
\begin{cor}
Let $\Pi_{1}$ and $\Pi_{2}$ be minimal projections in $\mathcal{A}$
such that $\tr\left[\Pi_{1}\right]\neq\tr\left[\Pi_{2}\right]$, then
$\Pi_{1}\Pi_{2}=0$.
\end{cor}
\begin{proof}
Consider $A=I$ in Lemma \ref{lem:construct of minimal isometr from minimal proj }.
Then $\tilde{S}=\Pi_{1}A\Pi_{2}=\Pi_{1}\Pi_{2}$ so if $\Pi_{1}\Pi_{2}\neq0$
we must have $\tr\left[\Pi_{1}\right]=\tr\left[\Pi_{2}\right]$.
\end{proof}
The set of operators that fully captures the structure of an operator
algebra is the following special spanning set of minimal isometries.
\begin{defn}
\label{def:minimal isom}A set of partial isometries $\left\{ S_{kl}^{q}\right\} \subset\mathcal{A}$
is called \emph{maximal set of minimal isometries }in $\mathcal{A}$
if all $S_{kl}^{q}$ are minimal, the set $\left\{ S_{kl}^{q}\right\} $
spans the algebra
\[
\mathcal{A}=\spn\left\{ S_{kl}^{q}\right\} ,
\]
and for all values of $q$, $k$, $l$ we have $S_{kl}^{q}=S_{lk}^{q\dagger}$
and $S_{kl}^{q}S_{l'k'}^{q'}=\delta_{qq'}\delta_{ll'}S_{kk'}^{q}$.
\end{defn}
The existence of\emph{ }maximal sets of minimal isometries in every
algebra is guaranteed by the following theorem.
\begin{thm}
\label{thm: exsitance of max set of min isom}Let $\mathcal{A}\in\mathcal{L}\left(\mathcal{H}\right)$
be an operator algebra, then, there is a maximal set of minimal isometries
$\left\{ S_{kl}^{q}\right\} \subset\mathcal{A}$ that spans it.
\end{thm}
\begin{proof}
Let $\left\{ \Pi_{k}^{q}\right\} $ be a maximal set of minimal projections
in $\mathcal{A}$ provided by Lemma \ref{lem:exsistance of max set of min proj}
and partitioned into equivalence classes $q$ according to Proposition
\ref{prop: equiv classes of minimal proj}. By the definition of these
equivalence classes, for every $q$, $k$, $l$ there is at least
one $A\in\mathcal{A}$ such that $\Pi_{k}^{q}A\Pi_{l}^{q}\neq0$.
Then, according to Lemma \ref{lem:construct of minimal isometr from minimal proj },
for each $q$, $k$, $l$ there is a minimal isometry $S_{kl}^{q}\in\mathcal{A}$
and a real positive constant $c$ such that $S_{kl}^{q}=c\Pi_{k}^{q}A\Pi_{l}^{q}$
for some $A\in\mathcal{A}$. In order to get the desired properties
of the Definition \ref{def:minimal isom} we can construct the maximal
set of minimal isometries with the following procedure. First, for
each $q$, $k$ arbitrarily choose $A\in\mathcal{A}$ such that $S_{k1}^{q}=c\Pi_{k}^{q}A\Pi_{1}^{q}\neq0$
(we fixed $l=1$ but it does not matter what value of $l$ is fixed).
Then, define $S_{1k}^{q}=\left(S_{k1}^{q}\right)^{\dagger}$ and $S_{kl}^{q}=S_{k1}^{q}S_{1l}^{q}$
which are also minimal isometries in $\mathcal{A}$. Thus, for all
values of $q$, $k$, $l$ we have
\[
S_{kl}^{q}=S_{k1}^{q}S_{1l}^{q}=\left(S_{l1}^{q}S_{1k}^{q}\right)^{\dagger}=S_{lk}^{q\dagger}
\]
\[
S_{kl}^{q}S_{l'k'}^{q'}=S_{k1}^{q}S_{1l}^{q}S_{l'1}^{q'}S_{1k'}^{q'}=\delta_{qq'}\delta_{ll'}S_{k1}^{q}S_{1k'}^{q'}=\delta_{qq'}\delta_{ll'}S_{kk'}^{q}.
\]
Lemma \ref{lem:uniqueness of min isometr } implies that each minimal
isometry is unique in $\mathcal{A}$ up to a phase factor. Therefore,
given the set $\left\{ S_{kl}^{q}\right\} $ as constructed above,
for any $A\in\mathcal{A}$ we either have $\Pi_{k}^{q}A\Pi_{l}^{q}=0$
or $\Pi_{k}^{q}A\Pi_{l}^{q}=\frac{e^{i\varphi}}{c}S_{kl}^{q}$ for
some phase $\varphi$ and a real $c$. Recalling that $\sum_{q,k}\Pi_{k}^{q}=I$,
we can express any $A\in\mathcal{A}$ as 
\[
A=\left(\sum_{q,k}\Pi_{k}^{q}\right)A\left(\sum_{q',l}\Pi_{l}^{q'}\right)=\sum_{q,k,l}\Pi_{k}^{q}A\Pi_{l}^{q}=\sum_{q,k,l}c_{kl}^{q}S_{kl}^{q}
\]
where we have used the fact that $\Pi_{k}^{q}A\Pi_{l}^{q'}=0$ for
$q\neq q'$ and introduced the complex coefficients $c_{kl}^{q}$.
Therefore $\mathcal{A}=\spn\left\{ S_{kl}^{q}\right\} $.
\end{proof}
Note that for $k=l$ the minimal isometries $S_{kk}^{q}$ are actually
the minimal projections $\Pi_{k}^{q}$, so the maximal set of minimal
projections $\left\{ \Pi_{k}^{q}\right\} $ is a subset of the maximal
set of minimal isometries $\left\{ S_{kl}^{q}\right\} $.

Although it was not very easy to get to the general result of Theorem
\ref{thm: exsitance of max set of min isom}, maximal sets of minimal
isometries are quite easy to find in some canonical examples. For
example, in the case of the full operator algebra $\mathcal{L}\left(\mathcal{H}_{qudit}\right)$
of a qudit $\mathcal{H}_{qudit}=\spn\left\{ \ket k\right\} _{k=1,...,d}$,
the minimal isometries are simply the matrix units
\[
S_{kl}=\ket k\bra l.
\]
Here the index $q$ is suppressed because all $\left\{ S_{kl}\right\} $
belong to the same block as we cannot partition them into subsets
that are completely orthogonal to each other.

For a a slightly more interesting example we may consider the qudit
$\mathcal{H}_{qudit}$ in a tensor product with a qubit $\mathcal{H}_{qubit}=\spn\left\{ \ket 0,\ket 1\right\} $.
Then, consider the algebra
\begin{equation}
\mathcal{A}=\left\{ \ket q\bra q\otimes A\,\,\,\,|\,q=0,1;\,\,\,A\in\mathcal{L}\left(\mathcal{H}_{qudit}\right)\right\} \label{eq: def examp. algebra 1}
\end{equation}
with the maximal set of minimal isometries 
\[
S_{kl}^{q}=\ket q\bra q\otimes\ket k\bra l.
\]
Here the index $q=0,1$ distinguishes the two blocks of completely
orthogonal isometries.

In both of the above examples, the initial and the final spaces of
the isometries are one-dimensional $\tr\left[S_{kl}^{q}S_{lk}^{q}\right]=1$.
In general, this is not the case and one should think of $\left\{ S_{kl}^{q}\right\} $
as generalized matrix units that map between orthogonal subspaces
of dimension one or higher. In the next section we will show how maximal
sets of minimal isometries fully capture the structure of irreducible
representations of operator algebras.

\section{Bipartition tables and the irreps structure\label{sec:Irrep-structures-and}}

In order to understand what a maximal set of minimal isometries tells
us about the algebra we will introduce a neat visual aid that captures
the implied structure. This visual aid is called a \emph{bipartition
table} and we will see that it specifies the structure of irreducible
representations. The correspondence between maximal sets of minimal
isometries and bipartition tables leads to the main result of the
representation theory of (finite-dimensional) operator algebras known
as the Wedderburn Decomposition.

Let us start with the definition.
\begin{defn}
\label{def:A-bipartition-table}A \emph{bipartition table (BPT)} is
an arrangement of some basis of the Hilbert space into a block-diagonal
table. This arrangement is specified by a choice of orthonormal basis
elements $\left\{ \ket{e_{ik}^{q}}\right\} $ labeled with the indices
of blocks $q$, rows $i$ and columns $k$. For each block $q$ we
construct the rectangular table

\noindent\begin{minipage}[c]{1\columnwidth}%
\begin{center}
\vspace{0.5\baselineskip}
\begin{tabular}{|c|c|c|}
\hline 
$e_{1,1}^{q}$ & $e_{1,2}^{q}$ & $\cdots$\tabularnewline
\hline 
$e_{2,1}^{q}$ & $e_{2,2}^{q}$ & $\cdots$\tabularnewline
\hline 
$\vdots$ & $\vdots$ & $\ddots$\tabularnewline
\hline 
\end{tabular} ,\vspace{0.5\baselineskip}
\par\end{center}%
\end{minipage} and the full bipartition table is given by the diagonal arrangement
of all the blocks

\noindent\begin{minipage}[c]{1\columnwidth}%
\begin{center}
\vspace{0.5\baselineskip}
\begin{tabular}{cc|cccccc}
\cline{1-2} \cline{2-2} 
\multicolumn{1}{|c|}{$e_{1,1}^{1}$} & $\cdots$ &  &  &  &  &  & \tabularnewline
\cline{1-2} \cline{2-2} 
\multicolumn{1}{|c|}{$\vdots$} & $\ddots$ &  &  &  &  &  & \tabularnewline
\cline{1-4} \cline{2-4} \cline{3-4} \cline{4-4} 
 &  & \multicolumn{1}{c|}{$e_{1,1}^{2}$} & \multicolumn{1}{c|}{$\cdots$} &  &  &  & \tabularnewline
\cline{3-4} \cline{4-4} 
 &  & \multicolumn{1}{c|}{$\vdots$} & \multicolumn{1}{c|}{$\ddots$} &  &  &  & \tabularnewline
\cline{3-4} \cline{4-4} 
 & \multicolumn{1}{c}{} &  &  & $\ddots$ &  &  & \tabularnewline
\cline{6-7} \cline{7-7} 
 & \multicolumn{1}{c}{} &  &  & \multicolumn{1}{c|}{} & \multicolumn{1}{c|}{$e_{1,1}^{q}$} & \multicolumn{1}{c|}{$\cdots$} & \tabularnewline
\cline{6-7} \cline{7-7} 
 & \multicolumn{1}{c}{} &  &  & \multicolumn{1}{c|}{} & \multicolumn{1}{c|}{$\vdots$} & \multicolumn{1}{c|}{$\ddots$} & \tabularnewline
\cline{6-7} \cline{7-7} 
 & \multicolumn{1}{c}{} &  &  &  &  &  & $\ddots$\tabularnewline
\end{tabular}\vspace{0.5\baselineskip}
\par\end{center}%
\end{minipage}
\end{defn}
What makes BPTs useful is that they tell us how to construct maximal
sets of minimal isometries. The construction is simple: Each pair
of columns $k$, $l$ in the block $q$, specifies the isometry
\begin{equation}
S_{kl}^{q}:=\sum_{i}\ket{e_{ik}^{q}}\bra{e_{il}^{q}},\label{eq: def of S_kl from BPT}
\end{equation}
where $i$ runs over all the rows in the block. The blocks of the
BPT partition the minimal isometries into orthogonal subsets, that
is $S_{kl}^{q}S_{k'l'}^{q'}=0$ for $q\neq q'$. The subsets of the
basis $\left\{ \ket{e_{il}^{q}}\right\} _{i=1,...}$ and $\left\{ \ket{e_{ik}^{q}}\right\} _{i=1,...}$
given by the columns $l$ and $k$ specify the initial and final spaces
of the isometry $S_{kl}^{q}$. The alignment of basis elements across
the rows specifies how the isometries map the vectors between the
subspaces, that is, the basis element $\ket{e_{il}^{q}}$ is mapped
to $\ket{e_{ik}^{q}}$ (these are the right and left singular vectors
of $S_{kl}^{q}$).

It is easy to show that the set of isometries constructed in this
way spans an algebra.
\begin{prop}
\label{prop: bpt isometries span an algebra}Let $\left\{ S_{kl}^{q}\right\} $
be the set of partial isometries constructed from a bipartition table
according to Eq. \eqref{eq: def of S_kl from BPT}. Then, $\mathcal{A}=\spn\left\{ S_{kl}^{q}\right\} $
is an operator algebra and $\left\{ S_{kl}^{q}\right\} $ is a maximal
set of minimal isometries in $\mathcal{A}$.
\end{prop}
\begin{proof}
Clearly, for all $q$, $k$, $l$ we have $S_{kl}^{q}=S_{lk}^{q\dagger}$
and $S_{kl}^{q}S_{l'k'}^{q'}=\delta_{qq'}\delta_{ll'}S_{kk'}^{q}$.
Therefore, for any $A_{1},A_{2}\in\spn\left\{ S_{kl}^{q}\right\} $
and $c_{1},c_{2}\in\mathbb{C}$ we have: $c_{1}A_{1}+c_{2}A_{2}\in\spn\left\{ S_{kl}^{q}\right\} $,
and $A_{1}A_{2}\in\spn\left\{ S_{kl}^{q}\right\} $, and $A_{1}^{\dagger},A_{2}^{\dagger}\in\spn\left\{ S_{kl}^{q}\right\} $.
Thus by definition $\mathcal{A}=\spn\left\{ S_{kl}^{q}\right\} $
is an operator algebra and $\left\{ S_{kl}^{q}\right\} $ is a maximal
set of minimal isometries in it.
\end{proof}
Let us consider some examples of BPTs that specify the minimal isometries
of some familiar algebras.

The minimal isometries $S_{kl}=\ket k\bra l$ of the full operator
algebra $\mathcal{L}\left(\mathcal{H}_{qudit}\right)$ of the qudit
$\mathcal{H}_{qudit}=\spn\left\{ \ket k\right\} _{k=1,...,d}$, are
constructed from the BPT

\noindent\begin{minipage}[c]{1\columnwidth}%
\begin{center}
\vspace{0.5\baselineskip}
\begin{tabular}{|c|c|c|c|}
\hline 
$1$ & $2$ & $\cdots$ & $d$\tabularnewline
\hline 
\end{tabular} .\vspace{0.5\baselineskip}
\par\end{center}%
\end{minipage} There is only one block here and this block has only one row. Following
the construction in Eq. \eqref{eq: def of S_kl from BPT} we can reproduce
all the minimal isometries of $\mathcal{L}\left(\mathcal{H}_{qudit}\right)$.

Adding a qubit $\mathcal{H}_{qubit}\otimes\mathcal{H}_{qudit}$ to
the qudit we consider again the algebra in Eq. \eqref{eq: def examp. algebra 1}.
Using the combined basis labels $\ket{q,k}\equiv\ket q_{qubit}\otimes\ket k_{qudit}$,
we can see that all the minimal isometries are of the form
\[
S_{kl}^{q}=\ket q\bra q\otimes\ket k\bra l\equiv\ket{q,k}\bra{q,l}.
\]
These isometries can be constructed from the BPT

\noindent\begin{minipage}[c]{1\columnwidth}%
\begin{center}
\vspace{0.5\baselineskip}
\begin{tabular}{|c|c|c|c|c|c|c|c|}
\cline{1-4} \cline{2-4} \cline{3-4} \cline{4-4} 
$0,1$ & $0,2$ & $\cdots$ & $0,d$ & \multicolumn{1}{c}{} & \multicolumn{1}{c}{} & \multicolumn{1}{c}{} & \multicolumn{1}{c}{}\tabularnewline
\hline 
\multicolumn{1}{c}{} & \multicolumn{1}{c}{} & \multicolumn{1}{c}{} &  & $1,1$ & $1,2$ & $\cdots$ & $1,d$\tabularnewline
\cline{5-8} \cline{6-8} \cline{7-8} \cline{8-8} 
\end{tabular} .\vspace{0.5\baselineskip}
\par\end{center}%
\end{minipage} Here we have two blocks with one row each.

Lastly, still with the Hilbert space $\mathcal{H}_{qubit}\otimes\mathcal{H}_{qudit}$,
consider the algebra of all the operators that act only on the qudit
\[
\mathcal{A}=\left\{ I_{qubit}\otimes A_{qudit}\,\,\,\,|\,A_{qudit}\in\mathcal{L}\left(\mathcal{H}_{qudit}\right)\right\} .
\]
The maximal set of minimal isometries in this case consists of
\[
S_{kl}=I_{qubit}\otimes\ket k\bra l=\ket{0,k}\bra{0,l}+\ket{1,k}\bra{1,l}
\]
and the BPT that produces them is

\noindent\begin{minipage}[c]{1\columnwidth}%
\begin{center}
\vspace{0.5\baselineskip}
\begin{tabular}{|c|c|c|c|}
\hline 
$0,1$ & $0,2$ & $\cdots$ & $0,d$\tabularnewline
\hline 
$1,1$ & $1,2$ & $\cdots$ & $1,d$\tabularnewline
\hline 
\end{tabular} .\vspace{0.5\baselineskip}
\par\end{center}%
\end{minipage} Here we have a single block with two rows.

In general, given a maximal set of minimal isometries of the algebra,
we can always find a BPT that produces it.
\begin{lem}
\label{lem: for all S_kl there is a BPT}Let $\left\{ S_{kl}^{q}\right\} $
be a maximal set of minimal isometries. Then, there is a BPT that
produces all $\left\{ S_{kl}^{q}\right\} $ according to Eq. \eqref{eq: def of S_kl from BPT}.
\end{lem}
\begin{proof}
Given $\left\{ S_{kl}^{q}\right\} $ let us explicitly construct this
BPT as follows:

\noindent\begin{minipage}[t]{1\columnwidth}%
\begin{enumerate}
\item Each $q$ corresponds to a separate block of the BPT constructed independently.
\item Arbitrarily choose orthonormal basis $\left\{ \ket{e_{i1}^{q}}\right\} _{i=1,...}$
for the eigenspace of $S_{11}^{q}$ and assign them to the first column\\
\noindent\begin{minipage}[c]{1\columnwidth}%
\begin{center}
\vspace{0.5\baselineskip}
\begin{tabular}{|c|}
\hline 
$e_{11}^{q}$\tabularnewline
\hline 
$e_{21}^{q}$\tabularnewline
\hline 
$\vdots$\tabularnewline
\hline 
\end{tabular} .\vspace{0.5\baselineskip}
\par\end{center}%
\end{minipage}
\item For every $k>1$ map the first column to a new column in the block
using the isometries $\ket{e_{ik}^{q}}=S_{k1}^{q}\ket{e_{i1}^{q}}$\\
\noindent\begin{minipage}[c]{1\columnwidth}%
\begin{center}
\vspace{0.5\baselineskip}
\begin{tabular}{|c|c|c|c|}
\hline 
$e_{11}^{q}$ & $\cdots$ & $e_{1k}^{q}$ & $\cdots$\tabularnewline
\hline 
$e_{21}^{q}$ & $\cdots$ & $e_{2k}^{q}$ & $\cdots$\tabularnewline
\hline 
$\vdots$ & $\vdots$ & $\vdots$ & $\ddots$\tabularnewline
\hline 
\end{tabular} .\vspace{0.5\baselineskip}
\par\end{center}%
\end{minipage}
\end{enumerate}
\end{minipage}

\noindent Then, according to Eq. \eqref{eq: def of S_kl from BPT},
the isometries constructed from this table are 
\[
\tilde{S}_{kl}^{q}=\sum_{i}\ket{e_{ik}^{q}}\bra{e_{il}^{q}}=S_{k1}^{q}\left[\sum_{i}\ket{e_{i1}^{q}}\bra{e_{i1}^{q}}\right]S_{1l}^{q}=S_{k1}^{q}S_{11}^{q}S_{1l}^{q}=S_{kl}^{q}.
\]
\end{proof}
Note that when $\left\{ S_{kl}^{q}\right\} $ is not supported on
the whole Hilbert space $\mathcal{H}$, the set of orthonormal basis
$\left\{ \ket{e_{ik}^{q}}\right\} $ constructed in the above lemma
is not complete, and it only spans a proper subspace of $\mathcal{H}$
where the algebra is supported. It can be shown that the algebra $\mathcal{A}$
is supported on the whole Hilbert space $\mathcal{H}$, if and only
if $I\in\mathcal{A}$.

The above Lemma closes the logical arc started with the Theorem \ref{thm: exsitance of max set of min isom}
and Proposition \ref{prop: bpt isometries span an algebra}: Every
operator algebra is spanned by a maximal set of minimal isometries
that can be constructed from a BPT, and every BPT constructs a set
of minimal isometries that span an operator algebra. Thus, we can
directly relate BPTs to operator algebras and operator algebras to
BPTs.

We already know that by construction \ref{eq: def of S_kl from BPT},
the columns of the BPT specify a maximal set of minimal projections
in the algebra. The rows of the BPT are also meaningful and they specify
the following subspaces.
\begin{defn}
Let $\mathcal{A}\subseteq\mathcal{L}\left(\mathcal{H}\right)$ be
an operator algebra. The subspace $\mathcal{V}\subseteq\mathcal{H}$
is called an \emph{invariant subspace} under $\mathcal{A}$ if for
all $\ket{\psi}\in\mathcal{V}$ and $A\in\mathcal{A}$ we have $A\ket{\psi}\in\mathcal{\mathcal{V}}$.
If, in addition, every proper subspace $\mathcal{\mathcal{V}}'\subset\mathcal{V}$
is not invariant, then $\mathcal{V}$ is called a \emph{minimal} invariant
subspace.
\end{defn}
\begin{prop}
Let $\left\{ \ket{e_{ik}^{q}}\right\} $ be the orthonormal basis
forming a BPT of the operator algebra $\mathcal{A}$. Then, every
subspace $\mathcal{V}:=\spn\left\{ \ket{e_{il}^{q}}\right\} _{l=1,...}$
spanned by the basis elements in a single row is a minimal invariant
subspace\emph{.}
\end{prop}
\begin{proof}
With the minimal isometries constructed as in Eq. \eqref{eq: def of S_kl from BPT},
we can express any $A\in\mathcal{A}=\spn\left\{ S_{kl}^{q}\right\} $
as 
\[
A=\sum_{q',k',l'}c_{k'l'}^{q'}S_{k'l'}^{q'}.
\]
Then, 
\begin{equation}
A\ket{e_{il}^{q}}=\sum_{q',k',l'}c_{k'l'}^{q'}S_{k'l'}^{q'}\ket{e_{il}^{q}}=\sum_{k'}c_{k'l}^{q}\ket{e_{ik'}^{q}}\in\mathcal{V},\label{eq: rows are min inv subspace identiy}
\end{equation}
so $\mathcal{V}$ is an invariant subspace under $\mathcal{A}$. If
$\mathcal{V}$ is not minimal then there is subspace $\mathcal{V}'\subset\mathcal{V}=\mathcal{V}'\oplus\mathcal{V}''$,
such that for every non-zero $\ket{\psi'}\in\mathcal{V}'$ and $\ket{\psi''}\in\mathcal{V}''$
we have $\bra{\psi''}A\ket{\psi'}=0$ for all $A\in\mathcal{A}$.
However, for every non-zero $\ket{\psi'},\ket{\psi''}\in\mathcal{V}=\spn\left\{ \ket{e_{il}^{q}}\right\} _{l=1,...}$
, there is always at least one $S_{kl}^{q}\in\mathcal{A}$ such that
$\bra{\psi''}S_{kl}^{q}\ket{\psi'}\neq0$, therefore $\mathcal{V}$
is minimal.
\end{proof}
The rows of the BPT identify the subspaces on which $\mathcal{A}$
acts irreducibly. Furthermore, it should be clear that the action
of $A\in\mathcal{A}$ is identical on every row in the same block
since the expression in Eq. \eqref{eq: rows are min inv subspace identiy}
does not depend on the row index $i$. Therefore, all the rows in
the same block carry equivalent irreducible representations of $\mathcal{A}$,
and the number of rows in the block is the multiplicity of that irreducible
representation.

The above statements are essentially the main result of the representation
theory of finite-dimensional operator algebras, albeit, in the non-standard
formulation that relies on the picture of BPTs. We will now present
this result in the standard form known as the Wedderburn Decomposition.
\begin{thm}
\label{thm: Wedderburn Decomp}Let $\mathcal{A}\subseteq\mathcal{L}\left(\mathcal{H}\right)$
be an operator algebra supported on the whole Hilbert space $\mathcal{H}$.
Then, there is a decomposition (Wedderburn Decomposition)
\begin{equation}
\mathcal{H}\cong\bigoplus_{q}\mathcal{H}_{\nu_{q}}\otimes\mathcal{H}_{\mu_{q}},\label{eq:Wedder. decomposition H space}
\end{equation}
such that
\[
\mathcal{A}\cong\bigoplus_{q}I_{\nu_{q}}\otimes\mathcal{L}\left(\mathcal{H}_{\mu_{q}}\right):=\left\{ \bigoplus_{q}I_{\nu_{q}}\otimes A_{q}\,|\,A_{q}\in\mathcal{L}\left(\mathcal{H}_{\mu_{q}}\right)\right\} .
\]
\end{thm}
\begin{proof}
Let $\left\{ S_{kl}^{q}\right\} $ be a maximal set of minimal isometries
in $\mathcal{A}$ as provided by Theorem \ref{thm: exsitance of max set of min isom},
and let $\left\{ \ket{e_{ik}^{q}}\right\} $ be the orthonormal basis
forming the BPT as provided by Lemma \ref{lem: for all S_kl there is a BPT}.
Since $\mathcal{A}$ is supported on the whole $\mathcal{H}$, we
can define the map
\[
V:\ket{e_{ik}^{q}}\longmapsto\ket{n_{i}^{q}}\otimes\ket{m_{k}^{q}}
\]
which isometrically maps the whole $\mathcal{H}$ to the tensor products
of $\mathcal{H}_{\nu_{q}}:=\spn\left\{ \ket{n_{i}^{q}}\right\} _{i=1,...}$
(associated with the row index) and $\mathcal{H}_{\mu_{q}}:=\spn\left\{ \ket{m_{k}^{q}}\right\} _{k=1,...}$
(associated with the column index). Thus, we identify the isometric
relation that specifies a decomposition of $\mathcal{H}$:
\[
\mathcal{H}\cong V\mathcal{H}=\bigoplus_{q}\mathcal{H}_{\nu_{q}}\otimes\mathcal{H}_{\mu_{q}}.
\]
The image of the algebra $\mathcal{A}$ under $V$ is then 
\[
\mathcal{A}\cong V\mathcal{A}V^{\dagger}=\spn\left\{ VS_{kl}^{q}V^{\dagger}\right\} .
\]
We can now see that 
\[
VS_{kl}^{q}V^{\dagger}=\sum_{i}V\ket{e_{ik}^{q}}\bra{e_{il}^{q}}V^{\dagger}=\sum_{i}\ket{n_{i}^{q}}\bra{n_{i}^{q}}\otimes\ket{m_{k}^{q}}\bra{m_{l}^{q}}=I_{\nu_{q}}\otimes\ket{m_{k}^{q}}\bra{m_{l}^{q}}
\]
and therefore
\[
\mathcal{A}\cong\spn\left\{ I_{\nu_{q}}\otimes\ket{m_{k}^{q}}\bra{m_{l}^{q}}\right\} =\bigoplus_{q}I_{\nu_{q}}\otimes\mathcal{L}\left(\mathcal{H}_{\mu_{q}}\right).
\]
\end{proof}
The decomposition in Eq. \eqref{eq:Wedder. decomposition H space}
is the general structure of irreducible representations of operator
algebras. In the broader mathematical context, this leads to the realization
that every operator algebra $\mathcal{A}\subseteq\mathcal{L}\left(\mathcal{H}\right)$
is just (up to an isomorphism) a direct sum of full operator algebras
$\mathcal{L}\left(\mathcal{H}_{\mu_{q}}\right)$. When the algebra
$\mathcal{A}$ is not supported on the whole $\mathcal{H}$, this
theorem applies to a proper subspace $\mathcal{H}'\subset\mathcal{H}$
where the operators of $\mathcal{A}$ are supported.\footnote{This should not be an issue for us since including the full identity
$I$ in the algebra will always be possible.}

In the qubit-qudit example with the algebra of operators that act
only on the qudit 
\[
\mathcal{A}=I_{qubit}\otimes\mathcal{L}\left(\mathcal{H}_{qudit}\right),
\]
the Wedderburn Decomposition is simply $\mathcal{H}=\mathcal{H}_{qubit}\otimes\mathcal{H}_{qudit}$
by the definition of $\mathcal{A}$. The matrix form of all $A\in\mathcal{A}$
is then
\[
A=I_{qubit}\otimes A_{qudit}=\ket 0\bra 0\otimes A_{qudit}+\ket 1\bra 1\otimes A_{qudit}=\begin{pmatrix}A_{qudit}\\
 & A_{qudit}
\end{pmatrix}.
\]

In general, Theorem \ref{thm: Wedderburn Decomp} tells us that there
is always a decomposition \eqref{eq:Wedder. decomposition H space}
where the operator algebra $\mathcal{A}$ acts as the identity on
$\mathcal{H}_{\nu_{q}}$'s and as the full operator algebra on $\mathcal{H}_{\mu_{q}}$'s,
and it does not map between the sectors $q$. That is, with respect
to the Wedderburn Decomposition, all $A\in\mathcal{A}$ are of the
form 
\[
A\cong\bigoplus_{q}I_{\nu_{q}}\otimes A_{q}=\bigoplus_{q}\left[\sum_{i=1}^{\dim\mathcal{H}_{\nu_{q}}}\ket{n_{i}^{q}}\bra{n_{i}^{q}}\otimes A_{q}\right]=\begin{pmatrix}\underbrace{\begin{array}{ccc}
A_{1}\\
 & \ddots\\
 &  & A_{1}
\end{array}}\\
\begin{array}{c}
^{\dim\mathcal{H}_{\nu_{1}}}\\
\\
\\
\end{array} & \underbrace{\begin{array}{ccc}
A_{2}\\
 & \ddots\\
 &  & A_{2}
\end{array}}\\
\begin{array}{c}
\\
\\
\end{array} & \begin{array}{c}
^{\dim\mathcal{H}_{\nu_{2}}}\\
\\
\end{array} & \begin{array}{c}
\ddots\\
\\
\end{array}
\end{pmatrix}.
\]

From the explicit block-diagonal matrix form we can see that for each
sector $q$, we have $\dim\mathcal{H}_{\nu_{q}}$ identical matrix
blocks where $\mathcal{A}$ acts irreducibly with the matrices $A_{q}$.
These matrix blocks correspond to the minimal invariant subspaces
spanned by a single row in the BPT
\[
\spn\left\{ \ket{e_{ik}^{q}}\right\} _{k=1,...}\cong\spn\left\{ \ket{n_{i}^{q}}\otimes\ket{m_{k}^{q}}\right\} _{k=1,...}.
\]
We can see now that the BPT block index $q$ distinguish between the
classes of minimal invariant subspaces on which the action of $A\in\mathcal{A}$
is represented independently with distinct $A_{q}$'s. Then, inside
the blocks, the BPT row index $i$ distinguishes between the minimal
invariant subspaces on which the action of $A\in\mathcal{A}$ is represented
with the same $A_{q}$. In other words, the \emph{rows} of the BPT
correspond to the \emph{irreducible} \emph{matrix blocks} of $A$,
while the \emph{blocks} of the BPT correspond to the\emph{ super-blocks
of identical irreducible matrix blocks }of $A$.

It should now be clear how BPTs specify the irreps structure by arranging
the basis into a table.\footnote{Note that BPTs only tell us how to arrange the basis \emph{labels}
into a table, they do not explicitly specify the basis themselves.
Defining the basis behind the labels in the BPT is an essential information
about the irreps structure.} Our earlier assertion that BPTs correspond to operator algebras can
now be restated in a stronger form: BPTs correspond to the irreps
structures behind operator algebras.

We will now consider group algebras as a special case and derive the
structure of group representations from the above results.
\begin{defn}
\label{def:group algebra}Given a finite or a Lie group $\mathcal{G}$
with the unitary representation $U\left(\mathcal{G}\right):=\left\{ U\left(g\right)\right\} _{g\in\mathcal{G}}\subset\mathcal{L}\left(\mathcal{H}\right)$,
the \emph{group algebra} is denoted and defined as 
\[
\mathcal{A}_{U\left(\mathcal{G}\right)}:=\spn\left\{ U\left(g\right)\right\} _{g\in\mathcal{G}}.
\]
\end{defn}
Clearly $U\left(\mathcal{G}\right)\subset\mathcal{A}_{U\left(\mathcal{G}\right)}$
so with respect to the Wedderburn Decomposition \eqref{eq:Wedder. decomposition H space},
for all $U\left(g\right)\in U\left(\mathcal{G}\right)$ there are
$U_{\mu_{q}}\left(g\right)\in\mathcal{L}\left(\mathcal{H}_{\mu_{q}}\right)$
such that 
\begin{equation}
U\left(g\right)\cong\bigoplus_{q}I_{\nu_{q}}\otimes U_{\mu_{q}}\left(g\right).\label{eq: unitary group rep decomp}
\end{equation}

\begin{thm}
\label{thm: group algebra irreps}Let $U\left(\mathcal{G}\right)$
be a unitary representation of the group $\mathcal{G}$ on $\mathcal{H}$,
and let Eq. \eqref{eq: unitary group rep decomp} be the decomposition
of the group action as given by the Theorem \ref{thm: Wedderburn Decomp}
for the group algebra $\mathcal{A}_{U\left(\mathcal{G}\right)}$.
Then, for all $q$, $U_{\mu_{q}}\left(\mathcal{G}\right)$ are inequivalent
irreducible unitary representations of $\mathcal{G}$ on $\mathcal{H}_{\mu_{q}}$.
\end{thm}
\begin{proof}
The fact that $U_{\mu_{q}}\left(\mathcal{G}\right)$ are unitary representations
of $\mathcal{G}$ follows directly from the fact that $U\left(\mathcal{G}\right)$
is a unitary representation of $\mathcal{G}$. According to Theorem
\ref{thm: Wedderburn Decomp}, the group algebra $\mathcal{A}_{U\left(\mathcal{G}\right)}$
acts on $\mathcal{H}_{\mu_{q}}$ as the full operator algebra $\mathcal{L}\left(\mathcal{H}_{\mu_{q}}\right)$.
Then, by the definition of group algebras, we must have
\[
\spn\left\{ U_{\mu_{q}}\left(g\right)\right\} _{g\in\mathcal{G}}=\mathcal{L}\left(\mathcal{H}_{\mu_{q}}\right).
\]
There can be no proper invariant subspaces of $\mathcal{H}_{\mu_{q}}$
under the action of $U_{\mu_{q}}\left(\mathcal{G}\right)$, because
the are no proper invariant subspaces under the action of $\mathcal{L}\left(\mathcal{H}_{\mu_{q}}\right)=\spn\left\{ U_{\mu_{q}}\left(g\right)\right\} _{g\in\mathcal{G}}$.
Therefore, $U_{\mu_{q}}\left(\mathcal{G}\right)$ acts irreducibly
on $\mathcal{H}_{\mu_{q}}$.

Furthermore, the general result of Theorem \ref{thm: Wedderburn Decomp}
implies that the algebra $\mathcal{A}_{U\left(\mathcal{G}\right)}$
includes the projection $I_{\nu_{q}}\otimes I_{\mu_{q}}$ on the sector
$q$. Then, there are coefficients $c\left(g\right)\in\mathbb{C}$
such that
\[
\sum_{g\in\mathcal{G}}c\left(g\right)U\left(g\right)\cong\bigoplus_{q'}I_{\nu_{q'}}\otimes\left[\sum_{g\in\mathcal{G}}c\left(g\right)U_{\mu_{q'}}\left(g\right)\right]=I_{\nu_{q}}\otimes I_{\mu_{q}},
\]
and so
\[
\sum_{g\in\mathcal{G}}c\left(g\right)U_{\mu_{q'}}\left(g\right)=\delta_{qq'}I_{\mu_{q}}.
\]
Therefore, for every $q'\neq q$ there must be some $g\in\mathcal{G}$
such that $U_{\mu_{q'}}\left(g\right)\neq U_{\mu_{q}}\left(g\right)$
and so the representations $U_{\mu_{q'}}\left(\mathcal{G}\right)$
and $U_{\mu_{q}}\left(\mathcal{G}\right)$ are not equivalent.
\end{proof}
Theorem \ref{thm: group algebra irreps} tells us that the irreps
structure of a group representation is, in fact, inherited from the
irreps structure of the group algebra. We can therefore use all the
insights about the irreps structure of operator algebras, in particular
BPTs, to characterize the representations of groups. 

As a simple example, consider the Hilbert space of two spins $\mathcal{H}=\underline{\frac{1}{2}}\otimes\underline{\frac{1}{2}}$
and the group of collective rotations. From group representation theory
of $SU\left(2\right)$ we know that this Hilbert space decomposes
as
\[
\mathcal{H}=\underline{\frac{1}{2}}\otimes\underline{\frac{1}{2}}=\underline{1}\oplus\underline{0}
\]
where the triplet (spin-$1$) and singlet (spin-$0$) subspaces are
spanned by the basis

\noindent\begin{minipage}[c]{1\columnwidth}%
\begin{center}
\vspace{0.5\baselineskip}
\begin{tabular}{l>{\raggedright}p{0.02\columnwidth}>{\raggedright}p{0.02\columnwidth}l}
$\ket{1;1}=\ket{\frac{1}{2},\frac{1}{2}}$ & \centering{} &  & $\ket{0;0}\propto\ket{\frac{1}{2},-\frac{1}{2}}-\ket{-\frac{1}{2},\frac{1}{2}}$\tabularnewline
$\ket{1;0}\propto\ket{\frac{1}{2},-\frac{1}{2}}+\ket{-\frac{1}{2},\frac{1}{2}}$ & \centering{}$\begin{array}{c}
\\
\\
\end{array}$ &  & \tabularnewline
$\ket{1;-1}=\ket{-\frac{1}{2},-\frac{1}{2}}.$ &  &  & \tabularnewline
\end{tabular}\vspace{0.5\baselineskip}
\par\end{center}%
\end{minipage} These basis identify the irreps structure of collective rotations
on two spins which can be summarized with a BPT as

\noindent\begin{minipage}[c]{1\columnwidth}%
\begin{center}
\vspace{0.5\baselineskip}
\begin{tabular}{|c|c|c|c|}
\cline{1-3} \cline{2-3} \cline{3-3} 
$1;1$ & $1;0$ & $1;-1$ & \multicolumn{1}{c}{}\tabularnewline
\hline 
\multicolumn{1}{c}{} & \multicolumn{1}{c}{} &  & $0;0$\tabularnewline
\cline{4-4} 
\end{tabular} .\vspace{0.5\baselineskip}
\par\end{center}%
\end{minipage} The two blocks here identify the two inequivalent irreps of $SU\left(2\right)$,
and each irrep is represented on a single invariant subspace, as per
the number of rows in each block.

If we add a third qubit, the Hilbert space will decompose under collective
rotations as
\[
\mathcal{H}=\underline{\frac{1}{2}}\otimes\underline{\frac{1}{2}}\otimes\underline{\frac{1}{2}}=\text{\ensuremath{\underline{\frac{3}{2}}}}\oplus\underline{\frac{1}{2}}\oplus\underline{\frac{1}{2}}
\]
with a single spin-$\frac{3}{2}$ subspace and two spin-$\frac{1}{2}$
subspaces. Given the basis of total spin $\ket{j;m}$, the BPT that
specifies the irreps structure is

\noindent\begin{minipage}[c]{1\columnwidth}%
\begin{center}
\vspace{0.5\baselineskip}
\begin{tabular}{cccc|c|c|}
\cline{1-4} \cline{2-4} \cline{3-4} \cline{4-4} 
\multicolumn{1}{|c|}{$\frac{3}{2};\frac{3}{2}$} & \multicolumn{1}{c|}{$\frac{3}{2};\frac{1}{2}$} & \multicolumn{1}{c|}{$\frac{3}{2};-\frac{1}{2}$} & $\frac{3}{2};-\frac{3}{2}$ & \multicolumn{1}{c}{} & \multicolumn{1}{c}{}\tabularnewline
\hline 
 &  &  &  & $\frac{1}{2};\frac{1}{2},1$ & $\frac{1}{2};-\frac{1}{2},1$\tabularnewline
\cline{5-6} \cline{6-6} 
 &  &  &  & $\frac{1}{2};\frac{1}{2},2$ & $\frac{1}{2};-\frac{1}{2},2$\tabularnewline
\cline{5-6} \cline{6-6} 
\end{tabular} .\vspace{0.5\baselineskip}
\par\end{center}%
\end{minipage} These two blocks identify the irreps of spin-$\frac{3}{2}$ and spin-$\frac{1}{2}$.
The second block has two rows since spin-$\frac{1}{2}$ is equivalently
represented on two separate invariant subspaces labeled with $i=1,2$.
The Wedderburn Decomposition implied by this BPT is
\[
\mathcal{H}=\text{\ensuremath{\underline{\frac{1}{2}}}}\otimes\underline{\frac{1}{2}}\otimes\underline{\frac{1}{2}}\cong\mathcal{H}_{\mu_{3/2}}\oplus\mathcal{H}_{\nu_{1/2}}\otimes\mathcal{H}_{\mu_{1/2}},
\]
where $\mathcal{H}_{\mu_{3/2}}$ and $\mathcal{H}_{\mu_{1/2}}$ are
the inequivalent irreps, and $\dim\mathcal{H}_{\nu_{1/2}}=2$ provides
the two-dimensional multiplicity to the spin-$\frac{1}{2}$ irrep.

Group representations are commonly used to identify the symmetries
of physically meaningful operators that commute with the group action.
The commutant algebra of a symmetry group representation is therefore
an interesting operator algebra that characterizes all the operators
that have that symmetry. The following theorem allows us to immediately
identify the commutant algebra from the BPT.
\begin{thm}
\label{thm:commutant is bpt transpose}Let $\left\{ \ket{e_{ik}^{q}}\right\} $
be the orthonormal basis forming the BPT of the operator algebra $\mathcal{A}$
supported on the whole Hilbert space $\mathcal{H}$. Then, the transposition
(interchanging rows with columns) of $\left\{ \ket{e_{ik}^{q}}\right\} $
produces the BPT of the commutant algebra
\[
\mathcal{A}'=\left\{ A'\in\mathcal{L}\left(\mathcal{H}\right)\,|\,\left[A',A\right]=0\,\,\,\,\forall A\in\mathcal{A}\right\} .
\]
\end{thm}
\begin{proof}
By construction \ref{eq: def of S_kl from BPT}, the minimal isometries
produced by the original and the transposed BPTs are 
\[
S_{kl}^{q}=\sum_{i}\ket{e_{ik}^{q}}\bra{e_{il}^{q}},\,\,\,\,\,\,\,\,\,\,\,\,\,\,\,\,\,\,\,\tilde{S}_{ij}^{q}=\sum_{k}\ket{e_{ik}^{q}}\bra{e_{jk}^{q}}.
\]
Since $\mathcal{A}$ is supported on the whole $\mathcal{H}$ we have
$I=\sum_{q,k}S_{kk}^{q}=\sum_{q,k}S_{kl}^{q}S_{lk}^{q}$ for any $l$.
By the definition of $\mathcal{A}'$, for every $A'\in\mathcal{A}'$
and $S_{kl}^{q}\in\mathcal{A}$ we have $S_{kl}^{q}A'=A'S_{kl}^{q}$,
and so 
\begin{align*}
A' & =A'\sum_{q,k}S_{kl}^{q}S_{lk}^{q}=\sum_{q,k}S_{kl}^{q}A'S_{lk}^{q}=\sum_{q,k}\sum_{i}\ket{e_{ik}^{q}}\bra{e_{il}^{q}}A'\sum_{j}\ket{e_{jl}^{q}}\bra{e_{jk}^{q}}\\
 & =\sum_{q,i,j}\bra{e_{il}^{q}}A'\ket{e_{jl}^{q}}\sum_{k}\ket{e_{ik}^{q}}\bra{e_{jk}^{q}}=\sum_{q,i,j}\bra{e_{il}^{q}}A'\ket{e_{jl}^{q}}\tilde{S}_{ij}^{q}.
\end{align*}
Therefore, $\mathcal{A}'\subseteq\spn\left\{ \tilde{S}_{ij}^{q}\right\} $.
By explicit multiplication we can see that $S_{kl}^{q'}\tilde{S}_{ij}^{q}=\tilde{S}_{ij}^{q}S_{kl}^{q'}=\delta_{qq'}\ket{e_{ik}^{q}}\bra{e_{jl}^{q}}$,
so $\left[\tilde{S}_{ij}^{q},S_{kl}^{q'}\right]=0$ and so $\spn\left\{ \tilde{S}_{ij}^{q}\right\} \subseteq\mathcal{A}'$.
Therefore, $\mathcal{A}'=\spn\left\{ \tilde{S}_{ij}^{q}\right\} $.
\end{proof}
In terms of the Wedderburn Decomposition, Theorem \ref{thm:commutant is bpt transpose}
tells us that
\[
\mathcal{A}\cong\bigoplus_{q}I_{\nu_{q}}\otimes\mathcal{L}\left(\mathcal{H}_{\mu_{q}}\right)\,\,\,\,\,\,\,\Leftrightarrow\,\,\,\,\,\,\mathcal{A}'\cong\bigoplus_{q}\mathcal{L}\left(\mathcal{H}_{\nu_{q}}\right)\otimes I_{\mu_{q}}.
\]
That is, operator algebras and their commutants have the same Wedderburn
Decomposition with the roles of $\mu_{q}$ and $\nu_{q}$ exchanged.
This theorem also trivially implies the following well known result.
\begin{cor}
(Bicommutant Theorem) Let $\mathcal{A}$ be an operator algebra supported
on the whole Hilbert space $\mathcal{H}$ and let $\mathcal{A}''$
be its bicommutant (commutant of a commutant) algebra. Then, $\mathcal{A}=\mathcal{A}''$.
\end{cor}
\begin{proof}
According to Theorem \ref{thm:commutant is bpt transpose}, the BPT
of $\mathcal{A}''$ is produced by transposing the BPT of $\mathcal{A}$
twice, which leaves it unchanged.
\end{proof}
In the example of three qubits, the commutant algebra of collective
rotations is then given by the BPT

\noindent\begin{minipage}[c]{1\columnwidth}%
\begin{center}
\vspace{0.5\baselineskip}
\begin{tabular}{|c|cc}
\cline{1-1} 
$\frac{3}{2};\frac{3}{2}$ &  & \tabularnewline
\cline{1-1} 
$\frac{3}{2};\frac{1}{2}$ &  & \tabularnewline
\cline{1-1} 
$\frac{3}{2};-\frac{1}{2}$ &  & \tabularnewline
\cline{1-1} 
$\frac{3}{2};-\frac{3}{2}$ &  & \tabularnewline
\hline 
\multicolumn{1}{c|}{} & \multicolumn{1}{c|}{$\frac{1}{2};\frac{1}{2},1$} & \multicolumn{1}{c|}{$\frac{1}{2};\frac{1}{2},2$}\tabularnewline
\cline{2-3} \cline{3-3} 
\multicolumn{1}{c|}{} & \multicolumn{1}{c|}{$\frac{1}{2};-\frac{1}{2},1$} & \multicolumn{1}{c|}{$\frac{1}{2};-\frac{1}{2},2$}\tabularnewline
\cline{2-3} \cline{3-3} 
\end{tabular} .\vspace{0.5\baselineskip}
\par\end{center}%
\end{minipage} Thus, all three-qubit operators that are symmetric under collective
rotations are spanned by the five partial isometries:

\noindent\begin{minipage}[c]{1\columnwidth}%
\begin{center}
\vspace{0.5\baselineskip}
\begin{tabular}{ll}
\multicolumn{2}{l}{$S^{\frac{3}{2}}=\ket{\frac{3}{2};\frac{3}{2}}\bra{\frac{3}{2};\frac{3}{2}}+\ket{\frac{3}{2};\frac{1}{2}}\bra{\frac{3}{2};\frac{1}{2}}+\ket{\frac{3}{2};-\frac{1}{2}}\bra{\frac{3}{2};-\frac{1}{2}}+\ket{\frac{3}{2};-\frac{3}{2}}\bra{\frac{3}{2};-\frac{3}{2}}$}\tabularnewline
 & \tabularnewline
$S_{11}^{\frac{1}{2}}=\ket{\frac{1}{2};\frac{1}{2},1}\bra{\frac{1}{2};\frac{1}{2},1}+\ket{\frac{1}{2};-\frac{1}{2},1}\bra{\frac{1}{2};-\frac{1}{2},1}$ & $S_{22}^{\frac{1}{2}}=\ket{\frac{1}{2};\frac{1}{2},2}\bra{\frac{1}{2};\frac{1}{2},2}+\ket{\frac{1}{2};-\frac{1}{2},2}\bra{\frac{1}{2};-\frac{1}{2},2}$\tabularnewline
$S_{12}^{\frac{1}{2}}=\ket{\frac{1}{2};\frac{1}{2},1}\bra{\frac{1}{2};\frac{1}{2},2}+\ket{\frac{1}{2};-\frac{1}{2},1}\bra{\frac{1}{2};-\frac{1}{2},2}$ & $S_{21}^{\frac{1}{2}}=\ket{\frac{1}{2};\frac{1}{2},2}\bra{\frac{1}{2};\frac{1}{2},1}+\ket{\frac{1}{2};-\frac{1}{2},2}\bra{\frac{1}{2};-\frac{1}{2},1}$.\tabularnewline
\end{tabular}\vspace{0.5\baselineskip}
\par\end{center}%
\end{minipage}

~\\
The above construction of commutants provides some indication that
there are benefit in using the BPT picture beyond the derivations
and proofs of this chapter. In the following chapters we will use
the BPT picture extensively. In Chapter \ref{chap:Identifying-the-irreps}
we will use it to describe the last step of the Scattering Algorithm
that finds the irreps structures of arbitrarily generated operator
algebras. In Chapters \ref{chap:Operational-reductions-of-dyn} and
\ref{chap:Operational-reductions-of-st} we will define the reductions
of states and Hamiltonians in terms of BPTs. In Chapter \ref{chap:Quantum-coarse-graining}
we will take advantage of the visual representation in terms of BPTs
to generalize state reductions beyond operator algebras. Thus, we
will see that BPTs can be a useful tool for specifying, manipulating
and producing tensor product structures, such as the structure of
irreducible representations.

\chapter{Finding the irreps structure with the Scattering Algorithm\label{chap:Identifying-the-irreps}}

In Theorem \ref{thm: Wedderburn Decomp} we have identified the general
irreps structure of operator algebras; unfortunately, it was not a
constructive result. We have learned that operator algebras can be
specified via a set of generators (see Definition \ref{def: generated OA})
but we do not know yet how to find the irreps structure of operator
algebras specified this way. In this chapter we will take a constructive
approach and address this problem.

Formally this problem can be stated as:
\begin{quote}
\emph{Given a finite set of self-adjoint operators $\mathcal{M}:=\left\{ M_{1},M_{2},...,M_{n}\right\} $
that generate the algebra $\left\langle \mathcal{M}\right\rangle $,
find the basis that identify the irreps structure of $\left\langle \mathcal{M}\right\rangle $
as promised by Theorem \ref{thm: Wedderburn Decomp}.}
\end{quote}
It can be equivalently formulated (but not solved) in a simpler mathematical
language:
\begin{quote}
\emph{Given a finite set of self-adjoint matrices $\mathcal{M}:=\left\{ M_{1},M_{2},...,M_{n}\right\} $,
find the basis in which all }$M_{i}\in\mathcal{M}$\emph{ are simultaneously
block-diagonal with the smallest possible blocks.}
\end{quote}
When $\mathcal{M}=\left\{ M\right\} $ is just one matrix, this means
find the basis that diagonalize $M$. When $\mathcal{M}$ is a set
of matrices that commute, this means find the basis that simultaneously
diagonalize all $M_{i}\in\mathcal{M}$. In general, for non-commuting
matrices, the basis that identify the irreps structure of $\left\langle \mathcal{M}\right\rangle $
are the basis that simultaneously block-diagonalize all $M_{i}\in\mathcal{M}$
with the smallest possible blocks. Therefore, we can think of this
problem as a problem of diagonalizing a set of matrices $\mathcal{M}$,
where not all matrices necessarily commute.

Solving this problem is essential for the practical applications that
require some form of reduction. More concretely, we would like to
be able to answer questions such as:
\begin{itemize}
\item If $\mathcal{M}\subset\mathcal{L}\left(\mathcal{H}\right)$ are terms
in a Hamiltonian, how can we restrict the dynamics to lower dimensional
subspaces where the Hamiltonian has a simpler form.
\item If $\mathcal{M}\subset\mathcal{L}\left(\mathcal{H}\right)$ is a subset
of observables, how can we reduce the dimension of the Hilbert space
while preserving all information about these observables?
\item If $\mathcal{M}\subset\mathcal{L}\left(\mathcal{H}\right)$ is a set
of error operators of a noisy quantum channel, how can we encode information
so it will not be affected by noise.
\end{itemize}
Just as we have a symbolic, not inherently numeric, algorithm for
diagonalizing matrices using pen and paper, our goal is to introduce
a symbolic algorithm for finding the irreps structure. The solution
we propose is called the Scattering Algorithm. The idea of this algorithm
was originally published in \citep{Kabernik2020Quantum}.

We are aware of two other approaches to this problem in the literature.
First, a numerical algorithm was proposed by Murota \emph{et al}.
\citep{Murota2010} in the context of semidefinite programming. A
key step in their algorithm involves sampling for a random matrix
in the algebra, which requires the ability to span the operator space
of the algebra. Second, in a more physical context, Holbrook \emph{et
al}. \citep{holbrook2003noiseless} proposed an algorithm for computing
the noise commutant of an error algebra associated with a noisy channel.
Similarly to what we intend to achieve here, they propose a symbolic
algorithm, however, this algorithm also requires the ability to span
the operator space of the algebra. Unlike these approaches, the Scattering
Algorithm does not require spanning the operator space of the algebra,
which is not a trivial task given only the generators $\mathcal{M}$.

In the following, Section \ref{sec:How-the-Scattering} is dedicated
to describing and demonstrating how the Scattering Algorithm works
without rigorous proofs. In Section \ref{sec:Why-the-Scattering}
we will go over the details with more rigor and prove the correctness
of the results.

\section{How the Scattering Algorithm works\label{sec:How-the-Scattering}}

\subsection{Overview}

The main idea behind the Scattering Algorithm is to take the spectral
projection of the generators and to break them down into minimal projections
from which the irreps structure is built. The whole process proceeds
in four steps:
\begin{enumerate}
\item Compile the initial set of projections from the spectral projections
of the generators.
\item Apply the rank-reducing operation called scattering\emph{ }on all
pairs of projections until no further reduction is possible.
\item Verify that all projections are minimal and the set is complete; fix
it if necessary.
\item Construct minimal isometries and then the BPT basis that identify
the irreps structure.
\end{enumerate}

\subsubsection{Phase 1}

The first phase of the algorithm is just the spectral decomposition
of all generators and the extraction of spectral projections on eigenspaces
with distinct eigenvalues. After this, the original generators are
left behind and their spectral projections move forward.

\subsubsection{Phase 2}

This phase is the heart of the algorithm where most of the calculations
take place. In this phase we will apply the scattering operation defined
as follows.
\begin{defn}
\emph{\label{def:Scattering }Scattering} is an operation that takes
a pair of projections and breaks each one into lower rank projections:
\[
\begin{array}{c}
\Pi_{1}\\
\\
\Pi_{2}
\end{array}\Diagram{fdA &  & fuA\\
 & f\\
fuA &  & fdA
}
\begin{array}{c}
\Pi_{1}^{\left(\lambda_{1}\right)},\,\Pi_{1}^{\left(\lambda_{2}\right)},...\,,\Pi_{1}^{\left(0\right)}\\
\\
\Pi_{2}^{\left(\lambda_{1}\right)},\,\Pi_{2}^{\left(\lambda_{2}\right)},...\,,\Pi_{2}^{\left(0\right)}.
\end{array}
\]
The lower rank projections are produced from the spectral decompositions
of 
\[
\Pi_{1}\Pi_{2}\Pi_{1}=\sum_{\lambda\neq0}\lambda\Pi_{1}^{\left(\lambda\right)}\textrm{\hspace{1cm}and\hspace{1cm}}\Pi_{2}\Pi_{1}\Pi_{2}=\sum_{\lambda\neq0}\lambda\Pi_{2}^{\left(\lambda\right)},
\]
with the addition of $\Pi_{i=1,2}^{\left(0\right)}:=\Pi_{i}-\sum_{\lambda\neq0}\Pi_{i}^{\left(\lambda\right)}$
called the \emph{null} \emph{projections}.
\end{defn}
Note that we do not yet assume that the spectrum $\left\{ \lambda\right\} $
is the same for both decompositions, however, we will later prove
that it is. Also note that the null projections are not the projections
on the kernel of $\Pi_{i}\Pi_{j}\Pi_{i}$ (the kernel projections
are given by $I-\sum_{\lambda\neq0}\Pi_{i}^{\left(\lambda\right)}$)
and it is possible that $\Pi_{i}^{\left(0\right)}=0$. We will treat
the null projections just as $\lambda=0$ elements of the set of spectral
projections $\left\{ \Pi_{i=1,2}^{\left(\lambda\right)}\right\} $.
The projections produced by scattering are therefore the set $\left\{ \Pi_{i=1,2}^{\left(\lambda\right)}\right\} $
of pairwise orthogonal projections that sum to their predecessor:
\[
\Pi_{i}=\Pi_{i}^{\left(\lambda_{1}\right)}+\Pi_{i}^{\left(\lambda_{2}\right)}+...+\Pi_{i}^{\left(0\right)}.
\]
Thus, in analogy with the scattering of particles, scattering of projections
breaks them into smaller (lower rank) constituents of the original
pair.

In Phase 2 of the algorithm we keep picking pairs of projections and
applying the scattering operation (after each scattering only the
resulting projections move forward) until all pairs have the following
property.
\begin{defn}
\label{def:Reflecting}A pair of projections $\Pi_{1}$, $\Pi_{2}$
is called \emph{reflecting }if both projections remain unbroken under
scattering, that is 
\begin{align}
\Pi_{1}\Pi_{2}\Pi_{1} & =\lambda\Pi_{1}\label{eq:def of reflecting 1}\\
\Pi_{2}\Pi_{1}\Pi_{2} & =\lambda\Pi_{2}.\label{eq:def of reflecting 2}
\end{align}
The coefficient $\lambda$ is then called a \emph{reflection} \emph{coefficient}
and we will say that $\Pi_{1}$, $\Pi_{2}$ are \emph{properly} \emph{reflecting
}if the reflection coefficient is not $0$ (i.e. they are not orthogonal
$\Pi_{1}\Pi_{2}\neq0$).
\end{defn}
We will later show that after one scattering, the resulting pairs
$\left\{ \Pi_{1}^{\left(\lambda\right)},\Pi_{2}^{\left(\lambda'\right)}\right\} $
are properly reflecting for all $\lambda=\lambda'$, and orthogonal
for $\lambda\neq\lambda'$. By repetitively applying the scattering
operation on non-reflecting pairs we are guaranteed to reach the state
where all pairs are reflecting. That is because scattering reduces
the ranks of projections (unless they are reflecting) and eventually
we will either reach all reflecting pairs with ranks higher than 1
or we will reach projections of rank 1, which are always reflecting.

The final output of Phase 2 of the algorithm is a graph of reflection
relations defined as follows.
\begin{defn}
\label{def:RefNetwork}A (proper)\emph{ reflection network }is a graph
$G=\left\{ V,E\right\} $ where the vertices set $V:=\left\{ \Pi_{\mathrm{v}}\right\} $
consists of pairwise reflecting projections and every properly reflecting
pair is connected with an edge (orthogonal pairs are unconnected)
\[
E:=\left\{ \left\{ \Pi_{\mathrm{v}},\Pi_{\mathrm{u}}\right\} \subset V\,\left|\begin{array}{c}
\Pi_{\mathrm{v}}\Pi_{\mathrm{u}}\Pi_{\mathrm{v}}=\lambda\Pi_{\mathrm{v}}\\
\Pi_{\mathrm{u}}\Pi_{\mathrm{v}}\Pi_{u}=\lambda\Pi_{\mathrm{u}}
\end{array}\right.,\,\lambda\neq0\right\} .
\]
An \emph{improper} reflection network\emph{ }is the generalization
where not all projections are known to be reflecting so there are
two kinds of edges: black edges for properly reflecting pairs and
red edges for unknown relations.
\end{defn}
Note that according to this definition only known orthogonal projections
are not connected by any edge. In general, reflection networks may
have multiple connected components formed by subsets of projections
that are orthogonal to every projection outside the subset. It does
not mean, however, that projections in the same connected component
cannot be orthogonal; as long as there is a sequence of proper reflection
(or unknown) relations connecting the projections, they will be in
the same connected component.

With the above definition we can say that Phase 2 begins with an improper
reflection notwork of projections produced in Phase 1. Then, as we
keep applying the scattering operation, the reflection network evolves
until it becomes a proper reflection network. The proper reflection
network is what proceeds to the next phase.

\subsubsection{Phase 3}

In order to construct the irreps structure we have to establish that
the reflection network has the following two properties:
\begin{enumerate}
\item \emph{(minimality}) All projections in the reflection network are
minimal projections.
\item \emph{(completeness}) The reflection network has a maximal set of
minimal projections
\end{enumerate}
Phase 3 is where we establish that the reflection network produced
in Phase 2 is indeed minimal and complete.

Although minimality is not guaranteed to hold for a proper reflection
network, in practice, purely on empirical grounds, reflection networks
produced in Phase 2 tend to always be minimal. Nevertheless, in the
next section we will show how to check if this property holds and
how to fix it if it does not.

Completeness is a rather trivial property that is guaranteed if any
of the initial generators is supported on the whole Hilbert space.
This can be arranged by adding the identity to the set of generators.
When adding the identity is not feasible we will show in the next
section how to complete the reflection network to have a maximal set.

\subsubsection{Phase 4}

In the last phase we take the proper reflection network that is minimal
and complete and construct the BPT basis that specify the irreps structure.
Each connected component in the reflection network will correspond
to a block in the BPT. We construct the BPT basis by following the
steps:
\begin{enumerate}
\item Select a maximal set of minimal projections $\left\{ \Pi_{k}^{q}\right\} _{k=1,...}$
in the connected component $q$.
\item For each $\Pi_{k}^{q}$ in the connected component $q$, take any
path form $\Pi_{1}^{q}$ to $\Pi_{k}^{q}$ and construct the minimal
isometry $S_{k1}^{q}\propto\Pi_{k}^{q}\cdots\Pi_{1}^{q}$ by taking
the product of projections along the path (the proportionality coefficient
is fixed after construction).
\item Use the minimal isometries $S_{k1}^{q}$ to construct the BPT basis
as described in the proof of Lemma \ref{lem: for all S_kl there is a BPT}.
\end{enumerate}

\subsection{Illustrative example}

In order to see how the Scattering Algorithm works we consider the
Hilbert space of three qubits $\mathcal{H}=\mathcal{H}_{qubit}^{\otimes3}$
and study a peculiar Hamiltonian whose choice is mainly motivated
by the fact that it presents a non-trivial problem in a relatively
simple setting.

The Hamiltonian we consider consists of two terms $H\left(\epsilon\right)=H_{int}+\epsilon H_{z_{1}}$.
Using the notation $\ket{\pm}=\left(\ket 0\pm\ket 1\right)/\sqrt{2}$,
the term $H_{int}$ is some interaction such that both $\ket{++0}$
and $\ket{+11}$ are the first (and only) excited states, and $\ket{-00}$,$\ket{-01}$,$\ket{-10}$,$\ket{-11}$,
$\ket{+01}$,$\ket{+-0}$ are the ground states. The second term is
$H_{z_{1}}=\sigma_{z}\otimes I_{23}$, where $\sigma_{z}$ is a Pauli
matrix acting on the first qubit so $\ket{000}$,$\ket{001}$,$\ket{010}$,$\ket{011}$
are the excited states. The excitation energy gap of $H_{int}$ is
normalized to $1$, while $\epsilon$ is a free parameter that controls
the gap of $H_{z_{1}}$. We would like to find out the spectrum and
the eigenstates of $H\left(\epsilon\right)$ as a function of $\epsilon$.

Since $H_{int}$ and $H_{z_{1}}$ do not commute, we cannot simultaneously
diagonalize them. If $\epsilon$ is small we could use perturbation
theory, but we do not want to assume that. What we can do instead
is observe that for all $\epsilon$, $H\left(\epsilon\right)$ is
an element of the operator algebra generated by $H_{int}$ and $H_{z_{1}}$.
Thus, with respect to the Wedderburn Decomposition of this algebra,
$H\left(\epsilon\right)$ may have a much simpler form. Another way
to say it is this: although $H_{int}$ and $H_{z_{1}}$ cannot be
simultaneously diagonalized, they can be simultaneously block-diagonalized.
Then, if the blocks are small and/or repetitive, the spectrum of $H\left(\epsilon\right)$
can be easier to analyze.

We will therefore find the irreps structure of the algebra $\left\langle H_{int},H_{z_{1}}\right\rangle $
which amounts to finding its BPT basis.

\subsubsection{Phase 1}

Recall that the energy gap of $H_{int}$ is $1$ and we can shift
the whole spectrum so that its ground energy is $0$. Then, this Hamiltonian
term is just a projection $H_{int}=\Pi_{int}$ on its exited states
\[
\Pi_{int}:=\ket{++0}\bra{++0}+\ket{+11}\bra{+11}.
\]
The second Hamiltonian term consists of two spectral projections $H_{z_{1}}=\Pi_{z_{1};0}-\Pi_{z_{1};1}$
where 
\[
\Pi_{z_{1};0}:=\ket 0\bra 0\otimes I_{23}\,\,\,\,\,\,\,\,\,\,\,\,\,\,\,\,\,\,\,\,\,\,\,\,\,\,\Pi_{z_{1};1}:=\ket 1\bra 1\otimes I_{23}.
\]
Overall, we compile the three spectral projections $\left\{ \Pi_{int},\Pi_{z_{1};0},\Pi_{z_{1};1}\right\} $.\footnote{We could shift the spectrum again and drop the second projection $\Pi_{z_{1};1}$
but then none of the generators will be supported on the whole Hilbert
space. This will result in an incomplete reflection network, which
is easy to fix, but there is no reason to deliberately create this
complication.}

\subsubsection{Phase 2}

The initial (improper) reflection network is shown Fig. \ref{fig:exmprefnet1}
where the red edges indicate unknown relations and the absent edge
between $\Pi_{z_{1};0}$ and $\Pi_{z_{1};1}$ indicates our prior
knowledge that they are orthogonal.

\begin{figure}[H]
\begin{centering}
\includegraphics[viewport=0bp 250bp 737bp 453bp,clip,width=0.6\columnwidth]{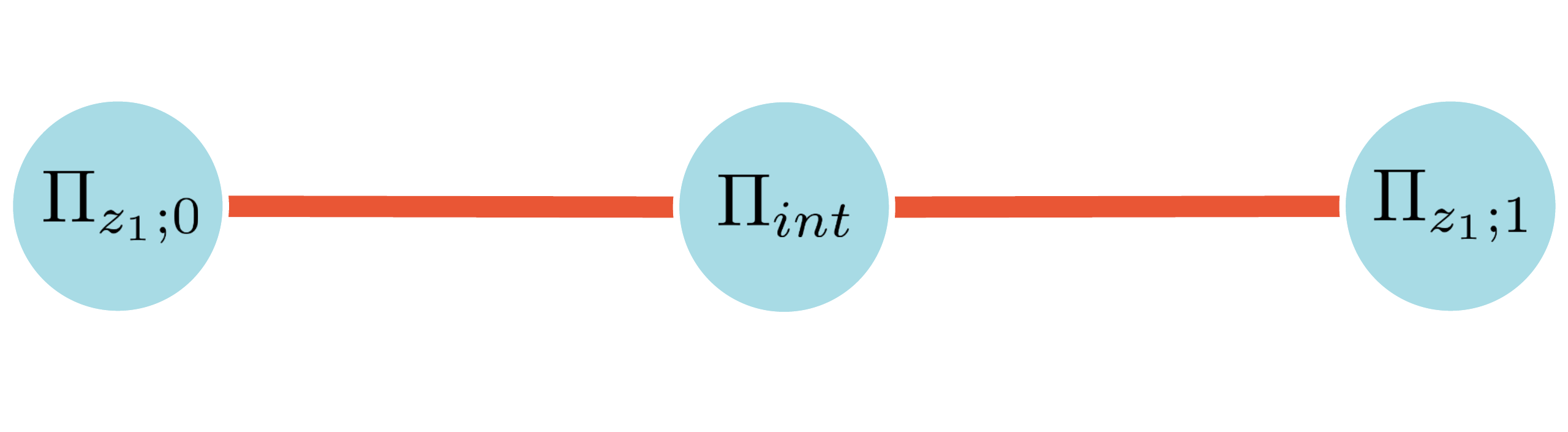}
\par\end{centering}
\centering{}\caption{\label{fig:exmprefnet1}An example of the initial (improper) reflection
network.}
\end{figure}

For the first scattering we pick any pair connected by a red edge,
say $\left\{ \Pi_{z_{1};0},\Pi_{int}\right\} $. For the scattering
calculation it is convenient to first calculate the product 
\[
\Pi_{z_{1};0}\Pi_{int}=\frac{1}{\sqrt{2}}\ket{0+0}\bra{++0}+\frac{1}{\sqrt{2}}\ket{011}\bra{+11},
\]
and then it is easy to get the scattering result for both projections
\begin{align*}
\Pi_{int}\Pi_{z_{1};0}\Pi_{int} & =\frac{1}{2}\ket{++0}\bra{++0}+\frac{1}{2}\ket{+11}\bra{+11}=\frac{1}{2}\Pi_{int}\\
\Pi_{z_{1};0}\Pi_{int}\Pi_{z_{1};0} & =\frac{1}{2}\ket{0+0}\bra{0+0}+\frac{1}{2}\ket{011}\bra{011}=:\frac{1}{2}\Pi_{z_{1};0}^{\left(1/2\right)}.
\end{align*}
After scattering, $\Pi_{int}$ remains unbroken and $\Pi_{z_{1};0}$
breaks into $\Pi_{z_{1};0}^{\left(1/2\right)}$ and the null projection
\[
\Pi_{z_{1};0}^{\left(0\right)}:=\Pi_{z_{1};0}-\Pi_{z_{1};0}^{\left(1/2\right)}=\ket{0-0}\bra{0-0}+\ket{001}\bra{001}.
\]
At this point, one can explicitly verify that $\Pi_{int}$ is reflecting
with $\Pi_{z_{1};0}^{\left(1/2\right)}$ and orthogonal to $\Pi_{z_{1};0}^{\left(0\right)}$
(in fact, this verification is unnecessary since this is a general
property of projections produced by scattering that we will prove
in Theorem \ref{thm: scattering of projections }). We also know that
$\Pi_{z_{1};1}$ is orthogonal to both $\Pi_{z_{1};0}^{\left(1/2\right)}$
and $\Pi_{z_{1};0}^{\left(0\right)}$ since it was orthogonal to their
predecessor. The updated reflection network is shown in Fig. \ref{fig:exmprefnet2}.

\begin{figure}[H]
\begin{centering}
\includegraphics[viewport=0bp 100bp 737bp 453bp,clip,width=0.6\columnwidth]{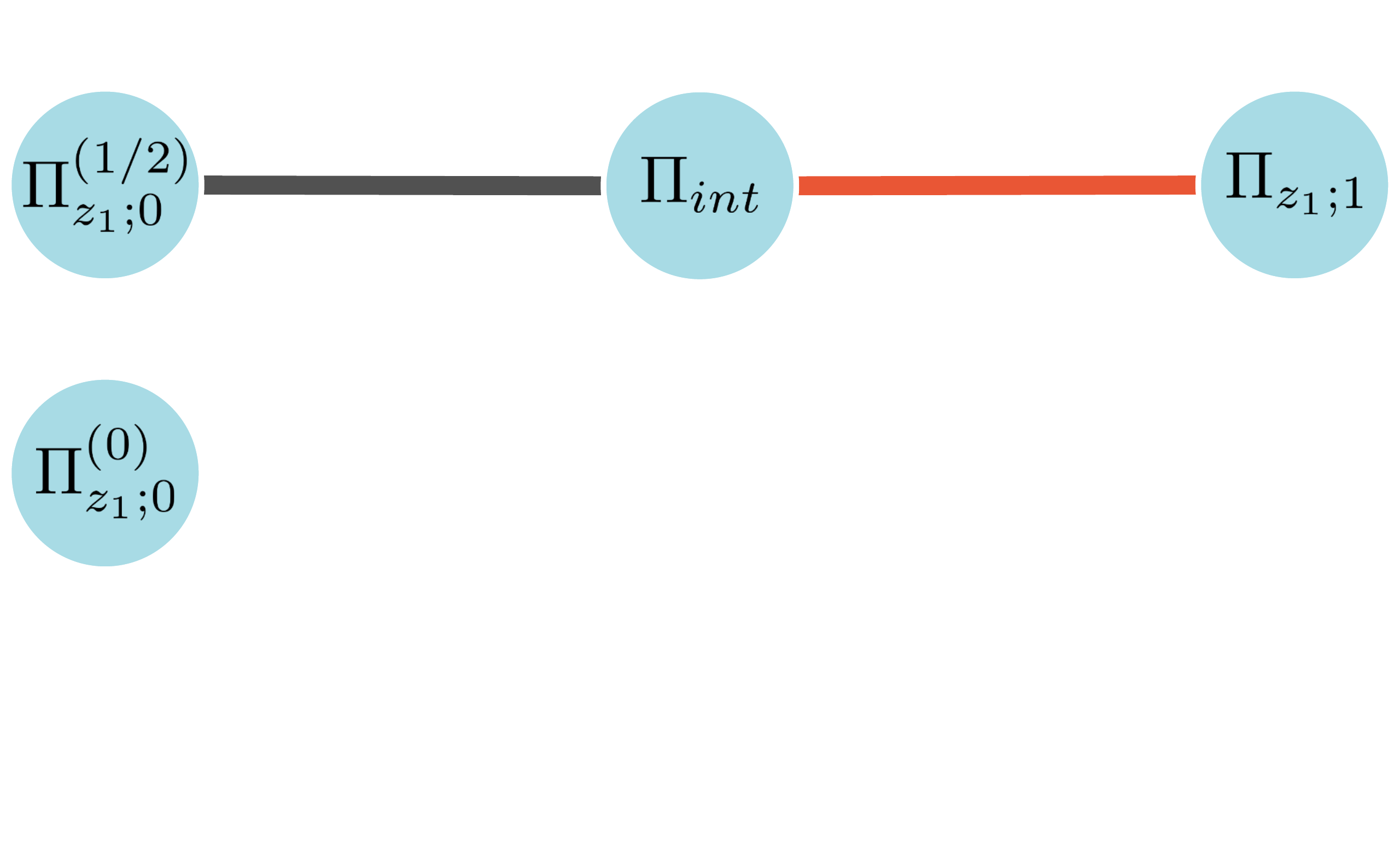}
\par\end{centering}
\centering{}\caption{\label{fig:exmprefnet2}An example of the intermediate reflection
network after one scattering.}
\end{figure}

The only remaining red edge is between the pair $\left\{ \Pi_{z_{1};1},\Pi_{int}\right\} $
which after scattering similarly yields
\begin{align*}
\Pi_{int}\Pi_{z_{1};1}\Pi_{int} & =\frac{1}{2}\ket{++0}\bra{++0}+\frac{1}{2}\ket{+11}\bra{+11}=\frac{1}{2}\Pi_{int}\\
\Pi_{z_{1};1}\Pi_{int}\Pi_{z_{1};1} & =\frac{1}{2}\ket{1+0}\bra{1+0}+\frac{1}{2}\ket{111}\bra{111}=:\frac{1}{2}\Pi_{z_{1};1}^{\left(1/2\right)}.
\end{align*}
Again, $\Pi_{int}$ remains unbroken and $\Pi_{z_{1};1}$ breaks into
$\Pi_{z_{1};1}^{\left(1/2\right)}$ and the null projection
\[
\Pi_{z_{1};1}^{\left(0\right)}:=\Pi_{z_{1};1}-\Pi_{z_{1};1}^{\left(1/2\right)}=\ket{1-0}\bra{1-0}+\ket{101}\bra{101}.
\]
The final and proper reflection network is shown in Fig. \ref{fig:exmprefnet3}.

\begin{figure}[H]
\centering{}\includegraphics[viewport=0bp 100bp 737bp 453bp,clip,width=0.6\columnwidth]{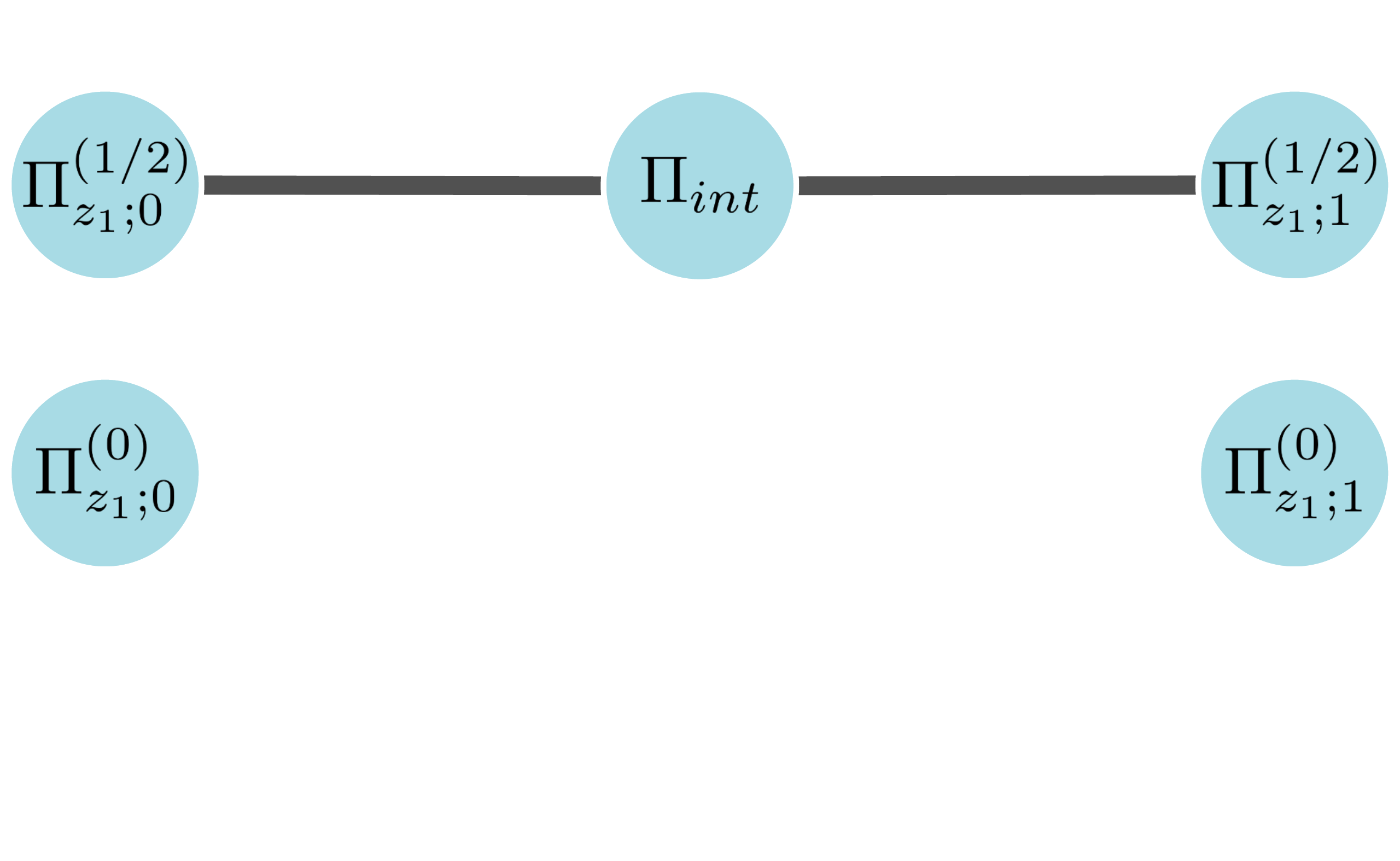}\caption{\label{fig:exmprefnet3}An example of the final (proper) reflection
network after two scatterings.}
\end{figure}

\subsubsection{Phase 3}

Since we had not discussed yet how to check minimality and completeness
of the reflection network, we will just assert that these properties
hold.

\subsubsection{Phase 4}

In the final reflection network in Fig. \ref{fig:exmprefnet3} we
identify three connected components that will correspond to three
blocks in the BPT. The two connected components with a single projection
have a single column given by the eigenspace of the projection. Since
we are free to choose the basis for the first column in each block,
we will stick with $\ket{0-0}$, $\ket{001}$ for $\Pi_{z_{1};0}^{\left(0\right)}$,
and with $\ket{1-0}$, $\ket{101}$ for $\Pi_{z_{1};1}^{\left(0\right)}$.

For the last block we need to choose a maximal subset of minimal projections
whose eigenspaces will correspond to the columns. In this case it
can only be $\left\{ \Pi_{z_{1};0}^{\left(1/2\right)},\Pi_{z_{1};1}^{\left(1/2\right)}\right\} $
and we pick $\Pi_{z_{1};0}^{\left(1/2\right)}$ to be the first column.
Again, we naturally choose the first column basis to be the eigenstates
$\ket{0+0}$, $\ket{011}$ of $\Pi_{z_{1};0}^{\left(1/2\right)}$.
For the second column in this block we cannot freely choose the basis.
Instead, we get the second column basis by mapping the first column
with the minimal isometry 
\[
S_{01}^{\left(1/2\right)}\propto\Pi_{z_{1};1}^{\left(1/2\right)}\Pi_{int}\Pi_{z_{1};0}^{\left(1/2\right)}=\frac{1}{2}\ket{1+0}\bra{0+0}+\frac{1}{2}\ket{111}\bra{011}
\]
constructed by taking the product of projections along the connecting
path in the reflection network (the factors $\frac{1}{2}$ are removed
by normalization). Thus, the second column also consists of the familiar
basis $\ket{1+0}$, $\ket{111}$ but we could not know that a priori.

The final result of the Scattering Algorithm is summarized by the
BPT

\noindent\begin{minipage}[c]{1\columnwidth}%
\begin{center}
\vspace{0.5\baselineskip}
\begin{tabular}{c|c|cc}
\cline{1-1} 
\multicolumn{1}{|c|}{$0-0$} & \multicolumn{1}{c}{} &  & \tabularnewline
\cline{1-1} 
\multicolumn{1}{|c|}{$001$} & \multicolumn{1}{c}{} &  & \tabularnewline
\cline{1-2} \cline{2-2} 
 & $1-0$ &  & \tabularnewline
\cline{2-2} 
 & $101$ &  & \tabularnewline
\cline{2-4} \cline{3-4} \cline{4-4} 
\multicolumn{1}{c}{} &  & \multicolumn{1}{c|}{$0+0$} & \multicolumn{1}{c|}{$1+0$}\tabularnewline
\cline{3-4} \cline{4-4} 
\multicolumn{1}{c}{} &  & \multicolumn{1}{c|}{$011$} & \multicolumn{1}{c|}{$111$}\tabularnewline
\cline{3-4} \cline{4-4} 
\end{tabular} .\vspace{0.5\baselineskip}
\par\end{center}%
\end{minipage} \\

Returning to our original question, the Wedderburn Decomposition given
by the above BPT is 
\[
\mathcal{H}_{qubit}^{\otimes3}\cong\mathcal{H}_{\nu_{a}}\oplus\mathcal{H}_{\nu_{b}}\oplus\mathcal{H}_{\nu_{c}}\otimes\mathcal{H}_{\mu_{c}}
\]
where we have labeled the three blocks as $a,b,c$. Since $H\left(\epsilon\right)\in\left\langle H_{int},H_{z_{1}}\right\rangle $,
for all $\epsilon$ this Hamiltonian must have the block-diagonal
form
\[
H\left(\epsilon\right)=I_{\nu_{a}}\alpha\left(\epsilon\right)\oplus I_{\nu_{b}}\beta\left(\epsilon\right)\oplus I_{\nu_{c}}\otimes H_{\mu_{c}}\left(\epsilon\right)=\begin{pmatrix}\alpha\left(\epsilon\right)\\
 & \alpha\left(\epsilon\right)\\
 &  & \beta\left(\epsilon\right)\\
 &  &  & \beta\left(\epsilon\right)\\
 &  &  &  & H_{\mu_{c}}\left(\epsilon\right)\\
 &  &  &  &  & H_{\mu_{c}}\left(\epsilon\right)
\end{pmatrix}
\]
where $\alpha\left(\epsilon\right)$, $\beta\left(\epsilon\right)$
are $\epsilon$-dependent scalars and $H_{\mu_{c}}\left(\epsilon\right)$
is an $\epsilon$-dependent $2\times2$ matrix.

We can calculate these scalars and matrix elements using the original
definition 
\[
H\left(\epsilon\right)=H_{int}+\epsilon H_{z_{1}}=\ket{++0}\bra{++0}+\ket{+11}\bra{+11}+\epsilon\sigma_{z}\otimes I_{23}.
\]
Since all rows in the same block of the BPT are identical representations
of $H\left(\epsilon\right)$, we only need to calculate the matrix
elements for a single row in each BPT block:

\[
\alpha\left(\epsilon\right)=\bra{001}H\left(\epsilon\right)\ket{001}=\epsilon\,\,\,\,\,\,\,\,\,\,\,\,\,\,\,\,\,\,\,\,\,\,\,\,\beta\left(\epsilon\right)=\bra{101}H\left(\epsilon\right)\ket{101}=-\epsilon
\]
\[
H_{\mu_{c}}\left(\epsilon\right)=\begin{pmatrix}\bra{011}H\left(\epsilon\right)\ket{011} & \bra{011}H\left(\epsilon\right)\ket{111}\\
\bra{111}H\left(\epsilon\right)\ket{011} & \bra{111}H\left(\epsilon\right)\ket{111}
\end{pmatrix}=\begin{pmatrix}\frac{1}{2}+\epsilon & \frac{1}{2}\\
\frac{1}{2} & \frac{1}{2}-\epsilon
\end{pmatrix}.
\]
Therefore, using the basis arranged in the BPT (reading the BPT top
to bottom, left to right)
\[
\left\{ \ket{0-0},\ket{001},\ket{1-0},\ket{101},\ket{0+0},\ket{1+0},\ket{011},\ket{111}\right\} ,
\]
results in the block-diagonal matrix representation of this Hamiltonian
\[
H\left(\epsilon\right)=\begin{pmatrix}\epsilon\\
 & \epsilon\\
 &  & -\epsilon\\
 &  &  & -\epsilon\\
 &  &  &  & \frac{1}{2}+\epsilon & \frac{1}{2}\\
 &  &  &  & \frac{1}{2} & \frac{1}{2}-\epsilon\\
 &  &  &  &  &  & \frac{1}{2}+\epsilon & \frac{1}{2}\\
 &  &  &  &  &  & \frac{1}{2} & \frac{1}{2}-\epsilon
\end{pmatrix}.
\]

The states $\left\{ \ket{0-0},\ket{001}\right\} $ and $\left\{ \ket{1-0},\ket{101}\right\} $
are clearly the eigenvectors with the eigenvalues $\epsilon$, $-\epsilon$
respectively. The $2\times2$ matrix block $H_{\mu_{c}}\left(\epsilon\right)$
can be decomposed into Pauli matrices
\[
H_{\mu_{c}}\left(\epsilon\right)=\frac{1}{2}\sigma_{x}+\epsilon\sigma_{z}+\frac{1}{2}I
\]
and we can disregard the identity as it only generates a phase factor.
Then we can see that $H_{\mu_{c}}\left(\epsilon\right)$ is just the
Hamiltonian of a single spin in transverse fields. The ``up'' and
``down'' states of this spin are $\left\{ \ket{0+0},\ket{1+0}\right\} $
for one matrix block and $\left\{ \ket{011},\ket{111}\right\} $ for
the other. Thus, the whole task reduces to analyzing a single spin
in transverse fields, which is a significant simplification of the
original problem.

\section{The Scattering Algorithm in detail\label{sec:Why-the-Scattering}}

With the above overview and example we are in a good position to formally
go over the details of the Scattering Algorithm and prove the correctness
of the solution that it finds.

The input of this algorithm is a finite set of self-adjoint matrices
$\mathcal{M}\subseteq\mathcal{L}\left(\mathcal{H}\right)$ that generate
the algebra $\mathcal{A}:=\left\langle \mathcal{M}\right\rangle $.
The output is a set of BPT basis $\left\{ \ket{e_{ik}^{q}}\right\} $
where the indices $q$ specify the distinct irreps, $i$ specify the
multiple instances of identical irreps, and $k$ specify the distinct
basis elements inside each irrep. The BPT basis specify the irreps
structure (Wedderburn Decomposition) of $\mathcal{A}$ as described
in the proof of Theorem \ref{thm: Wedderburn Decomp}.

The main procedure of the algorithm is as follows:

\noindent 
\begin{algorithm}[H]

\begin{algorithmic}[1] 

\Procedure{FindIrrepStructure}{$\mathcal{M}$}     
 
\State $SpecProjections \gets$  \Call{GetAllSpectralProjections}{$\mathcal{M}$}

\State $ReflectNet \gets$ \Call{ScatterProjections}{$SpecProjections$}

\State $ReflectNet \gets$ \Call{EstablishMinimality}{$ReflectNet$}

\State $ReflectNet \gets$ \Call{EstablishCompleteness}{$ReflectNet$}

\State $BptBasis \gets$ \Call{ConstructBptBasis}{$ReflectNet$}

\State \Return $BptBasis$
\EndProcedure

\end{algorithmic}

\caption{\label{alg:The-Scattering-Algorithm}The Scattering Algorithm}
\end{algorithm}

\noindent We will now go over the details of each procedure (except
the trivial first step of getting all the spectral projections from
the generators) and prove the accompanying facts. In Section \ref{subsec:Why-the-Scattering Alg works}
we will prove the correctness of the whole algorithm.

\subsection{Scattering of projections}

Following the Definitions \ref{def:Scattering } and \ref{def:Reflecting}
of scattering and reflecting projections, we will now prove a few
useful facts.

The most important fact about the scattering operation is that regardless
of what the initial projections $\Pi_{1}$, $\Pi_{2}$ are, the resulting
projections are always a series of reflecting pairs $\Pi_{1}^{\left(\lambda\right)},\Pi_{2}^{\left(\lambda\right)}$
with reflection coefficients $\lambda$, and every pair is orthogonal
to any other pair.
\begin{thm}
\label{thm: scattering of projections }Let $\Pi_{1}$, $\Pi_{2}$
be a pair of projections before scattering and let $\left\{ \Pi_{1}^{\left(\lambda\right)}\right\} $,
$\left\{ \Pi_{2}^{\left(\lambda\right)}\right\} $ be the resulting
projections after scattering. Then:

\noindent (1) The non-zero eigenvalues $\left\{ \lambda\right\} $
are the same for both $\Pi_{1}\Pi_{2}\Pi_{1}$ and $\Pi_{2}\Pi_{1}\Pi_{2}$.

\noindent (2) For all $\lambda\neq\lambda'$ the pairs of projections
$\Pi_{1}^{\left(\lambda\right)}$, $\Pi_{2}^{\left(\lambda'\right)}$
are orthogonal.

\noindent (3) For all $\lambda=\lambda'$ the pairs of projections
$\Pi_{1}^{\left(\lambda\right)}$, $\Pi_{2}^{\left(\lambda\right)}$
are reflecting with reflection coefficient $\lambda$.
\end{thm}
\begin{proof}
We will assume that $\left\{ \lambda\right\} $ are the eigenvalues
of $\Pi_{1}\Pi_{2}\Pi_{1}$ while the eigenvalues of $\Pi_{2}\Pi_{1}\Pi_{2}$
are unknown. Since all $\left\{ \Pi_{1}^{\left(\lambda\right)}\right\} $
are pairwise orthogonal and sum to $\Pi_{1}$, we have $\Pi_{1}^{\left(\lambda\right)}\Pi_{1}=\Pi_{1}\Pi_{1}^{\left(\lambda\right)}=\Pi_{1}^{\left(\lambda\right)}$
for all $\lambda$. Then, if we multiply the definition of scattering
$\Pi_{1}\Pi_{2}\Pi_{1}=\sum_{\lambda''\neq0}\lambda''\Pi_{1}^{\left(\lambda''\right)}$
from left and right with $\Pi_{1}^{\left(\lambda\right)}$ and $\Pi_{1}^{\left(\lambda'\right)}$,
we get the identity

\begin{equation}
\Pi_{1}^{\left(\lambda\right)}\Pi_{2}\Pi_{1}^{\left(\lambda'\right)}=\delta_{\lambda\lambda'}\lambda\Pi_{1}^{\left(\lambda\right)}.\label{eq:scat of proj proof identity 1}
\end{equation}
This equation holds for all $\lambda$ including $\lambda=0$ regardless
of whether $\Pi_{1}^{\left(0\right)}=0$ or not. In particular 
\[
\left(\Pi_{1}^{\left(0\right)}\Pi_{2}\right)\left(\Pi_{1}^{\left(0\right)}\Pi_{2}\right)^{\dagger}=\Pi_{1}^{\left(0\right)}\Pi_{2}\Pi_{1}^{\left(0\right)}=0
\]
so $\Pi_{1}^{\left(0\right)}\Pi_{2}=\Pi_{2}\Pi_{1}^{\left(0\right)}=0$.
Therefore, 
\begin{align}
\Pi_{2}\Pi_{1}\Pi_{2} & =\Pi_{2}\left(\Pi_{1}-\Pi_{1}^{\left(0\right)}\right)\Pi_{2}=\Pi_{2}\left(\sum_{\lambda\neq0}\Pi_{1}^{\left(\lambda\right)}\right)\Pi_{2}=\sum_{\lambda\neq0}\lambda\left(\frac{1}{\lambda}\Pi_{2}\Pi_{1}^{\left(\lambda\right)}\Pi_{2}\right)\nonumber \\
 & =\sum_{\lambda\neq0}\lambda\tilde{\Pi}_{2}^{\left(\lambda\right)},\label{eq: scat of proj proof identity 2}
\end{align}
where the last step suggests the definition $\tilde{\Pi}_{2}^{\left(\lambda\right)}:=\frac{1}{\lambda}\Pi_{2}\Pi_{1}^{\left(\lambda\right)}\Pi_{2}$.
The operators $\tilde{\Pi}_{2}^{\left(\lambda\right)}$ are clearly
self-adjoint and, using Eq. \eqref{eq:scat of proj proof identity 1},
we have 
\[
\tilde{\Pi}_{2}^{\left(\lambda\right)}\tilde{\Pi}_{2}^{\left(\lambda'\right)}=\frac{1}{\lambda\lambda'}\Pi_{2}\Pi_{1}^{\left(\lambda\right)}\Pi_{2}\Pi_{1}^{\left(\lambda'\right)}\Pi_{2}=\delta_{\lambda\lambda'}\frac{1}{\lambda}\Pi_{2}\Pi_{1}^{\left(\lambda\right)}\Pi_{2}=\delta_{\lambda\lambda'}\tilde{\Pi}_{2}^{\left(\lambda\right)}.
\]
Therefore, the operators $\left\{ \tilde{\Pi}_{2}^{\left(\lambda\right)}\right\} $
form a set of pairwise orthogonal projections. In that case, Eq. \eqref{eq: scat of proj proof identity 2}
is the spectral decomposition of $\Pi_{2}\Pi_{1}\Pi_{2}$. Since the
spectral decomposition is unique we must have $\tilde{\Pi}_{2}^{\left(\lambda\right)}=\Pi_{2}^{\left(\lambda\right)}$
for all $\lambda\neq0$ and so the non-zero eigenvalues are the same
for both $\Pi_{1}\Pi_{2}\Pi_{1}$ and $\Pi_{2}\Pi_{1}\Pi_{2}$. This
proves claim 1 and produces the identity
\begin{equation}
\Pi_{2}^{\left(\lambda\right)}=\frac{1}{\lambda}\Pi_{2}\Pi_{1}^{\left(\lambda\right)}\Pi_{2}.\label{eq:scat of proj proof identity 3}
\end{equation}
Using the identities \eqref{eq:scat of proj proof identity 3} and
\eqref{eq:scat of proj proof identity 1}, we get another identity
\[
\Pi_{1}^{\left(\lambda\right)}\Pi_{2}^{\left(\lambda'\right)}=\frac{1}{\lambda'}\Pi_{1}^{\left(\lambda\right)}\Pi_{2}\Pi_{1}^{\left(\lambda'\right)}\Pi_{2}=\delta_{\lambda\lambda'}\Pi_{1}^{\left(\lambda\right)}\Pi_{2},
\]
which proves claim 2. In particular, for $\lambda=\lambda'$, we can
multiply the last identity with its own adjoint from both sides
\begin{align*}
\Pi_{1}^{\left(\lambda\right)}\Pi_{2}^{\left(\lambda\right)}\Pi_{1}^{\left(\lambda\right)} & =\Pi_{1}^{\left(\lambda\right)}\Pi_{2}\Pi_{1}^{\left(\lambda\right)}\\
\Pi_{2}^{\left(\lambda\right)}\Pi_{1}^{\left(\lambda\right)}\Pi_{2}^{\left(\lambda\right)} & =\Pi_{2}\Pi_{1}^{\left(\lambda\right)}\Pi_{2}.
\end{align*}
Then, using the identity \eqref{eq:scat of proj proof identity 1}
in the first line, and the identity \eqref{eq:scat of proj proof identity 3}
in the second, we get
\begin{align*}
\Pi_{1}^{\left(\lambda\right)}\Pi_{2}^{\left(\lambda\right)}\Pi_{1}^{\left(\lambda\right)} & =\lambda\Pi_{1}^{\left(\lambda\right)}\\
\Pi_{2}^{\left(\lambda\right)}\Pi_{1}^{\left(\lambda\right)}\Pi_{2}^{\left(\lambda\right)} & =\lambda\Pi_{2}^{\left(\lambda\right)},
\end{align*}
which proves claim 3.
\end{proof}
Note that Eq. \eqref{eq:scat of proj proof identity 3} tells us how
to calculate the projections $\left\{ \Pi_{2}^{\left(\lambda\neq0\right)}\right\} $
if we know $\left\{ \Pi_{1}^{\left(\lambda\neq0\right)}\right\} $.
That is, we only need to calculate one spectral decomposition of $\Pi_{1}\Pi_{2}\Pi_{1}$,
and then get the spectral decomposition of $\Pi_{2}\Pi_{1}\Pi_{2}$
for free. In practice, it is often easier to get the spectral projections
of both $\Pi_{1}\Pi_{2}\Pi_{1}$ and $\Pi_{2}\Pi_{1}\Pi_{2}$ from
the left and right singular vectors of $\Pi_{1}\Pi_{2}$.

Another useful fact that we will need is:
\begin{prop}
\label{prop:reflecting proj}Let $\Pi_{1}$, $\Pi_{2}$ be a pair
of properly reflecting projections with the reflection coefficient
$\lambda\neq0$, then, $\Pi_{1}$ and $\Pi_{2}$ have the same rank.
If in addition $\lambda=1$ then $\Pi_{1}=\Pi_{2}$.
\end{prop}
\begin{proof}
If we take the trace on both sides of Eqs. \eqref{eq:def of reflecting 1},
\eqref{eq:def of reflecting 2} we get 
\[
\tr\left(\Pi_{1}\Pi_{2}\right)=\lambda\tr\left(\Pi_{1}\right),\,\,\,\,\,\,\,\,\,\,\,\,\,\,\,\,\,\,\,\,\,\,\,\,\,\,\,\,\,\,\,\tr\left(\Pi_{1}\Pi_{2}\right)=\lambda\tr\left(\Pi_{2}\right).
\]
Since $\lambda$'s are the same (see Theorem \ref{thm: scattering of projections })
then $\tr\left(\Pi_{1}\right)=\frac{\tr\left(\Pi_{1}\Pi_{2}\right)}{\lambda}=\tr\left(\Pi_{2}\right)$
and so they have the same rank. If $\lambda=1$ then
\begin{align*}
\left(\Pi_{1}-\Pi_{1}\Pi_{2}\right)\left(\Pi_{1}-\Pi_{1}\Pi_{2}\right)^{\dagger} & =\Pi_{1}-\Pi_{1}\Pi_{2}\Pi_{1}=\Pi_{1}-\lambda\Pi_{1}=0\\
\left(\Pi_{2}-\Pi_{2}\Pi_{1}\right)\left(\Pi_{2}-\Pi_{2}\Pi_{1}\right)^{\dagger} & =\Pi_{2}-\Pi_{2}\Pi_{1}\Pi_{2}=\Pi_{2}-\lambda\Pi_{2}=0.
\end{align*}
Therefore, $\Pi_{1}-\Pi_{1}\Pi_{2}=0$ and $\left(\Pi_{2}-\Pi_{2}\Pi_{1}\right)^{\dagger}=0$
so $\Pi_{1}=\Pi_{1}\Pi_{2}=\Pi_{2}$.
\end{proof}
The first statement of the above proposition implies that all the
projections that belong to the same connected component of a proper
reflection network (see Definition \ref{def:RefNetwork}) have the
same rank. The second statement of the above proposition implies that
whenever we scatter the pair $\Pi_{1}$, $\Pi_{2}$ and there is a
$\lambda=1$ in the spectrum, then $\Pi_{1}^{\left(\lambda=1\right)}=\Pi_{2}^{\left(\lambda=1\right)}$
(recall claim 3 in Theorem \ref{thm: scattering of projections }).
This situation occurs when the eigenspaces of $\Pi_{1}$ and $\Pi_{2}$
have a common subspace so $\Pi_{i=1,2}^{\left(\lambda=1\right)}$
is the common projection on it. During the scattering procedure we
can eliminate either $\Pi_{1}^{\left(\lambda=1\right)}$ or $\Pi_{2}^{\left(\lambda=1\right)}$
in order to avoid redundant operations in the future (it is not strictly
necessary though).

We can now consider how a single scattering operation changes the
reflection network (recall Definition \ref{def:RefNetwork} of the
reflection network). According to Theorem \ref{thm: scattering of projections },
a pair of projections $\Pi_{1}$, $\Pi_{2}$ whose relation is initially
unknown (red edge) scatters into a series of reflecting pairs (black
edges except for $\lambda=0$), and each pair is orthogonal (no edges)
to all other pairs. Since both projections $\Pi_{1}$, $\Pi_{2}$
are part of a larger network, we have to specify how the resulting
projections $\left\{ \Pi_{i=1,2}^{\left(\lambda\right)}\right\} $
inherit the relations with the rest of the network; see Fig. \ref{fig:scattering update rule}.
\begin{figure}[H]
\begin{centering}
\includegraphics[height=0.2\paperheight]{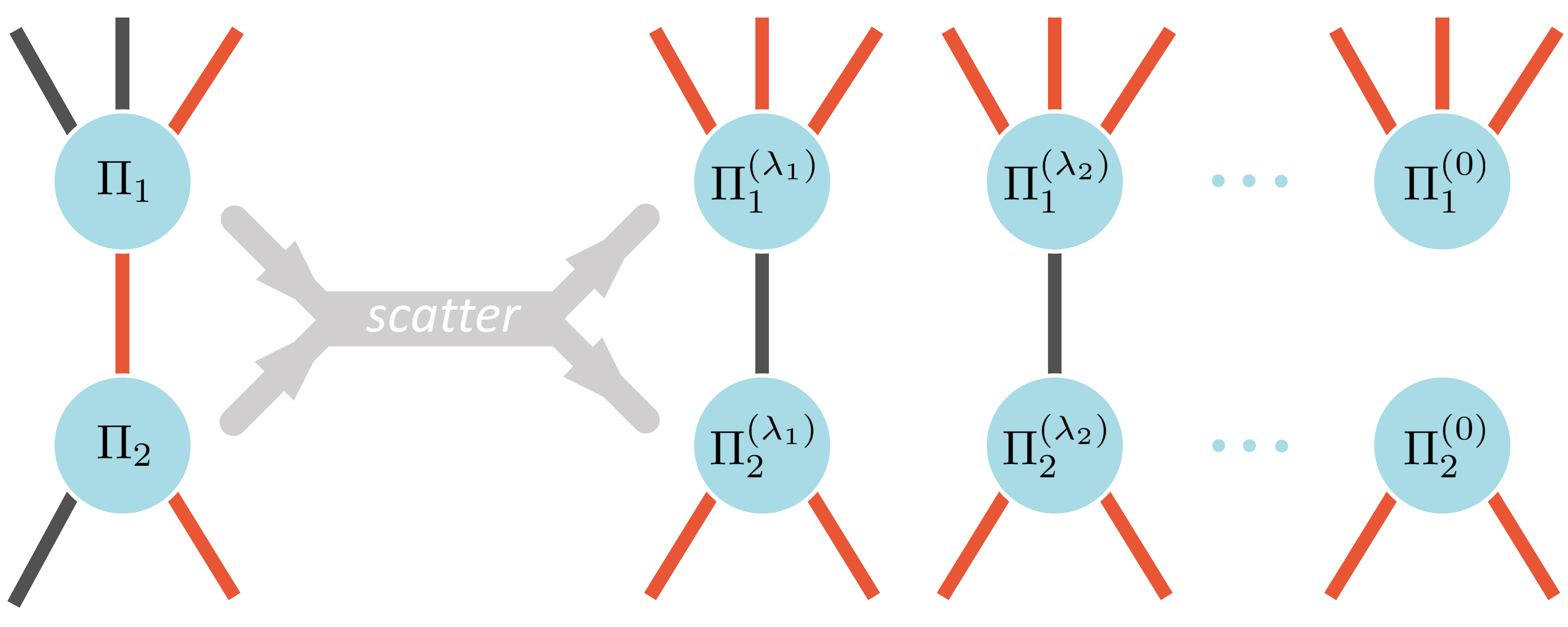}
\par\end{centering}
\raggedright{}\caption{\label{fig:scattering update rule}Generic update rules for the reflection
relations after scattering. The red edges represent unknown reflection
relations, black edges represent properly reflecting pairs, absent
edges represent orthogonal pairs. Open ended edges stand for the reflection
relations with the other projections in the network. In the generic
case each $\Pi_{i=1,2}$ breaks into $\left\{ \Pi_{i=1,2}^{\left(\lambda\right)}\right\} $
and the result is a series of properly reflecting pairs (for $\lambda=0$
the pair is orthogonal) as described in Theorem \ref{thm: scattering of projections }.
The open ended (external) edges are inherited from $\Pi_{i=1,2}$
by each of $\left\{ \Pi_{i=1,2}^{\left(\lambda\right)}\right\} $
with the black edges being reset to red (assuming $\Pi_{i=1,2}$ did
break under scattering).}
\end{figure}

First, note that orthogonality with other (external) projections is
preserved under scattering so we do not need to add new edges that
we did not already have. Second, the external red edges also do not
need to be updated since every unknown relation that $\Pi_{i=1,2}$
had, remains unknown for $\left\{ \Pi_{i=1,2}^{\left(\lambda\right)}\right\} $.
The external black edges, however, do not survive when a projection
is broken into smaller rank projections. That is because properly
reflecting pairs must have the same rank (see Proposition \ref{prop:reflecting proj})
so when one of the projections in the pair is broken, the resulting
projections are necessarily of lower rank than the projection that
was on the other side of that black edge. Therefore, the black edges
that $\Pi_{i=1,2}$ had before scattering have to be reset to red
when inherited by $\left\{ \Pi_{i=1,2}^{\left(\lambda\right)}\right\} $,
unless $\Pi_{i=1,2}$ did not break under scattering.

The special case where only one of the projections in the pair breaks
under scattering is presented in Fig. \ref{fig:scattering update rule special case}.
When both projections in the pair do not break (this is not shown
in the figures), we only need to update the connecting red edge to
black.
\begin{figure}[H]
\centering{}\includegraphics[viewport=0bp 0bp 860bp 453bp,clip,height=0.2\paperheight]{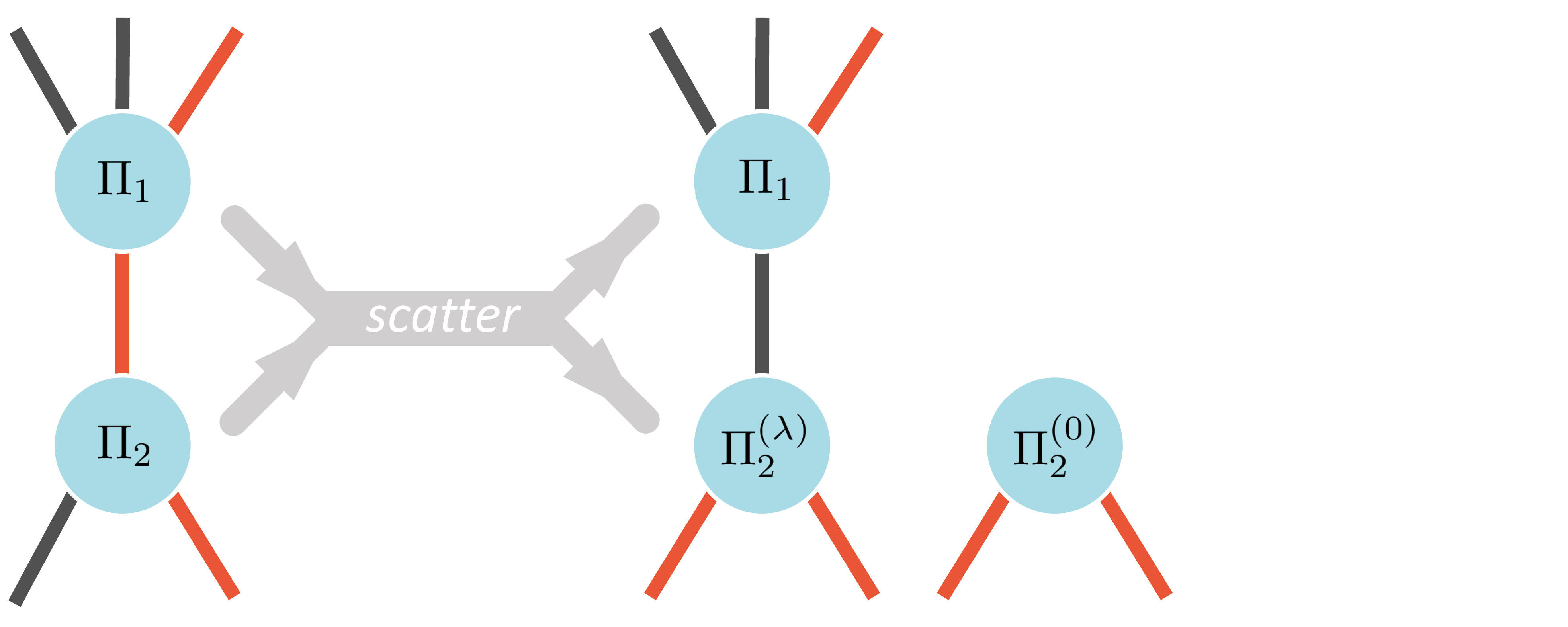}\caption{\label{fig:scattering update rule special case}Update rules for the
reflection relations in the case where only one of the projections
breaks. In this case, $\Pi_{2}$ may break to at most two projections
(if $\Pi_{2}$ also does not break then the red edge of $\left\{ \Pi_{1},\Pi_{2}\right\} $
is just set to black). The difference from the generic case is that
the external black edges of the unbroken projection $\Pi_{1}$ are
not reset to red.}
\end{figure}

The procedure \textproc{ScatterProjections} in the Scattering Algorithm
\ref{alg:The-Scattering-Algorithm} proceeds as follows:
\begin{enumerate}
\item Construct the improper reflection network from the initial spectral
projections and initializing all edges to red except for the ones
that are known to be reflecting (rank 1 and orthogonal projections).
\item Repeat until all edges are black: pick a pair of projections connected
by a red edge, scatter it and update the relations in the network
according to the rules given by Figs. \ref{fig:scattering update rule}
and \ref{fig:scattering update rule special case}.
\end{enumerate}
As was discussed after Theorem \ref{thm: scattering of projections },
we don't have to keep the duplicates if the scattered projections
share a common subspace. Also, we can argue heuristically that lower
rank projection are less likely to break under scattering which triggers
the resets of the previously established black edges. Thus, we may
reduce the overall number of scatterings needed if we prioritize scattering
the projections of lower ranks first.

The above procedure finishes when the reflection network is proper,
that is, when all edges are black. The fact that it always successfully
finishes in a finite number of steps is proven in the following lemma.
\begin{lem}
\label{lem: scat proc. always stops }Given a finite number of input
projections $\left\{ \Pi_{i}\right\} $, the procedure \textproc{ScatterProjections}
described above finishes in a finite number of steps and produces
a proper reflection network.
\end{lem}
\begin{proof}
Let $\left\{ \Pi_{i;t}\right\} $ be the set of projections in the
reflection network at step $t=0,1,...$ of the procedure. At each
$t$ consider the total number of projections $N_{t}$, the total
number of red edges $R_{t}$, and the total rank of all the projection
$\gamma_{t}=\sum_{i}\mathbf{rank}\Pi_{i;t}$. By the Definition \ref{def:Scattering },
the scattering operation does not change the total rank of projections
between input and output so for all $t$ we have $\gamma_{t}=\gamma_{0}$.
At each step, only two things can happen: Either both projections
do not break so $N_{t+1}=N_{t}$ and the connecting red edge becomes
black so $R_{t+1}=R_{t}-1$; or at least one projection breaks so
for some $n_{t},r_{t}>0$ we have $N_{t+1}=N_{t}+n_{t}$ and the value
$R_{t+1}=R_{t}\pm r_{t}$ can increase or decrease (it will decrease
only if the broken projections were not connected to any other projection
in the network). Now, consider the progression of the coordinate $\left(N_{t},R_{t}\right)$
on a two-dimensional grid. At each step it can either move one position
down or it can move diagonally but always to the right:
\[
\left(N_{t},R_{t}\right)\mapsto\left(N_{t+1},R_{t+1}\right)=\begin{cases}
(N_{t},R_{t}-1)\\
(N_{t}+n_{t},R_{t}\pm r_{t}).
\end{cases}
\]
Furthermore, we have the upper bound $N_{t}\leq\gamma_{t}=\gamma_{0}$
because projections cannot have ranks lower than $1$. We also know
that $N_{0}$ and $R_{0}$ are finite because $\left\{ \Pi_{i}\right\} $
is finite. Therefore, after a finite number of steps we will either
reach $R_{t}=0$ which means the reflection network is now proper,
or we will reach $N_{t}=\gamma_{0}$, after which every step will
decrease $R_{t}$ by $1$ until it reaches $0$.
\end{proof}

\subsection{Establishing minimality of the reflection network}

Establishing minimality of a reflection network means making sure
that all the projections in the network are minimal projections in
the algebra that they generate. Minimality can be established by considering
the paths in the network.

A path in a reflection network is given by an ordered set of vertices
$\pi=\left(v_{1},v_{2},...,v_{n}\right)$ that identifies a sequence
of connected projections in the reflection network. By taking the
product of all projections along the path and normalizing we define
an isometry
\begin{equation}
S_{\pi}:=\frac{1}{\lambda_{\pi}}\Pi_{v_{1}}\Pi_{v_{2}}\cdots\Pi_{v_{n}},\label{eq:def of path isometry}
\end{equation}
where the normalization $\lambda_{\pi}$ is the unique non-zero singular
value of the product of projections. We will refer to these operators
as \emph{path-isometries} that map from the initial space given by
$\Pi_{v_{n}}$ to the final space given by $\Pi_{v_{1}}$, along the
path $\pi$. It should be clear that $S_{\pi}^{\dagger}$ is a path-isometry
along the same path as $S_{\pi}$ but in reverse direction.

The minimality of reflection network can then be established using
the following lemma.
\begin{lem}
\label{lem:Minimality of reflecting projections }Let $\left\{ \Pi_{v}\right\} $
be a set of projections forming a proper reflection network such that
all path-isometries in the network are proportional $S_{\pi}\propto S_{\pi'}$
wherever the paths $\pi$ and $\pi'$ have the same initial and final
vertices. Then, all the projections $\left\{ \Pi_{v}\right\} $ are
minimal in the algebra that they generate $\mathcal{A}:=\left\langle \left\{ \Pi_{v}\right\} \right\rangle $.
\end{lem}
\begin{proof}
Every element $A\in\mathcal{A}$ is a linear combination of products
of $\left\{ \Pi_{v}\right\} $, therefore $\mathcal{A}=\mathsf{\spn}\left\{ S_{\pi}\right\} $
where $\left\{ S_{\pi}\right\} $ are all the path isometries in the
network. Then, all the projections $\left\{ \Pi_{v}\right\} $ are
minimal if $\Pi_{v_{0}}S_{\pi}\Pi_{v_{0}}\propto\Pi_{v_{0}}$ for
all $v_{0}$ and $\pi$. When the path $\pi$ does not start or end
next to $v_{0}$, we have $\Pi_{v_{0}}S_{\pi}=0$ or $S_{\pi}\Pi_{v_{0}}=0$
so the relation $\Pi_{v_{0}}S_{\pi}\Pi_{v_{0}}=0\propto\Pi_{v_{0}}$
trivially holds. Let us now consider $\pi:=\left(v_{1},v_{2},...,v_{n}\right)$
such that $\Pi_{v_{0}}S_{\pi}\Pi_{v_{0}}\neq0$. We can therefore
append $v_{0}$ to the beginning and the end of $\pi$ to get the
circular path $\tilde{\pi}:=\left(v_{0},v_{1},v_{2},...,v_{n},v_{0}\right)$.
Another circular path from $v_{0}$ to itself is the trivial path
$\left(v_{0},v_{0}\right)$, and so 
\[
\Pi_{v_{0}}S_{\pi}\Pi_{v_{0}}\propto S_{\tilde{\pi}}\propto S_{\left(v_{0},v_{0}\right)}\propto\Pi_{v_{0}}\Pi_{v_{0}}=\Pi_{v_{0}}.
\]
Therefore, all $\left\{ \Pi_{v_{0}}\right\} $ are minimal in $\mathcal{A}=\mathsf{\spn}\left\{ S_{\pi}\right\} $.
\end{proof}
By checking whether the path-isometries in a reflection network depend
only on the initial and final vertices independently from the paths,
we can verify that all projections are minimal. Note that projections
of rank $1$ are minimal so connected components with rank $1$ projections
are always minimal.

In case minimality could not be established, Lemma \ref{lem:Minimality of reflecting projections }
also implies a correction that can be implemented.
\begin{lem}
\label{lem:fixing minimality}In the setting of Lemma \ref{lem:Minimality of reflecting projections },
let $\pi$, $\pi'$ be two paths that share the same initial $v_{in}$
and final $v_{fin}$ vertices, but $S_{\pi}\not\propto S_{\pi'}$.
Then, the spectral projections of $U_{\pi\pi'}:=S_{\pi}S_{\pi'}^{\dagger}$
are \emph{not} reflecting with $\Pi_{v_{fin}}$.
\end{lem}
\begin{proof}
The operator $U_{\pi\pi'}$ is an isometry from the eigenspace of
$\Pi_{v_{fin}}$ to itself, so it is a unitary on the eigenspace of
$\Pi_{v_{fin}}$ . We therefore have the spectral projections $\Pi^{\left(\omega\right)}$
with the non-zero eigenvalues $\omega$ such that $U_{\pi\pi'}=\sum_{\omega}\omega\Pi^{\left(\omega\right)}$
and $\sum_{\omega}\Pi^{\left(\omega\right)}=\Pi_{v_{fin}}$. If there
was only one non-zero eigenvalue $\omega$ then $U_{\pi\pi'}=\omega\Pi^{\left(\omega\right)}=\omega\Pi_{v_{fin}}$
so $U_{\pi\pi'}=S_{\pi}S_{\pi'}^{\dagger}\propto\Pi_{v_{fin}}=S_{\pi'}S_{\pi'}^{\dagger}$,
but that implies $S_{\pi}\propto S_{\pi'}$. Therefore, there is more
than one spectral projection $\left\{ \Pi^{\left(\omega\right)}\right\} $
and so $\Pi_{v_{fin}}\Pi^{\left(\omega\right)}\Pi_{v_{fin}}=\Pi^{\left(\omega\right)}\not\propto\Pi_{v_{fin}}$.
\end{proof}
The procedure \textproc{EstablishMinimality} in the Scattering Algorithm
\ref{alg:The-Scattering-Algorithm} proceeds as follows:
\begin{enumerate}
\item Check whether path-isometries in the reflection network depend only
on the initial and final vertices and if they are, finish.
\item If not, given the paths $\pi$, $\pi'$ that violate the premise of
Lemma \ref{lem:Minimality of reflecting projections }, take $U_{\pi\pi'}=S_{\pi}S_{\pi'}^{\dagger}$
and add its spectral projections $\left\{ \Pi^{\left(\omega\right)}\right\} $
to the connected component where $\pi$ and $\pi'$ reside.
\item Initiate another round of scatterings on the connected component with
the new projections until the reflection network is proper again,
then repeat step 1.
\end{enumerate}
In step 1 only connected components with projections of rank higher
than $1$ need to be checked. This procedure is guaranteed to stop
because it either finishes on step 1 or it reduces the rank of projections
in the connected component, and minimality trivially holds if it reaches
projections of rank $1$.

\subsection{Establishing completeness of the reflection network\label{subsec:Establishing-completeness-of}}

Completeness of a reflection network means that there is a subset
of vertices in the network that forms a maximal set of minimal projections.
If any of the initial generators of the algebra are supported on the
whole Hilbert space, which means their spectral projections sum to
the identity, then completeness is guaranteed and we don't have to
do anything here. That is because after scattering, the descendants
of these spectral projections in the proper reflection network will
still sum to the identity, so they will form a maximal set of minimal
projections.

If we can always add the identity to the initial set of generators,
completeness becomes a trivial property. Nevertheless, it is also
possible (and sometimes easier) to scatter the initial projections
regardless of them being supported on the whole Hilbert space, and
then fix the completeness of the final reflection network after the
fact.

Given the projections $\left\{ \Pi_{v}\right\} $ forming a reflection
network we will assume that at this point minimality has been established.
Consider the largest subset of pairwise orthogonal projections $\left\{ \Pi_{v_{k}}\right\} \subseteq\left\{ \Pi_{v}\right\} $
where $\Pi_{v_{k}}\Pi_{v_{l}}=\delta_{kl}\Pi_{v_{k}}$. From the perspective
of graphs this is the maximal independent set of vertices in the network
and it does not have to be unique. The subset $\left\{ \Pi_{v_{k}}\right\} $
is a maximal set of minimal projections if the operator
\begin{equation}
I_{\mathcal{A}}:=\sum_{k}\Pi_{v_{k}}\label{eq:def of identity in reflect. net.}
\end{equation}
acts as the identity on every operator in the algebra $\mathcal{A}:=\left\langle \left\{ \Pi_{v}\right\} \right\rangle $,
meaning $I_{\mathcal{A}}\Pi_{v}=\Pi_{v}$ for all $v$. If it is not,
we can use the result of the following lemma to complete $I_{\mathcal{A}}$
to act as the identity.
\begin{lem}
\label{lem:Completness}Let $\Pi_{v}$ be a minimal projection in
a reflection network such that $I_{\mathcal{A}}\Pi_{v}\neq\Pi_{v}$
with $I_{\mathcal{A}}$ defined in Eq. \eqref{eq:def of identity in reflect. net.}.
Then, with the appropriate normalization factor $c$, the operator
\begin{equation}
\tilde{\Pi}_{v}:=\frac{1}{c}\left(I-I_{\mathcal{A}}\right)\Pi_{v}\left(I-I_{\mathcal{A}}\right),\label{eq:def of complementary projection}
\end{equation}
where $I$ is the full identity matrix, has the following properties:

(1) $\tilde{\Pi}_{v}$ is a minimal projection in $\mathcal{A}:=\left\langle \left\{ \Pi_{v}\right\} \right\rangle $.

(2) $\tilde{\Pi}_{v}$ is orthogonal to all $\Pi_{v_{k}}$ in Eq.
\eqref{eq:def of identity in reflect. net.}.

(3) The operator $\tilde{I}_{\mathcal{A}}:=I_{\mathcal{A}}+\tilde{\Pi}_{v}$
is such that $\tilde{I}_{\mathcal{A}}\Pi_{v}=\Pi_{v}$.
\end{lem}
\begin{proof}
If we distribute the terms in Eq. \eqref{eq:def of complementary projection}
we will get $c\tilde{\Pi}_{v}=\Pi_{v}-I_{\mathcal{A}}\Pi_{v}-\Pi_{v}I_{\mathcal{A}}+I_{\mathcal{A}}\Pi_{v}I_{\mathcal{A}}$
so clearly $\tilde{\Pi}_{v}$ is a self-adjoint operator in $\mathcal{A}$.
Since $\Pi_{v}$ is minimal we have
\begin{equation}
\Pi_{v}\left(I-I_{\mathcal{A}}\right)\Pi_{v}=\Pi_{v}-\Pi_{v}I_{\mathcal{A}}\Pi_{v}=\left(1-\lambda\right)\Pi_{v}.\label{eq:minimality of Pi_v}
\end{equation}
Here $\lambda$ is the proportionality factor in the minimality relation
$\Pi_{v}I_{\mathcal{A}}\Pi_{v}\propto\Pi_{v}$ and $\lambda$ is not
$1$ because that would contradict $I_{\mathcal{A}}\Pi_{v}\neq\Pi_{v}$.
Then, choosing $c=1-\lambda$ and taking the square of $\tilde{\Pi}_{v}$
we get
\begin{align*}
\tilde{\Pi}_{v}\tilde{\Pi}_{v} & =\frac{1}{c^{2}}\left(I-I_{\mathcal{A}}\right)\Pi_{v}\left(I-I_{\mathcal{A}}\right)\Pi_{v}\left(I-I_{\mathcal{A}}\right)\\
 & =\frac{1}{\left(1-\lambda\right)^{2}}\left(I-I_{\mathcal{A}}\right)\left(1-\lambda\right)\Pi_{v}\left(I-I_{\mathcal{A}}\right)=\tilde{\Pi}_{v}.
\end{align*}
Therefore, $\tilde{\Pi}_{v}$ is a projection. It is minimal because
for any $A\in\mathcal{A}$ we have 
\[
\tilde{\Pi}_{v}A\tilde{\Pi}_{v}=\frac{1}{c^{2}}\left(I-I_{\mathcal{A}}\right)\Pi_{v}\tilde{A}\Pi_{v}\left(I-I_{\mathcal{A}}\right)
\]
where $\tilde{A}:=\left(I-I_{\mathcal{A}}\right)A\left(I-I_{\mathcal{A}}\right)$.
Since $\tilde{A}\in\mathcal{A}$ and $\Pi_{v}$ is minimal we get
$\Pi_{v}\tilde{A}\Pi_{v}\propto\Pi_{v}$ and so $\tilde{\Pi}_{v}A\tilde{\Pi}_{v}\propto\tilde{\Pi}_{v}$.
This proves statement 1. Statement 2 follows from $\left(I-I_{\mathcal{A}}\right)\Pi_{v_{k}}=\Pi_{v_{k}}-\Pi_{v_{k}}=0$
so $\tilde{\Pi}_{v}\Pi_{v_{k}}=0$. Finally, recalling that $c=1-\lambda$
and using the identity \eqref{eq:minimality of Pi_v} again, we get
\[
\tilde{\Pi}_{v}\Pi_{v}=\frac{1}{c}\left(I-I_{\mathcal{A}}\right)\Pi_{v}\left(I-I_{\mathcal{A}}\right)\Pi_{v}=\left(I-I_{\mathcal{A}}\right)\Pi_{v}.
\]
Thus, $\tilde{I}_{\mathcal{A}}\Pi_{v}=I_{\mathcal{A}}\Pi_{v}+\left(I-I_{\mathcal{A}}\right)\Pi_{v}=\Pi_{v}$,
which proves statement 3.
\end{proof}
\noindent The procedure \textproc{EstablishCompleteness} in the Scattering
Algorithm \ref{alg:The-Scattering-Algorithm} proceeds as follows:
\begin{enumerate}
\item Choose the largest subset of pairwise orthogonal projections in the
network $\left\{ \Pi_{v_{k}}\right\} $ and if $I_{\mathcal{A}}=\sum_{k}\Pi_{v_{k}}$
acts as the identity on all other projections, finish.
\item If it does not, then for each projection such that $I_{\mathcal{A}}\Pi_{v}\neq\Pi_{v}$
construct the complementary projection $\tilde{\Pi}_{v}$ as defined
in Eq. \eqref{eq:def of complementary projection} and add it to the
network.
\end{enumerate}
\noindent Since by construction $\tilde{\Pi}_{v}$'s are minimal projections
in the same algebra, they do not render the reflection network improper.
Lemma \ref{lem:Completness} then ensures that after the completion
of the network all the new $\tilde{\Pi}_{v}$'s will join the largest
subset of pairwise orthogonal projections in the network and sum to
$\tilde{I}_{\mathcal{A}}$ that acts as the identity on every element.
It should be noted again that this procedure is only needed if none
of the original generators were supported on the whole Hilbert space.

\subsection{Constructing the bipartition table}

We already know from Lemma \ref{lem: for all S_kl there is a BPT}
how to construct BPTs from maximal sets of minimal isometries. What
we need then is to construct a maximal set of minimal isometries from
the reflection network. This is achieved with the help of the following
lemma.
\begin{lem}
\label{lem:BPT construction from refnet}Let $\left\{ \Pi_{v}\right\} $
be the projections of a reflection network for which minimality and
completeness holds. Then, there is a set $\left\{ S_{kl}^{q}\right\} $
of path-isometries in the network that is a maximal set of minimal
isometries in the algebra $\mathcal{A}:=\left\langle \left\{ \Pi_{v}\right\} \right\rangle $.
\end{lem}
\begin{proof}
Let $\left\{ \Pi_{v_{k}^{q}}\right\} \subseteq\left\{ \Pi_{v}\right\} $
be a maximal set of minimal projections partitioned into connected
components $q$. Let $\Pi_{v_{1}^{q}}$ be an arbitrarily chosen first
element in this set for each connected component $q$. Then, for every
$k\geq1$ there is a path $\pi$ between $\Pi_{v_{1}^{q}}$ and $\Pi_{v_{k}^{q}}$
identifying the path-isometries $S_{k1}^{q}\equiv S_{\pi}$ and $S_{1k}^{q}\equiv S_{\pi}^{\dagger}$
as in Eq. \ref{eq:def of path isometry}. For all $k,l\geq1$ we can
identify $S_{kl}^{q}:=S_{k1}^{q}S_{1l}^{q}$ which are path-isometries
from $\Pi_{v_{l}^{q}}$ to $\Pi_{v_{k}^{q}}$ via $\Pi_{v_{1}^{q}}$.
Isometries defined this way have the properties $S_{kl}^{q}=S_{lk}^{q\dagger}$
and $S_{kl}^{q}S_{l'k'}^{q'}=\delta_{qq'}\delta_{ll'}S_{kk'}^{q}$
as required by the Definition \ref{def:minimal isom} of maximal sets
of minimal isometries. The final property that we need to show is
that $\left\{ S_{kl}^{q}\right\} $ spans $\mathcal{A}$. Since $\mathcal{A}$
is spanned by products of $\left\{ \Pi_{v}\right\} $ that are proportional
to path-isometries $\left\{ S_{\pi}\right\} $, it is sufficient to
show that for any path $\pi$ the path-isometry $S_{\pi}$ is spanned
by $\left\{ S_{kl}^{q}\right\} $. Since completeness holds, the sum
$\sum_{q,k}\Pi_{v_{k}^{q}}=I_{\mathcal{A}}$ acts as the identity
of the algebra. Therefore, 
\begin{equation}
S_{\pi}=I_{\mathcal{A}}S_{\pi}I_{\mathcal{A}}=\sum_{kl}\Pi_{v_{k}^{q}}S_{\pi}\Pi_{v_{l}^{q}},\label{eq:BPT const from refnet eq 1}
\end{equation}
where $q$ is the connected component that contains the path $\pi$.
By definition of path-isometries, every non-vanishing term $\Pi_{v_{k}^{q}}S_{\pi}\Pi_{v_{l}^{q}}$
is proportional to the path-isometry $S_{\left(v_{k}^{q},\pi,v_{l}^{q}\right)}$.
Since minimality holds, according to Lemma \ref{lem:uniqueness of min isometr }
the path-isometries $S_{\left(v_{k}^{q},\pi,v_{l}^{q}\right)}\propto S_{kl}^{q}$
are proportional because they have the same initial and final spaces.
Therefore, for all $k,l$ in Eq. \eqref{eq:BPT const from refnet eq 1},
$\Pi_{v_{k}^{q}}S_{\pi}\Pi_{v_{l}^{q}}\propto S_{kl}^{q}$ so $S_{\pi}$
is in the span of $\left\{ S_{kl}^{q}\right\} $.
\end{proof}
The procedure \textproc{ConstructBptBasis} in the Scattering Algorithm
\ref{alg:The-Scattering-Algorithm} proceeds as follows:
\begin{enumerate}
\item Identify a maximal set of orthogonal projections $\left\{ \Pi_{v_{k}^{q}}\right\} $
in each connected component $q$ and arbitrarily designate the first
element $\Pi_{v_{1}^{q}}$.
\item In each connected component $q$ construct the path-isometries $\left\{ S_{k1}^{q}\right\} $
from $\Pi_{v_{1}^{q}}$ to every other element $\Pi_{v_{k}^{q}}$.
\item Use the path-isometries $\left\{ S_{k1}^{q}\right\} $ to construct
the BPT basis as described in the proof of Lemma \ref{lem: for all S_kl there is a BPT}.
\end{enumerate}
Note that the reason that we can use Lemma \ref{lem: for all S_kl there is a BPT}
in step 3 is because $S_{k1}^{q}$ are minimal isometries as established
by Lemma \ref{lem:BPT construction from refnet}.

\subsection{Why the Scattering Algorithm works: putting it all together\label{subsec:Why-the-Scattering Alg works}}

Following the above results we are almost ready to prove that the
output of the Scattering Algorithm is correct. What remains before
we can put it all together is to show that the algebra generated by
the final reflection network is the same algebra generated by the
input $\mathcal{M}$.
\begin{lem}
\label{lem:final algebra is the same as initial}Let $\left\{ \Pi_{v}\right\} $
be the projections in the final reflection network (after minimality
and completeness have been established) produced by the Scattering
Algorithm \ref{alg:The-Scattering-Algorithm}. Then, the algebra generated
by the projections $\left\langle \left\{ \Pi_{v}\right\} \right\rangle $
and the algebra generated by the input $\left\langle \mathcal{M}\right\rangle $
is the same algebra.
\end{lem}
\begin{proof}
Let $\mathcal{A:=\left\langle \mathcal{M}\right\rangle }$. All the
spectral projections of the operators in $\mathcal{M}$ span the operators
in $\mathcal{M}$ and are themselves in the algebra $\mathcal{A}$
(with Eq. \ref{eq:spec proj from operator} we can show that this
is true for every spectral projection of any $A\in\mathcal{A}$).
Therefore, the algebra $\mathcal{A}$ is generated by the output of
\textproc{GetAllSpectralProjections}. During the procedure \textproc{ScatterProjections},
we repeat the scattering operation where we replace a pair of projections
$\Pi_{1}$, $\Pi_{2}$ with the spectral projections $\left\{ \Pi_{1}^{\left(\lambda\right)}\right\} $,
$\left\{ \Pi_{2}^{\left(\lambda\right)}\right\} $ as specified in
the Definition \ref{def:Scattering }. Since $\left\{ \Pi_{i=1,2}^{\left(\lambda\right)}\right\} $
are the spectral projections of $\Pi_{i}\Pi_{j}\Pi_{i}$, and $\Pi_{i}\Pi_{j}\Pi_{i}\in\mathcal{A}$,
each $\Pi_{i=1,2}^{\left(\lambda\right)}$ (including the null projections)
is an element of $\mathcal{A}$. Conversely, the sum of $\left\{ \Pi_{i=1,2}^{\left(\lambda\right)}\right\} $
gives back $\Pi_{i=1,2}$, so replacing the pair $\Pi_{1}$, $\Pi_{2}$
with $\left\{ \Pi_{1}^{\left(\lambda\right)}\right\} $, $\left\{ \Pi_{2}^{\left(\lambda\right)}\right\} $
does not change the generating power of projections in the reflection
network. Therefore, the set of projection in the output of \textproc{ScatterProjections}
still generates the same algebra $\mathcal{A}$. During the procedure
\textproc{EstablishMinimality}, we may add the spectral projections
of $U_{\pi\pi'}$, but once again, since $U_{\pi\pi'}\in\mathcal{A}$
its spectral projections are elements of $\mathcal{A}$ so the output
of \textproc{EstablishMinimality} still generates the same algebra
$\mathcal{A}$. During the procedure \textproc{EstablishCompleteness}
we may add more projections as provided by Lemma \ref{lem:Completness}
but it guarantees that they are all in $\mathcal{A}$, so the output
of \textproc{EstablishCompleteness} still generates $\mathcal{A}$.
\end{proof}
\begin{thm}
The BPT basis $\left\{ \ket{e_{ik}^{q}}\right\} $ produced by the
Scattering Algorithm \ref{alg:The-Scattering-Algorithm} on the input
$\mathcal{M}$, identify the irreps structure of the algebra $\left\langle \mathcal{M}\right\rangle $
as established by Theorem \ref{thm: Wedderburn Decomp}.
\end{thm}
\begin{proof}
The procedure \textproc{GetAllSpectralProjections} outputs the projections
that will form the initial improper reflection network. Lemma \ref{lem: scat proc. always stops }
ensures that the procedure \textproc{ScatterProjections} will take
the initial improper reflection network and output a proper reflection
network in a finite number of steps. Lemma \ref{lem:Minimality of reflecting projections }
ensures that the procedure \textproc{EstablishMinimality} correctly
identifies whether the reflection network consists of minimal projections.
Lemma \ref{lem:fixing minimality} ensures that \textproc{EstablishMinimality}
correctly modifies the reflection network to consist of minimal projections
if it did not initially. Lemma \ref{lem:Completness} ensures that
the procedure \textproc{EstablishCompleteness} correctly modifies
the reflection network to include a maximal set of minimal projections.
At this point we have a minimal and complete reflection network that
consists of projections $\left\{ \Pi_{v}\right\} $. Since minimality
and completeness hold, Lemma \ref{lem:BPT construction from refnet}
ensures that the procedure \textproc{ConstructBptBasis} finds minimal
isometries and, following Lemma \ref{lem: for all S_kl there is a BPT},
constructs the BPT basis $\left\{ \ket{e_{ik}^{q}}\right\} $ for
the algebra $\left\langle \left\{ \Pi_{v}\right\} \right\rangle $.
Lemma \ref{lem:final algebra is the same as initial} then ensures
that $\left\{ \ket{e_{ik}^{q}}\right\} $ are also the BPT basis for
the algebra $\left\langle \mathcal{M}\right\rangle $. Finally, the
proof of Theorem \ref{thm: Wedderburn Decomp} demonstrates how the
BPT basis $\left\{ \ket{e_{ik}^{q}}\right\} $ identify the irreps
structure of the algebra $\left\langle \mathcal{M}\right\rangle $.
\end{proof}

\chapter{Reduction of states\label{chap:Operational-reductions-of-st}}

In this chapter we will consider reductions of states and their implications
in the form of superselection and decoherence. By reduction of states
we loosely mean the reduction of information contained in the quantum
state as a result of some operational constraint. The best known reduction
of states is the partial trace map. By shifting the focus from subsystems
to operator algebras we will consider more general reductions of states.

The reduced states produced by the partial trace map were motivated
by the need to describe the states of individual subsystems, even
when they are entangled. Identifying such reduced states turned out
to be more than a mathematical exercise because without it we could
not define decoherence and understand its role in the emergence of
classicality (see \citep{zurek2003decoherence,Schlosshauer:2003zy}
for a review of the decoherence program). The idea that other physically
motivated (but more general) state reductions can lead to decoherence
and emergence of classicality has been explored in \citep{CastagninoLombardi2004,Castagnino2008,FortinLombardiCastagnino2014,piazza2010glimmers,KoflerBrukner2007Classical,Alicki09,cotler2019locality,carroll2020quantum}.

A shift in perspective on the notion of a subsystem and the accompanying
state reduction is due to Zanardi \emph{et al}. \citep{Zanardi2001Virtual,Zanardi:2004zz,Knill2000,viola2001constructing},
that have defined the concept of a \emph{virtual }subsystem via operator
algebras. This idea found many applications in the quantum error correction
community with the development of decoherence free subsystems and
operator quantum error correction (a.k.a. subsystem codes) \citep{Knill2000,Kempe01,lidar2003decoherence,lidar2014review,kribs2005operator,kribs2005unified,bacon2006operator,blume2008characterizing,BlumeKohout10}.
These ideas have also percolated into the study of bulk reconstruction
in AdS/CFT correspondence where the holographic-error-correcting-code
approach was introduced \citep{almheiri2015bulk,pastawski2015holographic}.
Beyond quantum error correction, the definition of subsystems via
operator algebras (in particular group algebras) plays a central role
in ideas such as generalized entanglement \citep{Barnum03,Barnum04},
quantum reference frames \citep{Bartlett07,kabernik2014quantum} and
quantum state compression \citep{Bluhm2018}.

In the following, Section \ref{sec:State-reductions-and} is dedicated
to re-examining the partial trace map and re-framing it as an instance
of a state reduction map that arises from an operational constraint.
We will derive an alternative representation of the partial trace
map and show how it is visually captured by a bipartition table. This
alternative representation will then be used to describe the process
of decoherence without referring to the interacting subsystems (without
the system-environment split).

In Section \ref{sec:State-reudcions-due} we will consider state reductions
due to more general operational constraints that go beyond inaccessible
subsystems. We will see that in general, operational constraints lead
to a combination of superselection and decoherence.

In order to clarify these ideas we will study a few examples. In the
first example we will consider the operational constraint of not having
a shared reference frame and see how that leads to superselection.
In the second example we will demonstrate how an operational constraint
can lead to decoherence in a simple system such as the Hydrogen atom
even when no interactions or couplings to an external environment
are present. In the last example we will identify the possible encodings
of quantum information into a decoherence free subsystem by considering
the noise as an operational constraint. In this example we will demonstrate
how the Scattering Algorithm allows us to expand the scope of treatable
operational constraints beyond group representations.

\section{State reductions and decoherence due to inaccessible subsystems\label{sec:State-reductions-and}}

The partial trace map is the prototypical example of quantum state
reduction. It is usually introduced by the following reasoning: We
are given the bipartite Hilbert space $\mathcal{H}_{AB}:=\mathcal{H}_{A}\otimes\mathcal{H}_{B}$
and the operational constraint that allows measurements only on subsystem
$B$. Then, we consider the map $\tr_{A}$ that reduces the full states
of $AB$ to the states of $B$ and preserves all information about
$B$. In other words, the partial trace map $\tr_{A}$ is defined
by the condition that for all $\rho\in\mathcal{L}\left(\mathcal{H}_{AB}\right)$
and all $O_{B}\in\mathcal{L}\left(\mathcal{H}_{B}\right)$ it produces
reduced states $\rho_{B}:=\tr_{A}\left[\rho\right]$ such that 
\begin{equation}
\tr\left[O_{B}\rho_{B}\right]=\tr\left[I_{A}\otimes O_{B}\,\rho\right].\label{eq:part trace condition}
\end{equation}

If $\left\{ \ket{a_{i}}\right\} $ are some basis in $\mathcal{H}_{A}$,
the map $\tr_{A}$ can be expressed in the operator sum representation
as
\begin{equation}
\tr_{A}\left[\rho\right]:=\sum_{i}K_{i}\,\rho\,K_{i}^{\dagger},\label{eq:part trace op sum rep}
\end{equation}
where $K_{i}:=\left\langle a_{i}\right|\otimes I_{B}$. Since $\sum_{i}K_{i}^{\dagger}K_{i}=I_{AB}$,
it is completely positive and trace preserving (CPTP) so it maps quantum
states to quantum states. Using the cyclical property of the trace
(not partial) and linearity we can show that the condition in Eq.
\eqref{eq:part trace condition} holds for the map in Eq. \eqref{eq:part trace op sum rep}:
\begin{align*}
\tr\left[O_{B}\rho_{B}\right] & =\tr\left[O_{B}\left(\sum_{i}\left\langle a_{i}\right|\otimes I_{B}\,\rho\,\left|a_{i}\right\rangle \otimes I_{B}\right)\right]\\
 & =\tr\left[\sum_{i}\left|a_{i}\right\rangle \left\langle a_{i}\right|\otimes O_{B}\,\rho\right]\\
 & =\tr\left[I_{A}\otimes O_{B}\,\rho\right].
\end{align*}

We will now derive an alternative representation of the partial trace
map in the framework of operator algebras. First, we note that the
operational constraint dictates that only observables of the form
$I_{A}\otimes O_{B}$ are physically relevant. Therefore, the operational
constraint identifies the operator algebra 
\begin{equation}
\mathcal{A}:=\left\{ I_{A}\otimes O_{B}\,|\,O_{B}\in\mathcal{L}\left(\mathcal{H}_{B}\right)\right\} \label{eq: bipartite Hilb space algebra}
\end{equation}
that contains all the relevant observables. Since $\mathcal{A}=I_{A}\otimes\mathcal{L}\left(\mathcal{H}_{B}\right)$,
it can be reduced to $\mathcal{H}_{B}$ by mapping $I_{A}\otimes O_{B}\longmapsto O_{B}$.
The accompanying state reduction map $\rho_{B}:=\tr_{A}\left[\rho\right]$
must comply with the condition \eqref{eq:part trace condition}. By
linearity, it is sufficient to satisfy this condition for the minimal
isometries $S_{kl}:=I_{A}\otimes\ket{b_{k}}\bra{b_{l}}$ that span
$\mathcal{A}$, so we require the condition
\[
\tr\left[\ket{b_{k}}\bra{b_{l}}\rho_{B}\right]=\tr\left[S_{kl}\,\rho\right].
\]

With the above condition and the resolution of identity $I_{B}=\sum_{l}\ket{b_{l}}\bra{b_{l}}$,
we can express 
\[
\rho_{B}=\left(\sum_{l}\ket{b_{l}}\bra{b_{l}}\right)\rho_{B}\left(\sum_{k}\ket{b_{k}}\bra{b_{k}}\right)=\sum_{kl}\tr\left[\ket{b_{k}}\bra{b_{l}}\rho_{B}\right]\text{\ensuremath{\ket{b_{l}}\bra{b_{k}}}}=\sum_{kl}\tr\left[S_{kl}\,\rho\right]\text{\ensuremath{\ket{b_{l}}\bra{b_{k}}}}.
\]
Since $\rho_{B}=\tr_{A}\left[\rho\right]$, we have derived above
a new representation for the partial trace map 
\begin{equation}
\tr_{A}\left[\rho\right]:=\sum_{kl}\tr\left[S_{kl}\rho\right]\text{\ensuremath{\ket{b_{l}}\bra{b_{k}}}}.\label{eq:part trace S_kl rep}
\end{equation}
Such representations of maps between operators are known as ``Input/Output''
or ``Tomographic'' \citep{milz2017introduction}. The easiest way
to see that \ref{eq:part trace S_kl rep} is the same partial trace
map as \eqref{eq:part trace op sum rep}, is to verify that both representations
yield
\[
\tr_{A}\left[\ket{a_{i}}\bra{a_{j}}\otimes\ket{b_{l}}\bra{b_{k}}\right]=\delta_{ij}\ket{b_{l}}\bra{b_{k}}.
\]
Then, by linearity both maps have to be identical since $\left\{ \ket{a_{i}}\bra{a_{j}}\otimes\ket{b_{l}}\bra{b_{k}}\right\} $
span all the operators in $\mathcal{L}\left(\mathcal{H}_{AB}\right)$.

As an illustration, consider the Hilbert space $\mathcal{H}:=\underline{l}\otimes\underline{\frac{1}{2}}$
of spin-$l$ and spin-$\frac{1}{2}$. The algebra of observables on
spin-$\frac{1}{2}$ is spanned by the minimal isometries $S_{mm'}:=I_{l}\otimes\ket{\frac{1}{2},m}\bra{\frac{1}{2},m'}$
for $m,m'=\pm\frac{1}{2}$. Then, for any pure state $\ket{\mathcal{\psi}}\in\mathcal{H}$
we can express the partial trace over spin-$l$ using \eqref{eq:part trace S_kl rep}
as 
\begin{equation}
\ket{\psi}\longmapsto\tr_{l}\left[\ketbra{\psi}{\psi}\right]=\sum_{m,m'=\pm\frac{1}{2}}\bra{\psi}S_{mm'}\ket{\psi}\ketbra{\frac{1}{2},m'}{\frac{1}{2},m}=\begin{pmatrix}\bra{\psi}S_{\frac{1}{2},\frac{1}{2}}\ket{\psi} & \bra{\psi}S_{-\frac{1}{2},\frac{1}{2}}\ket{\psi}\\
\bra{\psi}S_{\frac{1}{2},-\frac{1}{2}}\ket{\psi} & \bra{\psi}S_{-\frac{1}{2},-\frac{1}{2}}\ket{\psi}
\end{pmatrix}.\label{eq:part trace over spin l}
\end{equation}

Let us now consider how the partial trace map looks in the BPT picture.
The algebra of relevant observables for the partial trace over $A$
is $I_{A}\otimes\mathcal{L}\left(\mathcal{H}_{B}\right)$, so for
some product basis $\ket{e_{ik}}:=\ket{a_{i}}\ket{b_{k}}$, where
$i=1,...,d_{A}$ and $k=1,...,d_{B}$, the irreps of this algebra
are given by the BPT

\noindent\begin{minipage}[c]{1\columnwidth}%
\begin{center}
\vspace{0.5\baselineskip}
\begin{tabular}{|c|c|c|c|}
\hline 
$e_{1,1}$ & $e_{1,2}$ & $\cdots$ & $e_{1,d_{B}}$\tabularnewline
\hline 
$e_{2,1}$ & $e_{2,2}$ & $\cdots$ & $e_{2,d_{B}}$\tabularnewline
\hline 
$\vdots$ & $\vdots$ & $\ddots$ & $\vdots$\tabularnewline
\hline 
$e_{d_{A},1}$ & $e_{d_{A},2}$ & $\cdots$ & $e_{d_{A},d_{B}}$\tabularnewline
\hline 
\multicolumn{1}{c}{$\downarrow$} & \multicolumn{1}{c}{$\downarrow$} & \multicolumn{1}{c}{} & \multicolumn{1}{c}{$\downarrow$}\tabularnewline
\hline 
$b_{1}$ & $b_{2}$ & $\cdots$ & $b_{d_{B}}$\tabularnewline
\hline 
\end{tabular}.\vspace{0.5\baselineskip}
\par\end{center}%
\end{minipage} We have added an additional single row on the bottom which represents
the Hilbert space of reduced states.

This picture implies that states that are supported on a single column
$\ket{\varphi_{k}}=\sum_{i}c_{i}\ket{e_{ik}}$---we will call them
\emph{column kets}---reduce as
\[
\ket{\varphi_{k}}\longmapsto\ket{b_{k}}.
\]
That is because column kets are the product states $\ket{\varphi_{k}}=\ket{\varphi_{A;k}}\ket{b_{k}}$
for some $\ket{\varphi_{A;k}}\in\mathcal{H}_{A}$.

A general pure state $\ket{\mathcal{\psi}}\in\mathcal{H}$ that is
supported on multiple columns can then be expressed as a sum of unnormalized
column kets $\ket{\psi}=\sum_{k}\ket{\varphi_{k}}$. Then, using the
representation \eqref{eq:part trace S_kl rep} of the partial trace
map , all pure states reduce as 
\begin{equation}
\ket{\psi}\longmapsto\rho_{B}=\sum_{kl}\bra{\varphi_{k}}S_{kl}\ket{\varphi_{l}}\text{\ensuremath{\ket{b_{l}}\bra{b_{k}}}}.\label{eq:state reduc map with column kets-1}
\end{equation}

Observe that in the reduced state $\rho_{B}$, the probability weights
of the diagonal terms $\ket{b_{k}}\bra{b_{k}}$ are given by the overlaps
$\bra{\varphi_{k}}S_{kk}\ket{\varphi_{k}}=\braket{\varphi_{k}}{\varphi_{k}}$
of the column kets with themselves (these are their square norms).
In general, the weights and phases of the reduced coherence terms
$\text{\ensuremath{\ket{b_{l}}\bra{b_{k}}}}$ are given by the overlaps
$\bra{\varphi_{k}}S_{kl}\ket{\varphi_{l}}$ of the corresponding column
kets. The overlap is calculated by mapping the kets to the same column
with the isometries $S_{kl}$.

In the following section we will use this perspective in order to
makes sense of the process of decoherence in reduced states that arise
from more general operational constraints. Before we do that, however,
let us describe the process of decoherence using this perspective
in a familiar setting where the reduced states are given by the partial
trace over a subsystem.

Going back to the composite system of spin-$l$ and spin-$\frac{1}{2}$,
we adopt the shorter notation for the product basis $\ket{m,\pm\frac{1}{2}}:=\ket{l,m}\ket{\frac{1}{2},\pm\frac{1}{2}}$.
The BPT picture of the partial trace over spin-$l$ is then

\noindent\begin{minipage}[c]{1\columnwidth}%
\begin{center}
\vspace{0.5\baselineskip}
\begin{tabular}{|c|c|}
\hline 
$+l,+\frac{1}{2}$ & $+l,-\frac{1}{2}$\tabularnewline
\hline 
$\vdots$ & $\vdots$\tabularnewline
\hline 
$0,+\frac{1}{2}$ & $0,-\frac{1}{2}$\tabularnewline
\hline 
$\vdots$ & $\vdots$\tabularnewline
\hline 
$-l,+\frac{1}{2}$ & $-l,-\frac{1}{2}$\tabularnewline
\hline 
\multicolumn{1}{c}{$\downarrow$} & \multicolumn{1}{c}{$\downarrow$}\tabularnewline
\hline 
$+\frac{1}{2}$ & $-\frac{1}{2}$\tabularnewline
\hline 
\end{tabular} .\vspace{0.5\baselineskip}
\par\end{center}%
\end{minipage} This BPT specifies the minimal isometries and the partial trace map
over spin-$l$ as given in Eq. \ref{eq:part trace over spin l}.

Now, consider the dynamics in the form of the interaction Hamiltonian

\noindent 
\[
H_{int}=-\epsilon\,L_{z}\otimes\sigma_{z},
\]

\noindent with the operator $\sigma_{z}=\Pi^{\left(+\right)}-\Pi^{\left(-\right)}$,
where $\Pi^{\left(\pm\right)}=\ketbra{\pm\frac{1}{2}}{\pm\frac{1}{2}}$.
Let us separate $H_{int}$ into two terms supported on the two columns
of the above BPT
\[
H_{int}=-\epsilon\,\left(L_{z}\otimes\Pi^{\left(+\right)}-\,L_{z}\otimes\Pi^{\left(-\right)}\right)=H_{int}^{\left(+\right)}\otimes\Pi^{\left(+\right)}+H_{int}^{\left(-\right)}\otimes\Pi^{\left(-\right)},
\]
where $H_{int}^{\left(\pm\right)}:=\mp\epsilon\,L_{z}$. Using the
fact that $\Pi^{\left(+\right)}$ and $\Pi^{\left(-\right)}$ are
orthogonal, the overall time evolution is given by 
\begin{align*}
e^{-itH_{int}} & =\sum_{n=0}^{\infty}\frac{1}{n!}\left(-itH_{int}^{\left(+\right)}\otimes\Pi^{\left(+\right)}-itH_{int}^{\left(-\right)}\otimes\Pi^{\left(-\right)}\right)^{n}\\
 & =\sum_{n=0}^{\infty}\frac{1}{n!}\left(-itH_{int}^{\left(+\right)}\right)^{n}\otimes\Pi^{\left(+\right)}+\sum_{n=0}^{\infty}\frac{1}{n!}\left(-itH_{int}^{\left(-\right)}\right)^{n}\otimes\Pi^{\left(-\right)}\\
 & =e^{-itH_{int}^{\left(+\right)}}\otimes\Pi^{\left(+\right)}+e^{-itH_{int}^{\left(-\right)}}\otimes\Pi^{\left(-\right)}.
\end{align*}

We can now see that $H_{int}^{\left(+\right)}$ generates the time
evolution inside the subspace of the left column of the BPT, and $H_{int}^{\left(-\right)}$
generates the evolution inside the right column. Since $H_{int}^{\left(+\right)}$
and $H_{int}^{\left(-\right)}$ differ in the overall sign (this traces
back to the eigenvalues of $\sigma_{z}$), the column kets evolve
in opposite directions inside the columns. From that we conclude that
the Hamiltonian $H_{int}$ will drive the column kets apart, which
will reduce their overlap, and that kills off the coherence terms
in the reduced states.

For concreteness, consider the initial product state
\[
\ket{\psi}_{AB}:=\ket{\varphi}\left(\alpha\ket{+\frac{1}{2}}+\beta\ket{-\frac{1}{2}}\right).
\]
After some time $t$ we will have
\[
\ket{\psi\left(t\right)}_{AB}=\ket{\varphi_{+}\left(t\right)}\ket{+\frac{1}{2}}+\ket{\varphi_{-}\left(t\right)}\ket{-\frac{1}{2}},
\]
where
\[
\ket{\varphi_{+}\left(t\right)}:=\alpha e^{-itH_{int}^{\left(+\right)}}\ket{\varphi},\hspace{1cm}\ket{\varphi_{-}\left(t\right)}:=\beta e^{-itH_{int}^{\left(-\right)}}\ket{\varphi}.
\]
The coefficient of the reduced coherence term $\text{\ensuremath{\ket{+\frac{1}{2}}\bra{-\frac{1}{2}}}}$
of spin-$\frac{1}{2}$ is given by the overlap of the corresponding
column kets
\[
\bra{\varphi_{-}\left(t\right)}\bra{-\frac{1}{2}}S_{-\frac{1}{2},\frac{1}{2}}\ket{\varphi_{+}\left(t\right)}\ket{+\frac{1}{2}}=\braket{\varphi_{-}\left(t\right)}{\varphi_{+}\left(t\right)}=\beta^{*}\alpha\bra{\varphi}e^{itH_{int}^{\left(-\right)}}e^{-itH_{int}^{\left(+\right)}}\ket{\varphi}.
\]
Unless $\ket{\varphi}$ is an eigenstate of $H_{int}^{\left(\pm\right)}$,
the overlap $\bra{\varphi}e^{itH_{int}^{\left(-\right)}}e^{-itH_{int}^{\left(+\right)}}\ket{\varphi}$
will vanish with time (the rate depends on the coupling strength $\epsilon$
and the magnitude of spin-$l$).

In Fig. \ref{fig:Purity-against-time} we have plotted the purity
$\tr\left(\rho_{B}^{2}\right)$ of the reduced state of spin-$\frac{1}{2}$
coupled to spin-$100$ with the initial state 
\[
\ket{\psi}_{AB}:=\ket{m_{x}=100}\left(\frac{1}{\sqrt{2}}\ket{+\frac{1}{2}}+\frac{1}{\sqrt{2}}\ket{-\frac{1}{2}}\right).
\]
Here $\ket{m_{x}=100}$ is the maximally $\hat{x}$-polarized eigenstate
of $L_{x}$.
\begin{figure}[H]
\centering{}\includegraphics[width=0.7\columnwidth]{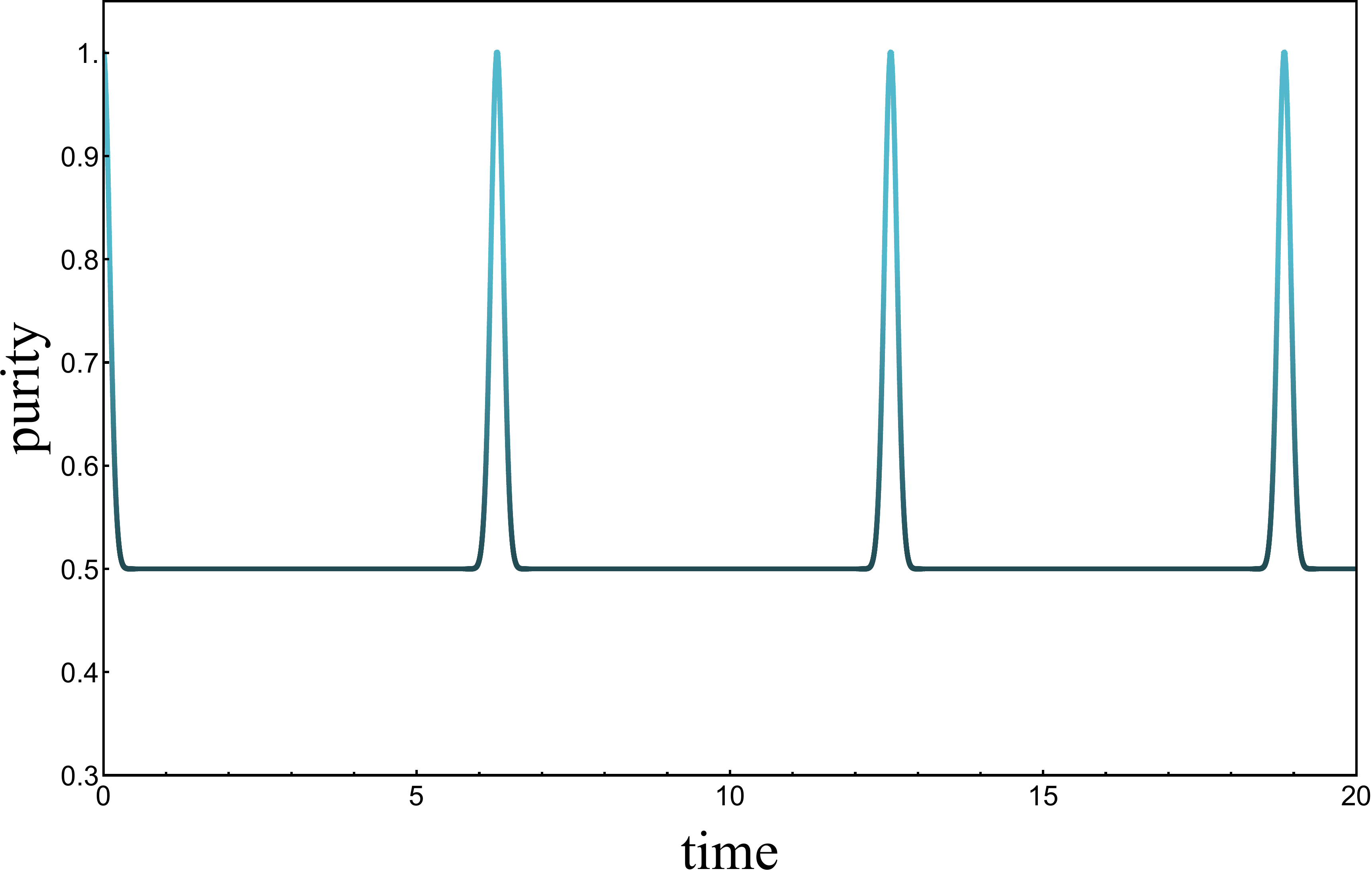}
\caption{\label{fig:Purity-against-time}Time dependent purity of spin-$\frac{1}{2}$
coupled to spin-$100$.}
\end{figure}

As we can see, starting with the maximal purity of $1$ the purity
rapidly drops to its minimal value of $0.5$ and stays there until
it reaches periodic brief revivals back to $1$. This can be understood
from the behavior of the overlaps between two column kets as they
evolve. Initially, both column kets correspond to the same state of
spin-$100$ that is polarized in the $\hat{x}$ direction. As they
evolve with the opposite Hamiltonians $H_{int}^{\left(\pm\right)}=\mp\frac{\epsilon}{2}\,L_{z}$,
they rotate in opposite directions in the $\hat{x}-\hat{y}$ plane;
see Fig. \ref{fig:opposite rotating spins}. Thus, the overlap between
the two column kets rapidly vanishes, and after a while it briefly
revives as they periodically meet in the $\hat{x}-\hat{y}$ plane.

\begin{figure}[H]
\centering{}\includegraphics[width=0.5\columnwidth]{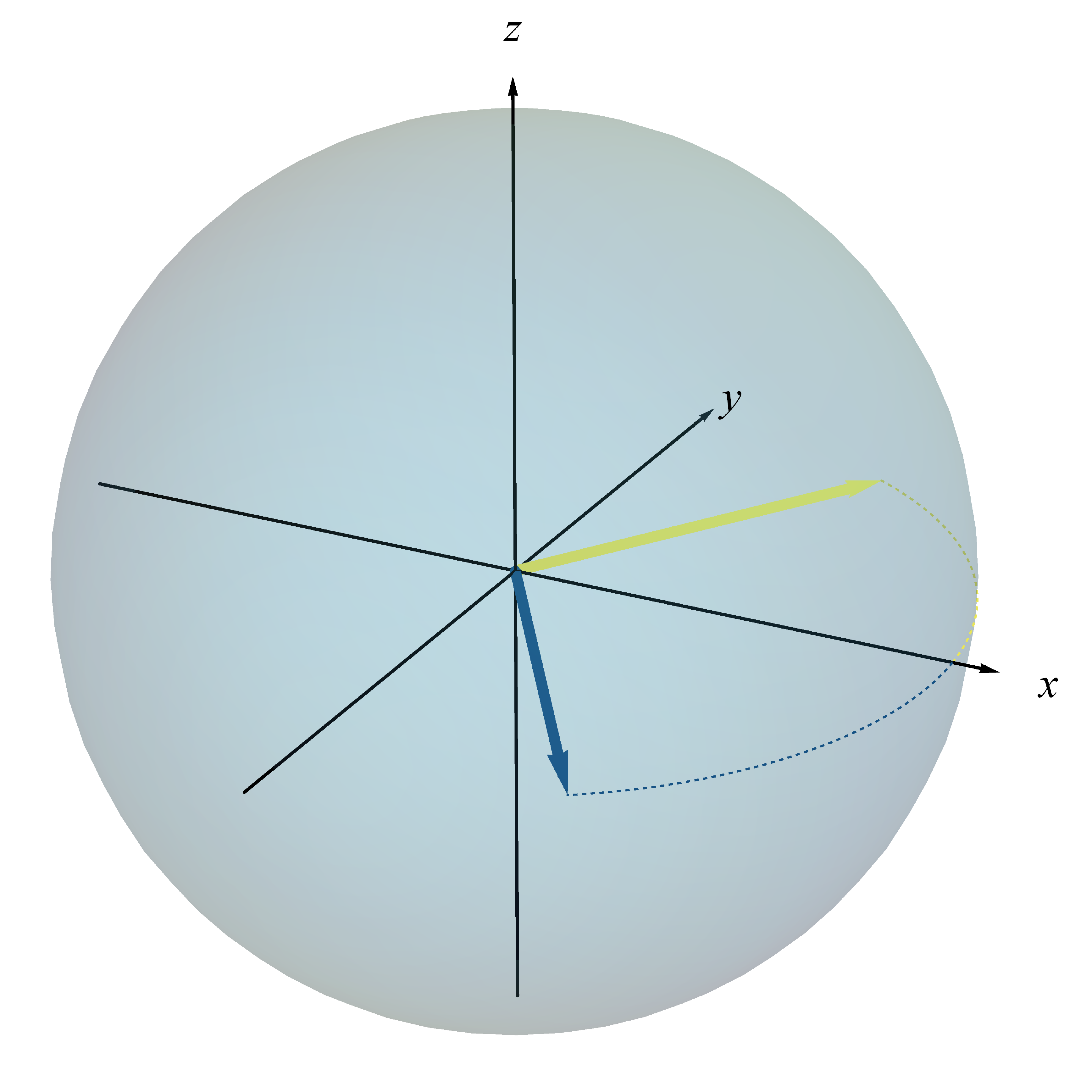}
\caption{\label{fig:opposite rotating spins}Column kets of maximally polarized
spin states rotate in opposite directions in the $\hat{x}-\hat{y}$
plane. The periodic brief alignments of the kets is responsible for
the periodic brief revivals of coherence.}
\end{figure}

As simplistic as the above example is, it demonstrates how the decoherence
of reduced states can be described without referring to subsystems
and instead focus on the BPT that specifies the state reduction. We
can solidify this observation by considering how the interaction terms
and the non-interaction terms of the Hamiltonian act on the BPT in
the generic bipartite system $\mathcal{H}_{A}\otimes\mathcal{H}_{B}$.

The self-Hamiltonian of system $A$ acts identically on all the columns
of the BPT as 
\[
H_{self:A}:=H_{A}\otimes I_{B}=\sum_{k}H_{A}\otimes\ketbra{b_{k}}{b_{k}}.
\]
This means that $H_{self:A}$ drives all the column kets in sync,
and that does not diminish their overlaps and does not cause decoherence.
The self-Hamiltonian of system $B$ does not generate dynamics inside
the columns at all, instead it generates dynamics inside the row subspaces
\[
H_{self:B}:=I_{A}\otimes H_{B}=\sum_{i}\ketbra{a_{i}}{a_{i}}\otimes H_{B}.
\]
This, of course, changes the reduced states but it does so unitarily.

What causes decoherence are the Hamiltonian terms that generate unsynchronized
evolutions inside the column subspaces which eliminates the overlaps
between the column kets; see Fig. \ref{fig:generic col kets} for
an illustration. This property is what characterizes the generic interaction
term

\[
H_{int}:=\sum_{k}H_{A_{k}}\otimes\ketbra kk,
\]
where the column Hamiltonians $H_{A_{k}}$ vary with $k$.

\begin{figure}[H]
\centering{}\includegraphics[viewport=75bp 150bp 1000bp 700bp,clip,width=0.49\columnwidth]{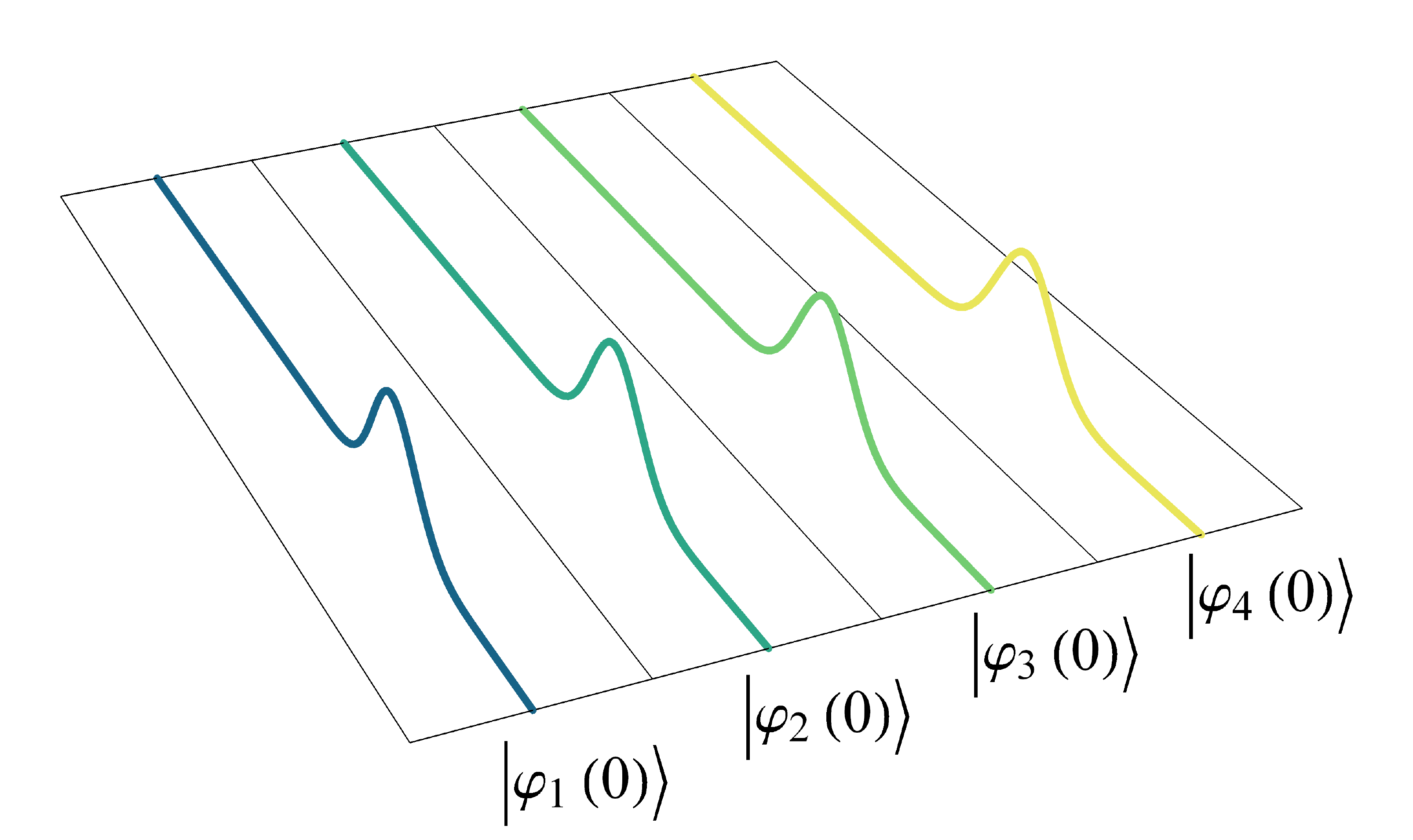}
\includegraphics[viewport=80bp 75bp 1000bp 625bp,clip,width=0.49\columnwidth]{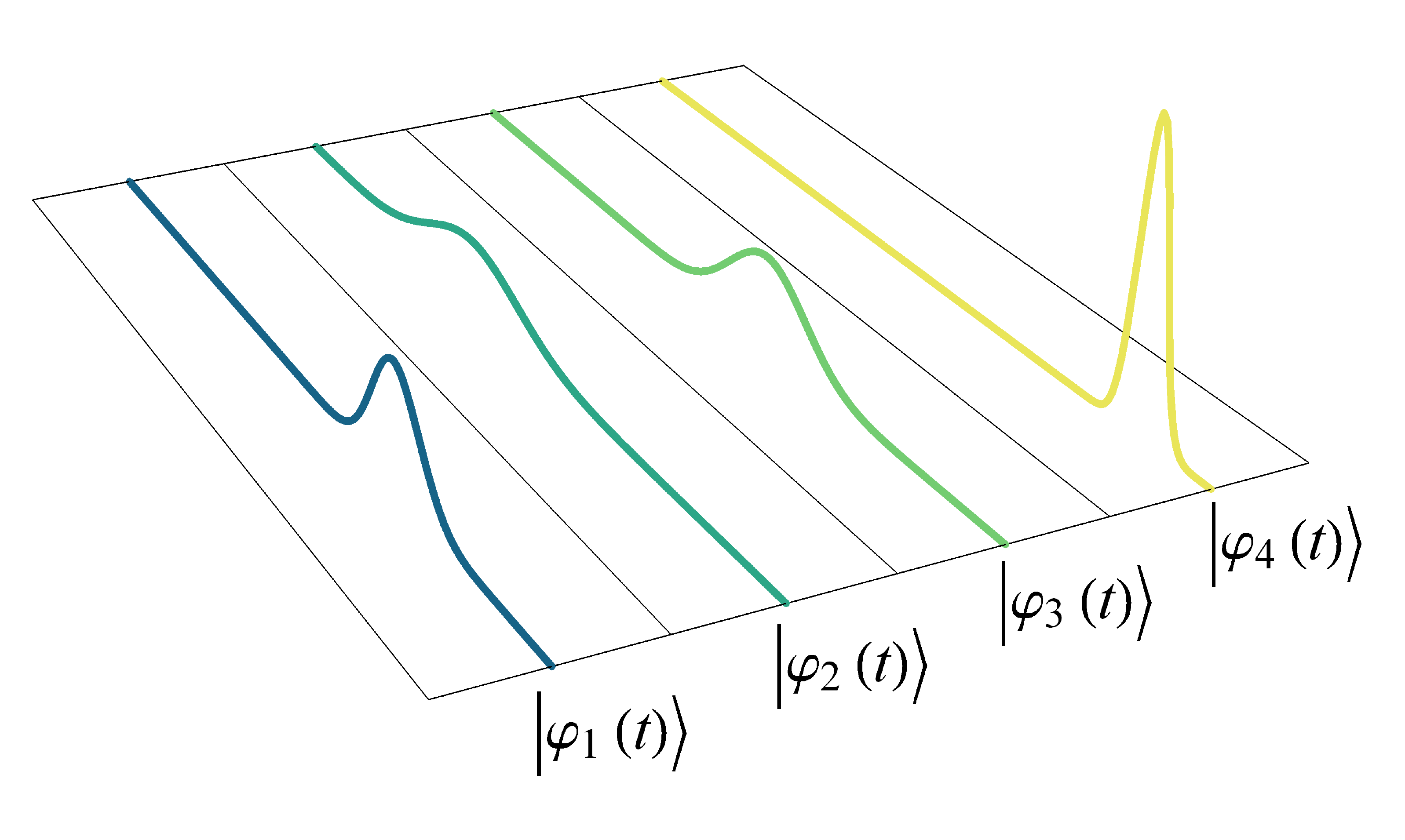}
\caption{\label{fig:generic col kets}Generic picture of column kets on a bipartition
table before (left) and after (right) unsynchronized evolution inside
the columns caused by the interaction term. The diminishing overlaps
between the different column kets translates into diminishing coherence
terms in the reduced state.}
\end{figure}

\section{State reductions and decoherence due to operational constraints\label{sec:State-reudcions-due}}

In the previous section we have derived the partial trace map by imposing
the operational constraint such that only the observables from the
algebra $I_{A}\otimes\mathcal{L}\left(\mathcal{H}_{B}\right)$ are
accessible. This does not imply that all the observables $I_{A}\otimes O_{B}$
have to be accessible, but it does exclude what definitely cannot
be accessed (information about $A$ in this case). The resulting state
reduction map accounts for the operational constraint of not having
access to one of the subsystems.

It is then natural to consider the operational constraints that go
beyond the restriction to physical subsystems. These constraints can
be the result of having only access to collective observables of a
composite system, or having a noisy quantum channel that renders some
observables irrelevant by randomizing their outcomes. In such scenarios,
what specifies the operational constraint is not the physical subsystem
decomposition, such as the system-environment split, but an algebra
$\mathcal{A}\subset\mathcal{L}\left(\mathcal{H}\right)$ of relevant
observables. This algebra, of course, may be an overstatement of the
practical reality and not all $O\in\mathcal{A}$ are necessarily accessible,
but it is still a useful notion that excludes what is definitely out
of reach.

For example, consider again the Hilbert space $\mathcal{H}:=\underline{l}\otimes\underline{\frac{1}{2}}$
of spin-$l$ and spin-$\frac{1}{2}$ and let us assume that in principle
we can measure whatever we want. However, there is some uncontrollable
effect that applies an unknown random rotation $U_{l}\left(R\right)\otimes I$
on the spin-$l$ subsystem, while leaving spin-$\frac{1}{2}$ unaffected.
Since we do not know what rotation has been applied, this effect renders
the orientation of spin-$l$ completely random. Under such circumstances
the only observables that remain relevant are the ones that are unaffected
by the unknown rotations. Therefore, this uncontrollable effect imposes
an operational constraint that restricts the algebra of relevant observables
to the ones that commute with all the rotations of spin-$l$:
\begin{equation}
\mathcal{A}=\left\{ O\in\mathcal{L}\left(\mathcal{H}\right)\,|\,\left[O,U_{l}\left(R\right)\otimes I\right]=0\,\,\,\forall R\right\} .\label{eq: commutant of spin_l rot}
\end{equation}

Once the algebra of relevant observables is identified, we would like
to define the state reduction map that accounts for the limitations
imposed by the operational constraints. Such reduced state will then
be the effective state of the system that we have access to in light
of the operational constraints.

Following the operator-algebraic derivation of the partial trace map
in the previous section, we can adopt the same approach to produce
state reduction maps for any operator algebra (provided we can find
its irreps structure). The derivation of the general state reduction
map is outlined as follows:
\begin{enumerate}
\item Identify the operator algebra $\mathcal{A}\subseteq\mathcal{L}\left(\mathcal{H}\right)$
of relevant observables.
\item Find its irreps structure $\mathcal{H}\cong\bigoplus_{q}\mathcal{H}_{\nu_{q}}\otimes\mathcal{H}_{\mu_{q}}$
and identify the minimal isometries $\left\{ S_{kl}^{q}\right\} $.
\item The Hilbert space of reduced states is given by 
\begin{equation}
\mathcal{H}_{\left\{ \mu_{q}\right\} }:=\bigoplus_{q}\mathcal{H}_{\mu_{q}},\label{eq: reduced Hilb space induced by Alg}
\end{equation}
and the state reduction map is
\begin{equation}
\tr_{\left\{ \nu_{q}\right\} }\left[\rho\right]:=\bigoplus_{q}\sum_{kl}\tr\left[S_{kl}^{q}\rho\right]\ket{m_{l}^{q}}\bra{m_{k}^{q}}.\label{eq:general state reduc map}
\end{equation}
\end{enumerate}
The physical meaning of state reduction maps depends on the physical
context. For the usual partial trace map, the reduced state is the
effective state of a physical subsystem. In the more general case
we can think of reduced states as the states of some virtual subsystems
which embody the degrees of freedom associated with the algebra of
relevant observables. Ultimately, it is the algebra of the relevant
observables that gives meaning to the reduced state.

Returning to our example of spin-$l$ and spin-$\frac{1}{2}$, the
irreps of rotations $U_{l}\left(R\right)\otimes I$ are specified
by the BPT

\noindent\begin{minipage}[c]{1\columnwidth}%
\begin{center}
\vspace{0.5\baselineskip}
\begin{tabular}{|c|c|c|c|c|}
\hline 
$+l,+\frac{1}{2}$ & $\cdots$ & $0,+\frac{1}{2}$ & $\cdots$ & $-l,+\frac{1}{2}$\tabularnewline
\hline 
$+l,-\frac{1}{2}$ & $\cdots$ & $0,-\frac{1}{2}$ & $\cdots$ & $-l,-\frac{1}{2}$\tabularnewline
\hline 
\end{tabular}.\vspace{0.5\baselineskip}
\par\end{center}%
\end{minipage} Since the algebra of relevant observables \eqref{eq: commutant of spin_l rot}
is the commutant of these rotations, the irreps we are interested
in are given by the transposition of this BPT

\noindent\begin{minipage}[c]{1\columnwidth}%
\begin{center}
\vspace{0.5\baselineskip}
\begin{tabular}{|c|c|}
\hline 
$+l,+\frac{1}{2}$ & $+l,-\frac{1}{2}$\tabularnewline
\hline 
$\vdots$ & $\vdots$\tabularnewline
\hline 
$0,+\frac{1}{2}$ & $0,-\frac{1}{2}$\tabularnewline
\hline 
$\vdots$ & $\vdots$\tabularnewline
\hline 
$-l,+\frac{1}{2}$ & $-l,-\frac{1}{2}$\tabularnewline
\hline 
\multicolumn{1}{c}{$\downarrow$} & \multicolumn{1}{c}{$\downarrow$}\tabularnewline
\hline 
$+\frac{1}{2}$ & $-\frac{1}{2}$\tabularnewline
\hline 
\end{tabular}\vspace{0.5\baselineskip}
\par\end{center}%
\end{minipage} As we have seen in the previous section, this BPT defines the partial
trace over spin-$l$. Not surprisingly, the appropriate state reduction
map that accounts for having the spin-$l$ subsystem randomly rotated,
is the erasure of all information about the state of spin-$l$.

In the more general cases of state reductions, the reduced Hilbert
space \eqref{eq: reduced Hilb space induced by Alg} is a direct sum
of orthogonal sectors $\left\{ \mathcal{H}_{\mu_{q}}\right\} $ identified
by the distinct inequivalent irreps of the algebra. These sectors
are commonly referred to as \emph{superselection sectors}. In the
case of the partial trace map we have only one distinct irrep so the
reduced Hilbert space has only one superselection sector. When multiple
superselection sectors are present (that is when the BPT has multiple
blocks), the state reduction map eliminates all coherence terms between
the basis elements belonging to distinct superselection sectors, regardless
of the state.

We will see how the superselection sectors arise from operational
constraints in the more elaborate examples below. There is one extremal
case, however, that we can briefly point out here.

Consider the operational constraint that allows only one observable
$O\in\mathcal{L}\left(\mathcal{H}\right)$ to be measured. In this
case, the algebra of relevant observables is just $\left\langle O\right\rangle $.
The minimal isometries of this algebra are the spectral projections
$\left\{ \Pi^{\left(\lambda\right)}\right\} $ of $O$, where the
eigenvalues $\lambda\neq0$ identify the distinct one-dimensional
irreps. Because the irreps are one-dimensional, the reduced Hilbert
space \eqref{eq: reduced Hilb space induced by Alg} in this case
is just $\bigoplus_{\lambda}\ket{\lambda}$. The state reduction map
\eqref{eq:general state reduc map} is then
\[
\rho\longmapsto\bigoplus_{\lambda}\tr\left[\Pi^{\left(\lambda\right)}\rho\right]\ket{\lambda}\bra{\lambda}.
\]
The resulting reduced state is completely diagonal and it represents
the probability distribution over the observable's outcomes $\left\{ \lambda\right\} $.

Not surprisingly, when we constrain the measurements to a single observable,
the state reduction map becomes the mapping of the quantum state $\rho$
to the classical probability distribution $p\left(\lambda\right):=\tr\left[\Pi^{\left(\lambda\right)}\rho\right]$
over the outcomes of that observable. As we can see, the extremal
constraint of having only one observable leads to the complete elimination
of coherence terms. Therefore, when more observables are available,
we expect an intermediate outcome where the coherence terms between
some subspaces are eliminated while the coherence terms inside these
subspaces are preserved.

Such elimination of coherence terms by the state reduction map is
what we call superselection, and it is distinct from the dynamical
elimination of coherence terms in the process of decoherence.\footnote{The dynamical elimination of coherence terms is sometimes referred
to as \emph{einselection}, which stands for environment-induced-superselection
\citep{zurek2003decoherence}.}

In order to clarify these ideas we will now study three more elaborate
examples.

\newcounter{example}
\refstepcounter{example}

\subsubsection*{Example \theexample \refstepcounter{example}}

In this example we will study a simple case where the superselection
sectors appear due to the lack of a shared reference frame. See \citep{Bartlett07}
for a review of this topic.

When the agent that prepares the states (Alice) and the agent that
measures them (Bob) do not share a common reference frame, it imposes
an operational constraint on the latter. Let us consider such situation
with the same system $\mathcal{H}:=\underline{l}\otimes\underline{\frac{1}{2}}$
of spin-$l$ and spin-$\frac{1}{2}$.

We will assume that Alice and Bob share a common reference frame for
the $\hat{z}$ axis but they are misaligned in the $\hat{x}-\hat{y}$
plane by an unknown angle $\theta$. If Alice can send multiple states
to Bob then he could implement some protocol for aligning his reference
frame with Alice by inferring the angle $\theta$ from the collection
of states. If, however, Bob receives only one state then from his
perspective it is rotated by an unknown angle around $\hat{z}$. The
only relevant observables that remain for Bob are the ones that commute
with $U_{\hat{z}}\left(\theta\right)$ for all $\theta$. Our goal
is to find the state reduction map that accounts for the lack of common
reference frame in the $\hat{x}-\hat{y}$ plane.

The abelian group $\left\{ U_{\hat{z}}\left(\theta\right)\right\} $
is generated by the single $J_{z}$ component of the total angular
momentum operator. The irreps of this group are all one-dimensional
and are given by the eigenvectors $\ket{j,m}$ of $J_{z}$, with distinct
eigenvalues $m$ identifying distinct irreps. For brevity, let us
specialize to $l=1$ so $j=1\pm\frac{1}{2}$.

The one-dimensional irreps of $\left\{ U_{\hat{z}}\left(\theta\right)\right\} $
are summarized by the BPT

\noindent\begin{minipage}[c]{1\columnwidth}%
\begin{center}
\vspace{0.5\baselineskip}
\begin{tabular}{c|c|c|c}
\cline{1-1} 
\multicolumn{1}{|c|}{$\frac{3}{2},\frac{3}{2}$} & \multicolumn{1}{c}{} & \multicolumn{1}{c}{} & \tabularnewline
\cline{1-2} \cline{2-2} 
 & $\frac{3}{2},\frac{1}{2}$ & \multicolumn{1}{c}{} & \tabularnewline
\cline{2-2} 
 & $\frac{1}{2},\frac{1}{2}$ & \multicolumn{1}{c}{} & \tabularnewline
\cline{2-3} \cline{3-3} 
\multicolumn{1}{c}{} &  & $\frac{3}{2},-\frac{1}{2}$ & \tabularnewline
\cline{3-3} 
\multicolumn{1}{c}{} &  & $\frac{1}{2},-\frac{1}{2}$ & \tabularnewline
\cline{3-4} \cline{4-4} 
\multicolumn{1}{c}{} & \multicolumn{1}{c}{} &  & \multicolumn{1}{c|}{$\frac{3}{2},-\frac{3}{2}$}\tabularnewline
\cline{4-4} 
\end{tabular}.\vspace{0.5\baselineskip}
\par\end{center}%
\end{minipage} Note that the abelian group $\left\{ U_{\hat{z}}\left(\theta\right)\right\} $
acts on the states in each BPT block with a different phase factor;
this is why the eigenvalues of $J_{z}$ distinguish the irreps.

The commutant algebra, and the implied state reduction map, are given
by the transposition of the BPT (note that we re-use the same labels
for the reduced basis on the bottom)

\noindent\begin{minipage}[c]{1\columnwidth}%
\begin{center}
\vspace{0.5\baselineskip}
\begin{tabular}{c|cc|cc|c}
\cline{1-1} 
\multicolumn{1}{|c|}{$\frac{3}{2},\frac{3}{2}$} &  & \multicolumn{1}{c}{} &  & \multicolumn{1}{c}{} & \tabularnewline
\cline{1-3} \cline{2-3} \cline{3-3} 
 & \multicolumn{1}{c|}{$\frac{3}{2},\frac{1}{2}$} & $\frac{1}{2},\frac{1}{2}$ &  & \multicolumn{1}{c}{} & \tabularnewline
\cline{2-5} \cline{3-5} \cline{4-5} \cline{5-5} 
\multicolumn{1}{c}{} &  &  & \multicolumn{1}{c|}{$\frac{3}{2},-\frac{1}{2}$} & $\frac{1}{2},-\frac{1}{2}$ & \tabularnewline
\cline{4-6} \cline{5-6} \cline{6-6} 
\multicolumn{1}{c}{} &  & \multicolumn{1}{c}{} &  &  & \multicolumn{1}{c|}{$\frac{3}{2},-\frac{3}{2}$}\tabularnewline
\cline{6-6} 
\multicolumn{1}{c}{$\downarrow$} & $\downarrow$ & \multicolumn{1}{c}{$\downarrow$} & $\downarrow$ & \multicolumn{1}{c}{$\downarrow$} & $\downarrow$\tabularnewline
\hline 
\multicolumn{1}{|c|}{$\frac{3}{2},\frac{3}{2}$} & \multicolumn{1}{c|}{$\frac{3}{2},\frac{1}{2}$} & $\frac{1}{2},\frac{1}{2}$ & \multicolumn{1}{c|}{$\frac{3}{2},-\frac{1}{2}$} & $\frac{1}{2},-\frac{1}{2}$ & \multicolumn{1}{c|}{$\frac{3}{2},-\frac{3}{2}$}\tabularnewline
\hline 
\end{tabular}.\vspace{0.5\baselineskip}
\par\end{center}%
\end{minipage} Instead of explicitly specifying all the minimal isometries needed
for the definition of the state reduction map as given in Eq. \eqref{eq:general state reduc map},
we can read the implied state reduction map directly from the BPT.
It tells us that all the basis elements reduce to themselves but only
the coherence terms between elements in the same row remain in tact;
all other coherence terms are eliminated.

One can verify explicitly using the definition in Eq. \eqref{eq:general state reduc map}
that the implied state reduction map reduces the pure states $\ket{\psi}=\sum_{j,m}c_{j,m}\ket{j,m}$
to
\[
\ket{\psi}\longmapsto\rho_{Bob}=\begin{pmatrix}\left|c_{\frac{3}{2},\frac{3}{2}}\right|^{2}\\
 & \left|c_{\frac{3}{2},\frac{1}{2}}\right|^{2} & c_{\frac{3}{2},\frac{1}{2}}c_{\frac{1}{2},\frac{1}{2}}^{*}\\
 & c_{\frac{1}{2},\frac{1}{2}}c_{\frac{3}{2},\frac{1}{2}}^{*} & \left|c_{\frac{1}{2},\frac{1}{2}}\right|^{2}\\
 &  &  & \left|c_{\frac{3}{2},-\frac{1}{2}}\right|^{2} & c_{\frac{3}{2},-\frac{1}{2}}c_{\frac{1}{2},-\frac{1}{2}}^{*}\\
 &  &  & c_{\frac{1}{2},-\frac{1}{2}}c_{\frac{3}{2},-\frac{1}{2}}^{*} & \left|c_{\frac{1}{2},-\frac{1}{2}}\right|^{2}\\
 &  &  &  &  & \left|c_{\frac{3}{2},-\frac{3}{2}}\right|^{2}
\end{pmatrix}.
\]
The coherence terms that got eliminated are exactly where the unknown
phase factors due to the unknown rotation $U_{\hat{z}}\left(\theta\right)$
were present. Since Bob has no access to this phase factor, whether
Alice sends him $\ket{\psi_{1}}$ or $\ket{\psi_{2}}=U_{\hat{z}}\left(\theta'\right)\ket{\psi_{1}}$
(for any $\theta'$), there is nothing he can do that will differentiate
the two cases. The resulting state reduction map accounts for that
by eliminating the coherence terms whose values remain unknown within
Bob's operational constraint.

In this example we saw a state reduction map that enforces superselection
between states with different eigenvalues of $J_{z}$ due to the lack
of common reference frame. In the next example we will see a state
reduction map that combines superselection with a partial-trace-like
map.

\subsubsection*{Example \theexample \refstepcounter{example}}

In this example we will consider the dynamics of a reduced state.
The main takeaway here is that decoherence does not have to be only
the consequence of interactions with inaccessible subsystems, it can
also arise from other combinations of dynamics and operational constraints.

Let us consider the composite system $\mathcal{H}:=\underline{l}\otimes\underline{\frac{1}{2}}\otimes\underline{\frac{1}{2}}$
of two spin-$\frac{1}{2}$'s and an integer angular momentum $l$,
such as the Hydrogen atom. Assume that we have a large ensemble of
$N$ Hydrogen atoms and we have come up with a procedure that allows
us to prepare all of them in the same arbitrary state $\ket{\psi}\in\mathcal{H}$.
Unfortunately, we cannot control the individual orientations of the
atoms so instead of having the collective state $\ket{\psi}^{\otimes N}$,
each atom ends up in the state $U\left(R\right)\ket{\psi}$, where
$U\left(R\right)$ is a random rotation that is independently chosen
for each atom. If we sample a single atom from this ensemble, what
is the effective state of this atom?

The operational constraint that the limitation of state preparations
imposes, is the restriction to rotationally invariant measurements.
The algebra of relevant observables is therefore the commutant of
the $SU\left(2\right)$ group. We can find the irreps structure of
this algebra from the representation theory of $SU\left(2\right)$.

We know that our Hilbert space decomposes into
\[
\underline{l}\otimes\underline{\frac{1}{2}}\otimes\underline{\frac{1}{2}}=\underline{l}\otimes\left(\underline{1}\oplus\underline{0}\right)=\underline{l+1}\oplus\underline{l}\oplus\underline{l-1}\oplus\underline{l}.
\]
That is, under $SU\left(2\right)$ rotations we have one irrep that
transforms as $j=l+1$, one that transforms as $j=l-1$, and two irreps
that transform as $j=l$. Note that the two $j=l$ irreps can be distinguished
by whether the spins are in the singlet or triplet states. We thus
have the $j=l\pm1$ total angular momentum basis $\ket{l\pm1,m}$,
and the singlet / triplet variants $\alpha=s,t$ of the $j=l$ total
angular momentum basis $\ket{l,\alpha,m}$. The irreps of $SU\left(2\right)$
can then be specified by the BPT (in each row $m=j,...,-j$).

\noindent\begin{minipage}[c]{1\columnwidth}%
\begin{center}
\vspace{0.5\baselineskip}
\begin{tabular}{ccc|c|c|c|ccc}
\cline{1-3} \cline{2-3} \cline{3-3} 
\multicolumn{1}{|c|}{$\cdots$} & \multicolumn{1}{c|}{$l+1,m$} & $\cdots$ & \multicolumn{1}{c}{} & \multicolumn{1}{c}{} & \multicolumn{1}{c}{} &  &  & \tabularnewline
\cline{1-6} \cline{2-6} \cline{3-6} \cline{4-6} \cline{5-6} \cline{6-6} 
 &  &  & $\cdots$ & $l,s,m$ & $\cdots$ &  &  & \tabularnewline
\cline{4-6} \cline{5-6} \cline{6-6} 
 &  &  & $\cdots$ & $l,t,m$ & $\cdots$ &  &  & \tabularnewline
\cline{4-9} \cline{5-9} \cline{6-9} \cline{7-9} \cline{8-9} \cline{9-9} 
 &  & \multicolumn{1}{c}{} & \multicolumn{1}{c}{} & \multicolumn{1}{c}{} &  & \multicolumn{1}{c|}{$\cdots$} & \multicolumn{1}{c|}{$l-1,m$} & \multicolumn{1}{c|}{$\cdots$}\tabularnewline
\cline{7-9} \cline{8-9} \cline{9-9} 
\end{tabular}\vspace{0.5\baselineskip}
\par\end{center}%
\end{minipage}

The irreps structure of the commutant algebra is given by its transposition

\noindent\begin{minipage}[c]{1\columnwidth}%
\begin{center}
\vspace{0.5\baselineskip}
\begin{tabular}{c|cc|c}
\cline{1-1} 
\multicolumn{1}{|c|}{$\vdots$} &  & \multicolumn{1}{c}{} & \tabularnewline
\cline{1-1} 
\multicolumn{1}{|c|}{$l+1,m$} &  & \multicolumn{1}{c}{} & \tabularnewline
\cline{1-1} 
\multicolumn{1}{|c|}{$\vdots$} &  & \multicolumn{1}{c}{} & \tabularnewline
\cline{1-3} \cline{2-3} \cline{3-3} 
 & \multicolumn{1}{c|}{$\vdots$} & $\vdots$ & \tabularnewline
\cline{2-3} \cline{3-3} 
 & \multicolumn{1}{c|}{$l,s,m$} & $l,t,m$ & \tabularnewline
\cline{2-3} \cline{3-3} 
 & \multicolumn{1}{c|}{$\vdots$} & $\vdots$ & \tabularnewline
\cline{2-4} \cline{3-4} \cline{4-4} 
\multicolumn{1}{c}{} &  &  & \multicolumn{1}{c|}{$\vdots$}\tabularnewline
\cline{4-4} 
\multicolumn{1}{c}{} &  &  & \multicolumn{1}{c|}{$l-1,m$}\tabularnewline
\cline{4-4} 
\multicolumn{1}{c}{} &  &  & \multicolumn{1}{c|}{$\vdots$}\tabularnewline
\cline{4-4} 
\multicolumn{1}{c}{$\downarrow$} & $\downarrow$ & \multicolumn{1}{c}{$\downarrow$} & $\downarrow$\tabularnewline
\hline 
\multicolumn{1}{|c|}{$l+1$} & \multicolumn{1}{c|}{$l,s$} & $l,t$ & \multicolumn{1}{c|}{$l-1$}\tabularnewline
\hline 
\end{tabular}\vspace{0.5\baselineskip}
\par\end{center}%
\end{minipage} The reduced Hilbert space consists of the basis $\ket{l+1}$, $\ket{l,s}$,
$\ket{l,t}$, $\ket{l-1}$.

The pure state $\ket{\psi}$ can now be expanded in the column kets
of the above BPT
\[
\ket{\psi}=\ket{\varphi_{l+1}}+\ket{\varphi_{l,s}}+\ket{\varphi_{l,t}}+\ket{\varphi_{l-1}}.
\]
Since different BPT blocks specify distinct irreps, no coherences
between column kets supported on different blocks are preserved; this
is superselection. In the central block, however, where we have two
columns, the reduced coherence terms are given by the overlaps $\bra{\varphi_{l,t}}S_{ts}\ket{\varphi_{l,s}}$
(where $S_{ts}:=\sum_{m=-l}^{l}\ket{l,t,m}\bra{l,s,m}$) between the
column kets; this is the partial-trace-like reduction. Overall, the
state reduction map is summarized as

\[
\ket{\psi}\longmapsto\rho=\begin{pmatrix}\braket{\varphi_{l+1}}{\varphi_{l+1}}\\
 & \braket{\varphi_{l,s}}{\varphi_{l,s}} & \bra{\varphi_{l,t}}S_{ts}\ket{\varphi_{l,s}}\\
 & \bra{\varphi_{l,s}}S_{st}\ket{\varphi_{l,t}} & \braket{\varphi_{l,t}}{\varphi_{l,t}}\\
 &  &  & \braket{\varphi_{l-1}}{\varphi_{l-1}}
\end{pmatrix}.
\]

Therefore, the only information we are left with is the one qubit
encoded between the triplet and singlet variants of the $j=l$ irrep,
and the overall probability distribution over the total angular momentum
$j=l+1,l,l-1$. The consequence of such operational constraint is
even more pronounced if we consider unitary dynamics acting on this
system.

Even the simple Hamiltonian of uniform magnetic field along the $\hat{y}$
axis (without spin-spin or spin-orbit interactions)
\[
H=\epsilon L_{y}+S_{1;y}+S_{2;y}
\]
can induce decoherence in such reduced states. Here $L_{y}$, $S_{1;y}$,
$S_{2;y}$ are the $\hat{y}$ components of the individual angular
momentum operators, and $\epsilon$ is the coupling strength of the
orbital angular momentum to the external field (for the two spins
it is normalized to $1$). We will now see that the non-uniformity
of the coupling strengths is responsible for the decoherence.

Let us separate the Hamiltonian into the uniform and the difference
parts: 
\[
H=\epsilon\left(L_{y}+S_{1;y}+S_{2;y}\right)+\left(1-\epsilon\right)\left(S_{1;y}+S_{2;y}\right).
\]
Now, consider how these two terms act on the above column kets. The
first term is the component $J_{y}=L_{y}+S_{1;y}+S_{2;y}$ of the
total angular momentum operator so it generates global rotations around
the $\hat{y}$ axis according to the representations $j=l+1,l,l-1$.
Therefore, the term $J_{y}$ generates identical time evolutions inside
the two columns of total angular momentum $j=l$, and it does not
map between the columns.

The second Hamiltonian term is the $\hat{y}$ component of the total
spin operator $S_{y}=S_{1;y}+S_{2;y}$. This operator acts trivially
on the singlet spin states $\ket{\varphi_{l,s}}$, but otherwise,
it does not preserve the total angular momentum $j$ and it is free
to map between all $\alpha\neq s$ columns. So, for $\epsilon\neq1$
the overlap between the column kets $\ket{\varphi_{l,s}}$ and $\ket{\varphi_{l,t}}$
will fluctuate as one column ket will remain stationary while the
other will not.

Following the above distinction between the two terms of the Hamiltonian,
we can label them as the effective ``self'' and ``interaction''
terms
\[
H=\epsilon H_{self}+\left(1-\epsilon\right)H_{int}
\]
\[
H_{self}:=L_{y}+S_{1;y}+S_{2;y}\,\,\,\,\,\,\,\,\,\,\,\,\,\,\,\,\,\,H_{int}:=S_{1;y}+S_{2;y}.
\]
The term $H_{self}$ only changes the global orientation of the system
(which we are completely ignorant of due to the operational constraint),
so we can think of $H_{self}$ as the self Hamiltonian of the inaccessible
environment. We can then think of the other term $H_{int}$, as the
effective interaction term because it couples the singlet-triplet
qubit to the global orientation of the system.

For concreteness, let us assume that $l=3$ and $\epsilon=0$ so there
is only the effective interaction term $H=H_{int}$. If the initial
unreduced state is 
\[
\ket{\psi}=\frac{1}{\sqrt{2}}\ket{l,s,0}+\frac{1}{\sqrt{2}}\ket{l,t,0},
\]
then the initial reduced state is 
\[
\rho=\frac{1}{2}\begin{pmatrix}0\\
 & 1 & 1\\
 & 1 & 1\\
 &  &  & 0
\end{pmatrix}.
\]

The purity of $\rho$ as a function of time under the evolution with
$H$ is shown in Fig. \ref{fig:purity of Bob's state}. This illustrates
how the singlet-triplet qubit periodically decoheres into the effective
environment imposed by the operational constraint.

\begin{figure}[H]
\centering{}\includegraphics[width=0.7\columnwidth]{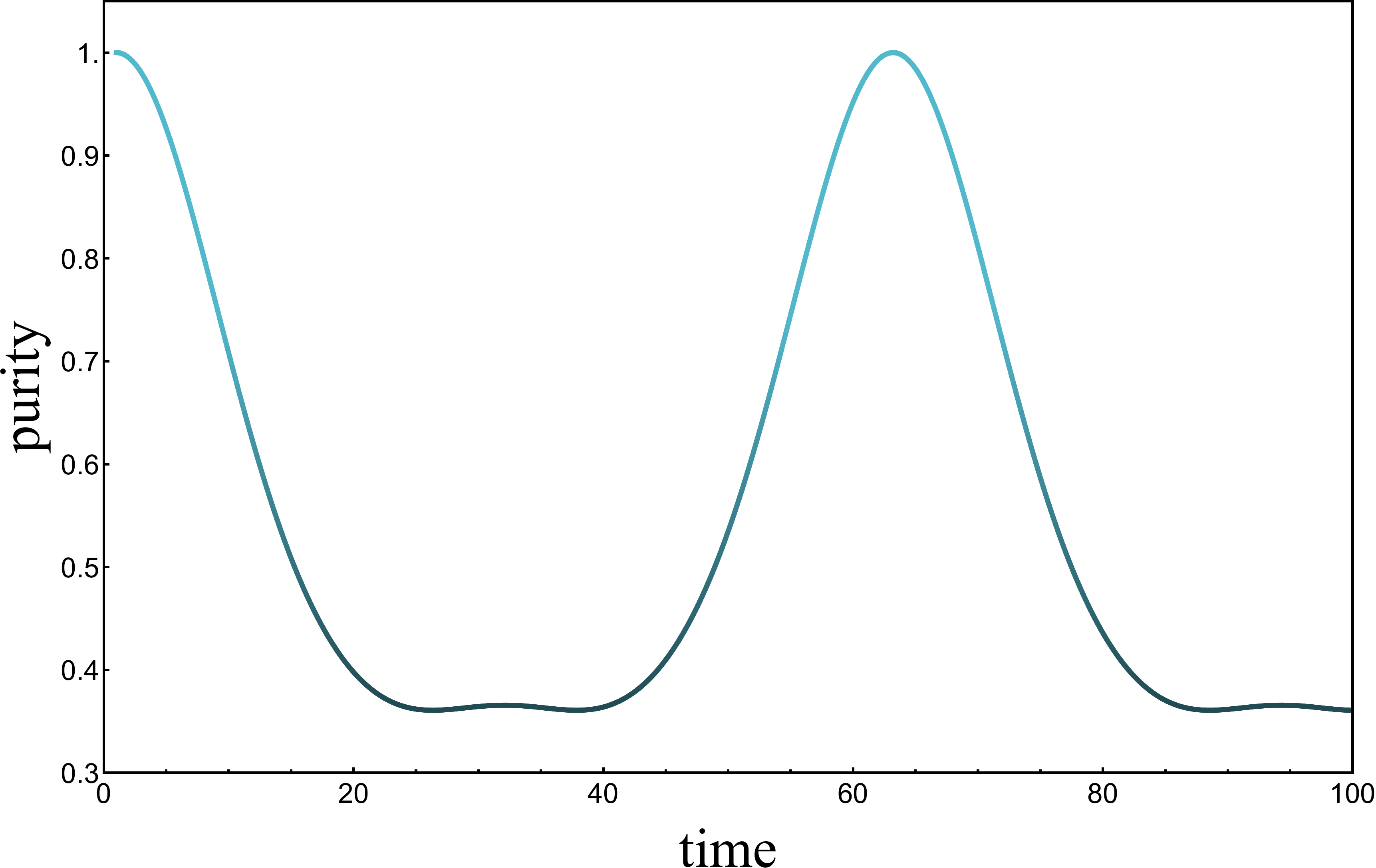}
\caption{\label{fig:purity of Bob's state}The purity of the reduced state
under the evolution with $H=H_{int}$.}
\end{figure}

In Fig. \ref{fig:Bob's column kets on the BPT } we can see the BPT
perspective on this decoherence process. At $t=0$ the initial state
is the even superposition of the two basis elements in the two $j=l$
columns. As time progresses, the column ket of the singlet remains
unchanged while the column ket of the triplet evolves both inside
the $j=l$ column and it leaks into the $j=l\pm1$ columns. Thus,
the overlap between the initial singlet and triplet column kets diminishes.
Periodically, as the spin rotations complete a full cycle, the triplet
column ket returns to its initial configuration (around $t=63$ for
the first time), which results in the revivals of coherence.

\begin{figure}[H]
\centering{}\includegraphics[viewport=0bp 80bp 720bp 540bp,clip,width=1\columnwidth]{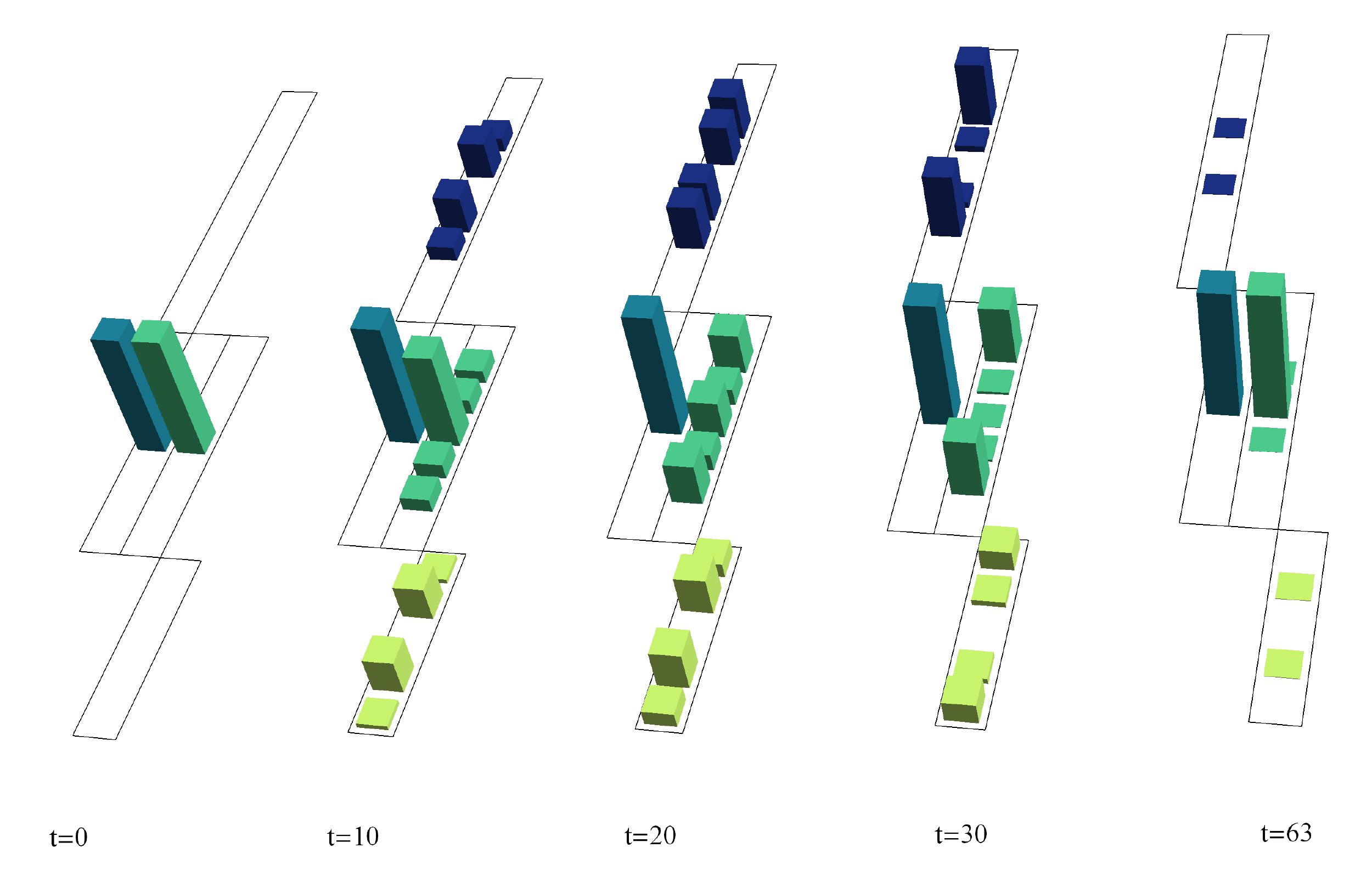}
\caption{\label{fig:Bob's column kets on the BPT }The column kets of the unreduced
state as it evolves on the BPT. The height of each bar corresponds
to the absolute value of the coefficient of the underlying BPT basis
element. The distinct bar color is assigned to each column for contrast
and has no numeric meaning.}
\end{figure}

In conclusion, there is no conceptual difference between decoherence
in this example and decoherence from coupling to a physical environment.
Ultimately, it is the combination of operational constraints and unitary
dynamics that leads to non-unitary evolutions of reduced states.

\subsubsection*{Example \theexample \refstepcounter{example}}

A basic question in quantum information is how to encode a logical
qubit in a physical system in a way that will be least susceptible
to noise. The idea of decoherence free subspaces and subsystems \citep{Knill2000,Kempe01,lidar2003decoherence,lidar2014review}
has emerged to address this question. The main obstacle to finding
decoherence free subspaces and subsystems is finding the relevant
irreps structure, so it is only the irreps of group representations
that are commonly treated. In this example we will demonstrate how
the Scattering Algorithm can expand the scope of treatable problems
beyond group representations.

As before, we consider the composite system $\mathcal{H}:=\underline{l}\otimes\underline{\frac{1}{2}}$
of spin-$l$ (either integer or half integer) and spin-$\frac{1}{2}$.
Assume that Alice wants to send some quantum information to Bob by
encoding it into the physical state of this system. Furthermore, Alice
knows that on its way it will be susceptible to noise dominated by
the Ising interaction $L_{z}\otimes S_{z}$, and spin-$\frac{1}{2}$
rotations generated by $I\otimes S_{x}$. The question is how can
Alice encode quantum information (and how many qubits) without it
being affected by the dominant sources of noise.

We can answer this question by considering the noise as an imposed
operational constraint such that only the observables that commute
with both $L_{z}\otimes S_{z}$ and $I\otimes S_{x}$ remain relevant.
The reduced states will then contain all the information that is unaffected
by noise and the state reduction map will tell us how to encode and
decode it.

In order to address this problem we need to find the irreps of the
commutant of the algebra generated by both $L_{z}\otimes S_{z}$ and
$I\otimes S_{x}$. We know the irreps of the individual terms---they
are given by their eigenvectors---but we do not know the irreps structure
of the combined algebra.

In the following we will use the product basis $\ket{m,\uparrow}$,
$\ket{m,\downarrow}$, where $m=l,...,-l$ and $\uparrow,\downarrow$
are the spin-up spin-down states along $\hat{z}$, or alternatively
$\ket{m,+}$,$\ket{m,-}$ where $+,-$ are the spin-up spin-down states
along $\hat{x}$. Before we address the problem for general $l$,
let us solve it for $l=\frac{1}{2}$.

The combination of the irreps of $L_{z}\otimes S_{z}$ and of $I\otimes S_{x}$
can be expressed as the addition of BPTs

\noindent\begin{minipage}[c]{1\columnwidth}%
\begin{center}
\vspace{0.5\baselineskip}
\begin{tabular}{|c|c|}
\cline{1-1} 
$\frac{1}{2},\uparrow$ & \multicolumn{1}{c}{}\tabularnewline
\cline{1-1} 
$-\frac{1}{2},\downarrow$ & \multicolumn{1}{c}{}\tabularnewline
\hline 
\multicolumn{1}{c|}{} & $\frac{1}{2},\downarrow$\tabularnewline
\cline{2-2} 
\multicolumn{1}{c|}{} & $-\frac{1}{2},\uparrow$\tabularnewline
\cline{2-2} 
\end{tabular}\hspace*{10bp}$\boldsymbol{+}$\hspace*{10bp} %
\begin{tabular}{|c|c|}
\cline{1-1} 
$\frac{1}{2},+$ & \multicolumn{1}{c}{}\tabularnewline
\cline{1-1} 
$-\frac{1}{2},+$ & \multicolumn{1}{c}{}\tabularnewline
\hline 
\multicolumn{1}{c|}{} & $\frac{1}{2},-$\tabularnewline
\cline{2-2} 
\multicolumn{1}{c|}{} & $-\frac{1}{2},-$\tabularnewline
\cline{2-2} 
\end{tabular} \vspace{0.5\baselineskip}
\par\end{center}%
\end{minipage} The left BPT consists of two single-column blocks that correspond
to the distinct eigenvalues of $L_{z}\otimes S_{z}$ with the two
degenerate eigenvectors in each columns. The right BPT comes from
$I\otimes S_{x}$ with similar interpretation. The scattering calculations
in the case of $l=\frac{1}{2}$ are very simple, however, in a anticipation
of the general case let us simplify things even further.

First, we note that separate blocks of BPTs can be detached into separate
terms in the sum, that is

\noindent\begin{minipage}[c]{1\columnwidth}%
\begin{center}
\vspace{0.5\baselineskip}
\begin{tabular}{|c|c|}
\cline{1-1} 
$\frac{1}{2},\uparrow$ & \multicolumn{1}{c}{}\tabularnewline
\cline{1-1} 
$-\frac{1}{2},\downarrow$ & \multicolumn{1}{c}{}\tabularnewline
\hline 
\multicolumn{1}{c|}{} & $\frac{1}{2},\downarrow$\tabularnewline
\cline{2-2} 
\multicolumn{1}{c|}{} & $-\frac{1}{2},\uparrow$\tabularnewline
\cline{2-2} 
\end{tabular} $\boldsymbol{+}$ %
\begin{tabular}{|c|c|}
\cline{1-1} 
$\frac{1}{2},+$ & \multicolumn{1}{c}{}\tabularnewline
\cline{1-1} 
$-\frac{1}{2},+$ & \multicolumn{1}{c}{}\tabularnewline
\hline 
\multicolumn{1}{c|}{} & $\frac{1}{2},-$\tabularnewline
\cline{2-2} 
\multicolumn{1}{c|}{} & $-\frac{1}{2},-$\tabularnewline
\cline{2-2} 
\end{tabular} \hspace*{10bp}$=$\hspace*{10bp}%
\begin{tabular}{|c|}
\hline 
$\frac{1}{2},\uparrow$\tabularnewline
\hline 
$-\frac{1}{2},\downarrow$\tabularnewline
\hline 
\end{tabular} $\boldsymbol{+}$ %
\begin{tabular}{|c|}
\hline 
$\frac{1}{2},\downarrow$\tabularnewline
\hline 
$-\frac{1}{2},\uparrow$\tabularnewline
\hline 
\end{tabular} $\boldsymbol{+}$ %
\begin{tabular}{|c|}
\hline 
$\frac{1}{2},+$\tabularnewline
\hline 
$-\frac{1}{2},+$\tabularnewline
\hline 
\end{tabular} $\boldsymbol{+}$ %
\begin{tabular}{|c|}
\hline 
$\frac{1}{2},-$\tabularnewline
\hline 
$-\frac{1}{2},-$\tabularnewline
\hline 
\end{tabular} \vspace{0.5\baselineskip}
\par\end{center}%
\end{minipage} We can do that because the set of operators that we can generate
from either side of this equation is the same, so it is the same algebra.

The second simplification is that we can drop redundant terms in the
combination. If we know that some terms can be generated by other
terms then they do not add anything to the combined algebra, and therefore
can be dropped. In this case we can drop any one of the columns because
the sole projection that it defines can be spanned by the other three.

We end up with the following combination of BPTs

\noindent\begin{minipage}[c]{1\columnwidth}%
\begin{center}
\vspace{0.5\baselineskip}
\begin{tabular}{|c|}
\hline 
$\frac{1}{2},\uparrow$\tabularnewline
\hline 
$-\frac{1}{2},\downarrow$\tabularnewline
\hline 
\end{tabular}\hspace*{10bp}$\boldsymbol{+}$\hspace*{10bp} %
\begin{tabular}{|c|}
\hline 
$\frac{1}{2},\downarrow$\tabularnewline
\hline 
$-\frac{1}{2},\uparrow$\tabularnewline
\hline 
\end{tabular} \hspace*{10bp}$\boldsymbol{+}$\hspace*{10bp} %
\begin{tabular}{|c|}
\hline 
$\frac{1}{2},+$\tabularnewline
\hline 
$-\frac{1}{2},+$\tabularnewline
\hline 
\end{tabular} .\vspace{0.5\baselineskip}
\par\end{center}%
\end{minipage} The three projections defined by these columns are
\begin{align*}
\Pi_{zz;\frac{1}{4}} & :=\ketbra{\frac{1}{2},\uparrow}{\frac{1}{2},\uparrow}+\ketbra{-\frac{1}{2},\downarrow}{-\frac{1}{2},\downarrow} & \,\,\,\,\,\,\,\,\, & \Pi_{zz;-\frac{1}{4}}:=\ketbra{\frac{1}{2},\downarrow}{\frac{1}{2},\downarrow}+\ketbra{-\frac{1}{2},\uparrow}{-\frac{1}{2},\uparrow}
\end{align*}
\[
\Pi_{x;\frac{1}{2}}:=\ketbra{\frac{1}{2},+}{\frac{1}{2},+}+\ketbra{-\frac{1}{2},+}{-\frac{1}{2},+}.
\]
The subscripts $zz;\pm\frac{1}{4}$ refers to the $\pm\frac{1}{4}$
eigenvalues of $L_{z}\otimes S_{z}$, and $x;\frac{1}{2}$ refers
to the $\frac{1}{2}$ eigenvalue of $I\otimes S_{x}$. By scattering
\[
\begin{array}{c}
\Pi_{zz;\frac{1}{4}}\\
\\
\Pi_{x;\frac{1}{2}}
\end{array}\Diagram{fdA &  & fuA\\
 & f\\
fuA &  & fdA
}
\begin{array}{c}
\Pi_{zz;\frac{1}{4}}\\
\\
\Pi_{x;\frac{1}{2}}
\end{array}\hspace{4cm}\begin{array}{c}
\Pi_{zz;-\frac{1}{4}}\\
\\
\Pi_{x;\frac{1}{2}}
\end{array}\Diagram{fdA &  & fuA\\
 & f\\
fuA &  & fdA
}
\begin{array}{c}
\Pi_{zz;-\frac{1}{4}}\\
\\
\Pi_{x;\frac{1}{2}}
\end{array}
\]
we learn that these projections are reflecting so the resulting reflection
network is shown in Fig \ref{fig:exmprefnet4}.

\begin{figure}[H]
\begin{centering}
\includegraphics[viewport=0bp 250bp 737bp 453bp,clip,width=0.6\columnwidth]{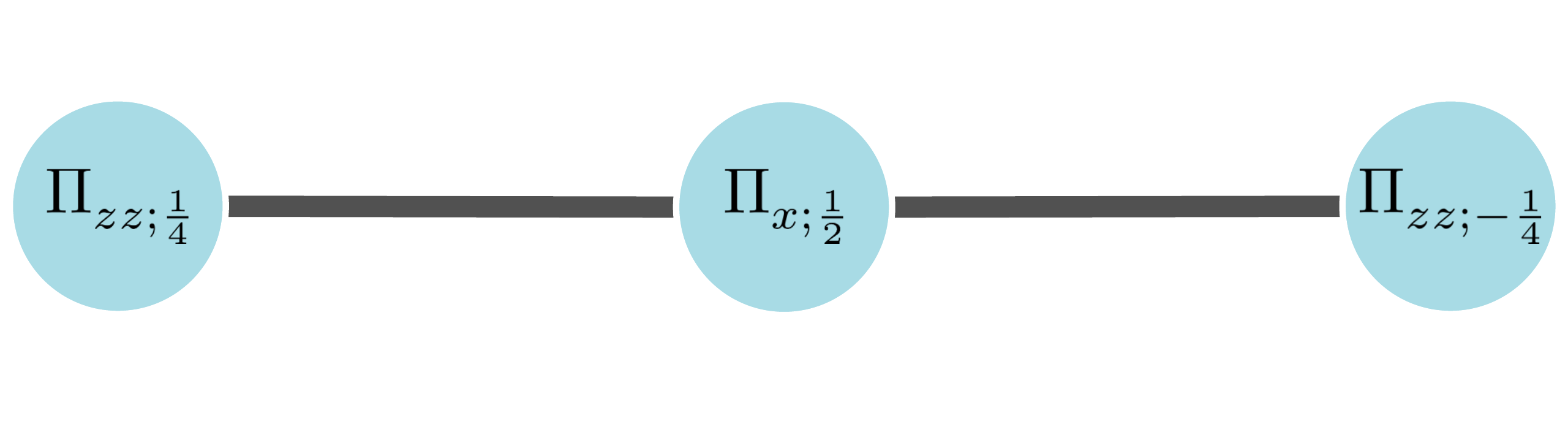}
\par\end{centering}
\centering{}\caption{\label{fig:exmprefnet4}The final reflection network for $l=\frac{1}{2}$.}
\end{figure}
The condition of minimality holds because the reflection network has
no cycles so all path-isometries between the same vertices have to
follow the same paths and therefore be equal. The condition of completeness
holds because $\Pi_{zz;\frac{1}{4}}+\Pi_{zz;-\frac{1}{4}}=I$.

The columns of the new BPT are determined by the maximal independent
set in the network, which is $\left\{ \Pi_{zz;\frac{1}{4}},\Pi_{zz;-\frac{1}{4}}\right\} $.
The choice and alignment of basis elements in the columns of the new
BPT is given by the path-isometry that connects the independent vertices,
which is 
\[
S_{-\frac{1}{4},\frac{1}{4}}\propto\Pi_{zz;-\frac{1}{4}}\Pi_{x;\frac{1}{2}}\Pi_{zz;\frac{1}{4}}\propto\ketbra{\frac{1}{2},\downarrow}{\frac{1}{2},\uparrow}+\ketbra{-\frac{1}{2},\uparrow}{-\frac{1}{2},\downarrow}.
\]
The resulting BPT is therefore

\noindent\begin{minipage}[c]{1\columnwidth}%
\begin{center}
\vspace{0.5\baselineskip}
\begin{tabular}{|c|}
\hline 
$\frac{1}{2},\uparrow$\tabularnewline
\hline 
$-\frac{1}{2},\downarrow$\tabularnewline
\hline 
\end{tabular}\hspace*{10bp}$\boldsymbol{+}$\hspace*{10bp} %
\begin{tabular}{|c|}
\hline 
$\frac{1}{2},\downarrow$\tabularnewline
\hline 
$-\frac{1}{2},\uparrow$\tabularnewline
\hline 
\end{tabular} \hspace*{10bp}$\boldsymbol{+}$\hspace*{10bp} %
\begin{tabular}{|c|}
\hline 
$\frac{1}{2},+$\tabularnewline
\hline 
$-\frac{1}{2},+$\tabularnewline
\hline 
\end{tabular}\hspace*{10bp}$\boldsymbol{=}$\hspace*{10bp} %
\begin{tabular}{|c|c|}
\hline 
$\frac{1}{2},\uparrow$ & $\frac{1}{2},\downarrow$\tabularnewline
\hline 
$-\frac{1}{2},\downarrow$ & $-\frac{1}{2},\uparrow$\tabularnewline
\hline 
\end{tabular} .\vspace{0.5\baselineskip}
\par\end{center}%
\end{minipage}

In order to get the irreps of the commutant we transpose the resulting
BPT and label the reduced basis according to the common state of spin-$l$
in each column:

\noindent\begin{minipage}[c]{1\columnwidth}%
\begin{center}
\vspace{0.5\baselineskip}
\begin{tabular}{|c|c|}
\hline 
$\frac{1}{2},\uparrow$ & $-\frac{1}{2},\downarrow$\tabularnewline
\hline 
$\frac{1}{2},\downarrow$ & $-\frac{1}{2},\uparrow$\tabularnewline
\hline 
\multicolumn{1}{c}{$\downarrow$} & \multicolumn{1}{c}{$\downarrow$}\tabularnewline
\hline 
$\frac{1}{2}$ & $-\frac{1}{2}$\tabularnewline
\hline 
\end{tabular} .\vspace{0.5\baselineskip}
\par\end{center}%
\end{minipage} For the state reduction map we explicitly define the minimal isometries
given by the alignment of columns 
\begin{align*}
S_{\frac{1}{2},\frac{1}{2}} & :=\ketbra{\frac{1}{2},\uparrow}{\frac{1}{2},\uparrow}+\ketbra{\frac{1}{2},\downarrow}{\frac{1}{2},\downarrow} & \,\,\,\,\,\,\,\,\, & S_{-\frac{1}{2},-\frac{1}{2}}:=\ketbra{-\frac{1}{2},\downarrow}{-\frac{1}{2},\downarrow}+\ketbra{-\frac{1}{2},\uparrow}{-\frac{1}{2},\uparrow}
\end{align*}
\[
S_{-\frac{1}{2},\frac{1}{2}}:=\ketbra{-\frac{1}{2},\downarrow}{\frac{1}{2},\uparrow}+\ketbra{-\frac{1}{2},\uparrow}{\frac{1}{2},\downarrow}.
\]
The state reduction map is then given by

\[
\ket{\psi}\longmapsto\rho_{Bob}=\begin{pmatrix}\bra{\psi}S_{\frac{1}{2},\frac{1}{2}}\ket{\psi} & \bra{\psi}S_{-\frac{1}{2},\frac{1}{2}}\ket{\psi}\\
\bra{\psi}S_{-\frac{1}{2},\frac{1}{2}}^{\dagger}\ket{\psi} & \bra{\psi}S_{-\frac{1}{2},-\frac{1}{2}}\ket{\psi}
\end{pmatrix}.
\]

The above BPT identifies a bipartition of the Hilbert space into two
virtual subsystems, and the state reduction map is the partial-trace-like
map over one of these subsystems. It may be tempting to think of this
state reduction map as the partial trace over the second spin, but
it is not quite the case. Although the product states such as $\ket{\pm\frac{1}{2},\uparrow}$
or $\ket{\pm\frac{1}{2},\downarrow}$ reduce with this map to $\ket{\pm\frac{1}{2}}$,
other product states such as $\ket{\psi}=\frac{1}{\sqrt{2}}\ket{\frac{1}{2},\downarrow}+\frac{1}{\sqrt{2}}\ket{-\frac{1}{2},\downarrow}$
do not reduce to $\frac{1}{\sqrt{2}}\ket{\frac{1}{2}}+\frac{1}{\sqrt{2}}\ket{-\frac{1}{2}}$
but to the completely mixed state
\[
\ket{\psi}\longmapsto\rho_{Bob}=\begin{pmatrix}\frac{1}{2} & 0\\
0 & \frac{1}{2}
\end{pmatrix}.
\]
We can summarize this distinction by observing that the BPT that corresponds
to the partial trace over the second spin is

\noindent\begin{minipage}[c]{1\columnwidth}%
\begin{center}
\vspace{0.5\baselineskip}
\begin{tabular}{|c|c|}
\hline 
$\frac{1}{2},\uparrow$ & $-\frac{1}{2},\uparrow$\tabularnewline
\hline 
$\frac{1}{2},\downarrow$ & $-\frac{1}{2},\downarrow$\tabularnewline
\hline 
\end{tabular} \vspace{0.5\baselineskip}
\par\end{center}%
\end{minipage} which is not quite the same as the one we have derived above.

Thus, Alice can encode one qubit of information into the subspaces
of $\ket{\frac{1}{2},\uparrow}$ and $\ket{-\frac{1}{2},\downarrow}$
(or $\ket{\frac{1}{2},\downarrow}$ and $\ket{-\frac{1}{2},\uparrow}$),
and Bob can decode it by applying the above state reduction map. Since
this qubit is encoded in a subsystem given by the commutant of $L_{z}\otimes S_{z}$
and $I\otimes S_{x}$, it is a decoherence free subsystem that is
not affected by such noise.

Now we will address the case of general $l$. The combination of the
irreps of $L_{z}\otimes S_{z}$ and of $I\otimes S_{x}$ is given
by the sum BPTs

\noindent\begin{minipage}[c]{1\columnwidth}%
\begin{center}
\vspace{0.5\baselineskip}
\begin{tabular}{ccccc}
\cline{1-1} 
\multicolumn{1}{|c|}{$l,\uparrow$} &  &  &  & \tabularnewline
\cline{1-1} 
\multicolumn{1}{|c|}{$-l,\downarrow$} &  &  &  & \tabularnewline
\cline{1-2} \cline{2-2} 
\multicolumn{1}{c|}{} & \multicolumn{1}{c|}{$l-1,\uparrow$} &  &  & \tabularnewline
\cline{2-2} 
\multicolumn{1}{c|}{} & \multicolumn{1}{c|}{$1-l,\downarrow$} &  &  & \tabularnewline
\cline{2-2} 
 &  & $\ddots$ &  & \tabularnewline
\cline{4-4} 
 &  & \multicolumn{1}{c|}{} & \multicolumn{1}{c|}{$1-l,\uparrow$} & \tabularnewline
\cline{4-4} 
 &  & \multicolumn{1}{c|}{} & \multicolumn{1}{c|}{$l-1,\downarrow$} & \tabularnewline
\cline{4-5} \cline{5-5} 
 &  &  & \multicolumn{1}{c|}{} & \multicolumn{1}{c|}{$-l,\uparrow$}\tabularnewline
\cline{5-5} 
 &  &  & \multicolumn{1}{c|}{} & \multicolumn{1}{c|}{$l,\downarrow$}\tabularnewline
\cline{5-5} 
\end{tabular}\hspace*{10bp}$\boldsymbol{+}$\hspace*{10bp} %
\begin{tabular}{|c|c|}
\cline{1-1} 
$l,+$ & \multicolumn{1}{c}{}\tabularnewline
\cline{1-1} 
$\vdots$ & \multicolumn{1}{c}{}\tabularnewline
\cline{1-1} 
$-l,+$ & \multicolumn{1}{c}{}\tabularnewline
\hline 
\multicolumn{1}{c|}{} & $l,-$\tabularnewline
\cline{2-2} 
\multicolumn{1}{c|}{} & $\vdots$\tabularnewline
\cline{2-2} 
\multicolumn{1}{c|}{} & $-l,-$\tabularnewline
\cline{2-2} 
\end{tabular} .\vspace{0.5\baselineskip}
\par\end{center}%
\end{minipage} The $2l+1$ column-blocks of the left BPT correspond to the eigenvalues
$\lambda=\frac{l}{2},...,-\frac{l}{2}$ of $L_{z}\otimes S_{z}$,
and the two column-blocks of the right BPT correspond to the two eigenvalues
$\pm\frac{1}{2}$ of $I\otimes S_{x}$. Using the same reasoning as
in the case of $l=\frac{1}{2}$, we rearrange this combination of
BPTs as follows

\noindent\begin{minipage}[c]{1\columnwidth}%
\begin{center}
\vspace{0.5\baselineskip}
 $\Biggl($ %
\begin{tabular}{|c|}
\hline 
$0,\uparrow$\tabularnewline
\hline 
$0,\downarrow$\tabularnewline
\hline 
\end{tabular} $\Biggr)$ $\boldsymbol{+}$ $\cdots$ $\boldsymbol{+}$ $\Biggl($
\begin{tabular}{|c|}
\hline 
$1-l,\uparrow$\tabularnewline
\hline 
$l-1,\downarrow$\tabularnewline
\hline 
\end{tabular} $\boldsymbol{+}$ %
\begin{tabular}{|c|}
\hline 
$l-1,\uparrow$\tabularnewline
\hline 
$1-l,\downarrow$\tabularnewline
\hline 
\end{tabular} $\Biggr)$ $\boldsymbol{+}$ $\Biggl($ %
\begin{tabular}{|c|}
\hline 
$l,\uparrow$\tabularnewline
\hline 
$-l,\downarrow$\tabularnewline
\hline 
\end{tabular} $\boldsymbol{+}$ %
\begin{tabular}{|c|}
\hline 
$-l,\uparrow$\tabularnewline
\hline 
$l,\downarrow$\tabularnewline
\hline 
\end{tabular} $\Biggr)$ $\boldsymbol{+}$ %
\begin{tabular}{|c|}
\hline 
$l,+$\tabularnewline
\hline 
$\vdots$\tabularnewline
\hline 
$-l,+$\tabularnewline
\hline 
\end{tabular} .\vspace{0.5\baselineskip}
\par\end{center}%
\end{minipage} Here we have dropped the redundant second column of the $-\frac{1}{2}$
eigenvalue of $I\otimes S_{x}$, and grouped together the columns
of $L_{z}\otimes S_{z}$ according to the magnitude of the eigenvalues
$\lambda$. Note that the $\lambda=0$ column on the left is special
since it is alone in its group and it only exists if $l$ is integer.
We will assume integer $l$ since this is the more general case while
the case of half integer $l$ can be considered as a simplification.
The above grouping and ordering of columns has been chosen in anticipation
of how the scattering calculations will unfold.

In the following, we will denote the spectral projections of $L_{z}\otimes S_{z}$
as
\[
\Pi_{zz;\lambda}:=\ketbra{\lambda,\uparrow}{\lambda,\uparrow}+\ketbra{-\lambda,\downarrow}{-\lambda,\downarrow}.
\]
We also conveniently define the projections
\[
\Pi_{x;m}:=\sum_{m'=-m}^{m}\ketbra{m',+}{m',+},
\]
such that for $m=l$ and $m=-l$ we get the $+\frac{1}{2}$ spectral
projection of $I\otimes S_{x}$. For any $m=l,...,-l$ we can calculate
\[
\Pi_{zz;m}\Pi_{x;m}=\frac{1}{\sqrt{2}}\ketbra{m,\uparrow}{m,+}+\frac{1}{\sqrt{2}}\ketbra{-m,\downarrow}{-m,+},
\]
which implies that for any $m\neq0$ we get the scatterings (the $m=0$
case is shown later)
\[
\begin{array}{c}
\Pi_{zz;m}\\
\\
\Pi_{x;m}
\end{array}\Diagram{fdA &  & fuA\\
 & f\\
fuA &  & fdA
}
\begin{array}{ccc}
\Pi_{zz;m}\\
\\
\Pi_{x;\left(m\right)} & , & \Pi_{x;\left|m\right|-1}
\end{array}
\]
Here we had to introduce another projection 
\[
\Pi_{x;\left(m\right)}:=\ketbra{m,+}{m,+}+\ketbra{-m,+}{-m,+}
\]
(note that $\Pi_{x;\left(m\right)}=\Pi_{x;\left(-m\right)}$).

The above general scattering calculation can then be used to combine
the BPT columns starting from the right side and proceeding toward
the left. In the first step we consider the reflection network of
$\Pi_{zz;l},\Pi_{zz;-l},\Pi_{x;l}$, which after scattering is shown
in Fig. \ref{fig:exmprefnet5}.

\begin{figure}[H]
\begin{centering}
\includegraphics[viewport=0bp 170bp 737bp 460bp,clip,width=0.6\columnwidth]{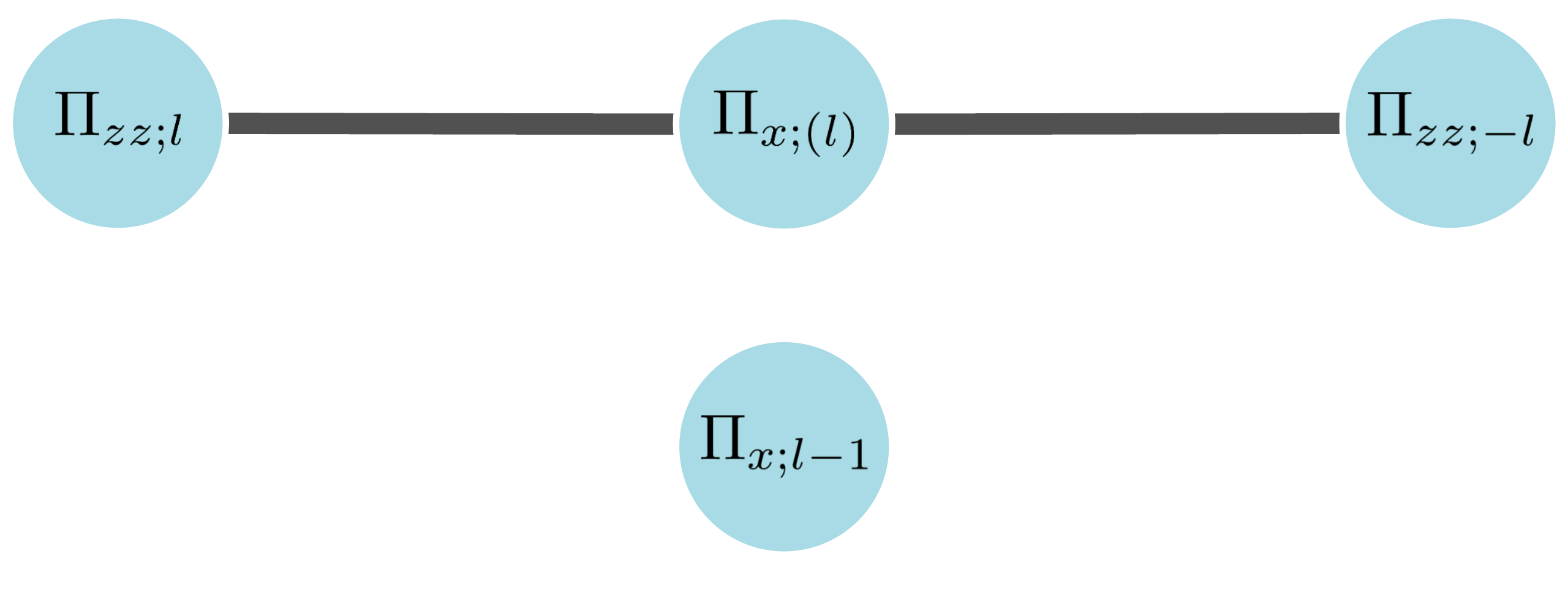}
\par\end{centering}
\centering{}\caption{\label{fig:exmprefnet5}The first step reflection network for general
$l$.}
\end{figure}
Then we include the next two columns $\Pi_{zz;l-1},\Pi_{zz;1-l}$
which are orthogonal to $\Pi_{zz;l},\Pi_{zz;-l},\Pi_{x;\left(l\right)}$
but not to $\Pi_{x;l-1}$. After scattering with $\Pi_{x;l-1}$ we
get the reflection network shown in Fig. \ref{fig:exmprefnet6}.

\begin{figure}[H]
\begin{centering}
\includegraphics[viewport=0bp 40bp 737bp 460bp,clip,width=0.6\columnwidth]{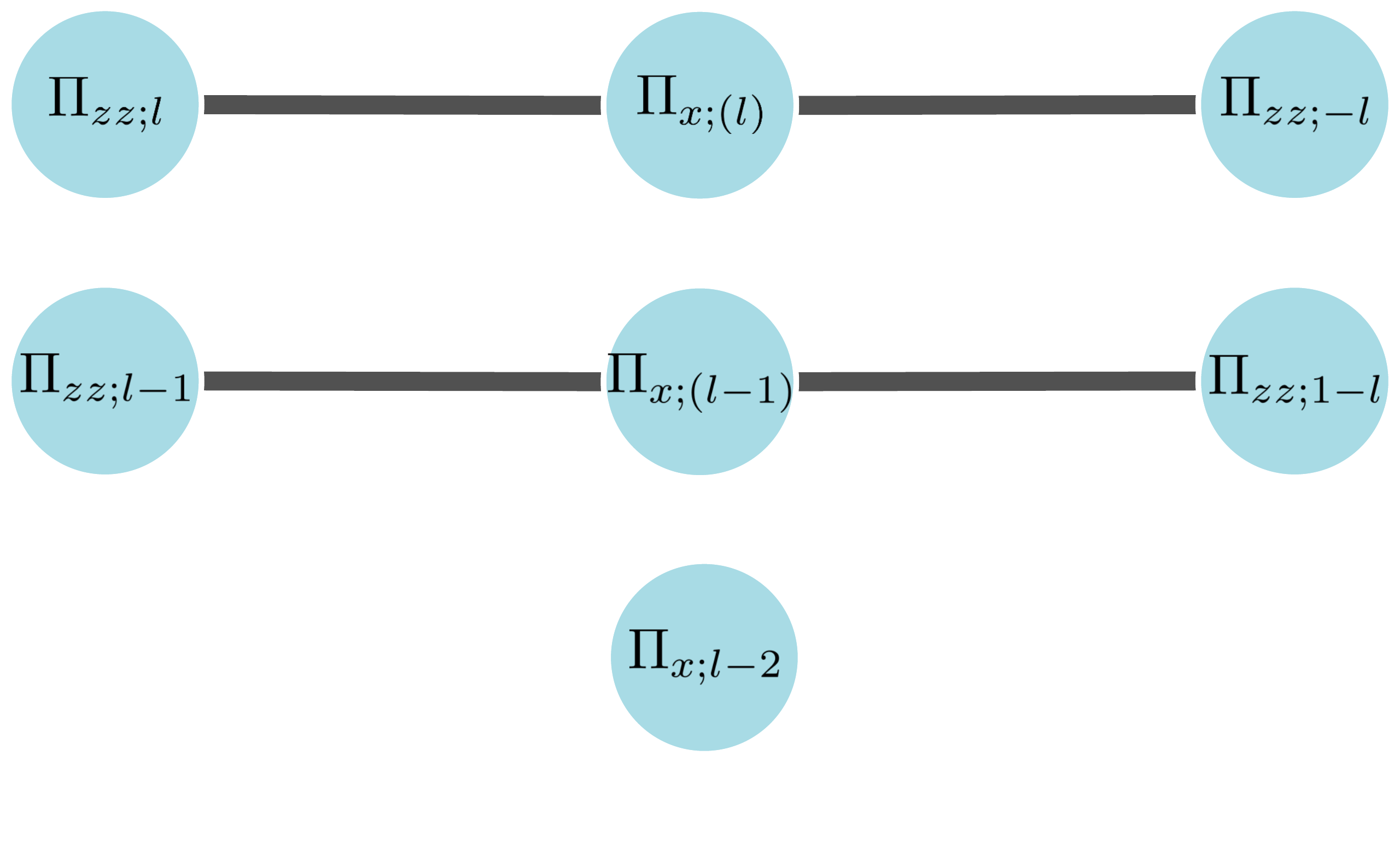}
\par\end{centering}
\centering{}\caption{\label{fig:exmprefnet6}The second step reflection network for general
$l$.}
\end{figure}
This pattern repeats as we fold in the columns until (assuming $l$
is integer) we are left with the last column $\Pi_{zz;0}=\ketbra{0,\uparrow}{0,\uparrow}+\ketbra{0,\downarrow}{0,\downarrow}$
and the leftover projection $\Pi_{x;0}=\ketbra{0,+}{0,+}$ from the
previous scatterings. They scatter differently than before:
\[
\begin{array}{c}
\Pi_{zz;0}\\
\\
\Pi_{x;0}
\end{array}\Diagram{fdA &  & fuA\\
 & f\\
fuA &  & fdA
}
\begin{array}{ccc}
\ketbra{0,+}{0,+}, &  & \ketbra{0,-}{0,-}\\
\\
\ketbra{0,+}{0,+}.
\end{array}
\]
The final reflection network is shown in Fig. \ref{fig:exmprefnet7}.

\begin{figure}[H]
\begin{centering}
\includegraphics[viewport=0bp 0bp 737bp 453bp,clip,width=0.6\columnwidth]{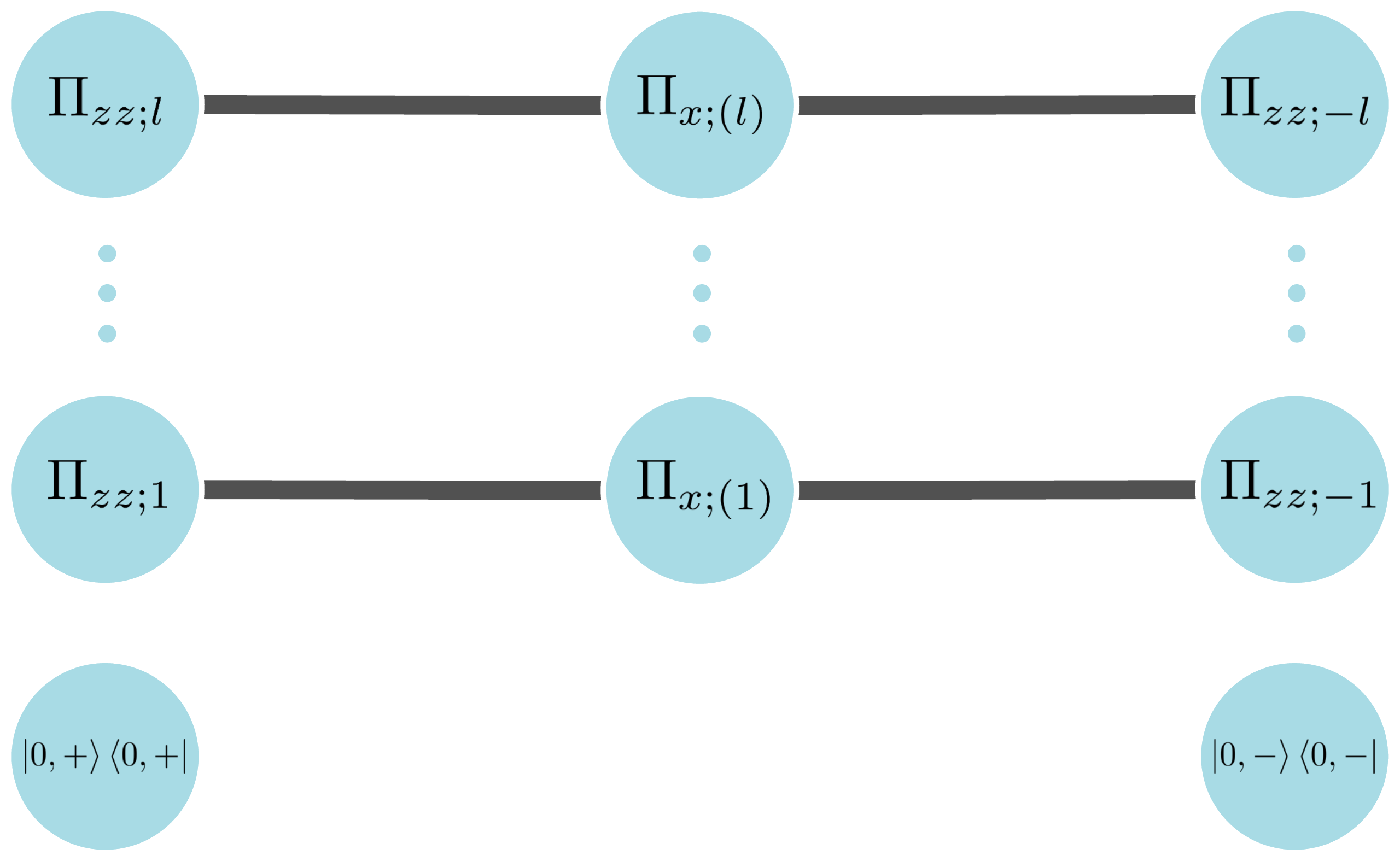}
\par\end{centering}
\centering{}\caption{\label{fig:exmprefnet7}The final reflection network for integer $l$.
When $l$ is half integer the last two projections $\protect\ketbra{0,+}{0,+}$
and $\protect\ketbra{0,-}{0,-}$ do not exist.}
\end{figure}

Each connected component with three vertices corresponds to a $2\times2$
BPT block, as in the $l=\frac{1}{2}$ case, and the last two isolated
vertices correspond to two single-celled blocks. As we did in the
$l=\frac{1}{2}$ case we construct the $2\times2$ blocks, and then
transpose them, which results in the BPT

\noindent\begin{minipage}[c]{1\columnwidth}%
\begin{center}
\vspace{0.5\baselineskip}
\begin{tabular}{ccccccc}
\cline{1-2} \cline{2-2} 
\multicolumn{1}{|c|}{$l,\uparrow$} & \multicolumn{1}{c|}{$-l,\downarrow$} &  &  &  &  & \tabularnewline
\cline{1-2} \cline{2-2} 
\multicolumn{1}{|c|}{$l,\downarrow$} & \multicolumn{1}{c|}{$-l,\uparrow$} &  &  &  &  & \tabularnewline
\cline{1-2} \cline{2-2} 
 &  & $\ddots$ &  &  &  & \tabularnewline
\cline{4-5} \cline{5-5} 
 &  & \multicolumn{1}{c|}{} & \multicolumn{1}{c|}{$1,\uparrow$} & \multicolumn{1}{c|}{$-1,\downarrow$} &  & \tabularnewline
\cline{4-5} \cline{5-5} 
 &  & \multicolumn{1}{c|}{} & \multicolumn{1}{c|}{$1,\downarrow$} & \multicolumn{1}{c|}{$-1,\uparrow$} &  & \tabularnewline
\cline{4-6} \cline{5-6} \cline{6-6} 
 &  &  &  & \multicolumn{1}{c|}{} & \multicolumn{1}{c|}{$0,+$} & \tabularnewline
\cline{6-7} \cline{7-7} 
 &  &  &  &  & \multicolumn{1}{c|}{} & \multicolumn{1}{c|}{$0,-$}\tabularnewline
\cline{7-7} 
\end{tabular} .\vspace{0.5\baselineskip}
\par\end{center}%
\end{minipage}

The form of the resulting BPT implies that there are $m=1,...,l$
alternative decoherence free subsystems that Alice can choose from;
she can still encode only one qubit though. Since there are $l$ alternative
orthogonal subspaces to choose from (not counting $l=0$), Alice can
also encode $\log l$ classical bits in addition to the qubit.

It is important to note that the derivation of the irreps structure
for general $l$ was only possible because we could carry out all
the calculations analytically, without specifying the value of $l$.
This demonstrates the key advantage of the Scattering Algorithm over
the numeric approaches in that it can be applied symbolically.

\chapter{Reduction of dynamics\label{chap:Operational-reductions-of-dyn}}

In this chapter we will consider two notions of reduction of dynamics.
The common idea here is that the Hamiltonian is block-diagonalizable
by the irrep basis of any non-trivial algebra that the Hamiltonian
belongs to. Such block-diagonalization leads to the reduction of dynamics
onto the irreps. With this basic idea in mind we will first consider
the reduction of dynamics with symmetries, and then proceed to the
symmetry-agnostic approach.

The idea of reduction of dynamics with symmetries traces back to the
seminal work by Emmy Noether \citep{Noether1918}. Today, symmetry
related methods are a well established staple in physics with many
dedicated textbooks such as \citep{cornwell1997group,tung1985group,Georgi99}.
With the advancement of finite-dimensional quantum mechanics driven
by the development of quantum information and quantum computing, the
central role of the irreps structures associated with symmetries was
gradually recognized in applications \citep{zanardi1997error,Zanardi97Noiseless,Knill2000,Bartlett07,marvian2013theory,marvian2014extending}.

In Section \ref{sec:Reduction-of-Hamiltonians-with-Sym} we will briefly
outline the role of the irreps structure in the reduction of dynamics
with symmetries. This will lead to the realization that the usual
notion of symmetry is too restrictive and even groups that do not
commute with the Hamiltonian may still be useful for the reduction
of dynamics. This idea is summarized in Theorem \ref{thm:reduction with generalized symm }
and we will illustrate it with an example of a quantum walk with a
broken symmetry. As a secondary goal we will use this example to demonstrate
how the the Scattering Algorithm constructs irreps of a non-trivial
finite group.

In Section \ref{sec:Symmetry-agnostic-reduction-of} we will demonstrate
how the same kind of reduction of dynamics can be performed without
the need to recognize symmetries. Such symmetry-agnostic approach
is possible with the Scattering Algorithm as it allows us to directly
focus on the irreps structure generated by the Hamiltonian terms.
We will illustrate this idea with two examples from the literature
on qubit implementations in quantum dots. Specifically, we will show
how the symmetry-agnostic approach can be used to reduce the control
Hamiltonian in order to find the possible qubit encodings.

\section{Reduction of Hamiltonians with symmetries\label{sec:Reduction-of-Hamiltonians-with-Sym}}

We will begin by describing the central role of the irreps structure
in the usual reduction of Hamiltonians with symmetries.

Let us consider the Hamiltonian $H\in\mathcal{L}\left(\mathcal{H}\right)$
and the group $\mathcal{G}$ represented by the unitaries $U\left(\mathcal{G}\right):=\left\{ U\left(g\right)\right\} _{g\in\mathcal{G}}$.
We say that $\mathcal{G}$ is a symmetry of $H$ (as represented by
$U\left(\mathcal{G}\right)$) if $\left[H,U\left(g\right)\right]=0$
for all $g\in\mathcal{G}$. This identifies $H$ as an element of
the commutant $U\left(\mathcal{G}\right)'$ of $U\left(\mathcal{G}\right)$.

In general, the group algebra $\mathcal{A}_{U\left(\mathcal{G}\right)}$
(recall Definition \ref{def:group algebra}) identifies the irreps
structure

\begin{equation}
\mathcal{H}\cong\bigoplus_{q}\mathcal{H}_{\nu_{q}}\otimes\mathcal{H}_{\mu_{q}}\label{eq: symmetry irrep strucutre}
\end{equation}
such that all the unitaries $U\left(\mathcal{G}\right)$ reduce to

\[
U\left(g\right)\cong\bigoplus_{q}I_{\nu_{q}}\otimes U_{\mu_{q}}\left(g\right),
\]
where $U_{\mu_{q}}\left(g\right)$ are the irreducible unitary representations
of $\mathcal{G}$ (see Theorem \ref{thm: group algebra irreps}).
Since $H\in U\left(\mathcal{G}\right)'$, it reduces in a complementary
manner (see Theorem \ref{thm:commutant is bpt transpose})
\[
H\cong\bigoplus_{q}H_{\nu_{q}}\otimes I_{\mu_{q}}.
\]

This block-diagonal form constitutes a reduction of dynamics where
we have reduced the action of $H$ from the whole $\mathcal{H}$ to
the smaller Hamiltonians $H_{\nu_{q}}$ acting on $\mathcal{H}_{\nu_{q}}$.
This form rules out any transitions between states supported on different
irreps, so the the irrep value $q$ is a conserved quantity. Thus,
the irreps structure \eqref{eq: symmetry irrep strucutre} of the
symmetry group $U\left(\mathcal{G}\right)$ identifies the constants
of motion and the subsystems $\mathcal{H}_{\nu_{q}}$ on which the
dynamics reduce.

As an example, consider the three spin Heisenberg interaction Hamiltonian
\[
H=\epsilon_{12}\vec{S}_{1}\cdot\vec{S}_{2}+\epsilon_{23}\vec{S}_{2}\cdot\vec{S}_{3},
\]
where $\epsilon_{ij}$ are arbitrary coupling strengths and $\vec{S}_{i}=\left(S_{i,x},S_{i,y},S_{i,z}\right)$
are the spin operators. This Hamiltonian commutes with the $SU\left(2\right)$
group of rotations $\left[H,U\left(R\right)\right]=0$ that have the
familiar irreps structure of total spin
\begin{equation}
\mathcal{H}\cong\mathcal{H}_{\mu_{3/2}}\oplus\mathcal{H}_{\nu_{1/2}}\otimes\mathcal{H}_{\mu_{1/2}}.\label{eq: three spin irreps}
\end{equation}
These irreps are identified by the total spin basis $\ket{j,\alpha,m}$
where $j=\frac{3}{2},\frac{1}{2}$, $m=-j,...,j$ and $\alpha=s,t$
distinguishes the two variants of the $j=\frac{1}{2}$ irrep. The
conserved quantity here is the total spin $j=\frac{3}{2},\frac{1}{2}$,
as identified by the irreps.

The irreps structure \eqref{eq: symmetry irrep strucutre} tells us
that the Hamiltonian reduces to 
\[
H\cong h_{\nu_{3/2}}\,I_{\mu_{3/2}}\oplus H_{\nu_{1/2}}\otimes I_{\mu_{1/2}}
\]
where $h_{\nu_{3/2}}$ is a scalar (since $\mathcal{H}_{\nu_{3/2}}$
is one-dimensional) and $H_{\nu_{1/2}}$ is a $2\times2$ matrix.
Therefore, in the total spin basis, $H$ is given by the five matrix
elements

\[
h_{\nu_{3/2}}=\bra{\frac{3}{2},m}H\ket{\frac{3}{2},m}=\frac{1}{4}\left(\epsilon_{12}+\epsilon_{23}\right),
\]
\begin{equation}
H_{\nu_{1/2}}=\begin{pmatrix}\bra{\frac{1}{2},s,m}H\ket{\frac{1}{2},s,m} & \bra{\frac{1}{2},s,m}H\ket{\frac{1}{2},t,m}\\
\bra{\frac{1}{2},t,m}H\ket{\frac{1}{2},s,m} & \bra{\frac{1}{2},t,m}H\ket{\frac{1}{2},t,m}
\end{pmatrix}=\frac{1}{4}\begin{pmatrix}-3\epsilon_{23} & \sqrt{3}\epsilon_{12}\\
\sqrt{3}\epsilon_{12} & \epsilon_{23}-2\epsilon_{12}
\end{pmatrix}\label{eq:Heisnberg nu_1/2 Hamil}
\end{equation}
where the choice of $m$ does not matter and all other matrix elements
are zero.

In the standard applications of symmetries, as in the above example,
only the groups such that $\left[H,U\left(g\right)\right]=0$ are
considered. The symmetry condition $\left[H,U\left(g\right)\right]=0$
, however, is too restrictive and the irreps structure \eqref{eq: symmetry irrep strucutre}
can still be useful with groups that fail to commute with the Hamiltonian.

For example, we can add the symmetry breaking term $S_{tot,z}=S_{1,z}+S_{2,z}+S_{3,z}$
to the Hamiltonian 
\[
H=\epsilon_{12}\vec{S}_{1}\cdot\vec{S}_{2}+\epsilon_{23}\vec{S}_{2}\cdot\vec{S}_{3}+S_{tot,z}\,.
\]
Now $\left[H,U\left(R\right)\right]\neq0$ so the $SU\left(2\right)$
group is not a symmetry and it would appear that the above reduction
of dynamics is no longer relevant. This, however, is not the case
and the irreps structure of the $SU\left(2\right)$ group is still
useful for the reduction of this Hamiltonian. The reason for that
is because the symmetry breaking term $S_{tot,z}$ is itself one of
the generators of the $SU\left(2\right)$ group (it breaks the symmetry
because the group is not Abelian). This means that $S_{tot,z}$ is
an element of the $SU\left(2\right)$ group algebra and so with respect
to its irreps \eqref{eq: three spin irreps} it is confined to the
form\footnote{In this example it is even more obvious because the symmetry breaking
term $S_{tot,z}$ is diagonal in the total spin basis $\ket{j,\alpha,m}$.} 
\[
S_{tot,z}\cong S_{\mu_{3/2},z}\oplus I_{\nu_{1/2}}\otimes S_{\mu_{1/2},z},
\]
so
\[
H\cong\left(h_{\nu_{3/2}}\,I_{\mu_{3/2}}+S_{\mu_{3/2},z}\right)\oplus\left(H_{\nu_{1/2}}\otimes I_{\mu_{1/2}}+I_{\nu_{1/2}}\otimes S_{\mu_{1/2},z}\right).
\]

Thus, we can still say that the original exchange interaction generates
dynamics via the $H_{\nu_{1/2}}$ term in the $\mathcal{H}_{\nu_{1/2}}$
subsystem, while the new $S_{\mu_{3/2},z}$ and $S_{\mu_{1/2},z}$
terms generates dynamics in the $\mathcal{H}_{\mu_{3/2}}$ and $\mathcal{H}_{\mu_{1/2}}$
subsystems. The critical detail here is that each subsystem evolves
independently as there are no interaction terms between them so the
dynamics can be reduced to the subsystems identified by the irreps
structure \eqref{eq: three spin irreps}. Also note that the total
spin is still a constant of motion even though the group that identifies
it is not a symmetry of the Hamiltonian.

The general result that extends the application of group representations
beyond symmetries is give by the following theorem.
\begin{thm}
\label{thm:reduction with generalized symm }Let $H\in\mathcal{L}\left(\mathcal{H}\right)$
and let $\mathcal{G}$ be a finite or a compact Lie group represented
by the unitaries $U\left(\mathcal{G}\right):=\left\{ U\left(g\right)\right\} _{g\in\mathcal{G}}\subset\mathcal{L}\left(\mathcal{H}\right)$,
such that
\[
\left[H,U\left(g\right)\right]\in\mathcal{A}_{U\left(\mathcal{G}\right)}\,\,\,\,\,\,\forall g\in\mathcal{G}.
\]
Then, with respect to the irreps structure of the group algebra $\mathcal{A}_{U\left(\mathcal{G}\right)}$
\[
\mathcal{H}\cong\bigoplus_{q}\mathcal{H}_{\nu_{q}}\otimes\mathcal{H}_{\mu_{q}},
\]
the operator $H$ reduces to 
\[
H=H_{\nu}+H_{\mu}\cong\bigoplus_{q}H_{\nu_{q}}\otimes I_{\mu_{q}}+\bigoplus_{q}I_{\nu_{q}}\otimes H_{\mu_{q}}
\]
for some $H_{\nu_{q}}\in\mathcal{L}\left(\mathcal{H}_{\nu_{q}}\right)$
and $H_{\mu_{q}}\in\mathcal{L}\left(\mathcal{H}_{\mu_{q}}\right)$.
\end{thm}
\begin{proof}
Let $A\left(g\right):=\left[H,U\left(g\right)\right]$ so 
\[
H=U\left(g\right)HU\left(g\right)^{\dagger}+A\left(g\right)U\left(g\right)^{\dagger}.
\]
If $\mathcal{G}$ is finite we can sum both sides over all $g\in\mathcal{G}$
and normalize it by the order of $\mathcal{G}$:
\[
H=\frac{1}{\mathcal{\left|G\right|}}\sum_{g\in\mathcal{G}}U\left(g\right)HU\left(g\right)^{\dagger}+\frac{1}{\mathcal{\left|G\right|}}\sum_{g\in\mathcal{G}}A\left(g\right)U\left(g\right)^{\dagger}.
\]
In the more general case, if $\mathcal{G}$ is a compact Lie group
there is a normalized invariant measure (Haar measure) $d\mu\left(g\right)$
over $\mathcal{G}$ such that
\[
H=\underset{H_{\nu}}{\underbrace{\int_{\mathcal{G}}d\mu\left(g\right)U\left(g\right)HU\left(g\right)^{\dagger}}}+\underset{H_{\mu}}{\underbrace{\int_{\mathcal{G}}d\mu\left(g\right)A\left(g\right)U\left(g\right)^{\dagger}}}.
\]
For all $g\in\mathcal{G}$ both $U\left(g\right)^{\dagger}$ and $A\left(g\right)$
are in the group algebra $\mathcal{A}_{U\left(\mathcal{G}\right)}$.
Therefore, $H_{\mu}\in\mathcal{A}_{U\left(\mathcal{G}\right)}$ so
according to Theorem \ref{thm: Wedderburn Decomp} it reduces to 
\[
H_{\mu}\cong\bigoplus_{q}I_{\nu_{q}}\otimes H_{\mu_{q}}.
\]
The term $H_{\nu}$, on the other hand, is in the commutant $\mathcal{A}_{U\left(\mathcal{G}\right)}'$
because it commutes with all $U\left(g\right)$: 
\[
H_{\nu}U\left(g\right)=\int_{\mathcal{G}}d\mu\left(g'\right)U\left(g'\right)HU\left(g^{-1}g'\right)^{\dagger}=\int_{\mathcal{G}}d\mu\left(gg'\right)U\left(gg'\right)HU\left(g'\right)^{\dagger}=U\left(g\right)H_{\nu}
\]
(here we have used the invariance of the measure $d\mu\left(gg'\right)=d\mu\left(g\right)$).
Therefore, $H_{\nu}\in\mathcal{A}_{U\left(\mathcal{G}\right)}'$ so
according to Theorem \ref{thm:commutant is bpt transpose} it reduces
to

\[
H_{\nu}\cong\bigoplus_{q}H_{\nu_{q}}\otimes I_{\mu_{q}}.
\]
Thus,
\[
H=H_{\nu}+H_{\mu}\cong\bigoplus_{q}H_{\nu_{q}}\otimes I_{\mu_{q}}+\bigoplus_{q}I_{\nu_{q}}\otimes H_{\mu_{q}}.
\]
\end{proof}
This theorem implies that symmetry groups are the special case when
$\left[H,U\left(g\right)\right]=0\in\mathcal{A}_{U\left(\mathcal{G}\right)}$.
The generalization is that now we can also consider groups such that
the Hamiltonian consists of both an invariant term---identified by
$H_{\nu}$---and a symmetry breaking term---identified by $H_{\mu}$.
The restriction is that the symmetry breaking term $H_{\mu}$ still
has to be an element of the group algebra $\mathcal{A}_{U\left(\mathcal{G}\right)}$.
Once we have identified such group, the dynamics reduce to the subsystems
of the group's irreps structure, that is
\[
e^{-itH}=e^{-itH_{\nu}}e^{-itH_{\mu}}=\bigoplus_{q}e^{-itH_{\nu_{q}}}\otimes e^{-itH_{\mu_{q}}}.
\]
In particular, the value $q$ that distinguishes the irreps is a conserved
quantity.

We will now consider a more elaborate example of such reduction of
dynamics.

\subsubsection*{Example}

In this example we will analyze the dynamics of a continuous-time
quantum walk (CTQW) on binary trees. CTQW is the quantum analog of
a continuous-time random walk on graphs. The idea that a CTQW model
can provide an exponentially faster way of searching for distinguished
vertices on certain problems, was first introduced by Farhi and Gutmann
in \citep{Farhi98}. Since some computational problems can be formulated
as searches on graphs, the CTQW model turned out to be an alternative
paradigm to quantum Fourier transform for designing quantum algorithms
with an exponential speed-up.

What is interesting about the CTQW paradigm is that it is relatively
easy to understand where the exponential speed-up is coming from.
It was observed in \citep{childs2002,childs2003exponential} that
the exponential speed-up can be explained by the exponential reduction
of dynamics. This observation was analyzed exactly for a search on
binary trees and the reduction was traced back to the symmetries of
the graph. In the following, we will reproduce this argument by finding
the irreps of symmetries of the binary tree and show that the exponential
reduction of dynamics holds even when the symmetry is broken.

The Hilbert space of a CTQW model is spanned by the vertices $V$
of a graph $G$: 
\[
\mathcal{H}:=\spn\,V.
\]
The Hamiltonian of a CTQW model can be defined by the edges $E$ of
a graph $G$ as follows
\begin{equation}
H:=-\sum_{\left\langle i,j\right\rangle \in E}\left(\ket i\bra j+\ket j\bra i-\ket i\bra i-\ket j\bra j\right).\label{eq:CTQW hamiltonian def}
\end{equation}
The graph $G$ that we will consider here is shown in Fig. \ref{fig:ctqw bin tree}.

\begin{figure}[H]
\centering{}\includegraphics[viewport=0bp 50bp 765bp 403bp,clip,width=0.8\columnwidth]{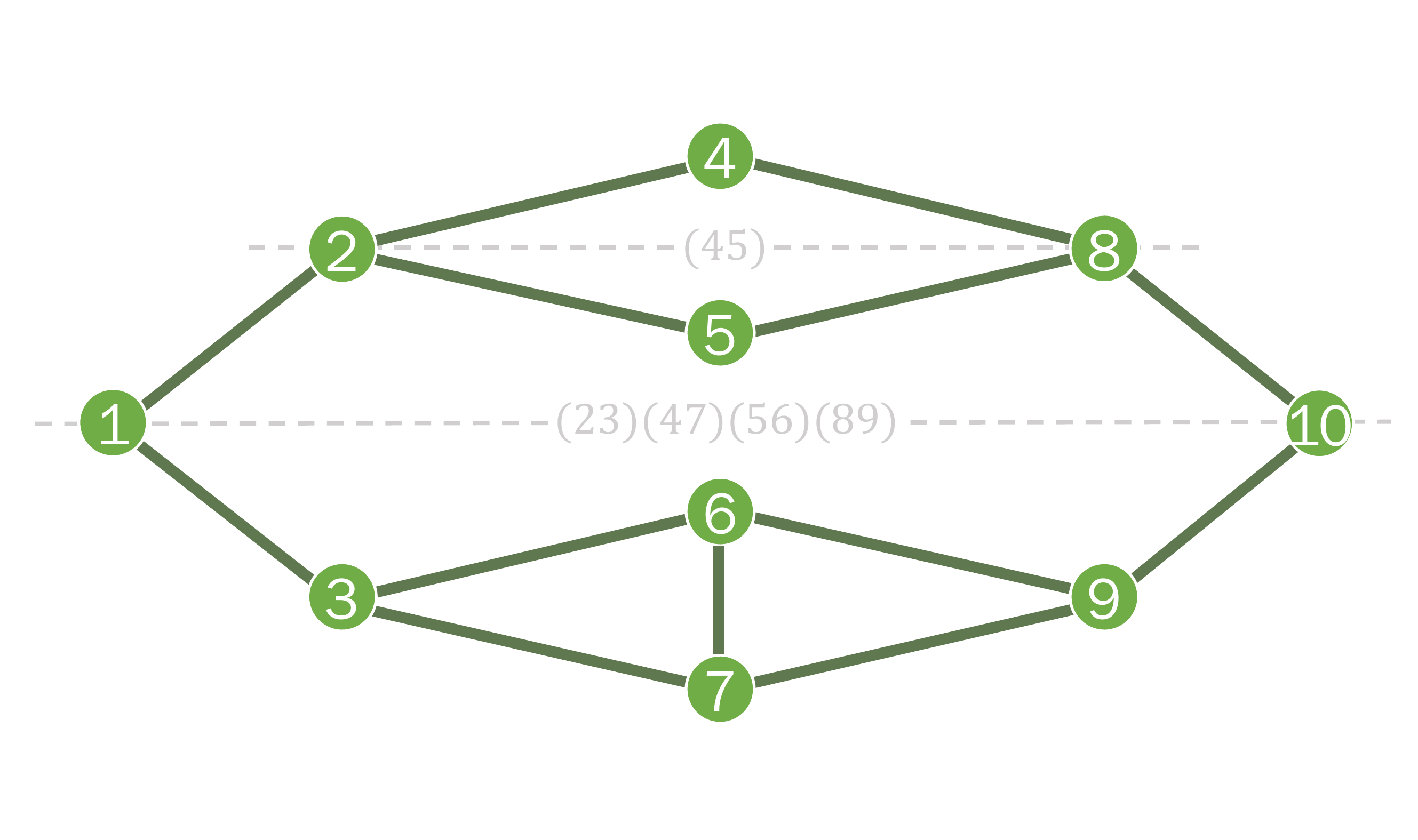}
\caption{\label{fig:ctqw bin tree}Two ``glued'' binary trees with a symmetry
breaking edge $\left\langle 6,7\right\rangle $. The two generators
of the permutation symmetry are shown in \foreignlanguage{canadian}{grey}
using the cyclic notation $\left(ij\right)$. The dashed lines are
the axis of reflection associated with the permutations.}
\end{figure}

Note that we can permute the vertices $4$ with $5$ or $6$ with
$7$ without changing the definition of the Hamiltonian \ref{eq:CTQW hamiltonian def}.
This is so because these permutations do not change how the vertices
are connected. This is not true, for example, for the permutation
of $2$ with $4$ because initially $4$ is not connected to $1$
or $5$, but after this permutation it is. Therefore, the permutations
$\pi_{2}:=\left(45\right)$ and $\pi_{3}:=\left(67\right)$ are symmetries
of this Hamiltonian.\footnote{The subscripts $2$ and $3$ refer to the root vertices of the sub-trees
where these permutations apply.}

Any permutation $\pi$ of the vertices $V$ is represented on the
Hilbert space $\mathcal{H}$ by the unitaries 
\[
U\left(\pi\right):=\sum_{i\in V}\ketbra{\pi\left(i\right)}i.
\]
In particular, the 2-cycle permutations $\left(ij\right)$ are represented
by
\[
U\left(\left(ij\right)\right):=\ketbra ij+\ketbra ji+\sum_{k\neq i,j}\ketbra kk=\ketbra ij+\ketbra ji+I-\ketbra ii-\ketbra jj.
\]

We can therefore express the Hamiltonian \eqref{eq:CTQW hamiltonian def}
as a sum of 2-cycle permutations
\begin{align*}
H & =-\sum_{\left\langle i,j\right\rangle \in E}\left(U\left(\left(ij\right)\right)-I\right)=-\sum_{\left\langle i,j\right\rangle \in E}U\left(\left(ij\right)\right)+\left|E\right|I
\end{align*}
and we can drop the constant identity $\left|E\right|I$. Observe
that the permutation $U\left(\left(67\right)\right)$ is both a symmetry
of $H$ and an additive term in $H$. Theorem \ref{thm:reduction with generalized symm }
then implies that we can consider symmetries that do not commute with
$U\left(\left(67\right)\right)$, as long as they commute with the
rest of $H$ and $U\left(\left(67\right)\right)$ is itself an element
of that symmetry. In the notation of Theorem \ref{thm:reduction with generalized symm }
we split $H=H_{\nu}+H_{\mu}$ where $H_{\mu}:=-U\left(\left(67\right)\right)$
and $H_{\nu}$ are all the other terms.

If we exclude the term $H_{\mu}$ (on the graph this means deleting
the edge $\left\langle 6,7\right\rangle $ ) the remaining term $H_{\nu}$
has more symmetry. That is, in addition to $\pi_{2}$ and $\pi_{3}$
another permutation is also a symmetry:
\[
\pi_{1}:=\left(23\right)\left(47\right)\left(56\right)\left(89\right).
\]
For the full $H$ it is not a symmetry because $U\left(\pi_{3}\right)$
and $U\left(\pi_{1}\right)$ do not commute. However, since the finite
group $\mathcal{G}$ generated by $\pi_{1}$, $\pi_{2}$ and $\pi_{3}$
is a symmetry of $H_{\nu}$ and $H=H_{\nu}+\left(-U\left(\pi_{3}\right)\right)$,
the condition of Theorem \ref{thm:reduction with generalized symm }
holds: 
\[
\left[H,U\left(\pi_{i}\right)\right]=\left[-U\left(\pi_{3}\right),U\left(\pi_{i}\right)\right]\in\mathcal{A}_{U\left(\mathcal{G}\right)}.
\]

In order to reduce the dynamics with the group $U\left(\mathcal{G}\right)$
we need to find its irreps structure. Note that the group element
$\pi_{3}$ is redundant as $\pi_{3}=\pi_{1}\pi_{2}\pi_{1}$ so we
only need to consider the generators $\pi_{1}$ and $\pi_{2}$. First,
we find the spectral projections of $U\left(\pi_{1}\right)$ and $U\left(\pi_{2}\right)$.
Using the shorthand notation for the states
\begin{equation}
\ket{_{-j_{1},j_{2},...}^{+i_{1},i_{2},...}}:=\frac{1}{\sqrt{N}}\left(\ket{i_{1}}+\ket{i_{2}}+...-\ket{j_{1}}-\ket{j_{2}}-...\right)\label{eq: shorthand state notation}
\end{equation}
we can diagonalize the generator 
\[
U\left(\pi_{2}\right)=\ketbra 45+\ketbra 54+\sum_{k\neq4,5}\ketbra kk=\ket{_{-}^{+4,5}}\bra{_{-}^{+4,5}}-\ket{_{-5}^{+4}}\bra{_{-5}^{+4}}+\sum_{k\neq4,5}\ketbra kk
\]
so its spectral projections are 
\begin{align*}
\Pi_{2;+} & :=\ket{_{-}^{+4,5}}\bra{_{-}^{+4,5}}+\ketbra 11+\ketbra 22+\ketbra 33+\ketbra 66+\ketbra 77+\ketbra 88+\ketbra 99+\ketbra{10}{10}\\
\Pi_{2;-} & :=\ket{_{-5}^{+4}}\bra{_{-5}^{+4}}.
\end{align*}
Similarly, we have
\[
U\left(\pi_{1}\right)=\Pi_{1;+}-\Pi_{1;-}
\]
where
\begin{align*}
\Pi_{1;+} & :=\ket 1\bra 1+\ket{_{-}^{+2,3}}\bra{_{-}^{+2,3}}+\ket{_{-}^{+4,7}}\bra{_{-}^{+4,7}}+\ket{_{-}^{+5,6}}\bra{_{-}^{+5,6}}+\ket{_{-}^{+8,9}}\bra{_{-}^{+8,9}}+\ket{10}\bra{10}\\
\Pi_{1;-} & :=\ket{_{-3}^{+2}}\bra{_{-3}^{+2}}+\ket{_{-7}^{+4}}\bra{_{-7}^{+4}}+\ket{_{-6}^{+5}}\bra{_{-6}^{+5}}+\ket{_{-9}^{+8}}\bra{_{-9}^{+8}}.
\end{align*}

Since both sets of spectral projections sum to $I$, one of the projections
is redundant so we will drop $\Pi_{2;+}$. It is now straight forward
to calculate the scatterings
\[
\begin{array}{c}
\Pi_{2;-}\\
\\
\Pi_{1;+}
\end{array}\Diagram{fdA &  & fuA\\
 & f\\
fuA &  & fdA
}
\begin{array}{cc}
\Pi_{2;-}\\
\\
\Pi_{1;+}^{\left(1/2\right)}, & \Pi_{1;+}^{\left(0\right)}
\end{array}\hspace{2cm}\begin{array}{c}
\Pi_{2;-}\\
\\
\Pi_{1;-}
\end{array}\Diagram{fdA &  & fuA\\
 & f\\
fuA &  & fdA
}
\begin{array}{cc}
\Pi_{2;-}\\
\\
\Pi_{1;-}^{\left(1/2\right)}, & \Pi_{1;-}^{\left(0\right)}
\end{array}
\]
where 
\[
\Pi_{1;+}^{\left(1/2\right)}:=\ket{_{-5,6}^{+4,7}}\bra{_{-5,6}^{+4,7}}\hspace{2cm}\Pi_{1;-}^{\left(1/2\right)}:=\ket{_{-5,7}^{+4,6}}\bra{_{-5,7}^{+4,6}}
\]
and
\begin{align*}
\Pi_{1;+}^{\left(0\right)} & :=\Pi_{1;+}-\Pi_{1;+}^{\left(1/2\right)}=\ket 1\bra 1+\ket{_{-}^{+2,3}}\bra{_{-}^{+2,3}}+\ket{_{-}^{+4,5,6,7}}\bra{_{-}^{+4,5,6,7}}+\ket{_{-}^{+8,9}}\bra{_{-}^{+8,9}}+\ket{10}\bra{10}\\
\\
\Pi_{1;-}^{\left(0\right)} & :=\Pi_{1;-}-\Pi_{1;-}^{\left(1/2\right)}=\ket{_{-3}^{+2}}\bra{_{-3}^{+2}}+\ket{_{-6,7}^{+4,5}}\bra{_{-6,7}^{+4,5}}+\ket{_{-9}^{+8}}\bra{_{-9}^{+8}}.
\end{align*}

The resulting reflection network consists of the three connected components
shown in Fig. \ref{fig: Bintree refnet}.

\begin{figure}[H]
\centering{}\includegraphics[viewport=0bp 100bp 737bp 453bp,clip,width=0.6\columnwidth]{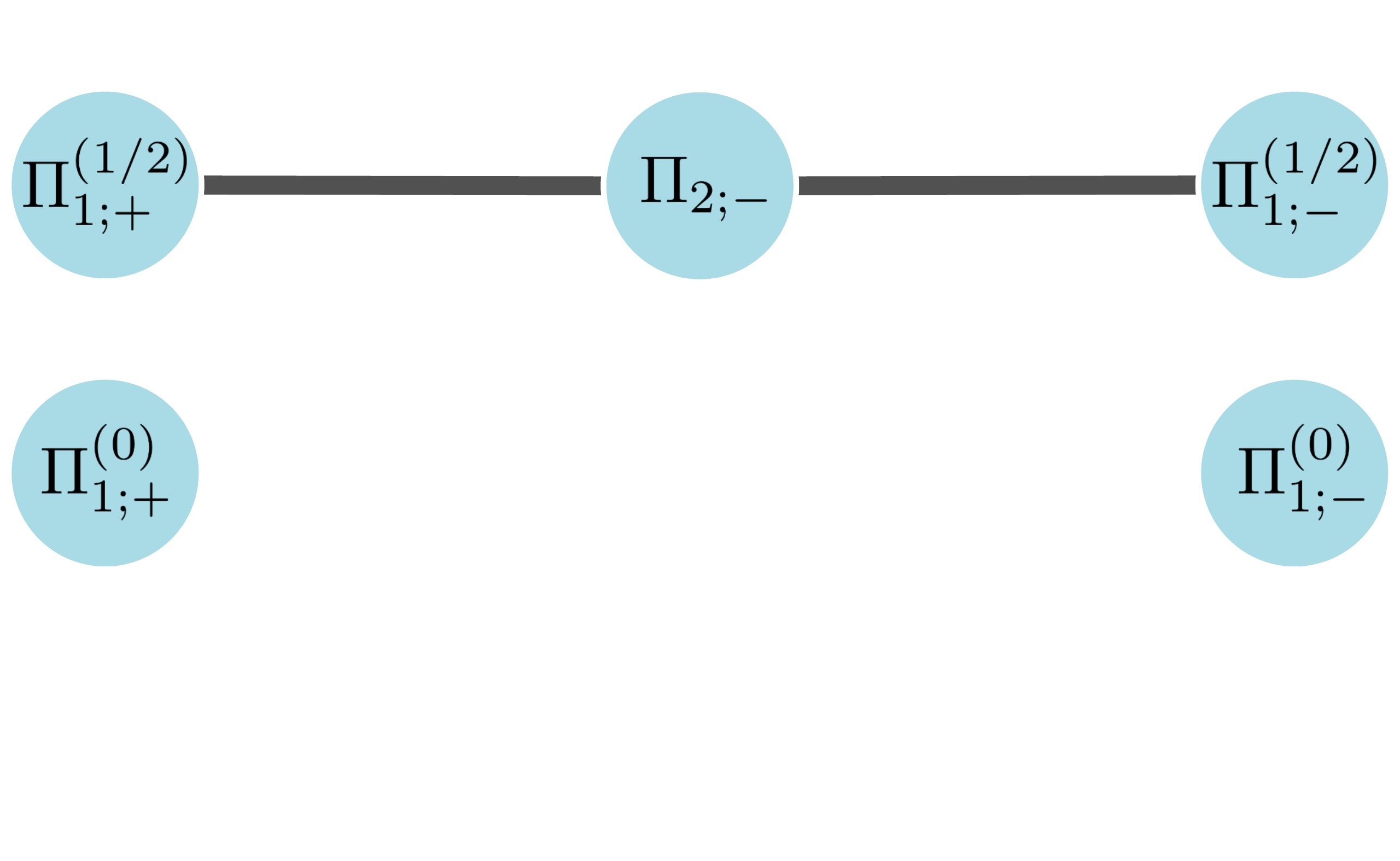}\caption{\label{fig: Bintree refnet}The reflection network from the generators
of the group $U\left(\mathcal{G}\right)$.}
\end{figure}
This, in turn, implies the following BPT

\noindent\begin{minipage}[c]{1\columnwidth}%
\begin{center}
\vspace{0.5\baselineskip}
\begin{tabular}{|c|ccc}
\cline{1-1} 
$1$ &  &  & \tabularnewline
\cline{1-1} 
$_{-}^{+2,3}$ &  &  & \tabularnewline
\cline{1-1} 
$_{-}^{+4,5,6,7}$ &  &  & \tabularnewline
\cline{1-1} 
$_{-}^{+8,9}$ &  &  & \tabularnewline
\cline{1-1} 
$10$ &  &  & \tabularnewline
\cline{1-2} \cline{2-2} 
\multicolumn{1}{c|}{} & \multicolumn{1}{c|}{$_{-3}^{+2}$} &  & \tabularnewline
\cline{2-2} 
\multicolumn{1}{c|}{} & \multicolumn{1}{c|}{$_{-6,7}^{+4,5}$} &  & \tabularnewline
\cline{2-2} 
\multicolumn{1}{c|}{} & \multicolumn{1}{c|}{$_{-9}^{+8}$} &  & \tabularnewline
\cline{2-4} \cline{3-4} \cline{4-4} 
\multicolumn{1}{c}{} & \multicolumn{1}{c|}{} & \multicolumn{1}{c|}{$_{-5,6}^{+4,7}$} & \multicolumn{1}{c|}{$_{-5,7}^{+4,6}$}\tabularnewline
\cline{3-4} \cline{4-4} 
\end{tabular} .\vspace{0.5\baselineskip}
\par\end{center}%
\end{minipage}

From the last block (bottom row) of the BPT we see that $U\left(\mathcal{G}\right)$
acts as a two-dimensional irrep on the subspace
\[
\mathcal{H}_{\mu_{3}}:=\spn\left\{ \text{\ensuremath{\ket{_{-5,6}^{+4,7}}}},\ket{_{-5,7}^{+4,6}}\right\} .
\]
Since the multiplicity of this irrep is one, the multiplicity subsystem
$\mathcal{H}_{\nu_{3}}$ is absorbed into $\mathcal{H}_{\mu_{3}}$.
The two other blocks in the BPT identify two distinct one-dimensional
irreps. Since these irreps are one-dimensional the irrep subsystems
$\mathcal{H}_{\mu_{1}}$ and $\mathcal{H}_{\mu_{2}}$ are absorbed
into the multiplicity subsystems
\begin{align*}
\mathcal{H}_{\nu_{1}} & :=\spn\left\{ \text{\ensuremath{\ket 1}},\ket{_{-}^{+2,3}},\ket{_{-}^{+4,5,6,7}},\ket{_{-}^{+8,9}},\ket{10}\right\} \\
\mathcal{H}_{\nu_{2}} & :=\spn\left\{ \ket{_{-3}^{+2}},\ket{_{-6,7}^{+4,5}},\ket{_{-9}^{+8}}\right\} .
\end{align*}
Overall, the $U\left(\mathcal{G}\right)$ irrep decomposition of the
Hilbert space is

\[
\mathcal{H}\cong\mathcal{H}_{\nu_{1}}\oplus\mathcal{H}_{\nu_{2}}\oplus\mathcal{H}_{\mu_{3}}.
\]
The two Hamiltonian terms $H=H_{\nu}+H_{\mu}$ are such that $H_{\mu}\in U\left(\mathcal{G}\right)$
and $H_{\nu}\in U\left(\mathcal{G}\right)'$. Therefore, with respect
to the above irreps structure they reduce to 
\begin{align*}
H_{\mu} & \cong I_{\nu_{1}}h_{\mu_{1}}\oplus I_{\nu_{2}}h_{\mu_{2}}\oplus H_{\mu_{3}}\\
H_{\nu} & \cong H_{\nu_{1}}\oplus H_{\nu_{2}}\oplus h_{\nu_{3}}I_{\mu_{3}}.
\end{align*}
Here $h_{\mu_{1}}$, $h_{\mu_{2}}$, $h_{\nu_{3}}$ are scalars and
$H_{\nu_{1}}$, $H_{\nu_{2}}$, $H_{\mu_{3}}$ are $5\times5$, $3\times3$
and $2\times2$ matrices respectively.

The important outcome from these analysis is that the dynamics are
restricted to the irrep sectors $\mathcal{H}_{\nu_{1}}$, $\mathcal{H}_{\nu_{2}}$,
$\mathcal{H}_{\mu_{3}}$. In particular, the sector $\mathcal{H}_{\nu_{1}}$
is spanned by the states $\text{\ensuremath{\ket 1}}$, $\ket{_{-}^{+2,3}}$,
$\ket{_{-}^{+4,5,6,7}}$, $\ket{_{-}^{+8,9}}$, $\ket{10}$ that dissect
the graph into the layers of the binary trees. Explicit construction
of the Hamiltonian term $H_{\nu_{1}}$ will show that it generates
a CTQW on a one-dimensional line constructed from these layer states.
In fact, we can present all the Hamiltonian terms as CTQW over the
irrep states that we have found.

In Fig. \ref{fig:CTQW over irrep states} we can see the term $H_{\nu_{1}}$
represented by the bottom line , the term $H_{\nu_{2}}$ represented
by the middle line, and the top vertical pair represents the term
$H_{\mu_{3}}$. The terms $H_{\nu_{1}}$ and $H_{\nu_{2}}$ traverse
the graph across layers while $H_{\mu_{3}}$ generates dynamics inside
the central layer. Since each term generates dynamics in a different
orthogonal subspace, each connected component in this graph evolves
independently from the others. In particular, this picture explains
the direct propagation from root $1$ to root $10$ over a subspace
that is exponentially smaller than the full tree. Therefore, what
we have shown here is that this speed-up holds even in the non-symmetric
version of the graph.

\begin{figure}[H]
\centering{}\includegraphics[clip,width=0.8\columnwidth]{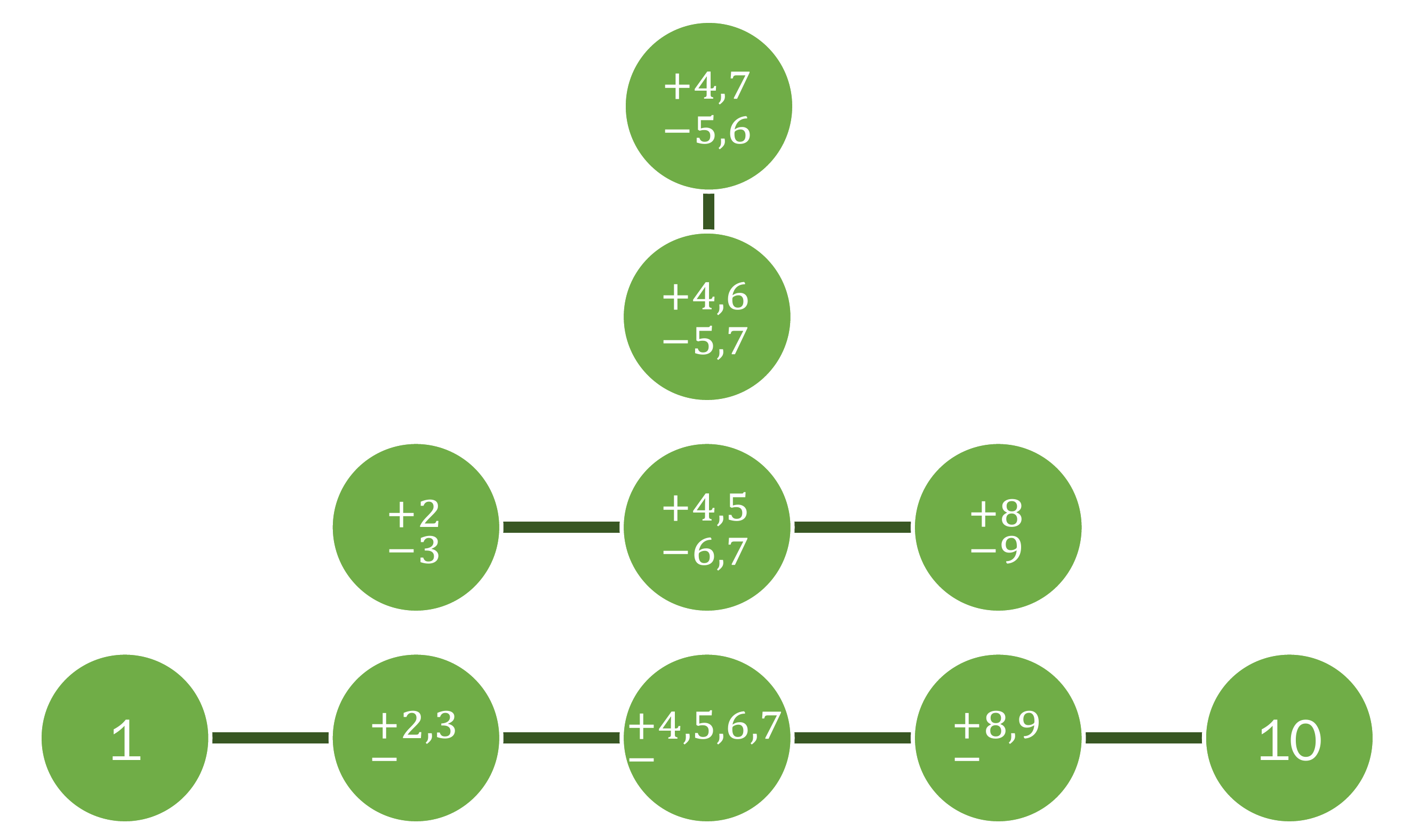}
\caption{\label{fig:CTQW over irrep states}The irrep basis representation
of the continuous-time quantum walk Hamiltonian over the binary trees
in Fig. \ref{fig:ctqw bin tree}. In the irrep basis, the Hamiltonian
decomposes into three terms that correspond to the three connected
components. The connected component in the button row indicates the
direct propagation from root $1$ to root $10$ that happens in a
subspace exponentially smaller than the full tree. The symmetry breaking
term of the Hamiltonian contributes to the dynamics in the top component,
but it does not affect the dynamics in the button row.}
\end{figure}

\section{Symmetry-agnostic reduction of Hamiltonians\label{sec:Symmetry-agnostic-reduction-of}}

In the previous section we have studied how the irreps structures
of symmetries lead to the reduction of Hamiltonians. In this section
we will demonstrate that it is not always necessary to identify the
symmetries in order to reduce Hamiltonians. Instead, we will focus
on directly finding the irreps structure that leads to the reduction.

The key takeaway from the discussion of symmetries is that the statement
``$U\left(\mathcal{G}\right)$ is a symmetry of $H$'' can be rephrased
as ``$H$ is an element of the commutant algebra $U\left(\mathcal{G}\right)'$''.
It is the observation that $H$ is an element of some non-trivial
algebra that leads to the reduction; the fact that this algebra happens
to be the commutant of a symmetry group is not important. Therefore,
if we can recognize that ``$H$ is an element of the algebra $\mathcal{A}$'',
for some non-trivial algebra $\mathcal{A}$, then we can reduce $H$
using the irreps structure of $\mathcal{A}$.

We can specialize the above idea as follows: The most obvious algebra
that we can use to restrict $H$ is the algebra generated by its additive
terms. That is, whenever we have $H=\sum_{k}\epsilon_{k}H_{k}$ we
can say that $H$ is an element of the algebra generated by the terms
$H_{k}$ and therefore it reduces to the irreps of $\left\langle \left\{ H_{k}\right\} \right\rangle $.
The illustrative example of the Scattering Algorithm that was given
in Section \ref{sec:How-the-Scattering} is exactly such a symmetry-agnostic
reduction of Hamiltonians.

As another simple example, consider again the spin-orbit coupled system
$\mathcal{H}=\underline{l}\otimes\underline{\frac{1}{2}}$ for some
integer $l$ and the Hamiltonian 
\[
H\left(\epsilon\right)=L_{z}\otimes S_{z}+\epsilon\,I\otimes S_{x}
\]
for arbitrary real constant $\epsilon$. Instead of identifying the
symmetries (which will still require finding the irreps structure)
we recognize that for all $\epsilon$, $H\left(\epsilon\right)$ is
an element of the algebra $\mathcal{A}$ generated by $L_{z}\otimes S_{z}$
and $I\otimes S_{x}$. In the last example of Section \ref{sec:State-reudcions-due}
we have derived the irreps structure of $\mathcal{A}$ to be given
by the BPT

\noindent\begin{minipage}[c]{1\columnwidth}%
\begin{center}
\vspace{0.5\baselineskip}
\begin{tabular}{ccccccc}
\cline{1-2} \cline{2-2} 
\multicolumn{1}{|c|}{$l,\uparrow$} & \multicolumn{1}{c|}{$l,\downarrow$} &  &  &  &  & \tabularnewline
\cline{1-2} \cline{2-2} 
\multicolumn{1}{|c|}{$-l,\downarrow$} & \multicolumn{1}{c|}{$-l,\uparrow$} &  &  &  &  & \tabularnewline
\cline{1-2} \cline{2-2} 
 &  & $\ddots$ &  &  &  & \tabularnewline
\cline{4-5} \cline{5-5} 
 &  & \multicolumn{1}{c|}{} & \multicolumn{1}{c|}{$1,\uparrow$} & \multicolumn{1}{c|}{$1,\downarrow$} &  & \tabularnewline
\cline{4-5} \cline{5-5} 
 &  & \multicolumn{1}{c|}{} & \multicolumn{1}{c|}{$-1,\downarrow$} & \multicolumn{1}{c|}{$-1,\uparrow$} &  & \tabularnewline
\cline{4-6} \cline{5-6} \cline{6-6} 
 &  &  &  & \multicolumn{1}{c|}{} & \multicolumn{1}{c|}{$0,+$} & \tabularnewline
\cline{6-7} \cline{7-7} 
 &  &  &  &  & \multicolumn{1}{c|}{} & \multicolumn{1}{c|}{$0,-$}\tabularnewline
\cline{7-7} 
\end{tabular} \vspace{0.5\baselineskip}
\par\end{center}%
\end{minipage} (note that in Section \ref{sec:State-reudcions-due} this BPT was
transposed since we were interested in the commutant of $\mathcal{A}$).
This BPT identifies the irrep decomposition
\[
\mathcal{H}\cong\ket{0,-}\oplus\ket{0,+}\oplus\left[\bigoplus_{q=1}^{l}\mathcal{H}_{\nu_{q}}\otimes\mathcal{H}_{\mu_{q}}\right],
\]
where both $\mathcal{H}_{\nu_{q}}$ and $\mathcal{H}_{\mu_{q}}$ are
two-dimensional virtual subsystems for all $q=1,...,l$.

Since $H\left(\epsilon\right)\in\mathcal{A}$, this irreps structure
implies that the Hamiltonian reduces to
\[
H\left(\epsilon\right)=h_{-}\left(\epsilon\right)\ketbra{0,-}{0,-}+h_{+}\left(\epsilon\right)\ketbra{0,+}{0,+}+\bigoplus_{q=1}^{l}I_{\nu_{q}}\otimes H_{\mu_{q}}\left(\epsilon\right)
\]
where $h_{\pm}\left(\epsilon\right)$ are scalars, $H_{\mu_{q}}\left(\epsilon\right)$
are $2\times2$ matrices and $I_{\nu_{q}}$ are $2\times2$ identities.
The explicit matrix elements are then given by 
\[
h_{\pm}\left(\epsilon\right)=\bra{0,\pm}H\left(\epsilon\right)\ket{0,\pm}=\pm\frac{\epsilon}{2}
\]
\[
H_{\mu_{q}}\left(\epsilon\right)=\begin{pmatrix}\bra{q,\uparrow}H\left(\epsilon\right)\ket{q,\uparrow} & \bra{q,\uparrow}H\left(\epsilon\right)\ket{q,\downarrow}\\
\bra{q,\downarrow}H\left(\epsilon\right)\ket{q,\uparrow} & \bra{q,\downarrow}H\left(\epsilon\right)\ket{q,\downarrow}
\end{pmatrix}=\frac{1}{2}\begin{pmatrix}q & \epsilon\\
\epsilon & -q
\end{pmatrix}.
\]

Thus, for each $q=1,...,l$, the Hamiltonian terms $L_{z}\otimes S_{z}$
and $\epsilon I\otimes S_{x}$ act as $q\sigma_{z}$ and $\epsilon\sigma_{x}$
(where $\sigma_{x}$, $\sigma_{z}$ are Pauli matrices) on the virtual
subsystems $\mathcal{H}_{\mu_{q}}$. If $\epsilon$ is a tunable parameter
then we can use $H_{\mu_{q}}\left(\epsilon\right)$ as the control
Hamiltonian for the logical qubit encoded in $\mathcal{H}_{\mu_{q}}$.
If, on the other hand, $\epsilon$ is the uncontrollable random noise
then $\mathcal{H}_{\nu_{q}}$ can be used as a decoherence free subsystem.

In order to further demonstrate the potential applications of the
symmetry-agnostic approach, we will analyze two examples dealing with
the qubit encodings in quantum dot arrays.

The idea of qubit implementations in quantum dots was first proposed
in \citep{Loss98Quantum}. In this setting, individual electrons are
trapped in manufactured potential wells (referred to as ``dots'')
where they can be controlled by the electric potentials that set the
barriers between adjacent dots, and by applying external magnetic
fields. The overall dynamics of such systems are described by the
Hubbard model \citep{Burkard99Coupled}, where the degrees of freedom
are the occupation numbers of electrons in the individual dots (also
referred to as ``orbital'' or ``charge'' degree of freedom), and
the spin degrees of freedom.

Because of the multiple degrees of freedom, there is a variety of
possible qubit encodings in quantum dots; see \citep{russ2017three}
for an overview. Different qubit encodings have different advantages
and disadvantages\footnote{The key characteristics are the levels and sources of noise from gate
operations and the complexity of two-qubit gates.} and it is not our goal to explore these issues here. What we will
focus on is how to identify the possible qubit encodings in the first
place, which at the very least should accommodate arbitrary Bloch
sphere rotations.

In the following two examples we will consider the effective control
Hamiltonian of a quantum dot system and find the possible qubit encodings
where arbitrary Bloch sphere rotations can be performed. This will
be achieved by adopting the symmetry-agnostic approach and finding
out how the independent terms of the Hamiltonian can be reduced. The
reduced subspaces (or subsystems) of the independent terms will then
identify the possible encodings.

\setcounter{example}{1}

\subsubsection*{Example \theexample \refstepcounter{example}}

In this example we will consider the charge quadrupole qubit that
was proposed in \citep{friesen2017decoherence}. The charge quadrupole
qubit is designed to be a more robust version of the charge dipole
qubit against the electric potential noise. By taking the symmetry-agnostic
approach we will show that there is a continuum of possible qubit
encodings between the charge quadrupole and the charge dipole cases
that has not been considered. Due to the systematic nature of this
approach we will also rule out the possibility of any other charge-qubit
encodings in this setting.

Our system consists of a single electron trapped in a triple quantum
dot where it can occupy the dots $\ket 1$, $\ket 2$, $\ket 3$.
We will disregard the spin degree of freedom so our Hilbert space
is just $\mathcal{H}=\spn\left\{ \ket 1,\ket 2,\ket 3\right\} $.
As discussed in \citep{friesen2017decoherence}, the effective control
Hamiltonian has five tunable parameters

\[
H=\begin{pmatrix}u_{1} & t_{12} & 0\\
t_{12} & u_{2} & t_{23}\\
0 & t_{23} & u_{3}
\end{pmatrix}=\begin{pmatrix}\epsilon_{d} & t_{12} & 0\\
t_{12} & \epsilon_{q} & t_{23}\\
0 & t_{23} & -\epsilon_{d}
\end{pmatrix}+\epsilon_{0}I
\]
where $t_{12}$, $t_{23}$ are the tunneling amplitudes, and $u_{1}$,
$u_{2}$, and $u_{3}$ are the dot potentials that can be re-stated
as the detuning parameters
\[
\epsilon_{d}=\frac{u_{1}-u_{3}}{2}\hspace{2cm}\epsilon_{q}=u_{2}-\frac{u_{1}+u_{3}}{2}\hspace{2cm}\epsilon_{0}=\frac{u_{1}+u_{3}}{2}.
\]

The independent Hamiltonian terms here are
\[
T_{12}:=\begin{pmatrix}0 & 1 & 0\\
1 & 0 & 0\\
0 & 0 & 0
\end{pmatrix}\hspace{1cm}T_{23}:=\begin{pmatrix}0 & 0 & 0\\
0 & 0 & 1\\
0 & 1 & 0
\end{pmatrix}\hspace{1cm}D:=\begin{pmatrix}1 & 0 & 0\\
0 & 0 & 0\\
0 & 0 & -1
\end{pmatrix}\hspace{1cm}Q:=\begin{pmatrix}0 & 0 & 0\\
0 & 1 & 0\\
0 & 0 & 0
\end{pmatrix},
\]
so
\[
H=t_{12}T_{12}+t_{23}T_{23}+\epsilon_{d}D+\epsilon_{q}Q
\]
and we have dropped the inconsequential identity term.

In principle, we can assert that $H\in\left\langle T_{12},T_{23},D,Q\right\rangle $
but it is not a very helpful assertion because, as we will see shortly,
the algebra $\left\langle T_{12},T_{23},D,Q\right\rangle $ is the
trivial irreducible algebra $\mathcal{L}\left(\mathcal{H}\right)$.
However, since there are excessive degrees of freedom here (we only
need two independent parameters for arbitrary Bloch sphere rotations)
we can constrain some of the parameters such that the constrained
terms will become reducible.

First, we show that $t_{12}$ and $t_{23}$ cannot be both independent
because the algebra $\left\langle T_{12},T_{23}\right\rangle $ is
irreducible. The spectral projections of $T_{ij}=\Pi_{ij;+}-\Pi_{ij;-}$
are given by

\[
\Pi_{ij;+}:=\ketbra{_{-}^{+i,j}}{_{-}^{+i,j}}\hspace{2cm}\Pi_{ij;-}:=\ketbra{_{-j}^{+i}}{_{-j}^{+i}}
\]
where we have used the shorthand state notation of Eq. \eqref{eq: shorthand state notation}.
Since these spectral projections are rank-$1$ they are all reflecting,
resulting in the reflection network shown in Fig. \ref{fig: T_ij refnet}.

\begin{figure}[H]
\centering{}\includegraphics[viewport=50bp 50bp 650bp 433bp,clip,width=0.4\columnwidth]{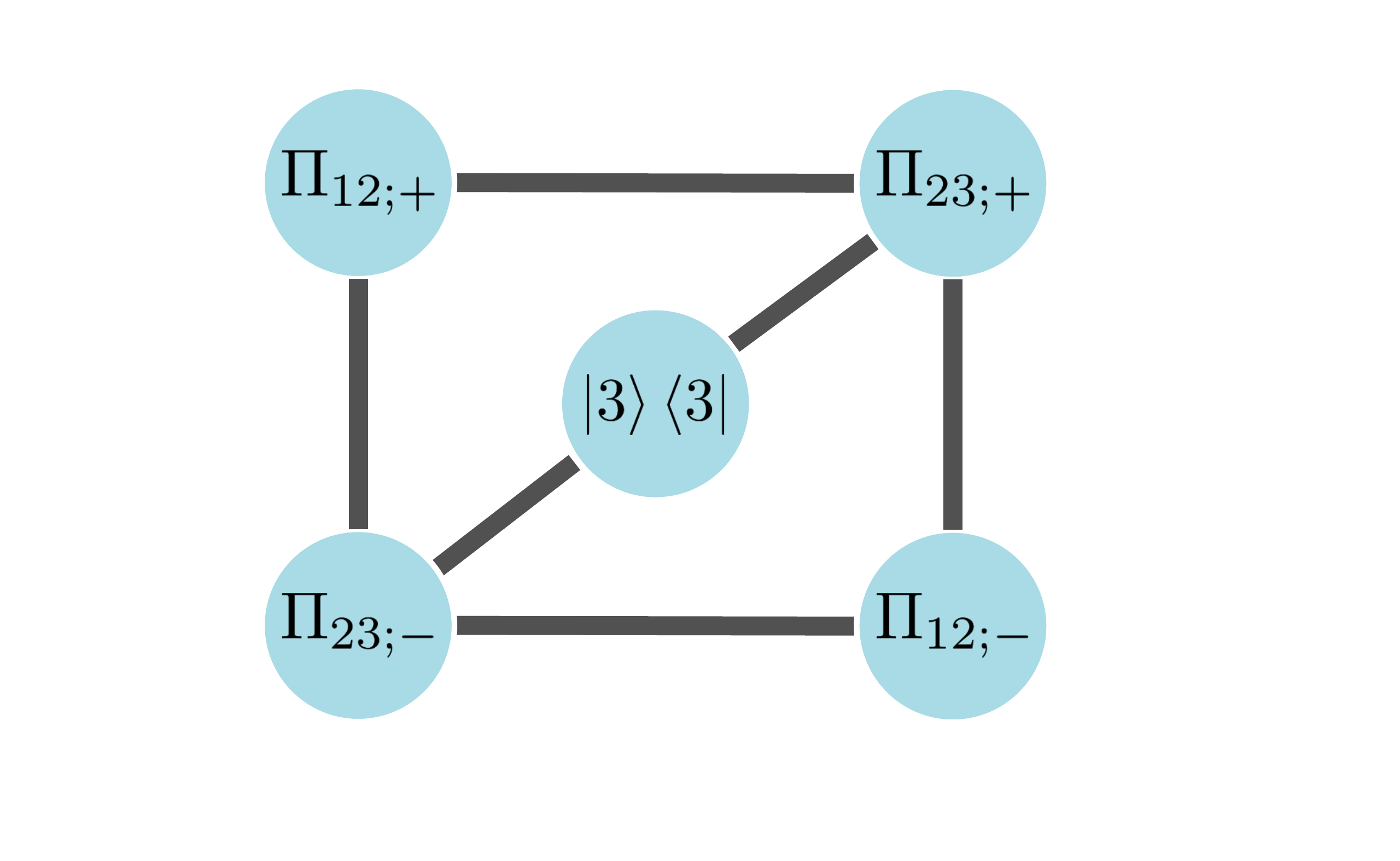}\caption{\label{fig: T_ij refnet}The reflection network from the generators
$T_{12}$, $T_{23}$.}
\end{figure}
Note that both $T_{12}$ and $T_{23}$ are not supported on the whole
Hilbert space ($\Pi_{ij;+}+\Pi_{ij;-}\neq I$) so the resulting reflection
network was incomplete. The projection $\ketbra 33$ was added by
the completion procedure described in Section \ref{subsec:Establishing-completeness-of}.
The resulting BPT is

\noindent\begin{minipage}[c]{1\columnwidth}%
\begin{center}
\vspace{0.5\baselineskip}
\begin{tabular}{|c|c|c|}
\hline 
$_{-}^{+1,2}$ & $_{-2}^{+1}$ & $3$\tabularnewline
\hline 
\end{tabular}\vspace{0.5\baselineskip}
\par\end{center}%
\end{minipage} which defines the full operator algebra $\mathcal{L}\left(\mathcal{H}\right)$.\footnote{This is always the case when the reflection network consists of a
single connected component of rank-$1$ projections and is supported
on the whole $\mathcal{H}$.} Therefore, the action of both $T_{12}$ and $T_{23}$ is irreducible
on $\mathcal{H}$.

Assuming that $t_{12}$ and $t_{23}$ have the same sign, we constrain
these two independent parameters to be $t_{12}=\alpha t$ and $t_{23}=(1-\alpha)t$
for a constant $0\leq\alpha\leq1$ and the new common parameter $t$.
The new constrained term is 
\[
T_{\alpha}:=\alpha T_{12}+\left(1-\alpha\right)T_{23}=\begin{pmatrix}0 & \alpha & 0\\
\alpha & 0 & 1-\alpha\\
0 & 1-\alpha & 0
\end{pmatrix}
\]
so the Hamiltonian is $H=tT_{\alpha}+\epsilon_{d}D+\epsilon_{q}Q$.

For $\alpha=0,1$ the problem reduces to a double quantum dot where
(assuming $\alpha=1$) the positions $\ket 1$, $\ket 2$ identify
the logical qubit basis $\ket{0_{L}}$, $\ket{1_{L}}$ and
\[
H=\begin{pmatrix}\epsilon_{d} & t & 0\\
t & \epsilon_{q} & 0\\
0 & 0 & -\epsilon_{d}
\end{pmatrix}.
\]
Thus, $T_{\alpha=1}=T_{12}$ serves as the Pauli $\sigma_{x}$ operator
and $D$ or $Q$ can serve as $\frac{1}{2}\left(I\pm\sigma_{z}\right)$
operators. This encoding is known as charge dipole.

For $0<\alpha<1$ we observe that the parameters $\epsilon_{d}$ and
$\epsilon_{q}$ also cannot be both independent since the three terms
$T_{\alpha}$, $D$, $Q$ cannot be reduced onto a two-dimensional
subspace. We can see this by considering the three spectral projections
$\Pi_{i=1,2,3}:=\ketbra ii$ provided by $D$ and $Q$. All three
$\Pi_{i=1,2,3}$ are rank-$1$ and orthogonal but we also have to
include the spectral projections of $T_{\alpha}$. Since $\Pi_{1}T_{\alpha}\Pi_{2}\neq0$
and $\Pi_{2}T_{\alpha}\Pi_{3}\neq0$ the resulting reflection network
will have a single connected component that contains all three projections
$\Pi_{i=1,2,3}$. The resulting BPT will therefore be that of the
trivial operator algebra $\mathcal{L}\left(\mathcal{H}\right)$ which
implies the irreducibility of $\left\langle T_{\alpha},D,Q\right\rangle $.

Assuming that $\epsilon_{d}$ and $\epsilon_{q}$ can have opposite
signs, we constrain them as $\epsilon_{d}=\beta\epsilon$ and $\epsilon_{q}=\left(1-\left|\beta\right|\right)\epsilon$
for a constant $-1\leq\beta\leq1$ and the new common parameter $\epsilon$.
The new constrained term is
\[
E_{\beta}:=\beta D+\left(1-\left|\beta\right|\right)Q=\begin{pmatrix}\beta & 0 & 0\\
0 & 1-\left|\beta\right| & 0\\
0 & 0 & -\beta
\end{pmatrix}
\]
so the Hamiltonian is $H=tT_{\alpha}+\epsilon E_{\beta}$.

For $\beta\neq0,\pm1,\pm\frac{1}{2}$ it is easy to see the spectral
projections of $E_{\beta}$ are still the three rank-$1$ projections
$\Pi_{i=1,2,3}$ so the algebra $\left\langle T_{\alpha},E_{\beta}\right\rangle $
is still irreducible. It is less obvious that $\left\langle T_{\alpha},E_{\beta}\right\rangle $
is irreducible even for $\beta=\pm1,\pm\frac{1}{2}$, however, that
is also the case. When $\beta=\pm1$ the projection $\Pi_{2}$ is
initially not a spectral projection of $E_{\beta}$ but the resulting
reflection network will be incomplete. The completion procedure will
add $\Pi_{2}$ and we will have $\Pi_{i=1,2,3}$ once again in the
same connected component. When $\beta=\pm\frac{1}{2}$ one of the
spectral projections of $E_{\beta}$ is rank-$2$ but it will scatter
with the spectral projection of $T_{\alpha}$ and form a single connected
component of rank-$1$ projections. Either way, for all $\beta\neq0$
we will have a single connected component of rank-$1$ projections
which implies the irreducibility of $\left\langle T_{\alpha},E_{\beta}\right\rangle $.

We are therefore left with $\beta=0$, that is, $E_{\beta=0}=Q$ and
the only spectral projection of $Q$ is $\Pi_{2}$. The eigenvalues
of $T_{\alpha}$ are $\pm\lambda_{\alpha}=\pm$$\sqrt{\alpha^{2}+\left(1-\alpha\right)^{2}}$
and the eigenvectors are
\[
\ket{\alpha;\pm}:=\frac{\alpha\ket 1\pm\lambda_{\alpha}\ket 2+\left(1-\alpha\right)\ket 3}{\lambda_{\alpha}\sqrt{2}}.
\]
The two spectral projections of $T_{\alpha}$ are therefore $\Pi_{\alpha;\pm}:=\ket{\alpha;\pm}\bra{\alpha;\pm}$
and the reflection network of the algebra $\left\langle T_{\alpha},Q\right\rangle $
is shown in Fig. \ref{fig: TQ refnet}.

\begin{figure}[H]
\centering{}\includegraphics[viewport=0bp 270bp 737bp 453bp,clip,width=0.5\columnwidth]{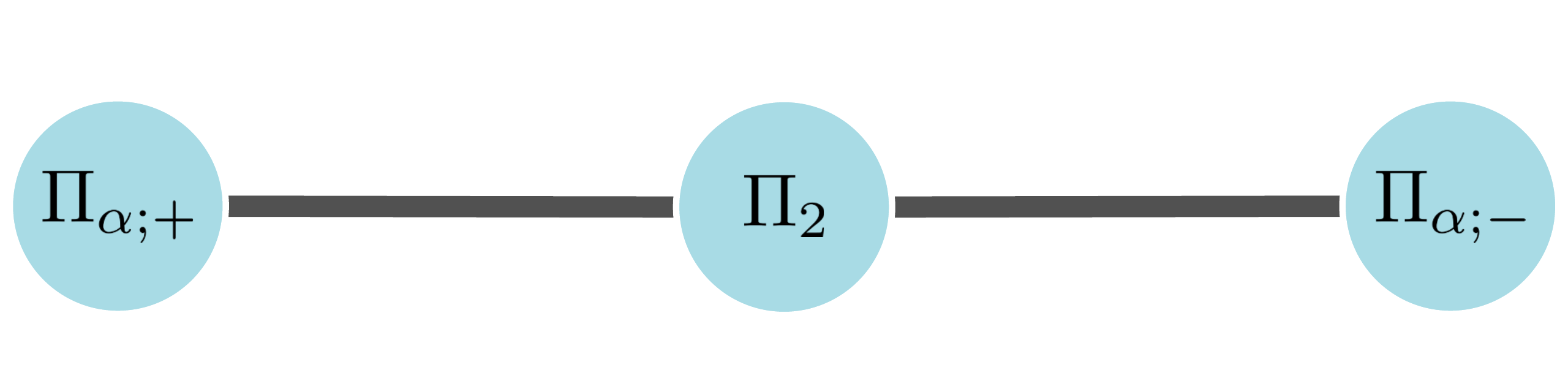}\caption{\label{fig: TQ refnet}The reflection network of $\left\langle T_{\alpha},Q\right\rangle $.}
\end{figure}

The resulting BPT is

\noindent\begin{minipage}[c]{1\columnwidth}%
\begin{center}
\vspace{0.5\baselineskip}
\begin{tabular}{|c|c|}
\hline 
$\alpha;+$ & $\alpha;-$\tabularnewline
\hline 
\end{tabular}\vspace{0.5\baselineskip}
\par\end{center}%
\end{minipage} so $T_{\alpha}$ and $Q$ reduce to the subspace spanned by $\ket{\alpha;+}$
and $\ket{\alpha;-}$. There is a freedom of choice for the logical
qubit basis, but if we want to be consistent with the choice we made
for the case $\alpha=0,1$ we will choose

\begin{align*}
\ket{0_{L}} & =\frac{\ket{\alpha;+}+\ket{\alpha;-}}{\sqrt{2}}=\frac{\alpha\ket 1+\left(1-\alpha\right)\ket 3}{\lambda_{\alpha}}\\
\ket{1_{L}} & =\frac{\ket{\alpha;+}-\ket{\alpha;-}}{\sqrt{2}}=\ket 2.
\end{align*}
The operator $T_{\alpha}$ then acts on $\ket{0_{L}}$, $\ket{1_{L}}$
as $\lambda_{\alpha}\sigma_{x}$, and the operator $Q$ acts as $\frac{1}{2}\left(I-\sigma_{z}\right)$.
By turning the terms $T_{\alpha}$ and $Q$ on and off we can perform
arbitrary Bloch sphere rotations.

In \citep{friesen2017decoherence} the authors have identified the
charge quadrupole qubit that corresponds to the case of $\alpha=\frac{1}{2}$.
Here, using the symmetry-agnostic approach, we have identified qubit
encodings for the continuum $0\leq\alpha\leq1$ where the edge cases
$\alpha=0,1$ correspond to the charge dipole encodings.

\subsubsection*{Example \theexample \refstepcounter{example}}

In quantum dot arrays where each dot is occupied by one electron,
a controlled variation of electric potentials allows the realization
of Heisenberg interactions between adjacent electrons \citep{Loss98Quantum}.
If the qubits are encoded in the spin degree of freedom of the trapped
electrons (as envisioned in \citep{Loss98Quantum}), then in addition
to the electric potentials the qubit gate operations require the use
of variable magnetic fields. The necessity of variable magnetic fields
complicates the design and reduces the performance of this qubits
due to the additional sources of noise and the relatively weak coupling
of the magnetic field to the electron spins.

The surprising resolution to this issue was brought forward by DiVincenzo
\emph{et al.} \citep{divincenzo2000universal} (see also \citep{Bacon2000Universal}),
where the encoding of a single qubit in a subspace of three (or four)
electron spins was proposed. The key advantage of this encoding is
that all qubit gates can be performed with the Heisenberg interaction
only, which in quantum dots is realized by electric potentials without
the need in variable magnetic fields.

In the following, we will first discuss how this qubit encoding is
identified by the irreps of the $SU\left(2\right)$ symmetry of the
Heisenberg interaction. Then, by adopting the symmetry-agnostic approach
we will show that additional encodings in four spins can be found.
These additional encodings are not identified by the $SU\left(2\right)$
symmetry alone and they were not recognized in the original proposal
\citep{Bacon2000Universal}.

The basic idea in \citep{divincenzo2000universal,Bacon2000Universal}
is to encode the qubit in a non-periodic spin-chain of three (as in
\citep{divincenzo2000universal}) or four (as in \citep{Bacon2000Universal})
spins with tunable nearest neighbor interactions 
\[
H=\epsilon_{12}\vec{S}_{1}\cdot\vec{S}_{2}+\epsilon_{23}\vec{S}_{2}\cdot\vec{S}_{3}+\epsilon_{34}\vec{S}_{3}\cdot\vec{S}_{4}
\]
(the term $\epsilon_{34}\vec{S}_{3}\cdot\vec{S}_{4}$ is only relevant
for the four-spin case). Since $H$ is an element of the commutant
of the $SU\left(2\right)$ group, it reduces to the irreps of this
commutant (these structures are sometimes referred to as decoherence
free subspaces and subsystems). The possible qubit encodings can therefore
be identified from the irreps of the commutant of $SU\left(2\right)$.

The commutant's irreps are given by the transposition of the BPT of
the $SU\left(2\right)$ irreps. In the total spin basis $\ket{j,\alpha,m}$,
the transposed BPT for the three-spin case is

\noindent\begin{minipage}[c]{1\columnwidth}%
\begin{center}
\vspace{0.5\baselineskip}
\begin{tabular}{c|cc}
\cline{1-1} 
\multicolumn{1}{|c|}{$\frac{3}{2},+\frac{3}{2}$} &  & \tabularnewline
\cline{1-1} 
\multicolumn{1}{|c|}{$\vdots$} &  & \tabularnewline
\cline{1-1} 
\multicolumn{1}{|c|}{$\frac{3}{2},-\frac{3}{2}$} &  & \tabularnewline
\hline 
 & \multicolumn{1}{c|}{$\frac{1}{2},s,+\frac{1}{2}$} & \multicolumn{1}{c|}{$\frac{1}{2},t,+\frac{1}{2}$}\tabularnewline
\cline{2-3} \cline{3-3} 
 & \multicolumn{1}{c|}{$\frac{1}{2},s,-\frac{1}{2}$} & \multicolumn{1}{c|}{$\frac{1}{2},t,-\frac{1}{2}$}\tabularnewline
\cline{2-3} \cline{3-3} 
\multicolumn{1}{c}{} & $\downarrow$ & $\downarrow$\tabularnewline
\cline{2-3} \cline{3-3} 
 & \multicolumn{1}{c|}{$0_{L}$} & \multicolumn{1}{c|}{$1_{L}$}\tabularnewline
\cline{2-3} \cline{3-3} 
\end{tabular} .\vspace{0.5\baselineskip}
\par\end{center}%
\end{minipage} From this BPT we can see that in the three-spin case, the candidates
for the qubit encodings are the two-dimensional invariant subspaces
of $j=\frac{1}{2}$ and $m=\frac{1}{2}$ or $j=\frac{1}{2}$ and $m=-\frac{1}{2}$
(top and bottom rows of the right block). As we have seen in Eq. \eqref{eq:Heisnberg nu_1/2 Hamil},
the Heisenberg interaction Hamiltonian acts on these subspaces as
\[
H_{\nu_{1/2}}=\frac{1}{4}\begin{pmatrix}-3\epsilon_{23} & \sqrt{3}\epsilon_{12}\\
\sqrt{3}\epsilon_{12} & \epsilon_{23}-2\epsilon_{12}
\end{pmatrix}=\frac{\epsilon_{12}}{4}\left(\sqrt{3}\sigma_{x}+\sigma_{z}-I\right)+\frac{\epsilon_{23}}{8}\left(\sigma_{z}+5I\right).
\]
One can change the basis inside these subspaces to get a more convenient
form of $H_{\nu_{1/2}}$ in terms of Pauli matrices. Regardless of
the choice of basis, we can implement arbitrary qubit rotations in
this subspace by controlling the couplings $\epsilon_{12}$, $\epsilon_{23}$.

Thus, the irreps of the commutant of $SU\left(2\right)$ identify
the two-dimensional subspaces on which the independent interaction
terms of the Hamiltonian can be reduced. Note that in the case of
three spins there are exactly two independent parameters $\epsilon_{12}$,
$\epsilon_{23}$ which is the minimal number of parameters needed
to control a qubit. We will now see that with four spins, where there
are three independent parameters, we can utilize the extra degree
of freedom to find qubit subspaces that are not captured by the $SU\left(2\right)$
symmetry alone.\footnote{In the case of three spins in a periodic configuration we could also
have three independent parameters but the Hilbert space is too small
to take advantage of that.}

Let us first identify the four-spin qubit encoding using $SU\left(2\right)$
symmetry. In this case, the representation of $SU\left(2\right)$
decomposes as
\begin{align*}
\underline{\frac{1}{2}}^{\otimes4} & =\left[\underline{\frac{1}{2}}\otimes\underline{\frac{1}{2}}\right]\otimes\left[\underline{\frac{1}{2}}\otimes\underline{\frac{1}{2}}\right]=\left[\underline{1}\oplus\underline{0}\right]\otimes\left[\underline{1}\oplus\underline{0}\right]\\
 & =(\underline{1}\otimes\underline{1})\oplus(\underline{1}\otimes\underline{0})\oplus(\underline{0}\otimes\underline{1})\oplus(\underline{0}\otimes\underline{0})\\
 & =\left(\underline{2}\oplus\underline{1}\oplus\underline{0}\right)\oplus\left(\underline{1}\right)\oplus\left(\underline{1}\right)\oplus\left(\underline{0}\right).
\end{align*}
In order to distinguish the different variant of equivalent irreps
we will use the labels $\alpha=tt,ts,st,ss$ that specify whether
the left and right pairs of spins are in the singlet or triplet states.
For $j=2$ there are no variants to distinguish and $\alpha=tt$;
for $j=1$ we have the three variants $\alpha=tt,ts,st$, and for
$j=0$ we have $\alpha=tt,ss$.

In the total spin basis $\ket{j,\alpha,m}$, the transposed BPT of
the irreps of the four-spin-representation of $SU\left(2\right)$
is

\noindent\begin{minipage}[c]{1\columnwidth}%
\begin{center}
\vspace{0.5\baselineskip}
\begin{tabular}{c|lll|cc}
\cline{1-1} 
\multicolumn{1}{|c|}{$2,tt,+2$} &  &  & \multicolumn{1}{l}{} &  & \tabularnewline
\cline{1-1} 
\multicolumn{1}{|c|}{$\vdots$} &  &  & \multicolumn{1}{l}{} &  & \tabularnewline
\cline{1-1} 
\multicolumn{1}{|c|}{$2,tt,-2$} &  &  & \multicolumn{1}{l}{} &  & \tabularnewline
\cline{1-4} \cline{2-4} \cline{3-4} \cline{4-4} 
 & \multicolumn{1}{l|}{$1,tt,+1$} & \multicolumn{1}{l|}{$1,ts,+1$} & $1,st,+1$ &  & \tabularnewline
\cline{2-4} \cline{3-4} \cline{4-4} 
 & \multicolumn{1}{l|}{$1,tt,0$} & \multicolumn{1}{l|}{$1,ts,0$} & $1,st,0$ &  & \tabularnewline
\cline{2-4} \cline{3-4} \cline{4-4} 
 & \multicolumn{1}{l|}{$1,tt,-1$} & \multicolumn{1}{l|}{$1,ts,-1$} & $1,st,-1$ &  & \tabularnewline
\cline{2-6} \cline{3-6} \cline{4-6} \cline{5-6} \cline{6-6} 
\multicolumn{1}{c}{} &  &  &  & \multicolumn{1}{c|}{$0,tt$} & \multicolumn{1}{c|}{$0,ss$}\tabularnewline
\cline{5-6} \cline{6-6} 
\multicolumn{1}{c}{} &  &  & \multicolumn{1}{l}{} & $\downarrow$ & $\downarrow$\tabularnewline
\cline{5-6} \cline{6-6} 
\multicolumn{1}{c}{} &  &  &  & \multicolumn{1}{c|}{$0_{L}$} & \multicolumn{1}{c|}{$1_{L}$}\tabularnewline
\cline{5-6} \cline{6-6} 
\end{tabular} .\vspace{0.5\baselineskip}
\par\end{center}%
\end{minipage} This irreps structure suggests that the only two-dimensional subspace
on which all the independent terms of $H$ can be reduced is the subspace
of $j=0$. Thus, as proposed in \citep{Bacon2000Universal}, the subspace
of $j=0$ can be used to encode the qubit 
\begin{align*}
\ket{0_{L}} & =\ket{0,tt}\\
\ket{1_{L}} & =\ket{0,ss}
\end{align*}
that is controlled by the independent terms of $H$.

From the $SU\left(2\right)$ symmetry perspective, the subspace of
$j=1$ is ruled out for qubit encodings because the independent terms
of $H$ reduce there only onto three-dimensional subspaces.\footnote{Note that these subspaces are three-dimensional not because $j=1$
are three-dimensional irreps but because there are three variants
$\alpha=tt,ts,st$ of the $j=1$ irreps.} This, however, ignores the fact that we do not actually need to reduce
all three independent terms to control a qubit; two will suffice.

For concreteness let us focus on the subspace of $j=1$ and $m=1$
(the top row in the middle block of the BPT) and simplify the basis
notation $\ket{\alpha}\equiv\ket{1,\alpha,1}$ where $\alpha=tt,ts,st$.
In terms of the product basis of the four spins we have
\begin{align*}
\ket{ts} & =\ket{\uparrow\uparrow}\frac{\ket{\uparrow\downarrow}-\ket{\downarrow\uparrow}}{\sqrt{2}}=\frac{\ket{\uparrow\uparrow\uparrow\downarrow}-\ket{\uparrow\uparrow\downarrow\uparrow}}{\sqrt{2}}\\
\ket{st} & =\frac{\ket{\uparrow\downarrow}-\ket{\downarrow\uparrow}}{\sqrt{2}}\ket{\uparrow\uparrow}=\frac{\ket{\uparrow\downarrow\uparrow\uparrow}-\ket{\downarrow\uparrow\uparrow\uparrow}}{\sqrt{2}}\\
\ket{tt} & =\frac{1}{\sqrt{2}}\left(\frac{\ket{\uparrow\downarrow}+\ket{\downarrow\uparrow}}{\sqrt{2}}\ket{\uparrow\uparrow}-\ket{\uparrow\uparrow}\frac{\ket{\uparrow\downarrow}+\ket{\downarrow\uparrow}}{\sqrt{2}}\right)=\frac{\ket{\uparrow\downarrow\uparrow\uparrow}+\ket{\downarrow\uparrow\uparrow\uparrow}-\ket{\uparrow\uparrow\uparrow\downarrow}-\ket{\uparrow\uparrow\downarrow\uparrow}}{2}.
\end{align*}

When restricted to this subspace (the restriction is denoted with
the brackets $\left[\,\,\right]$), the interaction terms of the left
and right pairs of spins are diagonal in the $\ket{\alpha}$ basis
\[
\left[\vec{S}_{1}\cdot\vec{S}_{2}\right]=\frac{1}{4}\ketbra{tt}{tt}+\frac{1}{4}\ketbra{ts}{ts}-\frac{3}{4}\ketbra{st}{st}
\]
\[
\left[\vec{S}_{3}\cdot\vec{S}_{4}\right]=\frac{1}{4}\ketbra{tt}{tt}+\frac{1}{4}\ketbra{st}{st}-\frac{3}{4}\ketbra{ts}{ts}.
\]
The interaction term of the central pair of spins $\vec{S}_{2}\cdot\vec{S}_{3}$
is similarly diagonal but in a different basis
\[
\left[\vec{S}_{2}\cdot\vec{S}_{3}\right]=\frac{1}{4}\ketbra{_{t}^{t}}{_{t}^{t}}+\frac{1}{4}\ketbra{_{s}^{t}}{_{s}^{t}}-\frac{3}{4}\ketbra{_{t}^{s}}{_{t}^{s}}.
\]
In this notation the top letter $t$ or $s$ refers to the triplet
or singlet state of the central pair of spins ($2,3$), and the button
letter refers to the state of the pair of boundary spins ($1,4$).
Explicitly these are
\begin{align*}
\ket{_{s}^{t}} & =\frac{\ket{\uparrow\uparrow\uparrow\downarrow}-\ket{\downarrow\uparrow\uparrow\uparrow}}{\sqrt{2}}=\frac{\ket{ts}+\ket{st}-\sqrt{2}\ket{tt}}{2}\\
\ket{_{t}^{s}} & =\frac{\ket{\uparrow\downarrow\uparrow\uparrow}-\ket{\uparrow\uparrow\downarrow\uparrow}}{\sqrt{2}}=\frac{\ket{ts}+\ket{st}+\sqrt{2}\ket{tt}}{2}\\
\ket{_{t}^{t}} & =\frac{\ket{\uparrow\uparrow\downarrow\uparrow}+\ket{\uparrow\downarrow\uparrow\uparrow}-\ket{\uparrow\uparrow\uparrow\downarrow}-\ket{\downarrow\uparrow\uparrow\uparrow}}{2}=\frac{\ket{st}-\ket{ts}}{\sqrt{2}}.
\end{align*}

Let us now constrain the two independent terms $\vec{S}_{1}\cdot\vec{S}_{2}$
and $\vec{S}_{3}\cdot\vec{S}_{4}$ into one independent term. That
is, for some $0\leq\beta\leq1$ we have the new term
\[
\left[S_{\beta}\right]:=\beta\left[\vec{S}_{1}\cdot\vec{S}_{2}\right]+\left(1-\beta\right)\left[\vec{S}_{3}\cdot\vec{S}_{4}\right]=\frac{1}{4}\ketbra{tt}{tt}+\left(\beta-\frac{3}{4}\right)\ketbra{ts}{ts}+\left(\frac{1}{4}-\beta\right)\ketbra{st}{st}.
\]
Note that the degeneracy of the eigenspaces of $\left[S_{\beta}\right]$
changes when $\beta=0,\frac{1}{2},1$. For $\beta\neq0,\frac{1}{2},1$,
its spectral projections are $\ketbra{tt}{tt}$, $\ketbra{ts}{ts}$,
$\ketbra{st}{st}$ which end up in a fully connected reflection network
if we scatter them with the spectral projections of $\left[\vec{S}_{2}\cdot\vec{S}_{3}\right]$.
This means that $\left[S_{\beta}\right]$ and $\left[\vec{S}_{2}\cdot\vec{S}_{3}\right]$
are irreducible for $\beta\neq0,\frac{1}{2},1$.

For $\beta=1$ or $\beta=0$ we disregard one of the interactions
$\left[\vec{S}_{3}\cdot\vec{S}_{4}\right]$ or $\left[\vec{S}_{1}\cdot\vec{S}_{2}\right]$
and focus on the interactions between three spins. This will lead
to a reduction of the remaining two terms onto a two dimensional subspace
inside the $j=1$, $m=1$ subspace. We will not elaborate on this
reduction since it is not much different than the three-spin encoding.

The genuinely different encoding that is possible here is for $\beta=\frac{1}{2}$.
The spectral projections of $\left[S_{\beta=\frac{1}{2}}\right]$
are 
\[
\Pi_{ts,st}=\ketbra{ts}{ts}+\ketbra{st}{st}\hspace{1cm}\Pi_{tt}=\ketbra{tt}{tt}
\]
and the spectral projections of $\left[\vec{S}_{2}\cdot\vec{S}_{3}\right]$
are
\[
\Pi_{_{t}^{t},{}_{s}^{t}}=\ketbra{_{t}^{t}}{_{t}^{t}}+\ketbra{_{s}^{t}}{_{s}^{t}}\hspace{1cm}\Pi_{_{t}^{s}}=\ketbra{_{t}^{s}}{_{t}^{s}}.
\]
Since both pairs of spectral projections sum to the identity on this
subspace, one of these projections is redundant so we drop $\Pi_{_{t}^{t},{}_{s}^{t}}$.

The scattering calculation is then 
\[
\begin{array}{c}
\Pi_{_{t}^{s}}\\
\\
\Pi_{ts,st}
\end{array}\Diagram{fdA &  & fuA\\
 & f\\
fuA &  & fdA
}
\begin{array}{cc}
\Pi_{_{t}^{s}}\\
\\
\Pi_{ts+st}, & \Pi_{ts-st}
\end{array}\hspace{2cm}\begin{array}{c}
\Pi_{_{t}^{s}}\\
\\
\Pi_{tt}
\end{array}\Diagram{fdA &  & fuA\\
 & f\\
fuA &  & fdA
}
\begin{array}{cc}
\Pi_{_{t}^{s}}\\
\\
\Pi_{tt}
\end{array}
\]
where the new projections are $\Pi_{ts\pm st}:=\ket{ts\pm st}\bra{ts\pm st}$
and
\[
\ket{ts\pm st}:=\frac{\ket{ts}\pm\ket{st}}{\sqrt{2}}.
\]
The resulting reflection network is shown in Fig. \ref{fig:The-reflection-network ss and s_b}.
\begin{figure}[H]
\centering{}\includegraphics[viewport=0bp 100bp 737bp 453bp,clip,width=0.5\columnwidth]{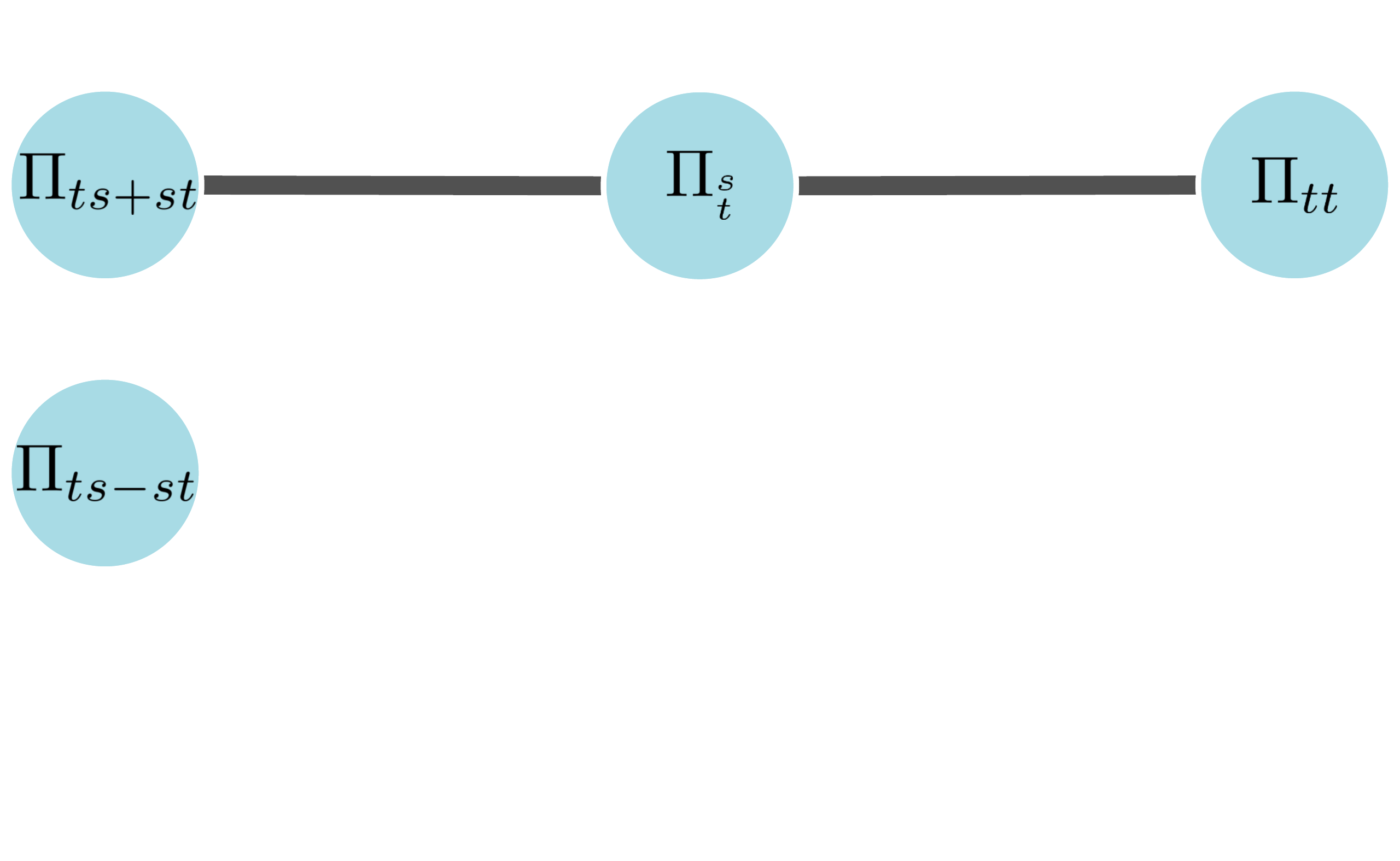}\caption{\label{fig:The-reflection-network ss and s_b}The reflection network
from the spectral projections of $\left[\vec{S}_{2}\cdot\vec{S}_{3}\right]$
and $\left[S_{\beta=\frac{1}{2}}\right]$.}
\end{figure}
Therefore, in the three-dimensional subspace of $j=1$, $m=1$ there
is a two-dimensional subspace on which both $\left[S_{\beta=\frac{1}{2}}\right]$
and $\left[\vec{S}_{2}\cdot\vec{S}_{3}\right]$ reduce. This subspace
is identified by the BPT constructed from the reflection network in
Fig. \ref{fig:The-reflection-network ss and s_b}:

\noindent\begin{minipage}[c]{1\columnwidth}%
\begin{center}
\vspace{0.5\baselineskip}
\begin{tabular}{c|c|c|}
\cline{1-1} 
\multicolumn{1}{|c|}{$ts-st$} & \multicolumn{1}{c}{} & \multicolumn{1}{c}{}\tabularnewline
\hline 
 & $ts+st$ & $tt$\tabularnewline
\cline{2-3} \cline{3-3} 
\multicolumn{1}{c}{} & \multicolumn{1}{c}{$\downarrow$} & \multicolumn{1}{c}{$\downarrow$}\tabularnewline
\cline{2-3} \cline{3-3} 
 & $0_{L}$ & $1_{L}$\tabularnewline
\cline{2-3} \cline{3-3} 
\end{tabular} .\vspace{0.5\baselineskip}
\par\end{center}%
\end{minipage}

The possible qubit basis for $j=1$, $m=1$ are then
\begin{eqnarray*}
\ket{0_{L}}= & \ket{ts+st} & =\frac{\ket{\uparrow\downarrow\uparrow\uparrow}-\ket{\downarrow\uparrow\uparrow\uparrow}+\ket{\uparrow\uparrow\uparrow\downarrow}-\ket{\uparrow\uparrow\downarrow\uparrow}}{2}\\
\ket{1_{L}}= & \ket{tt} & =\frac{\ket{\uparrow\downarrow\uparrow\uparrow}+\ket{\downarrow\uparrow\uparrow\uparrow}-\ket{\uparrow\uparrow\uparrow\downarrow}-\ket{\uparrow\uparrow\downarrow\uparrow}}{2}.
\end{eqnarray*}
It is easy to verify that the Hamiltonian terms restricted to this
subspace act as 
\begin{align*}
\left[S_{\beta=\frac{1}{2}}\right] & =-\frac{1}{4}\ketbra{0_{L}}{0_{L}}+\frac{1}{4}\ketbra{1_{L}}{1_{L}}=-\frac{1}{4}\sigma_{z}\\
\left[\vec{S}_{2}\cdot\vec{S}_{3}\right] & =\frac{1}{4}\ketbra{+_{L}}{+_{L}}-\frac{3}{4}\ketbra{-_{L}}{-_{L}}=\frac{1}{2}\sigma_{x}-\frac{1}{4}I
\end{align*}
where $\ket{\pm_{L}}\propto\ket{0_{L}}\pm\ket{1_{L}}$. Thus, we have
identified a controllable qubit encoded in the subspace of $j=1$,
$m=1$. Similar encodings can be found in the subspaces of $j=1$,
$m=0$ and $j=1$, $m=-1$.

In retrospect, we recognize that the new qubit encodings can be found
by imposing the permutation symmetry that swaps both $1\leftrightarrow4$
and $2\leftrightarrow3$ spins.\footnote{The original $H$ does not have this symmetry but it has redundant
degrees of freedom that can be fixed by imposing additional symmetries.} In order to comply with this symmetry we would have to constrain
the independent parameters $\epsilon_{12}=\epsilon_{34}$, which is
the same as fixing $\beta=\frac{1}{2}$. Then, the irreps of the commutant
of the combined symmetry of $SU\left(2\right)$ and these permutations
would point to this qubit encoding.

In the symmetry-agnostic approach we do not have to come up with the
additional symmetries to impose and construct the combined irreps.
Instead, we follow the systematic procedure where we introduce the
constraint $\beta$, note the special values $\beta=0,\frac{1}{2},1$
where the degeneracies change, and then identify the qubit by reducing
the constrained terms with the Scattering Algorithm.

\chapter{Beyond operator algebras\label{chap:Quantum-coarse-graining}}

In the previous chapters we have formulated the idea of reduction
based on the solid mathematical framework of operator algebras. In
this chapter we will push the formalism of bipartition tables to its
limits and advance the idea of reduction beyond operator algebras.

The fact that certain subspaces of observables, not algebras, can
identify new notions of bipartition and entanglement, has been brought
forward by Barnum, Viola, \emph{et al.} in \citep{Barnum04,viola2004entanglement,viola2010entanglement}
(see also \citep{Alicki09}). The need for generalized bipartition
beyond tensor products and operator algebras has also been expressed
in \citep{Ghosh18Quantum}, motivated by the definition of local entanglement
in the theories of quantum gravity.

In the following, we will derive such generalized notion of bipartition
by relaxing the rigid structure of bipartition tables and operator
algebras. Despite the relaxation, the essence of bipartition tables
will remain the same and the resulting structure naturally generalizes
tensor products to become \emph{partial }bipartitions, and virtual
subsystems to become\emph{ partial }subsystems. The associated state
reductions produced by tracing out a partial subsystem are analogous
to coarse-graining of classical probability distributions. Such reductions
will therefore be called \emph{quantum coarse-graining. }The ideas
of quantum coarse-graining were originally published in \citep{Kabernik2018Quantum}. 

Similar adaptations of the classical notion of coarse-graining in
quantum theory have been explored in \citep{faist2016quantum,duarte2017emerging,correia2020macro}
(for other approaches to the quantum notion of coarse-graining see
\citep{singh2018quantum,DiMatteo2017,duarte2020compatibility}). Our
main contribution is in deriving the general operational meaning of
such notion of coarse-graining in quantum theory, and rigorously demonstrating
its relation to the analogous classical notion.

Before formally introducing the new ideas, we will begin in Section
\ref{sec:Motivating-example} with a motivating example that demonstrates
why the operator-algebraic structure is too rigid and how it can be
relaxed. In Section \ref{sec:Partial-bipartitions} we will proceed
with an illustrative example of the classical notion of coarse-graining.
We will then generalize the formalism of bipartition tables to partial
bipartitions and make the connection with classical coarse-graining.
Finally, we will derive the operational meaning of quantum coarse-graining
in Theorem \ref{thm:oper meaning of part BP} which generalizes the
operational meaning of the partial trace map beyond tensor products.

\section{Motivating example\label{sec:Motivating-example}}

In this section we will investigate a simple example where the operator-algebraic
approach falls short of providing a satisfying solution. We will see
how the irreps structure of operator algebras can be too rigid and
that a relaxation of this structure has benefits. This motivates a
more serious investigation of such relaxed structures that we will
end up calling partial bipartitions in the next section.

Consider a communication scenario where Alice has a spin-$1$ system
in the state 
\[
\ket{\psi}=\alpha_{1}\ket 1+\alpha_{0}\ket 0+\alpha_{-1}\ket{-1}
\]
(the basis are the eigenvectors of the $z$-component of the spin
operator) and she wants to send it to Bob. Unfortunately, Alice cannot
send the spin-$1$ system directly but she can prepare and send a
single qubit in any (pure or mixed) state $\rho_{qubit}$. It is,
of course, not possible to genuinely encode a three-dimensional system
into a two-dimensional one without loosing some information. So, Alice
has to prioritize the observables of the spin-$1$ system that she
wants to preserve.

Assume that Alice's top priority is the observable
\[
Z=\ketbra 00-\ketbra 11-\ketbra{-1}{-1}
\]
that distinguishes the state $\ket 0$ from the rest. Alice can then
encode
\[
\ket{\psi}\longmapsto\rho_{qubit}=\begin{pmatrix}\left|\alpha_{0}\right|^{2} & 0\\
0 & \left|\alpha_{1}\right|^{2}+\left|\alpha_{-1}\right|^{2}
\end{pmatrix}
\]
and Bob can recover all the statistical information about $Z$ from
the Pauli observable $\sigma_{z}$. This encoding treats the qubit
as a classical bit so the interesting question here is how can Alice
take advantage of the full qubit to communicate additional information
about the spin-$1$ system.

Let us first try the operator-algebraic approach to address this question.

We may consider the encoding into a qubit as a quantum state reduction.
If only the observable $Z$ is considered, the state reduction is
given by the irreps structure of the algebra $\left\langle Z\right\rangle $:

\noindent\begin{minipage}[c]{1\columnwidth}%
\begin{center}
\vspace{0.5\baselineskip}
\begin{tabular}{c|c|}
\cline{1-1} 
\multicolumn{1}{|c|}{$0$} & \multicolumn{1}{c}{}\tabularnewline
\hline 
 & $1$\tabularnewline
\cline{2-2} 
 & $-1$\tabularnewline
\cline{2-2} 
\multicolumn{1}{c}{$\downarrow$} & \multicolumn{1}{c}{$\downarrow$}\tabularnewline
\hline 
\multicolumn{1}{|c|}{$0_{L}$} & $1_{L}$\tabularnewline
\hline 
\end{tabular} .\vspace{0.5\baselineskip}
\par\end{center}%
\end{minipage} The minimal isometries here are just the spectral projections of
$Z$ 
\[
\Pi_{0}:=\ketbra 00\hspace{1cm}\Pi_{\pm1}:=\ketbra 11+\ketbra{-1}{-1}.
\]
Using the general construction of state reduction maps in Eq. \ref{eq:general state reduc map},
we will get the encoding we know 
\[
\ket{\psi}\longmapsto\rho_{qubit}=\begin{pmatrix}\bra{\psi}\Pi_{0}\ket{\psi}\\
 & \bra{\psi}\Pi_{\pm1}\ket{\psi}
\end{pmatrix}=\begin{pmatrix}\left|\alpha_{0}\right|^{2} & 0\\
0 & \left|\alpha_{1}\right|^{2}+\left|\alpha_{-1}\right|^{2}
\end{pmatrix}.
\]

If in addition to $Z$ Alice wants encode information about any other
observable $X$, she has to consider the irreps of $\left\langle Z,X\right\rangle $.
There is, however, only five ``shapes'' that the BPT can have in
a three-dimensional Hilbert space. The possibilities are

\noindent\begin{minipage}[c]{1\columnwidth}%
\begin{center}
\vspace{0.5\baselineskip}
\begin{tabular}{|>{\centering}p{0.25cm}|}
\hline 
\tabularnewline
\hline 
\tabularnewline
\hline 
\tabularnewline
\hline 
\end{tabular}\hspace*{20bp}%
\begin{tabular}{>{\centering}p{0.25cm}|>{\centering}p{0.25cm}|}
\cline{1-1} 
\multicolumn{1}{|>{\centering}p{0.25cm}|}{} & \multicolumn{1}{>{\centering}p{0.25cm}}{}\tabularnewline
\hline 
 & \tabularnewline
\cline{2-2} 
 & \tabularnewline
\cline{2-2} 
\end{tabular}\hspace*{20bp}%
\begin{tabular}{>{\centering}p{0.25cm}|>{\centering}p{0.25cm}|>{\centering}p{0.25cm}}
\cline{1-1} 
\multicolumn{1}{|>{\centering}p{0.25cm}|}{} & \multicolumn{1}{>{\centering}p{0.25cm}}{} & \tabularnewline
\cline{1-2} \cline{2-2} 
 &  & \tabularnewline
\cline{2-3} \cline{3-3} 
\multicolumn{1}{>{\centering}p{0.25cm}}{} &  & \multicolumn{1}{>{\centering}p{0.25cm}|}{}\tabularnewline
\cline{3-3} 
\end{tabular}\hspace*{20bp}%
\begin{tabular}{|>{\centering}p{0.25cm}|>{\centering}p{0.25cm}|>{\centering}p{0.25cm}|}
\cline{1-1} 
 & \multicolumn{1}{>{\centering}p{0.25cm}}{} & \multicolumn{1}{>{\centering}p{0.25cm}}{}\tabularnewline
\hline 
\multicolumn{1}{>{\centering}p{0.25cm}|}{} &  & \tabularnewline
\cline{2-3} \cline{3-3} 
\end{tabular}\hspace*{20bp}%
\begin{tabular}{|>{\centering}p{0.25cm}|>{\centering}p{0.25cm}|>{\centering}p{0.25cm}|}
\hline 
 &  & \tabularnewline
\hline 
\end{tabular} .\vspace{0.5\baselineskip}
\par\end{center}%
\end{minipage}

The first shape (from the left) corresponds to a reduction onto a
one-dimensional system that preserves no information. The second shape
is what we got for $\left\langle Z\right\rangle $. Therefore, unless
$X$ has the same spectral projections as $Z$ (which makes it the
same observable up to eigenvalues), the irreps structure of $\left\langle Z,X\right\rangle $
will have one of the last three shapes. The last three shapes, however,
correspond to reductions onto a three-dimensional system so it will
not work as a qubit encoding.

Therefore, the operator-algebraic perspective suggests that it is
not possible to encode in a qubit another (distinct) observable together
with $Z$.

We should now point out that in the operator-algebraic approach the
state reduction map is not just concerned with the generators of the
algebra but with the whole algebra. For $\left\langle Z\right\rangle $
the associated state reduction map preserves the expectation values
not just of $Z$, but also of $Z^{n}$ for all $n$. In particular,
the spectral projections $\Pi_{0}$ and $\Pi_{\pm1}$ are also in
$\left\langle Z\right\rangle $ so their expectation values, which
are the probabilities of the two outcomes, are also preserved. This
suggests that if we want to include another observable $X$ we may
be able to compromise on preserving only its expectation values without
preserving the probabilities of the individual outcomes.

This brings us to the main point: The irreps structure of operator
algebras is too rigid for some tasks and a more flexible structure
is desired. Since BPTs are the visual representations of irreps structures,
we can relax the rigidity of irreps structures by relaxing the rigidity
of bipartition tables.

By their original Definition \ref{def:A-bipartition-table}, BPTs
can only have a block diagonal form where each block is rectangular,
that is, each row (or column) in the block has the same number of
cells. In Lemma \ref{lem: for all S_kl there is a BPT} we have showed
that all operator algebras correspond to BPTs defined this way. We
will now relax the requirement for BPT blocks to be rectangular and
such BPTs will not correspond to any operator algebra.

Consider, for example, the arrangement (and the implied state reduction)
given by the BPT

\noindent\begin{minipage}[c]{1\columnwidth}%
\begin{center}
\vspace{0.5\baselineskip}
\begin{tabular}{|c|c|}
\hline 
$0$ & $1$\tabularnewline
\hline 
\multicolumn{1}{c|}{} & $-1$\tabularnewline
\cline{2-2} 
\multicolumn{1}{c}{$\downarrow$} & \multicolumn{1}{c}{$\downarrow$}\tabularnewline
\hline 
$0_{L}$ & $1_{L}$\tabularnewline
\hline 
\end{tabular} .\vspace{0.5\baselineskip}
\par\end{center}%
\end{minipage} There is only one block here and the number of cells in the first
and second rows is not the same. Such BPT cannot arise from any operator
algebra, however, we can still use this arrangement to construct partial
isometries (we will not call them ``minimal'' anymore because we
no longer have an algebra). Following the original construction of
isometries from BPTs given in Eq. \eqref{eq: def of S_kl from BPT},
we get the projections $\Pi_{0}$ and $\Pi_{\pm1}$ as before, but
now we also get the proper isometries
\[
S_{0,1}:=\ketbra 01\hspace{1cm}S_{1,0}:=\ketbra 10.
\]

The state reduction map is defined as before only now we have the
additional proper isometries that preserve some of the coherences
\[
\ket{\psi}\longmapsto\rho_{qubit}=\begin{pmatrix}\bra{\psi}\Pi_{0}\ket{\psi} & \bra{\psi}S_{1,0}\ket{\psi}\\
\bra{\psi}S_{0,1}\ket{\psi} & \bra{\psi}\Pi_{\pm1}\ket{\psi}
\end{pmatrix}=\begin{pmatrix}\left|\alpha_{0}\right|^{2} & \alpha_{1}^{*}\alpha_{0}\\
\alpha_{0}^{*}\alpha_{1} & \left|\alpha_{1}\right|^{2}+\left|\alpha_{-1}\right|^{2}
\end{pmatrix}.
\]
It is easy to verify that $\rho_{qubit}$ is a positive operator of
trace $1$, so it is a proper quantum state.

As before, Bob can recover all the statistical information about $Z$
(including the probabilities of the individual outcomes) from the
Pauli observable $\sigma_{z}$. Since the coherence terms between
$\ket 0$ and $\ket 1$ are also preserved, Bob can now recover the
expectation values of observables such as
\begin{align*}
X & :=S_{1,0}+S_{0,1}=\ketbra 10+\ketbra 01\\
Y & :=iS_{1,0}-iS_{0,1}=i\ketbra 10-i\ketbra 01
\end{align*}
from the expectation values of the Pauli operators $\sigma_{x}$ and
$\sigma_{y}$. It is easy to verify that 
\begin{align*}
\tr\left[X\ketbra{\psi}{\psi}\right] & =\alpha_{1}^{*}\alpha_{0}+\alpha_{0}^{*}\alpha_{1}=\tr\left[\sigma_{x}\rho_{qubit}\right]\\
\tr\left[Y\ketbra{\psi}{\psi}\right] & =i\alpha_{1}^{*}\alpha_{0}-i\alpha_{0}^{*}\alpha_{1}=\tr\left[\sigma_{y}\rho_{qubit}\right].
\end{align*}

The surprising feature of this encoding is that even though the expectation
values of $X$ and $Y$ are preserved, the probabilities of the individual
outcomes are not. For example, since the $+1$ eigenvectors of $X$
and $\sigma_{x}$ are
\[
\ket{_{-}^{+0,1}}:=\frac{\ket 0+\ket 1}{\sqrt{2}}\hspace{1cm}\ket{+_{L}}:=\frac{\ket{0_{L}}+\ket{1_{L}}}{\sqrt{2}}
\]
we calculate the corresponding probabilities to be 
\begin{align*}
\tr\left[\ketbra{_{-}^{+0,1}}{_{-}^{+0,1}}\ketbra{\psi}{\psi}\right] & =\frac{\alpha_{0}^{*}\alpha_{1}+\alpha_{1}^{*}\alpha_{0}+\left|\alpha_{0}\right|^{2}+\left|\alpha_{1}\right|^{2}}{2}\\
\tr\left[\ketbra{+_{L}}{+_{L}}\rho_{qubit}\right] & =\frac{\alpha_{0}^{*}\alpha_{1}+\alpha_{1}^{*}\alpha_{0}+\left|\alpha_{0}\right|^{2}+\left|\alpha_{1}\right|^{2}+\left|\alpha_{-1}\right|^{2}}{2}.
\end{align*}
We see that in the qubit the probability of $+1$ has increased by
$\left|\alpha_{-1}\right|^{2}/2$. \footnote{It should not be too surprising since in the original system the probabilities
for the outcomes of $+1$ and $-1$ did not have to sum to $1$ as
there was another outcome, $0$. In the qubit the probabilities for
the two outcomes must sum to $1$ so we cannot expect the original
probabilities to stay unchanged.}

So far, the additional preserved observables $X$, $Y$ were the result
of arbitrary rearrangement of the basis elements in the BPT. If we
want a specific observable to be preserved in addition to $Z$, we
will have to be more deliberate about how we choose the basis elements
in the BPT.

We can change the basis in the subspace $\left\{ \ket 1,\ket{-1}\right\} $
of the second column

\noindent\begin{minipage}[c]{1\columnwidth}%
\begin{center}
\vspace{0.5\baselineskip}
\begin{tabular}{|c|c|}
\hline 
$0$ & $e_{1}$\tabularnewline
\hline 
\multicolumn{1}{c|}{} & $e_{2}$\tabularnewline
\cline{2-2} 
\multicolumn{1}{c}{$\downarrow$} & \multicolumn{1}{c}{$\downarrow$}\tabularnewline
\hline 
$0_{L}$ & $1_{L}$\tabularnewline
\hline 
\end{tabular} .\vspace{0.5\baselineskip}
\par\end{center}%
\end{minipage} This does not change the fact that the spectral projections 
\begin{align*}
\Pi_{0} & :=\ketbra 00\\
\Pi_{\pm1} & :=\ketbra{e_{1}}{e_{1}}+\ketbra{e_{2}}{e_{2}}=\ketbra 11+\ketbra{-1}{-1}
\end{align*}
of $Z$ are still the isometries constructed from this BPT. The proper
isometries, however, are now different

\[
S_{0,1}:=\ketbra 0{e_{1}}\hspace{1cm}S_{1,0}:=\ketbra{e_{1}}0.
\]
Thus, by changing the basis $\ket{e_{1}}$, $\ket{e_{2}}$ we can
choose the preserved observables.

Alice may choose, for example, to preserve the expectation value of
the $x$ component of spin. Using the subscripts $x,y,z$ to distinguish
the eigenvectors of different components of the spin operator (so
$\ket 1,$$\ket 0$,$\ket{-1}$ are now $\ket{1_{z}},$$\ket{0_{z}}$,$\ket{-1_{z}}$)
we note that 
\[
\ket{\pm1_{x}}=\frac{\ket{0_{y}}\pm\ket{0_{z}}}{\sqrt{2}}.
\]
The $x$ component of spin can then be expressed as 
\[
S_{x}=\ketbra{1_{x}}{1_{x}}-\ketbra{-1_{x}}{-1_{x}}=\ketbra{0_{y}}{0_{z}}+\ketbra{0_{z}}{0_{y}}.
\]

Before, when $\ket{e_{1}}=\ket{1_{z}}$ we preserved the expectations
of 
\[
X=\ketbra{1_{z}}{0_{z}}+\ketbra{0_{z}}{1_{z}}.
\]
This suggests that in order to preserve the expectations of $S_{x}$
we need to choose

\[
\ket{e_{1}}=\ket{0_{y}}=\frac{\ket{1_{z}}+\ket{-1_{z}}}{\sqrt{2}}.
\]
The resulting qubit encoding is 
\[
\ket{\psi}\longmapsto\rho_{qubit}=\begin{pmatrix}\bra{\psi}\Pi_{0}\ket{\psi} & \braket{\psi}{0_{y}}\braket{0_{z}}{\psi}\\
\braket{\psi}{0_{z}}\braket{0_{y}}{\psi} & \bra{\psi}\Pi_{\pm1}\ket{\psi}
\end{pmatrix}=\begin{pmatrix}\left|\alpha_{0}\right|^{2} & \frac{\alpha_{1}^{*}+\alpha_{-1}^{*}}{\sqrt{2}}\alpha_{0}\\
\alpha_{0}^{*}\frac{\alpha_{1}+\alpha_{-1}}{\sqrt{2}} & \left|\alpha_{1}\right|^{2}+\left|\alpha_{-1}\right|^{2}
\end{pmatrix}
\]
and the expectation value of $S_{x}$ is recovered from the expectation
value of $\sigma_{x}$ 
\[
\tr\left[S_{x}\ketbra{\psi}{\psi}\right]=\frac{\alpha_{1}^{*}+\alpha_{-1}^{*}}{\sqrt{2}}\alpha_{0}+\alpha_{0}^{*}\frac{\alpha_{1}+\alpha_{-1}}{\sqrt{2}}=\tr\left[\sigma_{x}\rho_{qubit}\right].
\]

In the following we will see that only the expectation values of observables
spanned by the isometries constructed from the BPT are preserved.
Then, the explanation to why some probabilities of outcomes may not
be preserved is as follows: Since the isometries constructed from
the relaxed BPTs do not span an algebra, it is possible for an observable
to be in the span but not for its spectral projections.

\section{Partial bipartitions and quantum coarse-graining\label{sec:Partial-bipartitions}}

The example in the previous section implies that it may be beneficial
to consider the structure of non-rectangular bipartition tables more
seriously. In this section we will identify this structure as a \emph{partial
}bipartition and establish its operational meaning. We will see that
such structure naturally generalizes virtual subsystems and lays the
foundation to more general state reduction maps called quantum coarse-graining\emph{.}

\subsection{Classical analogy\label{subsec:Classical-analogy}}

An illuminating perspective on partial bipartitions and coarse-graining
can be gained by considering its analogues in probability theory.
In the classical formalism it is quite natural to reduce one probabilistic
state into another state that provides a coarser probabilistic account
of the same system. Thus, it is instructive to first establish the
notions of partial bipartitions and coarse-graining in the classical
formalism where its reasoning is more natural.

As an illustrative example, consider the weather in Vancouver that
can be sunny or rainy, and warm or cold (say above or below $15$
C$\lyxmathsym{\textdegree}$). Precipitation and temperature are correlated
and historical data may tell us that on October 1st it is $50\%$
likely to be rainy and cold (rc), $30\%$ sunny and cold (sc), $15\%$
rainy and warm (rw), and $5\%$ sunny and warm (sw).\footnote{These probabilities come from a subjective approximation based on
the experiences of the author and not an actual meteorological data.} We can then produce a coarser account of the weather in Vancouver
by only distinguishing between sunny and rainy, or warm and cold.
These coarser accounts are associated with a (non-partial) bipartition
of the four-outcome state space into two state spaces of two outcomes
each. The reduced probabilities of it being sunny or rainy, and warm
or cold are given by summing over the columns and rows of the following
bipartition table

\noindent\begin{minipage}[c]{1\columnwidth}%
\begin{center}
\vspace{0.5\baselineskip}
\begin{tabular}{|c|c|c|l|}
\cline{1-2} \cline{2-2} \cline{4-4} 
$5\%$ sw & $15\%$ rw & $\rightarrow$ & $20\%$ warm\tabularnewline
\cline{1-2} \cline{2-2} \cline{4-4} 
$30\%$ sc & $50\%$ rc & $\rightarrow$ & $80\%$ cold\tabularnewline
\cline{1-2} \cline{2-2} \cline{4-4} 
\multicolumn{1}{c}{$\downarrow$} & \multicolumn{1}{c}{$\downarrow$} & \multicolumn{1}{c}{} & \multicolumn{1}{l}{}\tabularnewline
\cline{1-2} \cline{2-2} 
$35\%$ sunny & $65\%$ rainy & \multicolumn{1}{c}{} & \multicolumn{1}{l}{}\tabularnewline
\cline{1-2} \cline{2-2} 
\end{tabular} .\vspace{0.5\baselineskip}
\par\end{center}%
\end{minipage}

The above reduction of probabilities is called \emph{marginalization
}and it is the classical analogue of the partial trace map. Just like
the partial trace map, marginalization takes the joint probability
distribution of two random variables and produces the probability
distribution of one random variable. Since marginalized probability
distributions distinguish between fewer outcomes, we can say that
marginalization is a kind of coarse-graining of probability distributions.

In the context of classical probability theory it is also reasonable
to consider coarse-grainings that go beyond the usual notion of marginalization.
For example, whenever it is sunny and warm I wear a shirt, when it
is sunny and cold or rainy and warm I wear a jacket, and when it is
rainy and cold I wear a coat. Also, I wear some kind of hat when it
is cold or sunny. The probability of me wearing a shirt, a jacket,
or a coat, with or without a hat, on October 1st in Vancouver is given
by summations over the rows and column of the following partial BPT

\noindent\begin{minipage}[c]{1\columnwidth}%
\begin{center}
\vspace{0.5\baselineskip}
\begin{tabular}{|c|c|c|c|l|}
\cline{2-2} \cline{5-5} 
\multicolumn{1}{c|}{} & $15\%$ rw & \multicolumn{1}{c}{} & $\rightarrow$ & $15\%$ no hat\tabularnewline
\cline{1-3} \cline{2-3} \cline{3-3} \cline{5-5} 
$5\%$ sw & $30\%$ sc & $50\%$ rc & $\rightarrow$ & $85\%$ hat\tabularnewline
\cline{1-3} \cline{2-3} \cline{3-3} \cline{5-5} 
\multicolumn{1}{c}{$\downarrow$} & \multicolumn{1}{c}{$\downarrow$} & \multicolumn{1}{c}{$\downarrow$} & \multicolumn{1}{c}{} & \multicolumn{1}{l}{}\tabularnewline
\cline{1-3} \cline{2-3} \cline{3-3} 
$5\%$ shirt & $45\%$ jacket & $50\%$ coat & \multicolumn{1}{c}{} & \multicolumn{1}{l}{}\tabularnewline
\cline{1-3} \cline{2-3} \cline{3-3} 
\end{tabular} .\vspace{0.5\baselineskip}
\par\end{center}%
\end{minipage}

Since my clothing and hat combinations are perfectly correlated with
the weather, instead of ``sunny and cold'' we can label the same
outcome as ``jacket and hat'', and similarly for other outcomes.
Even though not all combinations of clothing and hat are possible
(``shirt and no hat'' or ``coat and no hat'' never happen), it
is still perfectly reasonable to coarse-grain the probability distribution
this way in order to get the reduced probabilities for my clothing
or hat choices.

Our goal is to import the same kind of reasoning into quantum theory
where instead of probability distributions we will coarse-grain quantum
states. In order to do that we will have to further develop our formalism
to incorporate partial BPTs.

\subsection{The formalism of partial bipartitions}

Let us forget for now about irreps and operator algebras and consider
what, in essence, rectangular bipartition tables tell us. By arranging
the basis elements into a two-dimensional grid we identify two distinct
degrees of freedom; one degree of freedom varies horizontally and
the other vertically. These two degrees of freedom can then be associated
with two subsystems (virtual or otherwise) that constitute a bipartition
of the Hilbert space.

In order to make this statement more precise, consider the generic
rectangular BPT with a single block\footnote{In this section we will only consider BPTs with a single block to
avoid unnecessary clutter but it naturally generalizes to multi-block
BPTs where each block is considered separately.} and the implied reductions of rows and columns

\noindent\begin{minipage}[c]{1\columnwidth}%
\begin{center}
\vspace{0.5\baselineskip}
\begin{tabular}{|c|c|c|c|c|c|}
\cline{1-4} \cline{2-4} \cline{3-4} \cline{4-4} \cline{6-6} 
$e_{1,1}$ & $e_{1,2}$ & $\cdots$ & $e_{1,d_{B}}$ & $\rightarrow$ & $a_{1}$\tabularnewline
\cline{1-4} \cline{2-4} \cline{3-4} \cline{4-4} \cline{6-6} 
$e_{2,1}$ & $e_{2,2}$ & $\cdots$ & $e_{2,d_{B}}$ & $\rightarrow$ & $a_{2}$\tabularnewline
\cline{1-4} \cline{2-4} \cline{3-4} \cline{4-4} \cline{6-6} 
$\vdots$ & $\vdots$ & $\ddots$ & $\vdots$ & $\rightarrow$ & $\vdots$\tabularnewline
\cline{1-4} \cline{2-4} \cline{3-4} \cline{4-4} \cline{6-6} 
$e_{d_{A},1}$ & $e_{d_{A},2}$ & $\cdots$ & $e_{d_{A},d_{B}}$ & $\rightarrow$ & $a_{d_{A}}$\tabularnewline
\cline{1-4} \cline{2-4} \cline{3-4} \cline{4-4} \cline{6-6} 
\multicolumn{1}{c}{$\downarrow$} & \multicolumn{1}{c}{$\downarrow$} & \multicolumn{1}{c}{$\downarrow$} & \multicolumn{1}{c}{$\downarrow$} & \multicolumn{1}{c}{} & \multicolumn{1}{c}{}\tabularnewline
\cline{1-4} \cline{2-4} \cline{3-4} \cline{4-4} 
$b_{1}$ & $b_{2}$ & $\cdots$ & $b_{d_{B}}$ & \multicolumn{1}{c}{} & \multicolumn{1}{c}{}\tabularnewline
\cline{1-4} \cline{2-4} \cline{3-4} \cline{4-4} 
\end{tabular} \vspace{0.5\baselineskip}
\par\end{center}%
\end{minipage} ($d_{A}$ and $d_{B}$ are the number of rows and columns in the
BPT). The choice and arrangement of basis in the BPT tells us how
to map the original Hilbert space $\mathcal{H}$ spanned by $\left\{ \ket{e_{ik}}\right\} $,
onto a bipartite Hilbert space $\mathcal{H}_{A}\otimes\mathcal{H}_{B}$
spanned by $\left\{ \ket{a_{i}}\otimes\ket{b_{k}}\right\} $. In other
words, the BPT defines a Hilbert space isometry
\begin{eqnarray}
V: & \mathcal{H} & \longrightarrow\mathcal{H}_{A}\otimes\mathcal{H}_{B}\label{eq: BPT def. isometry}\\
 & \ket{e_{ik}} & \longmapsto\ket{a_{i}}\otimes\ket{b_{k}}\nonumber 
\end{eqnarray}
from the original Hilbert space to the bipartite Hilbert space. Thus,
as the name suggests, the essence of a bipartition table is to identify
a tensor product bipartition of the Hilbert space
\[
\mathcal{H}\cong\mathcal{H}_{A}\otimes\mathcal{H}_{B}.
\]

With this perspective we realize that the same construction can also
be applied to non-rectangular BPTs such as\footnote{Earlier we assumed that $\dim\mathcal{H}=d_{A}d_{B}$, now we assume
that $\dim\mathcal{H}<d_{A}d_{B}$.}

\noindent\begin{minipage}[c]{1\columnwidth}%
\begin{center}
\vspace{0.5\baselineskip}
\begin{tabular}{|c|c|c|cc|c|}
\cline{1-4} \cline{2-4} \cline{3-4} \cline{4-4} \cline{6-6} 
$e_{1,1}$ & $e_{1,2}$ & $\cdots$ & \multicolumn{1}{c|}{$e_{1,d_{B}}$} & $\rightarrow$ & $a_{1}$\tabularnewline
\cline{1-4} \cline{2-4} \cline{3-4} \cline{4-4} \cline{6-6} 
$e_{2,1}$ & $e_{2,2}$ & $\cdots$ &  & $\rightarrow$ & $a_{2}$\tabularnewline
\cline{1-3} \cline{2-3} \cline{3-3} \cline{6-6} 
$\vdots$ & $\ddots$ & \multicolumn{1}{c}{} &  & $\rightarrow$ & $\vdots$\tabularnewline
\cline{1-2} \cline{2-2} \cline{6-6} 
$e_{d_{A},1}$ & $\cdots$ & \multicolumn{1}{c}{} &  & $\rightarrow$ & $a_{d_{A}}$\tabularnewline
\cline{1-2} \cline{2-2} \cline{6-6} 
\multicolumn{1}{c}{$\downarrow$} & \multicolumn{1}{c}{$\downarrow$} & \multicolumn{1}{c}{$\downarrow$} & $\downarrow$ & \multicolumn{1}{c}{} & \multicolumn{1}{c}{}\tabularnewline
\cline{1-4} \cline{2-4} \cline{3-4} \cline{4-4} 
$b_{1}$ & $b_{2}$ & $\cdots$ & \multicolumn{1}{c|}{$b_{d_{B}}$} & \multicolumn{1}{c}{} & \multicolumn{1}{c}{}\tabularnewline
\cline{1-4} \cline{2-4} \cline{3-4} \cline{4-4} 
\end{tabular} .
\par\end{center}
\begin{center}
\vspace{0.5\baselineskip}
\par\end{center}%
\end{minipage} Even though the dimensions of rows and columns can now vary, it still
defines a Hilbert space isometry as in Eq. \eqref{eq: BPT def. isometry},
only now the indices $i$,$k$ are constrained by the non-rectangular
shape of the BPT. As a result, $\mathcal{H}$ is not mapped \emph{onto}
but \emph{into} $\mathcal{H}_{A}\otimes\mathcal{H}_{B}$, that is,
$\mathcal{H}$ is mapped onto a subspace of $\mathcal{H}_{A}\otimes\mathcal{H}_{B}$.

We may still think of such mapping as a bipartition of $\mathcal{H}$
but it is no longer a tensor product bipartition. We will call such
generalized bipartitions \emph{partial} since not all product basis
$\ket{a_{i}}\otimes\ket{b_{k}}$ can be found in the original Hilbert
space. We introduce the notation
\[
\mathcal{H}\cong\mathcal{H}_{A}\oslash\mathcal{H}_{B}
\]
for partial bipartitions that emphasizes the fact that it is a generalization
of the tensor product. We will also say that $\mathcal{H}_{A}$ and
$\mathcal{H}_{B}$ are (the Hilbert spaces of) \emph{partial subsystems},
generalizing the notion of a virtual subsystem.

Now that we have defined partial bipartitions we can define state
reduction maps given by tracing out one of the partial subsystems.
By elevating the isometry \eqref{eq: BPT def. isometry} to act on
operators
\[
\mathcal{V}\left(\rho\right):=V\rho V^{\dagger},
\]
we define the state reduction map as the composition
\[
\tr_{\left(A\right)}:=\tr_{A}\circ\mathcal{V}.
\]
The map $\tr_{\left(A\right)}$ is CPTP (so it reduces proper quantum
states to proper quantum states) because it can be expressed in the
operator sum representation \citep{nielsen2002quantum}
\[
\tr_{\left(A\right)}\left(\rho\right)=\sum_{i}\left(K_{i}V\right)\rho\left(K_{i}V\right)^{\dagger}
\]
 where $K_{i}$ are the Kraus operators of $\tr_{A}$.

With this formalism we can reproduce the classical reasoning of Section
\ref{subsec:Classical-analogy} in a quantum setting. Let us map the
classical states of $\left\{ \textrm{sw, sc, rw, rc}\right\} $ (sunny
and warm, ... , rainy and cold) to the product basis of two spins
$\left\{ \ket{\uparrow\uparrow},\ket{\uparrow\downarrow},\ket{\downarrow\uparrow},\ket{\downarrow\downarrow}\right\} $
spanning $\mathcal{H}$. The coarse-graining that we have considered
before is now given by the partial BPT \footnote{Note that in the quantum case we do not specify any probabilities
in the BPT because they are complex amplitudes and their reduction
is more than a summation over the rows and columns.}

\noindent\begin{minipage}[c]{1\columnwidth}%
\begin{center}
\vspace{0.5\baselineskip}
\begin{tabular}{|c|c|c|c|l|}
\cline{2-2} \cline{5-5} 
\multicolumn{1}{c|}{} & $\downarrow\uparrow$ & \multicolumn{1}{c}{} & $\rightarrow$ & $\downarrow_{L}\land\uparrow_{R}$\tabularnewline
\cline{1-3} \cline{2-3} \cline{3-3} \cline{5-5} 
$\uparrow\uparrow$ & $\uparrow\downarrow$ & $\downarrow\downarrow$ & $\rightarrow$ & $\uparrow_{L}\lor\downarrow_{R}$\tabularnewline
\cline{1-3} \cline{2-3} \cline{3-3} \cline{5-5} 
\multicolumn{1}{c}{$\downarrow$} & \multicolumn{1}{c}{$\downarrow$} & \multicolumn{1}{c}{$\downarrow$} & \multicolumn{1}{c}{} & \multicolumn{1}{l}{}\tabularnewline
\cline{1-3} \cline{2-3} \cline{3-3} 
$1_{z}$ & $0_{z}$ & $-1_{z}$ & \multicolumn{1}{c}{} & \multicolumn{1}{l}{}\tabularnewline
\cline{1-3} \cline{2-3} \cline{3-3} 
\end{tabular} .\vspace{0.5\baselineskip}
\par\end{center}%
\end{minipage} Observe that the three columns distinguish between the states of
total spin $1,0,-1$ along $\hat{z}$. Similarly, the two rows distinguish
between states that can be described as ``left $\downarrow$ \emph{and}
right $\uparrow$'' and ``left $\uparrow$ \emph{or} right $\downarrow$''.
We then label the basis for the partial subsystems according to what
they distinguish 
\begin{align}
\mathcal{H}_{B} & :=\spn\left\{ \ket{1_{z}},\ket{0_{z}},\ket{-1_{z}}\right\} \label{eq:H_B from part. BPT}\\
\mathcal{H}_{A} & :=\spn\left\{ \ket{\downarrow_{L}\land\uparrow_{R}},\ket{\uparrow_{L}\lor\downarrow_{R}}\right\} .\label{eq:H_A from part. BPT}
\end{align}

The partial bipartition $\mathcal{H}\cong\mathcal{H}_{A}\oslash\mathcal{H}_{B}$
is defined by an isometry $V:\mathcal{H}\longrightarrow\mathcal{H}_{A}\otimes\mathcal{H}_{B}$
where
\begin{align*}
V\ket{\uparrow\uparrow} & =\ket{\uparrow_{L}\lor\downarrow_{R}}\otimes\ket{1_{z}}\\
V\ket{\uparrow\downarrow} & =\ket{\uparrow_{L}\lor\downarrow_{R}}\otimes\ket{0_{z}}\\
V\ket{\downarrow\uparrow} & =\ket{\downarrow_{L}\land\uparrow_{R}}\otimes\ket{0_{z}}\\
V\ket{\downarrow\downarrow} & =\ket{\uparrow_{L}\lor\downarrow_{R}}\otimes\ket{-1_{z}}.
\end{align*}
The remaining two states $\ket{\downarrow_{L}\land\uparrow_{R}}\otimes\ket{\pm1_{z}}$
are impossible spin states and they are not in the image of $V$ (so
they are annihilated by $V^{\dagger}$). The reduced states are given
by tracing out one of the partial subsystems $A$ or $B$. That is,
for any $\rho\in\mathcal{L}\left(\mathcal{H}\right)$ the reduced
states are 
\begin{align*}
\rho_{B} & =\tr_{\left(A\right)}\left[\rho\right]=\tr_{A}\left[V\rho V^{\dagger}\right]\in\mathcal{L}\left(\mathcal{H}_{B}\right)\\
\rho_{A} & =\tr_{\left(B\right)}\left[\rho\right]=\tr_{B}\left[V\rho V^{\dagger}\right]\in\mathcal{L}\left(\mathcal{H}_{A}\right).
\end{align*}

If $\rho$ is a classical probabilistic state such as
\[
\rho:=p_{\uparrow\uparrow}\ketbra{\uparrow\uparrow}{\uparrow\uparrow}+p_{\uparrow\downarrow}\ketbra{\uparrow\downarrow}{\uparrow\downarrow}+p_{\downarrow\uparrow}\ketbra{\downarrow\uparrow}{\downarrow\uparrow}+p_{\downarrow\downarrow}\ketbra{\downarrow\downarrow}{\downarrow\downarrow}
\]
then its reduced states are
\begin{align*}
\rho_{B} & =p_{\uparrow\uparrow}\ketbra{1_{z}}{1_{z}}+\left(p_{\uparrow\downarrow}+p_{\downarrow\uparrow}\right)\ketbra{0_{z}}{0_{z}}+p_{\downarrow\downarrow}\ketbra{-1_{z}}{-1_{z}}\\
\rho_{A} & =p_{\downarrow\uparrow}\ketbra{\downarrow_{L}\land\uparrow_{R}}{\downarrow_{L}\land\uparrow_{R}}+\left(p_{\uparrow\uparrow}+p_{\uparrow\downarrow}+p_{\downarrow\downarrow}\right)\ketbra{\uparrow_{L}\lor\downarrow_{R}}{\uparrow_{L}\lor\downarrow_{R}}.
\end{align*}
Apparently, when applied to classical states the reduction map is
just the classical coarse-graining of probabilities that sums them
up over the rows and columns of the BPT.

If $\rho$ is a pure quantum state such as
\[
\ket{\psi}:=\alpha_{\uparrow\uparrow}\ket{\uparrow\uparrow}+\alpha_{\uparrow\downarrow}\ket{\uparrow\downarrow}+\alpha_{\downarrow\uparrow}\ket{\downarrow\uparrow}+\alpha_{\downarrow\downarrow}\ket{\downarrow\downarrow},
\]
we calculate the reduced states to be (the ordering of the reduced
basis is as in Eqs. \eqref{eq:H_B from part. BPT} and \ref{eq:H_A from part. BPT})
\begin{align*}
\rho_{B} & =\begin{pmatrix}\left|\alpha_{\uparrow\uparrow}\right|^{2} & \alpha_{\uparrow\uparrow}\alpha_{\uparrow\downarrow}^{*} & \alpha_{\uparrow\uparrow}\alpha_{\downarrow\downarrow}^{*}\\
\alpha_{\uparrow\downarrow}\alpha_{\uparrow\uparrow}^{*} & \left|\alpha_{\uparrow\downarrow}\right|^{2}+\left|\alpha_{\downarrow\uparrow}\right|^{2} & \alpha_{\uparrow\downarrow}\alpha_{\downarrow\downarrow}^{*}\\
\alpha_{\downarrow\downarrow}\alpha_{\uparrow\uparrow}^{*} & \alpha_{\downarrow\downarrow}\alpha_{\uparrow\downarrow}^{*} & \left|\alpha_{\downarrow\downarrow}\right|^{2}
\end{pmatrix}\\
\\
\rho_{A} & =\begin{pmatrix}\left|\alpha_{\downarrow\uparrow}\right|^{2} & \alpha_{\downarrow\uparrow}\alpha_{\uparrow\downarrow}^{*}\\
\alpha_{\uparrow\downarrow}\alpha_{\downarrow\uparrow}^{*} & \left|\alpha_{\uparrow\uparrow}\right|^{2}+\left|\alpha_{\uparrow\downarrow}\right|^{2}+\left|\alpha_{\downarrow\downarrow}\right|^{2}
\end{pmatrix}.
\end{align*}
From the diagonal matrix elements we see that even in the pure quantum
case the classical coarse-graining of probabilities persists, however,
now it also preserves some of the coherence terms.

We conclude that the probability distribution over the outcomes of
observables that distinguish between the basis are coarse-grained
by the state reduction map as in the classical case. Thus, by promoting
distinguishable states to orthogonal basis and probability distributions
to quantum states we reproduce the classical notion of coarse-graining
by tracing out a partial subsystem. However, in the quantum setting
there are more observables than just a distinction of certain basis,
and there is more to tracing out a partial subsystem than just a summation
of probabilities. In order to understand how all observables are affected
by tracing out a partial subsystem we need to derive the operational
meaning of such state reductions.

\subsection{The operational meaning of quantum coarse-graining}

From here on, we will use the notions of state reduction, tracing
out a partial subsystem, and quantum coarse-graining interchangeably.
In order to derive the operational meaning of quantum coarse-graining
we will have to establish a few more facts.

Using the tomographic representation \eqref{eq:part trace S_kl rep}
of the partial trace map we derive
\[
\tr_{\left(A\right)}\left(\rho\right)=\tr_{A}\left[V\rho V^{\dagger}\right]=\sum_{k,l=1}^{d_{B}}\tr\left[\tilde{S}_{kl}V\rho V^{\dagger}\right]\text{\ensuremath{\ket{b_{l}}\bra{b_{k}}}}=\sum_{k,l=1}^{d_{B}}\tr\left[\left(V^{\dagger}\tilde{S}_{kl}V\right)\rho\right]\text{\ensuremath{\ket{b_{l}}\bra{b_{k}}}}
\]
where $\tilde{S}_{kl}:=I_{A}\otimes\ket{b_{k}}\bra{b_{l}}$ are partial
isometries in $\mathcal{H}_{A}\otimes\mathcal{H}_{B}$. We define
partial isometries in $\mathcal{H}$ as $S_{kl}:=V^{\dagger}\tilde{S}_{kl}V$
and then the map that traces out a partial subsystem can also be given
in the tomographic representation 
\begin{equation}
\tr_{\left(A\right)}\left(\rho\right)=\sum_{k,l=1}^{d_{B}}\tr\left[S_{kl}\rho\right]\text{\ensuremath{\ket{b_{l}}\bra{b_{k}}}}.\label{eq:partial state reudct map S_kl}
\end{equation}

This representation simplifies things since the partial isometries
$\left\{ S_{kl}\right\} $ can be constructed directly from the non-rectangular
BPT. In order to see that, let us explicitly define 
\[
V=\sum_{ik}\ket{a_{i}}\otimes\ket{b_{k}}\bra{e_{ik}},
\]
where the row and column indices $i$,$k$ are constrained by the
shape of the BPT, that is
\[
V^{\dagger}\ket{a_{i}}\otimes\ket{b_{k}}=\begin{cases}
\ket{e_{ik}} & \textrm{The cell \ensuremath{i,k} is present in the BPT}\\
0 & \textrm{The cell \ensuremath{i,k} is absent in the BPT}.
\end{cases}
\]
The aforementioned partial isometries then reduce to

\begin{align}
S_{kl} & =V^{\dagger}\tilde{S}_{kl}V=V^{\dagger}\left(\sum_{i=1}^{d_{A}}\ket{a_{i}}\bra{a_{i}}\right)\otimes\ket{b_{k}}\bra{b_{l}}V\nonumber \\
 & =\sum_{i\in\mathrm{CR}\left(k,l\right)}\ket{e_{ik}}\bra{e_{il}}\label{eq:def of S_kl non rec}
\end{align}
where the set $\mathrm{CR}\left(k,l\right)$ contains the common row
indices of cells that are present in both columns $k$ and $l$.

The result in Eq. \eqref{eq:def of S_kl non rec} tells us how $S_{kl}$
is constructed directly from the BPT according to the alignment of
elements in the columns $k$ and $l$. When the BPT is rectangular
the set $\mathrm{CR}\left(k,l\right)$ always contains all the rows.
When the BPT is non-rectangular some rows are shorter than others
so not every row is present in every column. Compare the original
construction of partial isometries from BPTs in Eq. \eqref{eq: def of S_kl from BPT}
(assuming a single block) to the new construction in Eq. \ref{eq:def of S_kl non rec},
and verify that the later generalizes the former.

So far, it seems like we do not really have to make a distinction
between rectangular and non-rectangular BPTs. The isometries are constructed
according to the same general prescription in Eq. \ref{eq:def of S_kl non rec},
and the state reduction map is given by the same Eq. \ref{eq:partial state reudct map S_kl}
in terms of the isometries. The obvious question then is how does
the non-rectangular shape generalize operator algebras associated
with the rectangular shape. The answer to that begins with the following
definition.
\begin{defn}
\label{def:Operator systems}An \emph{operator system }is a subset
of operators $\mathcal{O}\subseteq\mathcal{L}\left(\mathcal{H}\right)$
such that:

(1) For all $O_{1},O_{2}\in\mathcal{O}$ and $c_{1},c_{2}\in\mathbb{C}$
we have $c_{1}O_{1}+c_{2}O_{2}\in\mathcal{O}$.

(2) For all $O\in\mathcal{O}$ we have $O^{\dagger}\in\mathcal{O}$.

(3) There is a projection $I_{\mathcal{O}}\in\mathcal{O}$ such that
$I_{\mathcal{O}}O=O$ for all $O\in\mathcal{O}$.
\end{defn}
Operator systems generalize operator algebras in that the products
of operators do not have to remain in the set. In other words, operator
systems are just subspaces of operators that are closed under conjugation
and contain a projection that serves as the identity.\footnote{In finite-dimensional operator algebras the existence of the identity
was not part of the definition because it could be derived.}

It turns out that for rectangular BPTs the constructed isometries
span operator algebras, but when non-rectangular BPTs are considered
they span operator systems. Furthermore, the constructed isometries
$\left\{ S_{kl}\right\} $ form a basis for the operator system that
they span. These facts are shown in the following proposition.
\begin{prop}
\label{prop: S_kl are oper sys}Let $\mathcal{H}\cong\mathcal{H}_{A}\oslash\mathcal{H}_{B}$
be a partial bipartition and let $\left\{ S_{kl}\right\} $ be the
isometries constructed according to Eq. \eqref{eq:def of S_kl non rec}.
Then, $\spn\left\{ S_{kl}\right\} $ is an operator system and the
set $\left\{ S_{kl}\right\} $ forms an operator basis that are orthogonal
with respect to the Hilbert-Schmidt (HS) inner product 
\[
\left\langle S_{k'l'},S_{kl}\right\rangle _{HS}=\delta_{kk'}\delta_{ll'}\left|\mathrm{CR}\left(k,l\right)\right|.
\]
\end{prop}
\begin{proof}
Condition (1) of Definition \ref{def:Operator systems} trivially
holds as $\spn\left\{ S_{kl}\right\} $ is a vector space. Condition
(2) holds because $S_{kl}^{\dagger}=S_{lk}$. For condition (3) we
construct $I_{\mathcal{O}}:=\sum_{k'}S_{k'k'}$ so 
\begin{align*}
I_{\mathcal{O}}S_{kl} & =\sum_{k'}\left(\sum_{i'\in\mathrm{CR}\left(k',k'\right)}\ket{e_{i'k'}}\bra{e_{i'k'}}\right)\left(\sum_{i\in\mathrm{CR}\left(k,l\right)}\ket{e_{ik}}\bra{e_{il}}\right)\\
 & =\sum_{i\in\mathrm{CR}\left(k,l\right)}\sum_{k'}\sum_{i'\in\mathrm{CR}\left(k',k'\right)}\delta_{ii'}\delta_{kk'}\ket{e_{i'k'}}\bra{e_{il}}=\sum_{i\in\mathrm{CR}\left(k,l\right)}\ket{e_{ik}}\bra{e_{il}}=S_{kl}.
\end{align*}
Finally, using the definition of the HS inner product $\left\langle X,Y\right\rangle _{HS}:=\tr\left[X^{\dagger}Y\right]$
we derive
\begin{align*}
\left\langle S_{k'l'},S_{kl}\right\rangle _{HS} & =\tr\left[\left(\sum_{i'\in\mathrm{CR}\left(k',l'\right)}\ket{e_{i'l'}}\bra{e_{i'k'}}\right)\left(\sum_{i\in\mathrm{CR}\left(k,l\right)}\ket{e_{ik}}\bra{e_{il}}\right)\right]\\
 & =\sum_{i'\in\mathrm{CR}\left(k',l'\right)}\sum_{i\in\mathrm{CR}\left(k,l\right)}\delta_{kk'}\delta_{ll'}\delta_{ii'}=\delta_{kk'}\delta_{ll'}\left|\mathrm{CR}\left(k,l\right)\right|.
\end{align*}
\end{proof}
Note that the basis $\left\{ S_{kl}\right\} $ are not normalized
since
\[
\left\Vert S_{kl}\right\Vert _{HS}^{2}=\left\langle S_{kl},S_{kl}\right\rangle _{HS}=\left|\mathrm{CR}\left(k,l\right)\right|.
\]
The normalized basis will be denoted as
\[
\hat{S}_{kl}:=\frac{S_{kl}}{\sqrt{\left|\mathrm{CR}\left(k,l\right)\right|}}.
\]

Finally, we are ready to discuss the operational meaning of state
reductions associated with partial bipartitions. When the bipartition
$\mathcal{H}\cong\mathcal{H}_{A}\otimes\mathcal{H}_{B}$ is a proper
tensor product, the operational meaning of the reduced state $\rho_{B}=\tr_{A}\left(\rho\right)$
is that it preserves the expectation values of all the observables
of the form $O=I_{A}\otimes O_{B}$. That is, for every observable
$O\in I_{A}\otimes\mathcal{L}\left(\mathcal{H}_{B}\right)$ there
is an observable $O_{B}\in\mathcal{L}\left(\mathcal{H}_{B}\right)$
such that $\tr\left[O\rho\right]=\tr\left[O_{B}\rho_{B}\right]$,
and vice versa. The correspondence between these observables is trivially
given by
\begin{align*}
O_{B} & \longleftrightarrow O=I_{A}\otimes O_{B}.
\end{align*}
The following theorem generalizes this statement to partial bipartitions.
\begin{thm}
\label{thm:oper meaning of part BP}Let $\mathcal{H}\cong\mathcal{H}_{A}\oslash\mathcal{H}_{B}$
be a partial bipartition with the isometries $\left\{ S_{kl}\right\} $
as constructed in Eq. \eqref{eq:def of S_kl non rec}, and let $\rho\in\mathcal{L}\left(\mathcal{H}\right)$
and $\rho_{B}=\tr_{\left(A\right)}\left(\rho\right)$. Then, for every
observable $O\in\spn\left\{ S_{kl}\right\} $ there is an observable
$O_{B}\in\mathcal{L}\left(\mathcal{H}_{B}\right)$ such that $\tr\left[O\rho\right]=\tr\left[O_{B}\rho_{B}\right]$,
and vice versa. This correspondence of observables is explicitly given
by
\begin{align}
O_{B} & \longmapsto O:=\sum_{k,l=1}^{d_{B}}\tr\left[O_{B}\ket{b_{l}}\bra{b_{k}}\right]S_{kl}\label{eq:O_B to O}\\
O & \longmapsto O_{B}:=\sum_{k,l=1}^{d_{B}}\frac{\tr\left(S_{lk}O\right)}{\left|\mathrm{CR}\left(k,l\right)\right|}\ketbra{b_{k}}{b_{l}}.\label{eq:O to O_B}
\end{align}
\end{thm}
\begin{proof}
To see that the correspondence \eqref{eq:O_B to O} holds we recall
the definition \eqref{eq:partial state reudct map S_kl} of $\tr_{\left(A\right)}\left(\rho\right)$
and use the linearity of trace to derive

\[
\tr\left[O_{B}\rho_{B}\right]=\tr\left[O_{B}\left(\sum_{k,l=1}^{d_{B}}\ket{b_{l}}\bra{b_{k}}\tr\left[S_{kl}\rho\right]\right)\right]=\tr\left[\left(\sum_{k,l=1}^{d_{B}}\tr\left[O_{B}\ket{b_{l}}\bra{b_{k}}\right]S_{kl}\right)\rho\right]=\tr\left[O\rho\right].
\]
For the converse, using the fact that $O\in\spn\left\{ S_{kl}\right\} $
we can expand it in the $\left\{ \hat{S}_{kl}\right\} $ basis as
\[
O=\sum_{k,l=1}^{d_{B}}\left\langle \hat{S}_{kl},O\right\rangle _{HS}\hat{S}_{kl}=\sum_{k,l=1}^{d_{B}}\frac{\tr\left(S_{lk}O\right)}{\left|\mathrm{CR}\left(k,l\right)\right|}S_{kl}.
\]
so
\[
\tr\left[O\rho\right]=\sum_{k,l=1}^{d_{B}}\frac{\tr\left(S_{lk}O\right)}{\left|\mathrm{CR}\left(k,l\right)\right|}\tr\left[S_{kl}\rho\right].
\]
Then, using the correspondence \eqref{eq:O to O_B} we can check that
$\tr\left[O\rho\right]=\tr\left[O_{B}\rho_{B}\right]$ holds 
\begin{align*}
\tr\left[O_{B}\rho_{B}\right] & =\tr\left[\left(\sum_{k,l=1}^{d_{B}}\frac{\tr\left(S_{lk}O\right)}{\left|\mathrm{CR}\left(k,l\right)\right|}\ketbra{b_{k}}{b_{l}}\right)\left(\sum_{k',l'=1}^{d_{B}}\tr\left[S_{k'l'}\rho\right]\text{\ensuremath{\ket{b_{l'}}\bra{b_{k'}}}}\right)\right]\\
 & =\sum_{k,l=1}^{d_{B}}\frac{\tr\left(S_{lk}O\right)}{\left|\mathrm{CR}\left(k,l\right)\right|}\tr\left[S_{kl}\rho\right]=\tr\left[O\rho\right].
\end{align*}
\end{proof}
Thus, the above theorem tells us that the operational meaning of the
reduced state $\rho_{B}$ is that it preserves the expectation values
of all the observables $O\in\spn\left\{ S_{kl}\right\} $. It is also
worth noting that the correspondences \eqref{eq:O_B to O} and \eqref{eq:O to O_B}
can be concisely expressed element-wise as
\begin{align}
\ket{b_{l}}\bra{b_{k}} & \longleftrightarrow S_{kl}.\label{eq:O_B and O corres elem wise}
\end{align}

The surprising feature of tracing out a partial subsystem is that
even when the expectation values of some observable are preserved,
it does not mean that the probabilities of its individual outcomes
are preserved. The reason for that is because the statement $O\in\spn\left\{ S_{kl}\right\} $
does not imply that the spectral projections $\left\{ \Pi_{i}\right\} $
of $O$ are also in $\spn\left\{ S_{kl}\right\} $ (unless $\spn\left\{ S_{kl}\right\} $
is an algebra so the bipartition is not partial). Since the probabilities
of individual outcomes are given by the expectation values of the
spectral projections, they are not guaranteed to be preserved.

As an example, let us go back to the two spin Hilbert space $\mathcal{H}=\underline{\frac{1}{2}}\otimes\underline{\frac{1}{2}}$
and change to the total spin basis $\ket{j,m_{z}}$. The partial bipartition
that we now want to consider is given by the following BPT

\noindent\begin{minipage}[c]{1\columnwidth}%
\begin{center}
\vspace{0.5\baselineskip}
{\small{} }%
\begin{tabular}{|c|c|c|c|c|}
\cline{2-2} \cline{5-5} 
\multicolumn{1}{c|}{} & {\small{}$0,0_{z}$} & \multicolumn{1}{c}{} & $\rightarrow$ & $s$\tabularnewline
\cline{1-3} \cline{2-3} \cline{3-3} \cline{5-5} 
{\small{}$1,+1_{z}$} & {\small{}$1,0_{z}$} & {\small{}$1,-1_{z}$} & $\rightarrow$ & $t$\tabularnewline
\cline{1-3} \cline{2-3} \cline{3-3} \cline{5-5} 
\multicolumn{1}{c}{$\downarrow$} & \multicolumn{1}{c}{$\downarrow$} & \multicolumn{1}{c}{$\downarrow$} & \multicolumn{1}{c}{} & \multicolumn{1}{c}{}\tabularnewline
\cline{1-3} \cline{2-3} \cline{3-3} 
$1_{z}$ & $0_{z}$ & $-1_{z}$ & \multicolumn{1}{c}{} & \multicolumn{1}{c}{}\tabularnewline
\cline{1-3} \cline{2-3} \cline{3-3} 
\end{tabular}{\small{} .}\vspace{0.5\baselineskip}
\par\end{center}%
\end{minipage} As before, the columns distinguish between the states with different
$m_{z}$. The rows now distinguish between the singlet ($s$) and
triplet ($t$) states. The partial subsystems are therefore
\begin{align*}
\mathcal{H}_{B} & :=\spn\left\{ \ket{1_{z}},\ket{0_{z}},\ket{-1_{z}}\right\} \\
\mathcal{H}_{A} & :=\spn\left\{ \ket s,\ket t\right\} .
\end{align*}
The question then is what observable information is preserved if we
trace out the singlet-triplet subsystem?

The short answer is that the preserved information is given by the
expectation values of all the observables $O\in\spn\left\{ S_{kl}\right\} $,
where $S_{kl}$ are constructed from the BPT according to Eq. \eqref{eq:def of S_kl non rec}
as 
\[
S_{kl}=\sum_{j\in\mathrm{CR}\left(k,l\right)}\ketbra{j,k_{z}}{j,l_{z}}=\begin{cases}
\ketbra{0,0_{z}}{0,0_{z}}+\ketbra{1,0_{z}}{1,0_{z}} & k=l=0\\
\ketbra{1,k_{z}}{1,l_{z}} & \textrm{otherwise .}
\end{cases}
\]
In particular, for $k=l$ the projections $\left\{ S_{kk}\right\} $
are the spectral projections of the total spin component 
\[
J_{z}=\sum_{j,k}k\ketbra{j,k}{j,k}=\sum_{k}kS_{kk}\,.
\]
Thus, all the statistical information about the observable $J_{z}$
is preserved and, according to the correspondence \eqref{eq:O_B and O corres elem wise},
the reduced observable is 
\[
J_{z}\longmapsto J_{z;B}=\sum_{k}k\ketbra{k_{z}}{k_{z}}.
\]
As expected, $J_{z}$ corresponds to the $\hat{z}$ component of the
reduced spin-$1$ system and $\ket{k_{z}}$ are its eigenstates.

The total spin ladder operators $J_{\pm}$ are also in $\spn\left\{ S_{kl}\right\} $
since they can be expanded as 
\begin{align*}
J_{+} & =\sqrt{2}\ketbra{1,1_{z}}{1,0_{z}}+\sqrt{2}\ketbra{1,0_{z}}{1,-1_{z}}=\sqrt{2}S_{1,0}+\sqrt{2}S_{0,-1}\\
J_{-} & =\sqrt{2}\ketbra{1,-1_{z}}{1,0_{z}}+\sqrt{2}\ketbra{1,0_{z}}{1,1_{z}}=\sqrt{2}S_{-1,0}+\sqrt{2}S_{0,1}\,.
\end{align*}
This means that the other two total spin components $J_{x}$ and $J_{y}$
are in $\spn\left\{ S_{kl}\right\} $ as well. From the correspondence
\eqref{eq:O_B and O corres elem wise} we have 
\begin{align*}
J_{+} & \longmapsto J_{+;B}=\sqrt{2}\ketbra{1_{z}}{0_{z}}+\sqrt{2}\ketbra{0_{z}}{-1_{z}}\\
J_{-} & \longmapsto J_{-;B}=\sqrt{2}\ketbra{-1_{z}}{0_{z}}+\sqrt{2}\ketbra{0_{z}}{1_{z}}
\end{align*}
and so
\begin{align*}
J_{x} & \longmapsto J_{x;B}=\frac{J_{+;B}+J_{-;B}}{2}=\frac{\ketbra{1_{z}}{0_{z}}+\ketbra{0_{z}}{-1_{z}}+\ketbra{-1_{z}}{0_{z}}+\ketbra{0_{z}}{1_{z}}}{\sqrt{2}}\\
J_{y} & \longmapsto J_{-;B}=\frac{J_{+;B}-J_{-;B}}{2i}=\frac{\ketbra{1_{z}}{0_{z}}+\ketbra{0_{z}}{-1_{z}}-\ketbra{-1_{z}}{0_{z}}-\ketbra{0_{z}}{1_{z}}}{\sqrt{2}i}.
\end{align*}
That is, $J_{x}$ and $J_{y}$ correspond to the $\hat{x}$ and $\hat{y}$
components of the reduced spin-$1$ system.

Unlike $J_{z}$, however, the spectral projections of $J_{x}$ and
$J_{y}$ are not present in $\spn\left\{ S_{kl}\right\} $. If they
were, then $J_{x}^{2}$ (similarly $J_{y}^{2}$) would also be in
$\spn\left\{ S_{kl}\right\} $ but that is not the case as
\begin{align*}
J_{x}^{2} & =\left(\frac{J_{+}+J_{-}}{2}\right)^{2}=\frac{1}{2}\left(S_{1,0}+S_{0,-1}+S_{-1,0}+S_{0,1}\right)^{2}\\
 & =\frac{1}{2}\left(S_{1,-1}+S_{1,1}+S_{-1,-1}+S_{-1,1}\right)+\ketbra{1,0_{z}}{1,0_{z}}.
\end{align*}
Since $\ketbra{1,0_{z}}{1,0_{z}}$ is not in $\spn\left\{ S_{kl}\right\} $
(because $S_{00}=\ketbra{0,0_{z}}{0,0_{z}}+\ketbra{1,0_{z}}{1,0_{z}}$)
then neither is $J_{x}^{2}$.

Therefore, by tracing out the singlet-triplet partial subsystem we
can preserve all the statistical information about the component $J_{z}$
of total spin, but for $J_{x}$ and $J_{y}$ only the expectation
values are preserved and not their higher moments. These, of course,
are not all the observables in $\spn\left\{ S_{kl}\right\} $ and
there is more observable information that is preserved in the reduced
states.

\chapter{The uncertainty principle on a lattice \label{chap:The-uncertainty-principle}}

In this chapter we will carry out a case study of the uncertainty
principle on a lattice. Unlike previous chapters where the emphasis
was on the methods, here we will focus on specific physical questions.
Because a much simpler notion of coarse-graining is needed here, the
analysis in this chapter will not rely on the contents of previous
chapters. These results were originally published in \citep{kabernik2020quantifying}.

Heisenberg\textquoteright s uncertainty principle is colloquially
understood as the fact that arbitrarily precise values of position
and momentum cannot simultaneously be determined (see \citep{busch2006complementarity,busch2007heisenberg}
for a review). A rigorous formulation of the uncertainty principle
is often conflated with the uncertainty relations for states $\sigma_{x}\sigma_{p}\geq\hbar/2$,
where $\sigma_{x}$ and $\sigma_{p}$ refer to the standard deviations
of independently measured position and momentum of a particle in the
same state. This inequality rules out the possibility of quantum states
with arbitrarily sharp values of both position and moment. It does
not, however, rule out the possibility of measurements that simultaneously
determine both of these values with arbitrary precision. The essential
effect behind the uncertainty principle that rules out the latter
possibility is the mutual disturbance between measurements of incompatible
observables.

According to the original formulation by Heisenberg \citep{heisenberg1927},
due to the unavoidable disturbance by measurements, it is not possible
to localize a particle in a phase space cell of the size of the Planck
constant or smaller. However, when phase space cells much coarser
than the Planck constant are considered, Heisenberg argued that the
values of both observables can be estimated at the expense of lower
resolution. The picture that emerges from Heisenberg's original arguments
is that the Planck constant sets a resolution scale in phase space
that separates the quantum regime from the classical (see Fig. \ref{fig:phase space diag}(a)).
There is, of course, a continuum of scales so it is natural to ask
for a characteristic function that outlines how the uncertainty principle
transitions to the classical regime as the resolution of measurements
decreases.

A rigorous formulation of the measurement uncertainty principle has
been extensively debated in recent years \citep{Ozawa03Universally,Busch13proof,branciard2013error,korzekwa2014operational,buscemi2014noise,rozema2015note},
producing multiple perspectives on the fundamental limits of simultaneous
measurability of incompatible observables. These formulations are
similar to the uncertainty relations for states as they capture the
trade-off between the resolution and disturbance of measurements (which
may also depend on the states). However, the picture of how the uncertainty
principle transitions to the regime where joint measurability is possible
is not so clear from these perspectives.

Furthermore, the picture of continuous phase space as a fundamental
concept has been challenged by the various approaches to quantum gravity
\citep{ali2009discreteness}. The existence of minimal length in space
is indicated by many thought experiments that point to the impossibility
of probing length scales close to the Planck length $\delta x\sim10^{-35}\,\textrm{m}$
(see \citep{hossenfelder2013minimal} for a review). It then follows
that due to the existence of minimal length in space, the canonical
commutation relations and the associated mutual disturbance effects
have to be modified; this is known as the \emph{generalized uncertainty
principle }\citep{ali2009discreteness}. There is great interest in
identifying any observable effects associated with the modifications
of the uncertainty principle due to minimal length, and in recent
years there have been at least two experimental proposals \citep{ali2011proposal,pikovski2012probing}
based on this idea.

\begin{figure}
\begin{centering}
(a)\includegraphics[width=0.35\columnwidth]{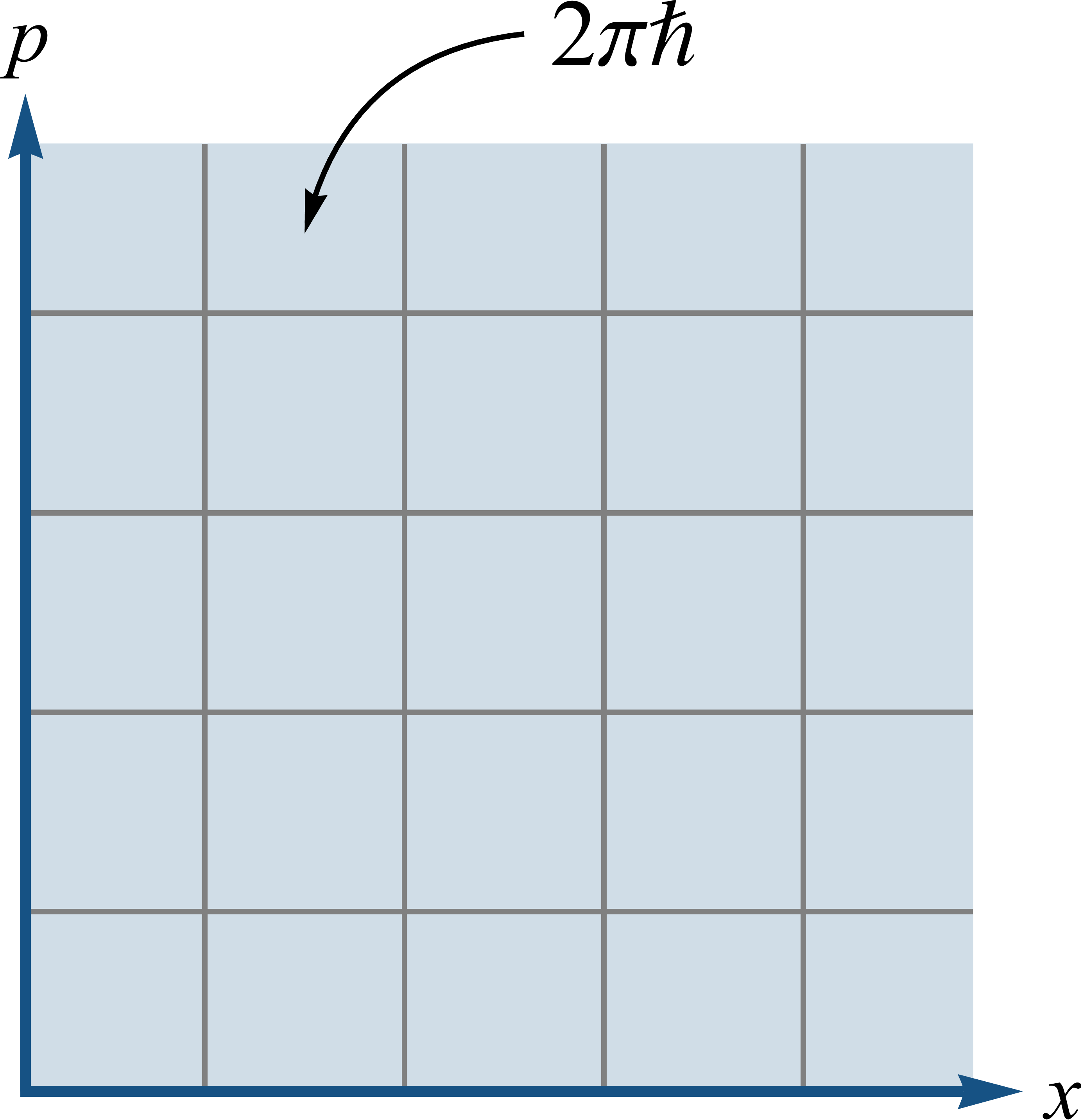}
\hspace*{0.1\columnwidth}(b)\includegraphics[width=0.35\columnwidth]{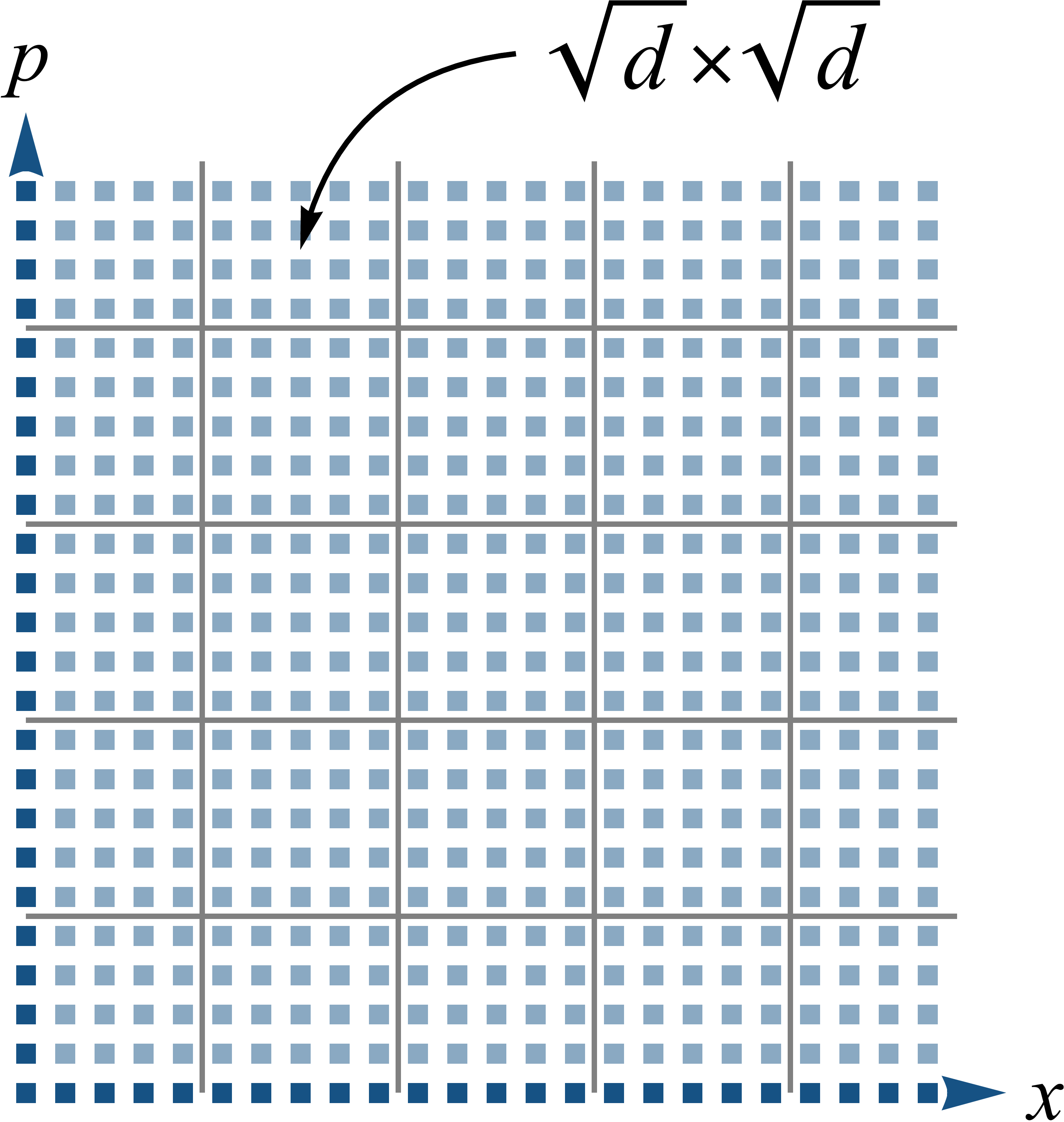}
\par\end{centering}
\caption{\label{fig:phase space diag}(a) The continuous phase space where
the cells with the area $2\pi\hbar$ represent the resolution scale
associated with the uncertainty principle. (b) The discretized phase
space of a lattice of integer length $d$. The cells with the area
$\sqrt{d}\times\sqrt{d}$ arise from the scale $\sqrt{d}$ associated
with the uncertainty principle on a lattice. The Planck constant $2\pi\hbar$
can be recovered from $\sqrt{d}$ by converting the phase space area
$\sqrt{d}\times\sqrt{d}$ to proper units.}
\end{figure}

In this chapter we will study the effects of the uncertainty principle
as a function of measurement resolution. In Section \ref{sec:From-quantum-to}
we will introduce and analyze an operationally defined measure of
mutual disturbance that is responsible for the uncertainty principle.
This measure is based on the probability that an instantaneous succession
of coarse-grained measurements of position-momentum-position will
agree on both outcomes of position. The analysis will be carried out
on a finite-dimensional periodic lattice of integer length $d$, where
the continuous space can be recovered by introducing the minimal length
$\delta x$ and taking the limits $d\rightarrow\infty$, $\delta x\rightarrow0$.
As a result, we will derive a rigorous characteristic function that
quantifies the transition of the uncertainty principle from quantum
to classical regimes, in both continuous and discrete settings.

In Section \ref{sec:The-implications-of} we will study the implications
of the uncertainty principle on a lattice. One implication is that
the transition of the uncertainty principle to the classical regimes
is perturbed by the discontinuity of the lattice. We will see how
this perturbation can be quantified by our operationally defined characteristic
function.

Another implication is related to the question of how classicality
emerges in isolated finite-dimensional systems. Such questions have
been considered in \citep{poulin2005macroscopic} and \citep{peres2006quantum},
and in particular Kofler and Brukner \citep{KoflerBrukner2007Classical}
have demonstrated that for a spin-$j$ system, incompatible spin components
can simultaneously be determined if the resolution of measurements
is coarse compared to $\sqrt{j}$.

Our analysis show that the same conclusion applies to position and
momentum on a lattice, where both variables can simultaneously be
determined if the resolution of measurements is coarse compared to
$\sqrt{d}$. We will then discuss how the unitless scale $\sqrt{d}$
factorizes the Planck constant (see Fig. \ref{fig:phase space diag}
(b)) and defines a new length scale given by the geometric mean $\sqrt{\delta xL}$
of the \emph{minimal} length $\delta x$ and the \emph{maximal} length
$L$.

\section{From quantum to classical regimes on a lattice\label{sec:From-quantum-to}}

Let us consider the simple, operationally meaningful quantity $\boldsymbol{p}_{\textrm{agree}}$,
which is the probability that an instantaneous succession of position-momentum-position
measurements will agree on both outcomes of position, regardless of
the outcomes. When all measurements have arbitrarily fine resolution,
the second measurement in this succession prepares a sharp momentum
state that is nearly uniformly distributed in position space. Then,
the probability that the first and the last measurements of position
will agree is vanishingly small $\boldsymbol{p}_{\textrm{agree}}\approx0$.
As we decrease the resolution of measurements, we expect the probability
$\boldsymbol{p}_{\textrm{agree}}$ to grow from $0$ to $1$ because
coarser momentum measurement will cause less spread in the position
space, and coarser position measurements will be more likely to agree
on the estimate of position.

Now, consider the average $\left\langle \boldsymbol{p}_{\textrm{agree}}\right\rangle $
over all states. In general, the average value $\left\langle \boldsymbol{p}_{\textrm{agree}}\right\rangle $
does not inform us about how strongly the measurements disturb each
other for any particular state $\rho$. However, when the average
$\left\langle \boldsymbol{p}_{\textrm{agree}}\right\rangle $ is close
to $0$ or $1$, the value of $\boldsymbol{p}_{\textrm{agree}}\left(\rho\right)$
has to converge to the average for almost all states $\rho$. That
is because $\boldsymbol{p}_{\textrm{agree}}\in\left[0,1\right]$ so
its variance has to vanish as the average gets close to the edges.
Therefore, the value of $\left\langle \boldsymbol{p}_{\textrm{agree}}\right\rangle $
indicates how close we are to the regime $\left\langle \boldsymbol{p}_{\textrm{agree}}\right\rangle \approx0$
where the measurements strongly disturb each other for almost all
states, or the regime $\left\langle \boldsymbol{p}_{\textrm{agree}}\right\rangle \approx1$
where the mutual disturbance is inconsequential for almost all states.
We can therefore utilize $\left\langle \boldsymbol{p}_{\textrm{agree}}\right\rangle $
as a characteristic function that quantifies the relevance of the
uncertainty principle and outlines the transition between quantum
and classical regimes.

For the rest of this section we will focus on deriving and studying
the explicit expression for $\left\langle \boldsymbol{p}_{\textrm{agree}}\right\rangle $
as a function of measurement resolution. The most technical calculations
concerned with the upper and lower bounds on $\left\langle \boldsymbol{p}_{\textrm{agree}}\right\rangle $
are deferred to the Appendix. The final result is the explicit expression
for $\left\langle \boldsymbol{p}_{\textrm{agree}}\right\rangle $
in Eq. \eqref{eq: lam_agree expr plane} along with the bounds \eqref{eq: lam agree upper bound},
\eqref{eq:  lam agree lower bound}, and the plot presented in Fig.
\ref{fig:ClassicalityVw}.

In order to calculate the value of $\left\langle \boldsymbol{p}_{\textrm{agree}}\right\rangle $
as a function of measurement resolution, we turn to the canonical
setting of finite-dimensional quantum mechanics. In this setting we
consider a particle on a periodic one-dimensional lattice with $d$
lattice sites. Initially, both lattice units of position and momentum
will be set to unity $\delta x\equiv1$, $\delta p\equiv1$. Later,
we will introduce proper units and consider the continuum limit.

Following the construction in \citep{vourdas2004quantum,jagannathan1981finite},
the Hilbert space of our system is given by the span of position basis
$\ket{X;n}$ for $n=0,...,d-1$. The momentum basis are related to
the position basis via the discrete Fourier transform $F$ 
\begin{align*}
\ket{X;n} & =F^{\dagger}\ket{P;n}=\frac{1}{\sqrt{d}}\sum_{m=0}^{d-1}e^{-i2\pi mn/d}\ket{P;m}\\
\ket{P;m} & =F\ket{X;m}=\frac{1}{\sqrt{d}}\sum_{n=0}^{d-1}e^{i2\pi mn/d}\ket{X;n}.
\end{align*}

In principle, realistic finite resolution measurements should be modeled
as unsharp POVMs \citep{busch1996quantum,peres2006quantum}. For our
purposes, however, it will be sufficient to consider the idealized
version in the form of coarse-grained projective measurements.

We introduce the integer parameters $w_{x}$, $w_{p}$ to specify
the widths of the coarse-graining intervals for the corresponding
observables (larger $w$ means lower resolution). The variable $k=d/w$
specifies the number of coarse-graining intervals which we will also
assume to be an integer. See Fig. \ref{fig:lattice diagram} for a
diagrammatic summary of the relevant lengths. 
\begin{figure}[t]
\begin{centering}
\includegraphics[width=0.8\columnwidth]{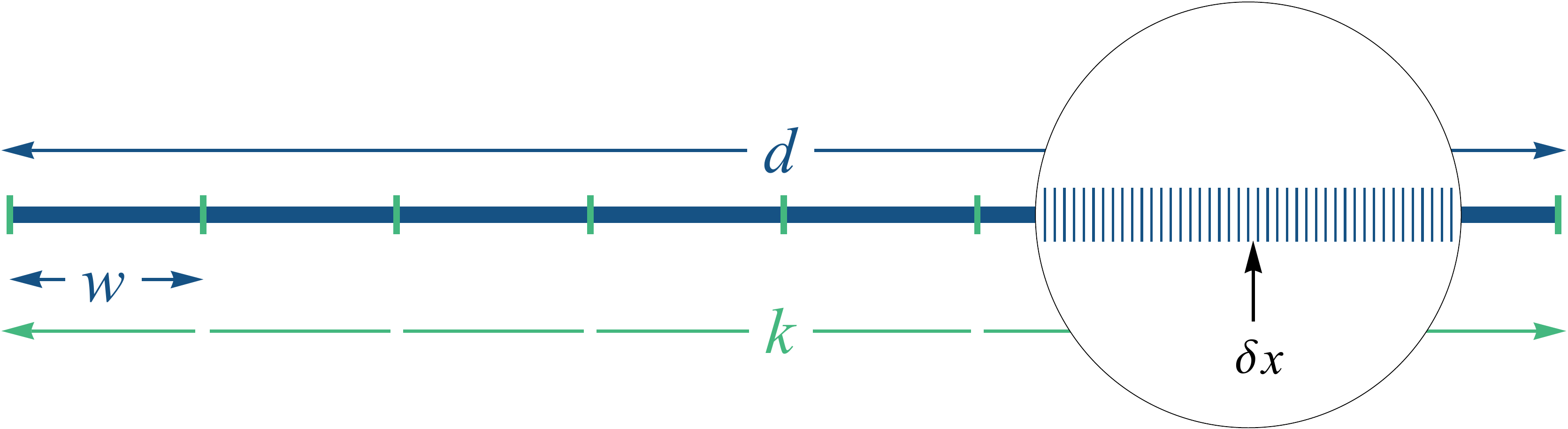}
\par\end{centering}
\caption{\label{fig:lattice diagram}Periodic one dimensional lattice with
$d$ lattice sites in total, $w$ lattice sites in each coarse-graining
interval, and $k=d/w$ intervals. The lattice unit of length is $\delta x$.}
\end{figure}

The coarse-grained position and momentum observables are constructed
from the spectral projections
\begin{align}
\Pi_{X;\nu} & =\sum_{n=\nu w_{x}}^{\nu w_{x}+w_{x}-1}\ket{X;n}\bra{X;n}\label{eq:Pi_X;nu}\\
\Pi_{P;\mu} & =\sum_{m=\mu w_{p}}^{\mu w_{p}+w_{p}-1}\ket{P;m}\bra{P;m}\label{eq:Pi_P;mu}
\end{align}
associated with the eigenvalues of coarse-grained position $\nu=0,...,k_{x}-1$
and momentum $\mu=0,...,k_{p}-1$. In the following we will only need
these spectral projections, so we do not have to explicitly define
the operators of coarse-grained position and momentum.

Let us now calculate the probability of getting the outcomes $\nu,\mu,\nu$
in an instantaneous sequence of position-momentum-position measurements
on the initial state $\rho$. Using the intermediate post-measurement
states in this sequence 
\[
\rho^{\left(\nu\right)}=\frac{\Pi_{X;\nu}\rho\Pi_{X;\nu}}{\tr\left[\Pi_{X;\nu}\rho\right]}\hspace{2cm}\rho^{\left(\nu\mu\right)}=\frac{\Pi_{X;\mu}\rho^{\left(\nu\right)}\Pi_{X;\mu}}{\tr\left[\Pi_{X;\mu}\rho^{\left(\nu\right)}\right]}
\]
we can express this probability as 
\begin{align}
\boldsymbol{p}_{xpx}\left(\nu,\mu,\nu|\rho\right) & =\tr\left[\Pi_{X;\nu}\rho\right]\tr\left[\Pi_{P;\mu}\rho^{\left(\nu\right)}\right]\tr\left[\Pi_{X;\nu}\rho^{\left(\nu\mu\right)}\right]\nonumber \\
 & =\tr\left[\left(\Pi_{X;\nu}\Pi_{P;\mu}\Pi_{X;\nu}\right)^{2}\rho\right].\label{eq:def p_xpx}
\end{align}
Then, the probability that both position outcomes agree, regardless
of the outcomes, is 
\begin{align}
\boldsymbol{p}_{\textrm{agree}}\left(\rho\right) & =\sum_{v=0}^{k_{x}-1}\sum_{\mu=0}^{k_{p}-1}\boldsymbol{p}_{xpx}\left(\nu,\mu,\nu|\rho\right)\nonumber \\
 & =\tr\left[\sum_{v=0}^{k_{x}-1}\sum_{\mu=0}^{k_{p}-1}\left(\Pi_{X;\nu}\Pi_{P;\mu}\Pi_{X;\nu}\right)^{2}\rho\right].\label{eq:p_agree}
\end{align}

From Eq. \eqref{eq:p_agree} we identify the observable 
\[
\Lambda_{\textrm{agree}}=\sum_{\nu=0}^{k_{x}-1}\sum_{\mu=0}^{k_{p}-1}\left(\Pi_{X;\nu}\Pi_{P;\mu}\Pi_{X;\nu}\right)^{2}
\]
whose expectation values are the probabilities $\boldsymbol{p}_{\textrm{agree}}\left(\rho\right)=\tr\left(\Lambda_{\textrm{agree}}\rho\right)$.
Since $\boldsymbol{p}_{\textrm{agree}}\left(\rho\right)$ is linear
in $\rho$, the average $\left\langle \boldsymbol{p}_{\textrm{agree}}\right\rangle $
is given by $\boldsymbol{p}_{\textrm{agree}}\left(\left\langle \rho\right\rangle \right)$,
where $\left\langle \rho\right\rangle =\frac{1}{d}I$ is the average
state, thus
\[
\left\langle \boldsymbol{p}_{\textrm{agree}}\right\rangle =\boldsymbol{p}_{\textrm{agree}}\left(\frac{1}{d}I\right)=\frac{1}{d}\tr\left[\Lambda_{\textrm{agree}}\right].
\]
In order to calculate $\left\langle \boldsymbol{p}_{\textrm{agree}}\right\rangle $
as an explicit function of $w_{x}$ and $w_{p}$, we will have to
establish a few identities.

Let us introduce the lattice translation operators $T_{X}$, $\text{\ensuremath{T_{P}} }$
in position and momentum defined by their action on the basis (addition
on the lattice is $\mathsf{mod}\,d$)
\begin{align*}
T_{X}\ket{X;n} & =\ket{X;n+1} &  & T_{X}^{\dagger}\ket{X;n}=\ket{X;n-1}\\
T_{P}\ket{P;m} & =\ket{P;m+1} &  & T_{P}^{\dagger}\ket{P;m}=\ket{P;m-1}.
\end{align*}
By expanding the position basis in momentum basis and vice versa,
it is straight forward to verify that
\begin{align*}
T_{P}\ket{X;n} & =e^{i2\pi n/d}\ket{X;n} &  & T_{P}^{\dagger}\ket{X;n}=e^{-i2\pi n/d}\ket{X;n}\\
T_{X}\ket{P;m} & =e^{-i2\pi m/d}\ket{P;m} &  & T_{X}^{\dagger}\ket{P;m}=e^{i2\pi m/d}\ket{P;m}.
\end{align*}
Therefore, $T_{P}$ commutes with $\ket{X;n}\bra{X;n}$ and $T_{X}$
commutes with $\ket{P;m}\bra{P;m}$. By extension, $T_{P}$ commutes
with $\Pi_{X;\nu}$ and $T_{X}$ commutes with $\Pi_{P;\mu}$.

Using the translation operators we can express the coarse-grained
position and momentum projections \eqref{eq:Pi_X;nu}, \eqref{eq:Pi_P;mu}
as 
\begin{align*}
\Pi_{X;\nu} & =\sum_{n=0}^{w_{x}-1}T_{X}^{\nu w_{x}}\ket{X;n}\bra{X;n}T_{X}^{\nu w_{x}\dagger}=T_{X}^{\nu w_{x}}\Pi_{X;0}T_{X}^{\nu w_{x}\dagger}\\
\Pi_{P;\mu} & =\sum_{m=0}^{w_{p}-1}T_{P}^{\mu w_{p}}\ket{P;m}\bra{P;m}T_{P}^{\mu w_{p}\dagger}=T_{P}^{\mu w_{p}}\Pi_{P;0}T_{P}^{\mu w_{p}\dagger}.
\end{align*}
 Then, using the commutativity of projections with translations we
get the identity

\[
\Pi_{X;\nu}\Pi_{P;\mu}\Pi_{X;\nu}=T_{P}^{\mu w_{p}}\left(\Pi_{X;\nu}\Pi_{P;0}\Pi_{X;\nu}\right)T_{P}^{\mu w_{p}\dagger}=T_{P}^{\mu w_{p}}T_{X}^{\nu w_{x}}\left(\Pi_{X;0}\Pi_{P;0}\Pi_{X;0}\right)T_{X}^{\nu w_{x}\dagger}T_{P}^{\mu w_{p}\dagger}.
\]
With this identity we can simplify
\begin{align}
\left\langle \boldsymbol{p}_{\textrm{agree}}\right\rangle  & =\frac{1}{d}\tr\left[\Lambda_{\textrm{agree}}\right]=\frac{1}{d}\sum_{\nu=0}^{k_{x}-1}\sum_{\mu=0}^{k_{p}-1}\tr\left[\left(\Pi_{X;\nu}\Pi_{P;\mu}\Pi_{X;\nu}\right)^{2}\right]\nonumber \\
 & =\frac{k_{x}k_{p}}{d}\tr\left[\left(\Pi_{X;0}\Pi_{P;0}\Pi_{X;0}\right)^{2}\right].\label{eq: p_agree =00003D tr=00005B(PPP)^2=00005D}
\end{align}

Let us then express

\begin{equation}
\Pi_{X;0}\Pi_{P;0}\Pi_{X;0}=\sum_{m=0}^{w_{p}-1}\Pi_{X;0}\ket{P;m}\bra{P;m}\Pi_{X;0}=\frac{1}{k_{x}}\sum_{m=0}^{w_{p}-1}\ket{P_{0};m}\bra{P_{0};m}.\label{eq:Pi_X0Pi_P0Pi_X0 in trunc basis}
\end{equation}
Here we have defined the \emph{truncated} momentum states 
\begin{align*}
\ket{P_{\nu};m} & :=\sqrt{k_{x}}\,\Pi_{X;\nu}\ket{P;m}=\frac{1}{\sqrt{w_{x}}}\sum_{n=\nu w_{x}}^{\nu w_{x}+w_{x}-1}e^{i2\pi mn/d}\ket{X;n}
\end{align*}
given by normalizing the support of the $m$'th momentum state on
the $\nu$'th position interval. In general, these states are not
orthogonal to each other and their overlap is given by
\begin{align*}
\braket{P_{\nu'};m'}{P_{\nu};m} & =\delta_{\nu',\nu}k_{x}\bra{P;m'}\Pi_{X;\nu}\ket{P;m}=\delta_{\nu',\nu}\frac{k_{x}}{d}\sum_{n=\nu w_{x}}^{\nu w_{x}+w_{x}-1}e^{i2\pi\left(m-m'\right)n/d}.
\end{align*}
It will be convenient to express such sums by defining the function
\begin{equation}
\varDelta_{q}\left(x\right):=\frac{1}{q}\sum_{n=0}^{q-1}e^{i2\pi xn/q}=\frac{e^{i\pi\left(x-x/q\right)}}{q}\frac{\sin\left(\pi x\right)}{\sin\left(\pi x/q\right)}\label{eq:def of delta}
\end{equation}
over real $x$ and integer $q\geq1$ (note that $\varDelta_{q}\left(0\right)=1$).
Then, for $\nu'=\nu=0$ the overlaps of truncated momentum states
are give by 
\begin{align}
\braket{P_{0};m'}{P_{0};m} & =\varDelta_{w_{x}}\left(\frac{m-m'}{k_{x}}\right).\label{eq:P_0 inner prod}
\end{align}

Then, with \eqref{eq:Pi_X0Pi_P0Pi_X0 in trunc basis} and \eqref{eq:def of delta}
we can express
\begin{align*}
\left\langle \boldsymbol{p}_{\textrm{agree}}\right\rangle  & =\frac{k_{x}k_{p}}{d}\tr\left[\left(\Pi_{X;0}\Pi_{P;0}\Pi_{X;0}\right)^{2}\right]=\frac{1}{d}\frac{k_{p}}{k_{x}}\sum_{m,m'=0}^{w_{p}-1}\left|\braket{P_{0};m'}{P_{0};m}\right|^{2}\\
 & =\frac{1}{d}\frac{k_{p}}{k_{x}}\sum_{m,m'=0}^{w_{p}-1}\left|\varDelta_{w_{x}}\left(\frac{m-m'}{k_{x}}\right)\right|^{2}.
\end{align*}
Noting that the summand depends only on the difference $n=m-m'$,
we simplify
\[
\left\langle \boldsymbol{p}_{\textrm{agree}}\right\rangle =\frac{1}{d}\frac{k_{p}}{k_{x}}\sum_{n=1-w_{p}}^{w_{p}-1}\left(w_{p}-\left|n\right|\right)\left|\varDelta_{w_{x}}\left(\frac{n}{k_{x}}\right)\right|^{2}.
\]
Since the summed function is symmetric $\left|\varDelta_{w_{x}}\left(x\right)\right|^{2}=\left|\varDelta_{w_{x}}\left(-x\right)\right|^{2}$,
we further simplify 
\[
\left\langle \boldsymbol{p}_{\textrm{agree}}\right\rangle =\frac{1}{d}\frac{k_{p}}{k_{x}}\left[w_{p}\left|\varDelta_{w_{x}}\left(0\right)\right|^{2}+2\sum_{n=1}^{w_{p}-1}\left(w_{p}-n\right)\left|\varDelta_{w_{x}}\left(\frac{n}{k_{x}}\right)\right|^{2}\right].
\]
Finally, by substituting the explicit form \eqref{eq:def of delta}
of $\varDelta_{w_{x}}$ and recalling that $k_{x}=d/w_{x}$ , $k_{p}=d/w_{p}$
and $\varDelta_{w_{x}}\left(0\right)=1$, we find out how $\left\langle \boldsymbol{p}_{\textrm{agree}}\right\rangle $
varies as a function of $w_{x}$ and $w_{p}$ :
\begin{align}
\left\langle \boldsymbol{p}_{\textrm{agree}}\right\rangle  & =\frac{1}{d}\frac{w_{x}}{w_{p}}\left[w_{p}+2\sum_{n=1}^{w_{p}-1}\left(w_{p}-n\right)\frac{1}{w_{x}^{2}}\frac{\sin^{2}\left(\frac{\pi nw_{x}}{d}\right)}{\sin^{2}\left(\frac{\pi n}{d}\right)}\right]\nonumber \\
 & =\frac{w_{x}}{d}+\frac{2}{w_{x}w_{p}d}\sum_{n=1}^{w_{p}-1}\left(w_{p}-n\right)\frac{\sin^{2}\left(\frac{\pi nw_{x}}{d}\right)}{\sin^{2}\left(\frac{\pi n}{d}\right)}.\label{eq: lam_agree expr plane}
\end{align}

The apparent asymmetry under the exchange of $w_{x}$ with $w_{p}$
traces back to the apparent asymmetry under the exchange between $\Pi_{X;0}$
and $\Pi_{P;0}$ in Eq. \eqref{eq: p_agree =00003D tr=00005B(PPP)^2=00005D}.
These asymmetries are only apparent because 
\[
tr\left[\left(\Pi_{X;0}\Pi_{P;0}\Pi_{X;0}\right)^{2}\right]=tr\left[\Pi_{X;0}\Pi_{P;0}\Pi_{X;0}\Pi_{P;0}\right]=tr\left[\left(\Pi_{P;0}\Pi_{X;0}\Pi_{P;0}\right)^{2}\right].
\]
If we were to exchange $\Pi_{X;0}$ with $\Pi_{P;0}$ we would have
to exchange $w_{x}$ with $w_{p}$, and end up with 
\begin{equation}
\left\langle \boldsymbol{p}_{\textrm{agree}}\right\rangle =\frac{w_{p}}{d}+\frac{2}{w_{x}w_{p}d}\sum_{n=1}^{w_{x}-1}\left(w_{x}-n\right)\frac{\sin^{2}\left(\frac{\pi nw_{p}}{d}\right)}{\sin^{2}\left(\frac{\pi n}{d}\right)}.\label{eq: lam_agree expr plane exch}
\end{equation}

The symmetry under the exchange of $w_{x}$ with $w_{p}$ can also
be seen in Fig. \ref{fig:ClassicalityVw}(a) where we have used Eq.
\ref{eq: lam_agree expr plane} to plot $\left\langle \boldsymbol{p}_{\textrm{agree}}\right\rangle $
as a function of $w_{x}$ and $w_{p}$ . In Fig. \ref{fig:ClassicalityVw}(b)
we plot $\left\langle \boldsymbol{p}_{\textrm{agree}}\right\rangle $
for the diagonal $w=w_{x}=w_{p}$, together with the upper and lower
bounds 
\begin{align}
 & \left\langle \boldsymbol{p}_{\textrm{agree}}\right\rangle \leq w^{2}/d &  & w<\sqrt{d}\label{eq: lam agree upper bound}\\
 & \left\langle \boldsymbol{p}_{\textrm{agree}}\right\rangle \geq1-\frac{2}{\pi^{2}}\frac{\ln\left(w^{2}/d\right)+3\pi^{2}/2}{w^{2}/d} &  & w>\sqrt{d}\label{eq:  lam agree lower bound}
\end{align}
See the Appendix for the derivation of these bounds. 
\begin{figure}
\begin{centering}
(a)\includegraphics[width=0.65\columnwidth]{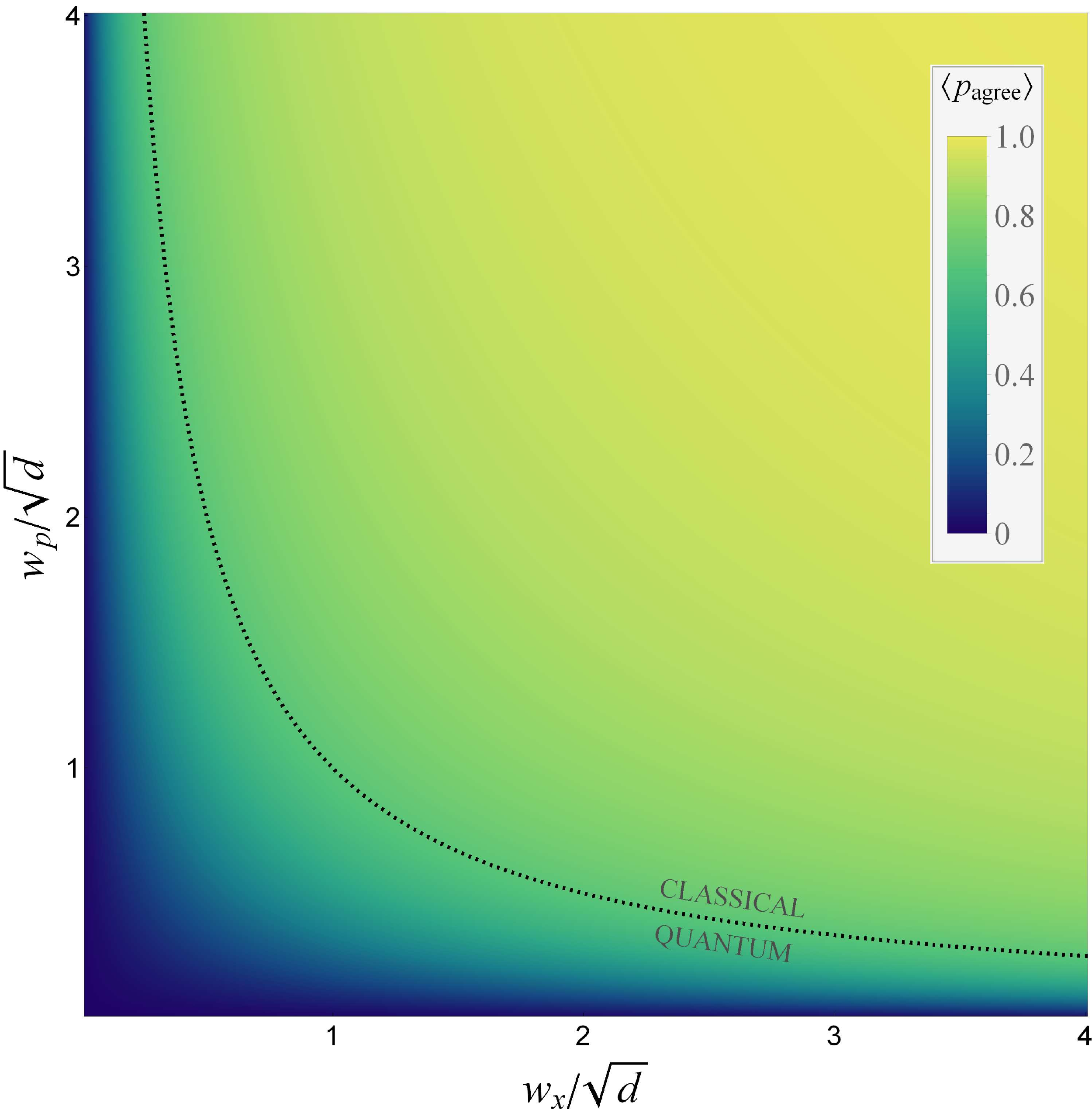}\vspace{1cm}
\par\end{centering}
\centering{}(b)\includegraphics[width=0.65\columnwidth]{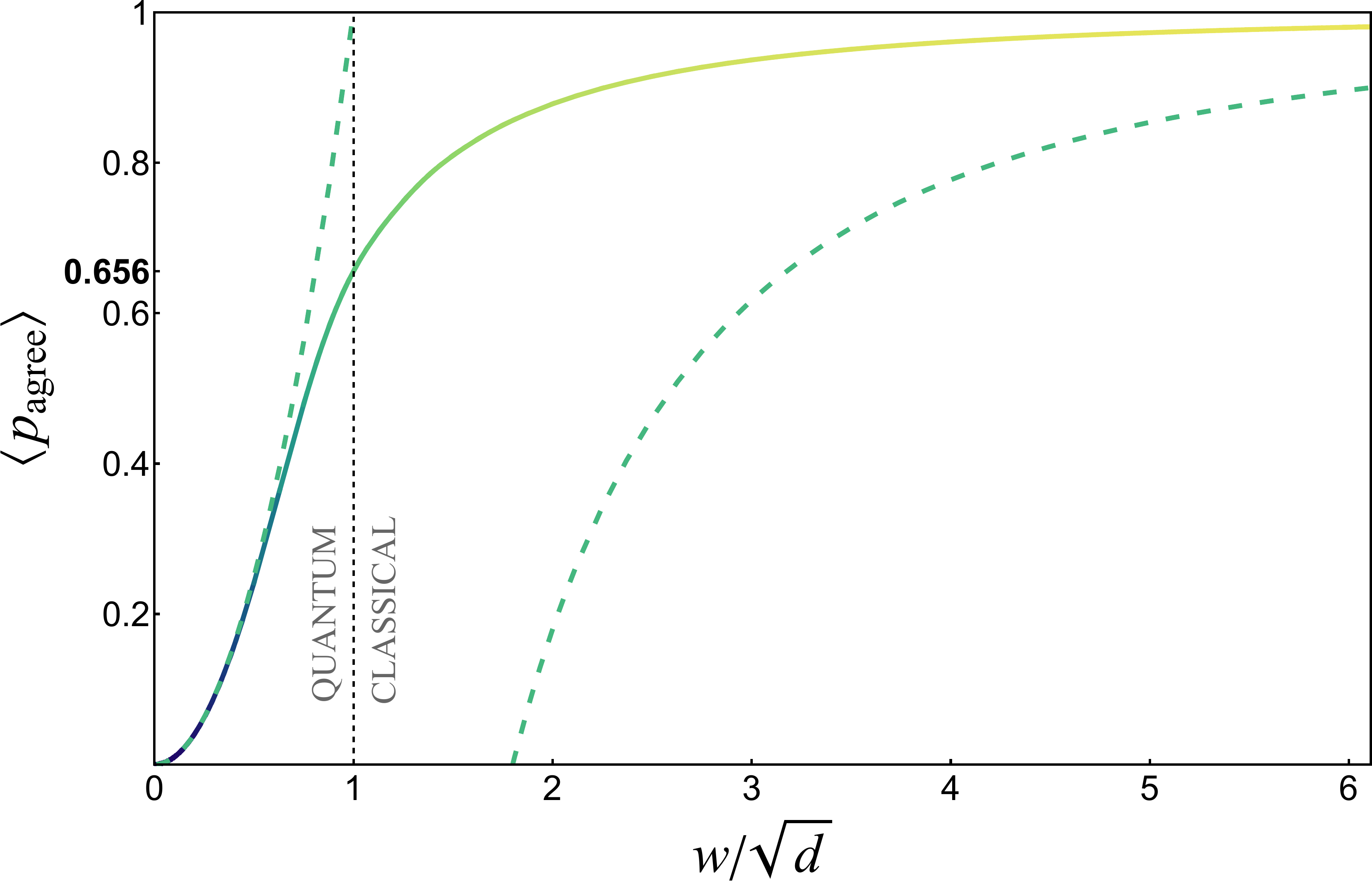}\caption{\label{fig:ClassicalityVw}(a) The plot of the average probability
$\left\langle \boldsymbol{p}_{\textrm{agree}}\right\rangle $ that
an instantaneous succession of position-momentum-position measurements
will agree on both outcomes of position as a function of the resolution
parameters $w_{x}$, $w_{p}$ on a lattice of length $d$. The dotted
curve $w_{x}w_{p}=d$ is the boundary that outlines the transitional
scale with respect to which we distinguish the quantum and classical
regimes. (b) The plot of $\left\langle \boldsymbol{p}_{\textrm{agree}}\right\rangle $
(solid) along the diagonal $w=w_{x}=w_{p}$ with the upper and lower
bounds (dashed) from Eqs. \eqref{eq: lam agree upper bound} and \eqref{eq:  lam agree lower bound}.}
\end{figure}

The upper bound \eqref{eq: lam agree upper bound} tells us that when
$w<\sqrt{d}$, the value of $\left\langle \boldsymbol{p}_{\textrm{agree}}\right\rangle $
falls to $0$ at least as fast as $\sim w^{2}$. The lower bound \eqref{eq:  lam agree lower bound}
tells us that when $w>\sqrt{d}$, the value of $\left\langle \boldsymbol{p}_{\textrm{agree}}\right\rangle $
climbs to $1$ at least as fast as $\sim1-\frac{\ln w^{2}}{w^{2}}$.
This implies that along the diagonal $w=w_{x}=w_{p}$, there is an
inflection in $\left\langle \boldsymbol{p}_{\textrm{agree}}\right\rangle $
around $w=\sqrt{d}$. Therefore, $\sqrt{d}$ is a natural boundary
that separates the scales of the quantum and classical regimes.

The above observation can be extended to the entire plane of $w_{x}$,
$w_{p}$, where the curve $w_{x}w_{p}=d$ generalizes the boundary
$w=\sqrt{d}$. According to the plot in Fig. \ref{fig:ClassicalityVw}(a),
as we get farther from the curve $w_{x}w_{p}=d$, we get deeper into
one of the regimes, and an inflection in $\left\langle \boldsymbol{p}_{\textrm{agree}}\right\rangle $
occurs somewhere near the curve. The fact that the curve $w_{x}w_{p}=d$
separates the scales of the quantum and classical regimes follows
from the observation that $\left\langle \boldsymbol{p}_{\textrm{agree}}\right\rangle \approx0.656$
holds almost everywhere on this curve (except for the far ends).

In order to see that, we assume that $d$ is large (otherwise the
lattice does not approximate a continuum) but finite. On the curve
$w_{x}w_{p}=d$ so Eq. \eqref{eq: lam_agree expr plane} simplifies
to
\begin{equation}
\left\langle \boldsymbol{p}_{\textrm{agree}}\right\rangle =\frac{1}{w_{p}}+\frac{2}{d^{2}}\sum_{n=1}^{w_{p}-1}\left(w_{p}-n\right)\frac{\sin^{2}\left(\frac{\pi n}{w_{p}}\right)}{\sin^{2}\left(\frac{\pi n}{d}\right)}.\label{eq: lam_agree on wxwp=00003Dd}
\end{equation}
In the intermediate range of values $1\ll w_{p}\ll d$ we have $\frac{n}{d}\ll1$
since $n<w_{p}$. We can then approximate $\sin^{-2}\left(\frac{\pi n}{d}\right)\approx\left(\frac{\pi n}{d}\right)^{-2}$
, so 
\begin{equation}
\left\langle \boldsymbol{p}_{\textrm{agree}}\right\rangle \approx\frac{1}{w_{p}}+\frac{2}{\pi^{2}}\sum_{n=1}^{w_{p}-1}\left(w_{p}-n\right)\frac{\sin^{2}\left(\frac{\pi n}{w_{p}}\right)}{n^{2}}.\label{eq: lam_agree on wxwp=00003Dd wp<<d}
\end{equation}
Since the intermediate range also implies that $1\ll w_{p}$, we can
approximate the sum with an integral by introducing the variable $\alpha=\frac{n}{w_{p}}\in\left[0,1\right]$
and $d\alpha=\frac{1}{w_{p}}$. Then,
\begin{align*}
\left\langle \boldsymbol{p}_{\textrm{agree}}\right\rangle  & \approx\frac{1}{w_{p}}+\frac{2}{\pi^{2}}\sum_{n=1}^{w_{p}-1}\frac{1}{w_{p}}\left(1-\frac{n}{w_{p}}\right)\frac{\sin^{2}\left(\pi\frac{n}{w_{p}}\right)}{n^{2}/w_{p}^{2}}\\
 & \approx d\alpha+\frac{2}{\pi^{2}}\int_{0}^{1}d\alpha\left(1-\alpha\right)\frac{\sin^{2}\left(\pi\alpha\right)}{\alpha^{2}}\approx0.656\,.
\end{align*}
Thus, for the intermediate range $1\ll w_{p}\ll d$ on the curve of
$w_{x}w_{p}=d$, we have $\left\langle \boldsymbol{p}_{\textrm{agree}}\right\rangle \approx0.656$.

Then we ask, for what values of $w_{p}$ does the approximation $\left\langle \boldsymbol{p}_{\textrm{agree}}\right\rangle \approx0.656$
breaks? For $w_{p}\sim1$ Eq. \eqref{eq: lam_agree on wxwp=00003Dd wp<<d}
still holds (since $w_{p}\ll d$) and its numeric values are
\begin{center}
\begin{tabular}{|c|cccccccc|}
\hline 
$w_{p}$ & 1 & 2 & 3 & 4 & ... & 15 & 16 & ...\tabularnewline
$\left\langle \boldsymbol{p}_{\textrm{agree}}\right\rangle $ & $1.00$ & $0.703$ & $0.675$ & $0.667$ & ... & $0.657$ & $0.656$ & $0.656$\tabularnewline
\hline 
\end{tabular} .
\par\end{center}

\noindent Thus, on one end of the curve $w_{x}w_{p}=d$, the approximation
$\left\langle \boldsymbol{p}_{\textrm{agree}}\right\rangle \approx0.656$
breaks for $w_{p}<16$ (considering $3$ significant figures). Since
$w_{x}$ and $w_{p}$ are interchangeable, on the other end of this
curve (where $w_{x}\sim1$) the approximation breaks for $w_{x}<16$.
Therefore, if $d$ is large then the approximation $\left\langle \boldsymbol{p}_{\textrm{agree}}\right\rangle \approx0.656$
holds almost everywhere on the curve $w_{x}w_{p}=d$, with the exception
of the far ends $w_{p}<16$ or $w_{x}<16$.

The curve $w_{x}w_{p}=d$ is significant not because there is something
special about the value $0.656$, but because it allows us to say
that
\[
\begin{cases}
\left\langle \boldsymbol{p}_{\textrm{agree}}\right\rangle \approx1 & w_{x}w_{p}\gg d\\
\left\langle \boldsymbol{p}_{\textrm{agree}}\right\rangle \approx0 & w_{x}w_{p}\ll d.
\end{cases}
\]
In other words, the significance of the curve $w_{x}w_{p}=d$ is that
it outlines the transitional scale in phase space with respect to
which we distinguish the quantum scale from the classical scale.

\section{The implications of the uncertainty principle on a lattice \label{sec:The-implications-of}}

\subsection{Inferring the size of the lattice}

Let us briefly point out one simple implication: it is possible to
infer the size of the lattice from the scale at which the transition
to the classical regime takes place. The general idea is that if we
take a generic state $\rho$ and probe the probability $\boldsymbol{p}_{\textrm{agree}}\left(\rho\right)$
at various scales of coarse-graining, the scale where $\boldsymbol{p}_{\textrm{agree}}\left(\rho\right)\sim0.656$
is the scale where $w_{x}w_{p}\sim d$ so the product $w_{x}w_{p}$
is an estimation of the value of $d$. Given concrete assumptions
about the limitations of state preparation and measurements, a more
specific protocol for determining $d$ can be designed around this
general idea.

\subsection{The continuum limit and lattice perturbations}

We will now introduce proper units to the lattice.

The total length of the lattice in proper units is $L=\delta xd$,
where $\delta x$ is the smallest unit of length associated with one
lattice spacing. The smallest unit of inverse length, or a wavenumber,
is then $1/L$. With the de Broglie relation $p=2\pi\hbar/\lambda$,
we can convert wavenumbers $1/\lambda$ to momenta, so the smallest
unit of momentum is $\delta p=2\pi\hbar/L$.\footnote{Note that the de Broglie relation is the source of the Planck constant
in all of the following equations} The coarse-graining intervals $w_{x}$ and $w_{p}$ become $\Delta x=\delta xw_{x}$
and $\Delta p=\delta pw_{p}$ when expressed in proper units.

The continuum limit is achieved by taking $\delta x\rightarrow0$
and $d\rightarrow\infty$ while keeping $L$ constant. The coarse-graining
interval of position $\Delta x=\delta xw_{x}$ is kept constant by
fixing the total number of intervals $k_{x}=d/w_{x}$ while $w_{x}\rightarrow\infty$.
Unlike $\delta x$, $\delta p=2\pi\hbar/L$ does not vanish in the
continuum limit (the momentum of a particle in a box remains quantized)
so the coarse-graining intervals of momentum $\Delta p=\delta pw_{p}$
are unaffected and $w_{p}$ remains a finite integer.

We may now ask what happens to $\left\langle \boldsymbol{p}_{\textrm{agree}}\right\rangle $
as we take the continuum limit. Since $w_{x}/d=\Delta x/L$, the expression
\eqref{eq: lam_agree expr plane} can be re-stated using the proper
units of length as 
\begin{equation}
\left\langle \boldsymbol{p}_{\textrm{agree}}\right\rangle =\frac{\Delta x}{L}+\frac{L}{\Delta x}\frac{2}{w_{p}}\sum_{n=1}^{w_{p}-1}\left(w_{p}-n\right)\frac{\sin^{2}\left(\frac{\pi n\Delta x}{L}\right)}{\left[d\sin\left(\frac{\pi n}{d}\right)\right]^{2}}.\label{eq: lam_agree expr plane cont}
\end{equation}
We do not have to change to the proper units of momentum because 
\[
w_{p}=\frac{\Delta p}{\delta p}=\frac{\Delta p}{2\pi\hbar}L,
\]
which is a legitimate quantity even in the continuum limit (provided
that $L$ is finite).

The only evidence for the lattice structure that remains in Eq. \eqref{eq: lam_agree expr plane cont}
is the $d$-dependence of the factors 
\begin{equation}
\left[d\sin\left(\frac{\pi n}{d}\right)\right]^{-2}=\frac{1}{\pi^{2}n^{2}}+\frac{1}{3d^{2}}+O\left(\frac{1}{d^{3}}\right).
\end{equation}
In the continuum limit these factors reduce to $1/\pi^{2}n^{2}$,
but when the minimal length $\delta x=L/d$ is above $0$, these factors
are perturbed with the leading order contribution of $1/3d^{2}=\delta x^{2}/3$$L^{2}$.

\subsection{Factorizing the Planck constant}

Observe that the smallest unit of phase space area on a lattice is
$\delta x\delta p=2\pi\hbar/d$.\footnote{This is a well known constraint that comes up in the construction
of Generalized Clifford Algebras in finite-dimensional quantum mechanics.
See \citep{singh2018modeling} for an overview and the references
therein.} Therefore, the curve $w_{x}w_{p}=d$ that outlines the transitional
scale in phase space becomes 
\begin{equation}
\Delta x\Delta p=\delta x\delta p\,w_{x}w_{p}=\delta x\delta p\,d=2\pi\hbar.\label{eq:class. boundary with units}
\end{equation}
Thus, we have recovered Heisenberg's original argument where the Planck
constant identifies the transitional scale in phase space. We now
see that in the unitless lattice setting (where $\delta x\equiv1$
and $\delta p\equiv1$) the constant $d$ is the unitless ``Planck
constant''.\footnote{Note that unlike $2\pi\hbar$, the constant $d$ depends on the size
of the system. This inconstancy traces back to the fact that in the
unitless case we define $\delta p\equiv1$, while in proper units
we have $\delta p=2\pi\hbar/L$, which depends on the total length
$L$.}

In the continuous phase space the uncertainty principle is only associated
with the constant $2\pi\hbar$, which does not admit a preferred factorization
into position and momentum. On the lattice, however, the same constant
is given by $\delta x\delta pd$, which can be factorized as $\delta x\sqrt{d}$
and $\delta p\sqrt{d}$. We will now argue that the constants $\delta x\sqrt{d}$
and $\delta p\sqrt{d}$ are more than arbitrary factors of the Planck
constant. In fact, these are the primary scales associated with the
uncertainty principle on a lattice and the Planck constant is a secondary
quantity derived from their product.

Returning to the unitless picture of Fig. \ref{fig:ClassicalityVw}(a),
observe that if the localization in position $w_{x}$ approaches $\sqrt{d}$
from above, in order to stay in the classical regime the localization
in momentum $w_{p}$ has to diverge faster than the convergence in
$w_{x}$. In contrast, as long as both $w_{x},w_{p}\gg\sqrt{d}$,
the classical regime is insensitive to the variations in these variables
and there is no need to compensate the increase in localization for
one variable with the decrease in localization for the other.

We can then define the transitional scale for a single variable as
the scale around which increases in localization for one variable
(say position) result in\emph{ higher} decreases in localization for
the other variable (say momentum). This definition is only meaningful
on a lattice because it requires the fundamental units $\delta x$
and $\delta p$ in terms of which we can compare the changes in localization
for both variables.\footnote{In the continuum we cannot tell how the localization for one variable
compares to the other because the answer depends on the arbitrary
choice of units.} From the plot in Fig. \ref{fig:ClassicalityVw}(a) we see that $\sqrt{d}$
is the transitional scale for a single unitless variable. It then
follows that the uncertainty principle on a lattice is primarily associated
with the unitless constant $\sqrt{d}$, that in turn defines the transitional
scales $\delta x\sqrt{d}$ and $\delta p\sqrt{d}$ for position and
momentum, and then the transitional scale in phase space is given
by

\[
\left(\delta x\sqrt{d}\right)\left(\delta p\sqrt{d}\right)=\delta x\delta p\,d=2\pi\hbar.
\]

With proper units we conclude that on a lattice, in addition to the
minimal length $\delta x$ and the total length $L$, quantum mechanics
imposes another fundamental length 
\[
l_{u}=\delta x\sqrt{d}.
\]
The length $l_{u}$ is directly related to the minimal length $\delta x$
via the total length $L=\delta xd$ as $l_{u}=\sqrt{\delta x\,L}$
or $\delta x=l_{u}^{2}/L$. The length $l_{u}$ is therefore the geometric
mean of the minimal length $\delta x$ and the maximal length $L$.
It can also be framed as the length for which there are as many intervals
$l_{u}$ in $L$ as there are $\delta x$ in $l_{u}$. In the continuum
limit, where the minimal length $\delta x$ vanishes, the length $l_{u}=\sqrt{\delta x\,L}$
must also be $0$. Therefore, if we can establish that $l_{u}>0$
then it follows that $\delta x>0$.

The advantage of $l_{u}$ as an indicator of the discontinuity of
space is that it is greater than $\delta x$ by orders of magnitude.
For instance, for $L\sim1\,\textrm{m}$ of the order of a macroscopic
box and $\delta x\sim10^{-35}\,\textrm{m}$ of the order of Planck
length, we have $l_{u}\sim10^{-17.5}\,\textrm{m}$ which is much closer
to the scale of experiments than $10^{-35}\,\textrm{m}$.

It is not clear at this point what are the observable effects associated
with the fundamental length $l_{u}$. However, if such effects can
be identified then the discontinuity of space can be probed at scales
that are many orders of magnitude greater than the Planck length.

\chapter{Conclusion}

Inspired by the methods of symmetries, we have studied an operator
algebraic approach to reductions in finite-dimensional quantum mechanics,
and its extension to operator systems. For this purpose we have identified
a convenient representation of the irreps structure in the form of
bipartition tables, and introduced the Scattering Algorithm to find
the irreps structures of operator algebras. The applications of this
approach have been subdivided into reductions of states and reductions
of dynamics, and studied separately. The extension of operator algebras
to operator systems has led to the formulation of the quantum notion
of coarse-graining that is analogous to its classical counterpart.

We started with an observation that in finite-dimensional settings
the structure of irreducible representations of groups is in fact
associated with operator algebras. Thus, the simplifications that
are usually associated with symmetries can be attributed to operator
algebras. We studied the representation theory of finite-dimensional
operator algebras in Chapter \ref{chap:Operator-algebras-and} and
concluded that all the important aspects of irreps of operator algebras,
such as minimal projections and invariant subspaces, are captured
by bipartition tables.

Throughout this thesis we saw many examples of bipartition tables.
We summarize below the possible shapes of bipartition tables and the
corresponding reductions.\footnote{To this list we have added the change of basis transformation that
can be specified by linearly arranging the new basis in place of the
old basis.}

\noindent\begin{minipage}[c]{1\columnwidth}%
\begin{center}
\vspace{0.5\baselineskip}
\begin{tabular}{>{\centering}p{0.25cm}|>{\centering}p{0.25cm}|>{\centering}p{0.25cm}|>{\centering}p{0.25cm}|>{\centering}p{0.25cm}|>{\centering}p{0.25cm}|>{\centering}p{0.25cm}|>{\centering}p{0.25cm}}
\cline{2-7} \cline{3-7} \cline{4-7} \cline{5-7} \cline{6-7} \cline{7-7} 
 &  &  &  &  &  &  & \tabularnewline
\cline{2-7} \cline{3-7} \cline{4-7} \cline{5-7} \cline{6-7} \cline{7-7} 
\multicolumn{8}{c}{Change of basis}\tabularnewline
\end{tabular}\hspace*{1cm}%
\begin{tabular}{>{\centering}p{0.25cm}>{\centering}p{0.25cm}|>{\centering}p{0.25cm}|>{\centering}p{0.25cm}|>{\centering}p{0.25cm}}
\cline{3-3} 
 &  &  & \multicolumn{1}{>{\centering}p{0.25cm}}{} & \tabularnewline
\cline{3-3} 
 &  &  & \multicolumn{1}{>{\centering}p{0.25cm}}{} & \tabularnewline
\cline{3-3} 
 &  &  & \multicolumn{1}{>{\centering}p{0.25cm}}{} & \tabularnewline
\cline{3-4} \cline{4-4} 
 & \multicolumn{1}{>{\centering}p{0.25cm}}{} &  &  & \tabularnewline
\cline{4-4} 
 & \multicolumn{1}{>{\centering}p{0.25cm}}{} &  &  & \tabularnewline
\cline{4-4} 
 & \multicolumn{1}{>{\centering}p{0.25cm}}{} &  &  & \tabularnewline
\cline{4-4} 
\multicolumn{5}{c}{Measurement}\tabularnewline
\end{tabular}\hspace*{1cm}%
\begin{tabular}{>{\centering}p{0.25cm}|>{\centering}p{0.25cm}|>{\centering}p{0.25cm}|>{\centering}p{0.25cm}|>{\centering}p{0.25cm}|>{\centering}p{0.25cm}}
\cline{1-3} \cline{2-3} \cline{3-3} 
\multicolumn{1}{|>{\centering}p{0.25cm}|}{} &  &  & \multicolumn{1}{>{\centering}p{0.25cm}}{} & \multicolumn{1}{>{\centering}p{0.25cm}}{} & \tabularnewline
\hline 
\multicolumn{1}{>{\centering}p{0.25cm}}{} & \multicolumn{1}{>{\centering}p{0.25cm}}{} &  &  &  & \multicolumn{1}{>{\centering}p{0.25cm}|}{}\tabularnewline
\cline{4-6} \cline{5-6} \cline{6-6} 
\multicolumn{6}{c}{Superselection}\tabularnewline
\end{tabular}\vspace{0.5\baselineskip}
\begin{tabular}{>{\centering}p{0.25cm}|>{\centering}p{0.25cm}|>{\centering}p{0.25cm}|>{\centering}p{0.25cm}|>{\centering}p{0.25cm}}
\cline{2-4} \cline{3-4} \cline{4-4} 
 &  &  &  & \tabularnewline
\cline{2-4} \cline{3-4} \cline{4-4} 
 &  &  &  & \tabularnewline
\cline{2-4} \cline{3-4} \cline{4-4} 
\multicolumn{5}{c}{Subsystem}\tabularnewline
\end{tabular}\hspace*{5cm}%
\begin{tabular}{>{\centering}p{0.25cm}>{\centering}p{0.25cm}|>{\centering}p{0.25cm}|>{\centering}p{0.25cm}|>{\centering}p{0.25cm}>{\centering}p{0.25cm}>{\centering}p{0.25cm}}
\cline{3-5} \cline{4-5} \cline{5-5} 
 &  &  &  & \multicolumn{1}{>{\centering}p{0.25cm}|}{} &  & \tabularnewline
\cline{3-5} \cline{4-5} \cline{5-5} 
 &  &  &  &  &  & \tabularnewline
\cline{3-4} \cline{4-4} 
 &  &  & \multicolumn{1}{>{\centering}p{0.25cm}}{} &  &  & \tabularnewline
\cline{3-3} 
\multicolumn{7}{c}{Partial subsystem}\tabularnewline
\end{tabular}\vspace{0.5\baselineskip}
\par\end{center}%
\end{minipage} Thus, bipartition tables unite a broad class of important concepts
in finite-dimensional quantum mechanics in a single picture.

The principal problem that arises in applications of operator algebras
is the derivation of the irreps structure from the generators of the
algebra. In Chapter \ref{chap:Identifying-the-irreps} we have addressed
this problem by introducing the Scattering Algorithm. The idea of
the algorithm is to apply the scattering operation to break the initial
spectral projections of the generators into minimal projections, and
use them to construct the bipartition tables specifying the irreps
structure.

As we have emphasized, the Scatting Algorithm is designed to allow
analytical derivations of the irreps structure without having to specify
the operators numerically. The execution of the algorithm mostly involves
multiplications and diagonalizations of operators, and we saw multiple
non-trivial examples that are simple enough to derive the irreps structure
with  pen and paper.

Applications associated with the reduction of states were studied
in Chapter \ref{chap:Operational-reductions-of-st}. We first observed
that the prototypical state reduction in the form of the partial trace
map can be understood as a map that accounts for operational constraints.
By adopting this perspective we defined state reductions as maps that
account for operational constraints given by a restriction of observables
to a subalgebra. Such state reduction maps were illustrated with examples
that involve lacking a common reference frame and encoding of quantum
information in a noiseless subsystem.

An important consequence of constraining the observables to a subsystem
is the decoherence of the reduced state. There is nothing special,
however, about constraining the observables to a subsystem, and in
principle decoherence can be the consequence of any operational constraint.
We saw how simple rotations (without interactions) can cause decoherence
under the operational constraint of not having a common reference
frame of direction in space. Not only such reduced states decohere,
but we can also single out the effective interaction term of the Hamiltonian
by considering the irreps structure of the operational constraint.
Thus, from this broader perspective all the implications of the decoherence
program follow primarily from operational constraints, of which the
restriction to subsystems is a special case.

Applications associated with the reduction of dynamics were studied
in Chapter \ref{chap:Operational-reductions-of-dyn}. The reduction
of Hamiltonians with symmetries was reexamined and the condition for
a group to be a symmetry was relaxed. Specifically, we showed that
the Hamiltonian may have a symmetry breaking term and still be reducible,
as long as this term is itself an element of the group algebra. We
then introduced the symmetry-agnostic approach to the reduction of
Hamiltonians where we shifted the focus from symmetries to operator
algebras. This approach was demonstrated in two problems concerned
with finding the possible qubit encodings for a control Hamiltonian
in quantum dot arrays.

The aim of the symmetry-agnostic approach is not to replace the concept
of symmetries but rather provide an alternative for problems where
identifying the symmetries is not easy. In particular, when dealing
with Hamiltonians that have multiple different terms it may not be
obvious what their common symmetry group is. In addition, when the
symmetry group is identified, it is still necessary to find the irreps
structure of the group in order to reduce the dynamics. In the symmetry-agnostic
approach we also have to find the irreps structure of the algebra
generated by the Hamiltonian terms. However, it is no longer necessary
to identify any symmetries and we can start with the problem of finding
the irreps directly.

When considering a simple reduction problem of compressing a qutrit
into a qubit we observed that it does not seem to have a satisfying
solution in the framework of operator algebras. This has lead in Chapter
\ref{chap:Quantum-coarse-graining} to the extension of the mathematical
framework of state reductions from operator algebras to operator systems.
The resulting state reduction maps turned out to be the quantum analogue
of the classical notion of coarse-graining that so far did not have
an equivalent in quantum theory.

In applications such as quantum state compression or tomography, we
have a set of physically available observables and we want to find
a reduction map that represents the state compression or the tomographic
reconstruction of the state. If we assume that these observables form
an operator algebra, we can use the Scattering Algorithm to produce
the bipartition tables from which the reduction map is constructed.
Physically available observables, however, do not usually form an
operator algebra so it is more realistic to assume an operator system
instead.

Finally, in Chapter \ref{chap:The-uncertainty-principle} we studied
the effects of the uncertainty principle as a function of measurement
resolution on a lattice. By introducing a measure of mutual disturbance
between incompatible observables we characterized the transition of
the uncertainty principle to the classical regime with decreasing
resolution of measurements. From this characteristic function we were
able to conclude that the resolution scale that separates the quantum
and classical regimes is given by the square root of the unitless
length of the lattice.

The analysis of the uncertainty principle on a lattice implies certain
effects that can be associated with the discontinuity of space. Specifically,
we saw that the probability that a successive measurement of position-momentum-position
will agree on both outcomes of position is perturbed by the existence
of minimal length on the lattice. We also noted that if the minimal
length exists, then the geometric mean of the minimal length and the
maximal length is a special length scale that is singled out by the
uncertainty principle. In principle, this special length scale is
directly related to the discontinuity of space, but it is much longer
than the minimal length itself. However, it is not yet clear what
measurable effects can be associated with it.

Regarding the directions of future research, there are a few questions
that are worth exploring further.

We saw that bipartition tables can represent various maps such as
the partial trace, unitary transformation, and a measurement; in the
most general case partial bipartition tables represent quantum coarse-graining.
It would be interesting to find out what class of CPTP maps can be
represented with bipartition tables, and whether we can use bipartition
tables to represent CPTP maps in general.

Even though we have designed the Scattering Algorithm for purely analytical
uses, it would be good to have a rigorous complexity analysis of its
runtime and compare it to numeric implementations. For calculations
it is also desirable to have a computer implementation of the Scattering
Algorithm in a symbolic calculation software such as the Wolfram Mathematica.

More importantly, just as we have the Scattering Algorithm for constructing
bipartition tables from a generating set of an operator algebra, we
want to be able to construct partial bipartition tables from a spanning
set of an operator system. Without something like the Scattering Algorithm
for operator systems, the idea of reduction by quantum coarse-graining
is difficult to implement in applications.

We have pointed out that all the implications of the decoherence program
follow primarily from operational constraints. That is, decoherence
is not just the result of how the observed system interact with other
systems, it is also the result of how the observer interacts with
the observed system. It would be interesting to find out whether implications
such as the emergence of classicality can be attributed to operational
limitations that go beyond the paradigm of the system-environment
split.

\bibliographystyle{plainnat}
\bibliography{ThesBib}

\chapter*{Appendix}
\addcontentsline{toc}{part}{Appendix}

\section*{{\normalsize{}A calculation of the bounds \eqref{eq: lam agree upper bound}
and \eqref{eq:  lam agree lower bound}\label{app: calc of bounds}}}

Here we will assume $w=w_{x}=w_{p}$ and $k=k_{x}=k_{p}$.

In order to calculate the bounds on $\left\langle \boldsymbol{p}_{\textrm{agree}}\right\rangle $
we will have to find a different way to express $\Pi_{X;0}\Pi_{P;0}\Pi_{X;0}$.
Recalling Eq. \eqref{eq:P_0 inner prod} and the function \eqref{eq:def of delta}
we now have
\[
\left|\braket{P_{0};m'}{P_{0};m}\right|=\left|\varDelta_{w}\left(\frac{m-m'}{k}\right)\right|=\frac{\sin\left(\pi\frac{m-m'}{k}\right)}{w\sin\left(\pi\frac{m-m'}{d}\right)}.
\]
Observe that the truncated momentum states are orthogonal when the
difference $m-m'$ is an integer number of $k$'s. That is, for any
integers $c$, $c'$ and $n$ the states $\ket{P_{0};ck+n}$ and $\ket{P_{0};c'k+n}$
are orthogonal.

In Eq. \ref{eq:Pi_X0Pi_P0Pi_X0 in trunc basis} we have derived the
form

\begin{equation}
\Pi_{X;0}\Pi_{P;0}\Pi_{X;0}=\frac{1}{k}\sum_{m=0}^{w-1}\ket{P_{0};m}\bra{P_{0};m}\label{eq:P0P0P0 form recall}
\end{equation}
where $\ket{P_{0};m}\bra{P_{0};m}$ are rank 1 projections. Since
some of these projections are pairwise orthogonal, we can group them
together and express $\Pi_{X;0}\Pi_{P;0}\Pi_{X;0}$ as a smaller sum
of higher rank projections.

In order to do that, let us first assume that $\gamma=w/k$ is a non-zero
integer (we will not need this assumption in general). Then the set
of integers $\left\{ m=0,...,w-1\right\} $ can be partitioned into
$k$ subsets $\Omega_{n}=\left\{ ck+n\,|\,c=0,...,\gamma-1\right\} $
with $n=0,...,k-1$. Thus, we can group up the orthogonal elements
in the sum \eqref{eq:P0P0P0 form recall} as 
\[
\Pi_{X;0}\Pi_{P;0}\Pi_{X;0}=\frac{1}{k}\sum_{n=0}^{k-1}\sum_{m\in\Omega_{n}}\ket{P_{0};m}\bra{P_{0};m}=\frac{1}{k}\sum_{n=0}^{k-1}\Pi^{\left(n\right)}
\]
where we have introduced the rank $\gamma$ projections 
\[
\Pi^{\left(n\right)}=\sum_{m\in\Omega_{n}}\ket{P_{0};m}\bra{P_{0};m}=\sum_{c=0}^{\gamma-1}\ket{P_{0};ck+n}\bra{P_{0};ck+n}.
\]

When $\gamma=w/k$ is not an integer, the accounting of indices is
more involved. We have to introduce the integer part $g=\left\lfloor \gamma\right\rfloor $
and the remainder part $r=w-k\left\lfloor \gamma\right\rfloor $ of
$\gamma$. As before, we partition the set $\left\{ m=0,...,w-1\right\} $
into subsets 
\[
\Omega_{n}:=\begin{cases}
\left\{ ck+n\,|\,c=0,...,g\right\}  & n<r\\
\left\{ ck+n\,|\,c=0,...,g-1\right\}  & n\geq r
\end{cases}
\]
but now they are not of equal size and the range of $n$ depends on
whether $\gamma\ge1$. When $\gamma\ge1$ then $\left|\Omega_{n}\right|$
is $g+1$ for $n<r$ and $g$ for $n\geq r$. When $\gamma<1$ so
$g=0$ and $r=w$, then $\left|\Omega_{n}\right|=1$ for $n<w$ but
$\left|\Omega_{n}\right|=0$ for $n\geq w$ so we do not need to count
$\Omega_{n}$ for $n\geq w$. Noting that the condition $\gamma\geq1$
is equivalent to $\min\left(k,w\right)=k$ and the condition $\gamma<1$
is equivalent to $\min\left(k,w\right)=w$, we conclude that we only
have to count $\Omega_{n}$ for $n<\min\left(k,w\right)$. Therefore,
for the general $\gamma$ we have
\begin{equation}
\Pi_{X;0}\Pi_{P;0}\Pi_{X;0}=\frac{1}{k}\sum_{n=0}^{\min\left(k,w\right)-1}\sum_{m\in\Omega_{n}}\ket{P_{0};m}\bra{P_{0};m}=\frac{1}{k}\sum_{n=0}^{\min\left(k,w\right)-1}\Pi^{\left(n\right)}\label{eq: P0P0P0 in terms of P^(n)}
\end{equation}
and the projections 
\[
\Pi^{\left(n\right)}=\sum_{m\in\Omega_{n}}\ket{P_{0};m}\bra{P_{0};m}=\sum_{c=0}^{g_{n}-1}\ket{P_{0};ck+n}\bra{P_{0};ck+n}
\]
are now of the rank 
\[
g_{n}=\begin{cases}
g+1 & n<r\\
g & n\geq r.
\end{cases}
\]

Using the new form \eqref{eq: P0P0P0 in terms of P^(n)}, we can re-express
Eq. \ref{eq: p_agree =00003D tr=00005B(PPP)^2=00005D} as
\begin{equation}
\left\langle \boldsymbol{p}_{\textrm{agree}}\right\rangle =\frac{k^{2}}{d}\tr\left[\left(\Pi_{X;0}\Pi_{P;0}\Pi_{X;0}\right)^{2}\right]=\frac{1}{d}\sum_{n,n'=0}^{\min\left(k,w\right)-1}\tr\left[\Pi^{\left(n\right)}\Pi^{\left(n'\right)}\right].\label{eq: lam_agree in terms of P^n P^n'}
\end{equation}

\subsubsection*{The upper bound}

The quantity $\tr\left[\Pi^{\left(n\right)}\Pi^{\left(n'\right)}\right]$
is the Hilbert-Schmidt inner product $\left\langle \Pi^{\left(n\right)},\Pi^{\left(n'\right)}\right\rangle $
(also known as Frobenius inner product) of the operators $\Pi^{\left(n\right)}$
and $\Pi^{\left(n'\right)}$. Therefore, it obeys the Cauchy--Schwarz
inequality 
\[
\left|\tr\left[\Pi^{\left(n\right)}\Pi^{\left(n'\right)}\right]\right|^{2}=\left|\left\langle \Pi^{\left(n\right)},\Pi^{\left(n'\right)}\right\rangle \right|^{2}\leq\left\langle \Pi^{\left(n\right)},\Pi^{\left(n\right)}\right\rangle \left\langle \Pi^{\left(n'\right)},\Pi^{\left(n'\right)}\right\rangle =\tr\left[\Pi^{\left(n\right)}\right]\tr\left[\Pi^{\left(n'\right)}\right].
\]
Since the value 
\[
\tr\left[\Pi^{\left(n\right)}\Pi^{\left(n'\right)}\right]=\sum_{m\in\Omega_{n}}\sum_{m'\in\Omega_{n'}}\left|\braket{P_{0};m}{P_{0};m'}\right|^{2}
\]
is clearly real and positive, we get 
\[
\tr\left[\Pi^{\left(n\right)}\Pi^{\left(n'\right)}\right]\leq\sqrt{\tr\left[\Pi^{\left(n\right)}\right]\tr\left[\Pi^{\left(n'\right)}\right]}.
\]

The value of $\tr\left[\Pi^{\left(n\right)}\right]$ is the rank of
the projection which is either $g$ or $g+1$ so
\[
\tr\left[\Pi^{\left(n\right)}\Pi^{\left(n'\right)}\right]\leq g+1.
\]
Therefore, the form of $\left\langle \boldsymbol{p}_{\textrm{agree}}\right\rangle $
in Eq. \eqref{eq: lam_agree in terms of P^n P^n'} implies that 
\[
\left\langle \boldsymbol{p}_{\textrm{agree}}\right\rangle \leq\frac{1}{d}\sum_{n,n'=0}^{\min\left(k,w\right)-1}\left(g+1\right)=\left(g+1\right)\frac{\min\left(k,w\right)^{2}}{d}.
\]

When $\gamma\geq1$, this upper bound is greater or equal to $1$
because
\[
\left(g+1\right)\frac{\min\left(k,w\right)^{2}}{d}=\left(g+1\right)\frac{k^{2}}{d}\geq\gamma\frac{k^{2}}{d}=w\frac{k}{d}=1
\]
which is not helpful since we already know that $\left\langle \boldsymbol{p}_{\textrm{agree}}\right\rangle \leq1$
for it is a probability. When $\gamma<1$, on the other hand, we have
$g=0$ and so 
\[
\left(g+1\right)\frac{\min\left(k,w\right)^{2}}{d}=\frac{w^{2}}{d}.
\]
Thus, when $\gamma<1$, which translates to $w<k=d/w$ so $w<\sqrt{d}$,
we have the upper bound
\[
\left\langle \boldsymbol{p}_{\textrm{agree}}\right\rangle \leq\frac{w^{2}}{d}.
\]

\subsubsection*{The lower bound}

We will now focus on the lower bound of the inner product $\tr\left[\Pi^{\left(n\right)}\Pi^{\left(n'\right)}\right]$
for the case $\gamma\geq1$ (so $w\geq\sqrt{d}$ and $\min\left(k,w\right)=k$)
and then substitute the result in Eq. \eqref{eq: lam_agree in terms of P^n P^n'}.

Since we are interested in the lower bound, we can simplify the expression
by discarding the terms $c,c'=g$ in the sum
\[
\tr\left[\Pi^{\left(n\right)}\Pi^{\left(n'\right)}\right]=\sum_{c=0}^{g_{n}-1}\sum_{c'=0}^{g_{n'}-1}\left|\braket{P_{0};c'k+n'}{P_{0};ck+n}\right|^{2}\geq\sum_{c,c'=0}^{g-1}\left|\braket{P_{0};c'k+n'}{P_{0};ck+n}\right|^{2}.
\]
According to Eq. \eqref{eq:P_0 inner prod} we have
\[
\left|\braket{P_{0};c'k+n'}{P_{0};ck+n}\right|^{2}=\left|\varDelta_{w}\left(c-c'+\alpha\right)\right|^{2}
\]
where we have introduced the variable $\alpha=\frac{n-n'}{k}$. We
can now identify the sum 
\[
S\left(\alpha\right)=\sum_{c,c'=0}^{g-1}\left|\varDelta_{w}\left(c-c'+\alpha\right)\right|^{2}\,\,\,\,\,\leq\tr\left[\Pi^{\left(n\right)}\Pi^{\left(n'\right)}\right]
\]
and focus on lower bounding $S\left(\alpha\right)$ for all possible
$\alpha$.

Since $\left|\varDelta_{w}\left(x\right)\right|^{2}$ is a symmetric
function of $x$ we have
\[
\left|\varDelta_{w}\left(c-c'+\alpha\right)\right|^{2}=\left|\varDelta_{w}\left(-c+c'-\alpha\right)\right|^{2}
\]
and since the values of $c$ and $c'$ are interchangeable in the
sum, we conclude that $S\left(\alpha\right)$ is a symmetric function
of $\alpha$. Therefore, we only need to consider positive $\alpha=\frac{n-n'}{k}$,
and since $n,n'=0,...,k-1$, it takes the values $\alpha=0,\frac{1}{k},...,\frac{k-1}{k}\in\left[0,1\right]$.

Since the summand in $S\left(\alpha\right)$ only depends on the differences
$l=c-c'$, we can simplify the sum
\[
S\left(\alpha\right)=\sum_{l=-g+1}^{g-1}\left(g-\left|l\right|\right)\left|\varDelta_{w}\left(l+\alpha\right)\right|^{2}=\sum_{l=-g+1}^{g-1}\frac{\left(g-\left|l\right|\right)}{w^{2}}\frac{\sin^{2}\left(\pi\left(l+\alpha\right)\right)}{\sin^{2}\left(\pi\left(l+\alpha\right)/w\right)}
\]
where in the last step we substituted the explicit form of $\varDelta_{w}$.
Note that $\sin^{2}\left(\pi\left(l+\alpha\right)\right)=\sin^{2}\left(\pi\alpha\right)$
for integer $l$ and also $\sin^{-2}\left(\frac{\pi\left(l+\alpha\right)}{w}\right)\geq\left(\frac{\pi\left(l+\alpha\right)}{w}\right)^{-2}$
so we get
\begin{align}
S\left(\alpha\right) & \geq\frac{\sin^{2}\left(\pi\alpha\right)}{\pi^{2}}\sum_{l=-g+1}^{g-1}\frac{g-\left|l\right|}{\left(l+\alpha\right)^{2}}.\label{eq: S as sin times s}
\end{align}
We will now focus on evaluating the lower bound of the sum
\begin{equation}
s\left(\alpha\right)=\sum_{l=-g+1}^{g-1}\frac{g-\left|l\right|}{\left(l+\alpha\right)^{2}}.\label{eq: s alpha sum}
\end{equation}

We can rearrange the elements of this sum as follows:

\[
s\left(\alpha\right)=\frac{g}{\alpha^{2}}+\sum_{l=1}^{g-1}\left[\frac{g-l}{\left(l+\alpha\right)^{2}}+\frac{g-l}{\left(l-\alpha\right)^{2}}\right]=\frac{g}{\alpha^{2}}+\sum_{l=1}^{g-1}\left[\frac{l}{\left(g-l+\alpha\right)^{2}}+\frac{l}{\left(g-l-\alpha\right)^{2}}\right]
\]
where in the last step we simply reversed the order of the elements
in the sum. Now we can introduce the auxiliary variables $\beta_{\pm}=g\pm\alpha$,
so
\begin{align}
s\left(\alpha\right) & =\frac{g}{\alpha^{2}}+\sum_{l=1}^{g-1}\left[\frac{l}{\left(l-\beta_{+}\right)^{2}}+\frac{l}{\left(l-\beta_{-}\right)^{2}}\right]=\frac{g}{\alpha^{2}}+\sum_{l=1}^{g-1}\left[\frac{\beta_{+}}{\left(l-\beta_{+}\right)^{2}}+\frac{1}{\left(l-\beta_{+}\right)}+\frac{\beta_{-}}{\left(l-\beta_{-}\right)^{2}}+\frac{1}{\left(l-\beta_{-}\right)}\right]\nonumber \\
 & =\frac{g}{\alpha^{2}}+s_{1}\left(\alpha\right)+s_{2}\left(\alpha\right)\label{eq: s alpha rearanged sum}
\end{align}
where we have identified the sums of harmonic-like series
\[
s_{1}\left(\alpha\right)=\sum_{l=1}^{g-1}\left[\frac{1}{\left(l-\beta_{-}\right)}+\frac{1}{\left(l-\beta_{+}\right)}\right]\,\,\,\,\,\,\,\,\,\,\,\,\,\,\,\,\,\,\,\,\,\,\,\,\,\,s_{2}\left(\alpha\right)=\sum_{l=1}^{g-1}\left[\frac{\beta_{-}}{\left(l-\beta_{-}\right)^{2}}+\frac{\beta_{+}}{\left(l-\beta_{+}\right)^{2}}\right].
\]

Such sums can be evaluated using the polygamma functions \citep{abramowitz1948handbook}
\[
\psi^{\left(j\right)}\left(x\right):=\frac{d^{j}}{dx^{j}}\ln\Gamma\left(x\right)
\]
where $\Gamma$ is the gamma function that interpolates the factorial
for all real (and complex) values. The two key properties of the polygamma
functions that we will need are the recursion and reflection relations
\begin{align}
\psi^{\left(j\right)}\left(1+x\right) & =\psi^{\left(j\right)}\left(x\right)+\left(-1\right)^{j}\frac{j!}{x^{j+1}}\label{eq:polygamma recursion}\\
\psi^{\left(j\right)}\left(1-x\right) & =\left(-1\right)^{j}\psi^{\left(j\right)}\left(x\right)+\left(-1\right)^{j}\pi\frac{d^{j}}{dx^{j}}\cot\left(\pi x\right).\label{eq:polygamma reflection}
\end{align}

For integer $g$ we can expand $\psi^{\left(j\right)}\left(g-x\right)$
for $j=0,1$ using the recursion relation \eqref{eq:polygamma recursion}
to get 
\begin{align*}
\psi^{\left(0\right)}\left(g-x\right) & =\psi^{\left(0\right)}\left(1-x\right)+\sum_{l=1}^{g-1}\frac{1}{l-x}\\
\psi^{\left(1\right)}\left(g-x\right) & =\psi^{\left(1\right)}\left(1-x\right)-\sum_{l=1}^{g-1}\frac{1}{\left(l-x\right)^{2}}.
\end{align*}
Applying the reflection relation \eqref{eq:polygamma reflection}
and rearranging yields
\begin{align}
 & \sum_{l=1}^{g-1}\frac{1}{l-x}=\psi^{\left(0\right)}\left(g-x\right)-\psi^{\left(0\right)}\left(x\right)-\pi\cot\left(\pi x\right)\label{eq: sum as polygamma 1}\\
 & \sum_{l=1}^{g-1}\frac{1}{\left(l-x\right)^{2}}=-\psi^{\left(1\right)}\left(g-x\right)-\psi^{\left(1\right)}\left(x\right)+\frac{\pi^{2}}{\sin^{2}\left(\pi x\right)}.\label{eq: sum as polygamma 2}
\end{align}

Now, using \eqref{eq: sum as polygamma 1} and recalling that $g-\beta_{\pm}=\mp\alpha$
we can express $s_{1}\left(\alpha\right)$ as 
\[
s_{1}\left(\alpha\right)=\psi^{\left(0\right)}\left(\alpha\right)-\psi^{\left(0\right)}\left(\beta_{-}\right)+\psi^{\left(0\right)}\left(-\alpha\right)-\psi^{\left(0\right)}\left(\beta_{+}\right)
\]
where the trigonometric terms cancel each other out as they are anti-symmetric
and periodic with integer $g$. We can re-express $\psi^{\left(0\right)}\left(\alpha\right)$
and $\psi^{\left(0\right)}\left(-\alpha\right)$ as $\psi^{\left(0\right)}\left(\alpha+1\right)$
using the recursion \eqref{eq:polygamma recursion} and reflection
relations \eqref{eq:polygamma reflection} respectively:
\[
\psi^{\left(0\right)}\left(\alpha\right)+\psi^{\left(0\right)}\left(-\alpha\right)=2\psi^{\left(0\right)}\left(\alpha+1\right)+\pi\cot\left(\pi\alpha\right)-\frac{1}{\alpha}.
\]
We can replace $2\psi^{\left(0\right)}\left(\alpha+1\right)$ with
its lower bound $2\psi^{\left(0\right)}\left(1\right)$ on the interval
$0\leq\alpha<1$ as the function $\psi^{\left(0\right)}\left(x\right)$
is monotonically increasing for $0\leq x$. For the same reason we
can also use the bound $\psi^{\left(0\right)}\left(\beta_{\pm}\right)\leq\psi^{\left(0\right)}\left(g+1\right)$
so we end up with the overall lower bound on the sum
\begin{align}
s_{1}\left(\alpha\right) & \geq2\psi^{\left(0\right)}\left(1\right)-2\psi^{\left(0\right)}\left(g+1\right)+\pi\cot\left(\pi\alpha\right)-\frac{1}{\alpha}.\label{eq:s1 LB}
\end{align}

Similarly, using \eqref{eq: sum as polygamma 2} we can express $s_{2}\left(\alpha\right)$
as
\[
s_{2}\left(\alpha\right)=-\left[\beta_{-}\psi^{\left(1\right)}\left(\alpha\right)+\beta_{+}\psi^{\left(1\right)}\left(-\alpha\right)\right]-\left[\beta_{-}\psi^{\left(1\right)}\left(\beta_{-}\right)+\beta_{+}\psi^{\left(1\right)}\left(\beta_{+}\right)\right]+\frac{\beta_{-}\pi^{2}}{\sin^{2}\left(\pi\beta_{-}\right)}+\frac{\beta_{+}\pi^{2}}{\sin^{2}\left(\pi\beta_{+}\right)}.
\]
Using the recursion \eqref{eq:polygamma recursion} and reflection
\eqref{eq:polygamma reflection} relations, we express
\[
-\left[\beta_{-}\psi^{\left(1\right)}\left(\alpha\right)+\beta_{+}\psi^{\left(1\right)}\left(-\alpha\right)\right]=2\alpha\psi^{\left(1\right)}\left(\alpha+1\right)-\frac{\beta_{-}}{\alpha^{2}}-\frac{\beta_{+}\pi^{2}}{\sin^{2}\left(\pi\alpha\right)}\geq-\frac{\beta_{-}}{\alpha^{2}}-\frac{\beta_{+}\pi^{2}}{\sin^{2}\left(\pi\alpha\right)}
\]
where in the last step we have replaced $2\alpha\psi^{\left(1\right)}\left(\alpha+1\right)$
with its lower bound $0$ at $\alpha=0$. Since $\psi^{\left(1\right)}\left(x\right)$
is monotonically decreasing for $0\leq x$ we also use the lower bound
\[
-\left[\beta_{-}\psi^{\left(1\right)}\left(\beta_{-}\right)+\beta_{+}\psi^{\left(1\right)}\left(\beta_{+}\right)\right]\geq-2g\psi^{\left(1\right)}\left(g-1\right).
\]
Thus, the overall lower bound for $s_{2}\left(\alpha\right)$ is
\begin{align}
s_{2}\left(\alpha\right)\geq & -\frac{\beta_{-}}{\alpha^{2}}-\frac{\beta_{+}\pi^{2}}{\sin^{2}\left(\pi\alpha\right)}-2g\psi^{\left(1\right)}\left(g-1\right)+\frac{\beta_{-}\pi^{2}}{\sin^{2}\left(\pi\beta_{-}\right)}+\frac{\beta_{+}\pi^{2}}{\sin^{2}\left(\pi\beta_{+}\right)}\nonumber \\
= & -\frac{\beta_{-}}{\alpha^{2}}-2g\psi^{\left(1\right)}\left(g-1\right)+\frac{\beta_{-}\pi^{2}}{\sin^{2}\left(\pi\alpha\right)}.\label{eq: s2 LB}
\end{align}
where in the last step we have used the fact that $\sin^{2}\left(\pi\beta_{\pm}\right)=\sin^{2}\left(\pi\alpha\right)$.

Combining the lower bounds \eqref{eq:s1 LB} and \eqref{eq: s2 LB}
into Eqs. \eqref{eq: S as sin times s}, \eqref{eq: s alpha sum},
\eqref{eq: s alpha rearanged sum}, we get 
\[
S\left(\alpha\right)\geq\,\,g-\frac{2\sin^{2}\left(\pi\alpha\right)}{\pi^{2}}\left[\psi^{\left(0\right)}\left(g+1\right)+g\psi^{\left(1\right)}\left(g-1\right)\right]-\alpha+\frac{2\sin^{2}\left(\pi\alpha\right)}{\pi^{2}}\psi^{\left(0\right)}\left(1\right)+\frac{\sin\left(2\pi\alpha\right)}{2\pi}.
\]

On the interval $0\leq\alpha<1$, the minimum value of 
\[
-\alpha+\frac{2\sin^{2}\left(\pi\alpha\right)}{\pi^{2}}\psi^{\left(0\right)}\left(1\right)+\frac{\sin\left(2\pi\alpha\right)}{2\pi}
\]
 is given by $-\epsilon_{1}\approx-1.005$ and the minimum value of
the coefficient $-\frac{2\sin^{2}\left(\pi\alpha\right)}{\pi^{2}}$
is $-\frac{2}{\pi^{2}}$. With that, we can get rid of the dependence
on $\alpha$: 
\begin{align*}
S\left(\alpha\right) & \geq S_{\min}=g-\frac{2}{\pi^{2}}\left(\psi^{\left(0\right)}\left(g+1\right)+g\psi^{\left(1\right)}\left(g-1\right)\right)-\epsilon_{1}.
\end{align*}

We know that $\psi^{\left(0\right)}\left(x\right)$ is a smooth function
for $x>0$ and it is bounded by \citep{alzer1997some}
\[
\ln x-\frac{1}{x}<\psi^{\left(0\right)}\left(x\right)<\ln x-\frac{1}{2x}
\]
so asymptotically the function $\psi^{\left(0\right)}\left(x+1\right)\sim\ln\left(x+1\right)$
and it converges to $\ln x$ from above. Since $\psi^{\left(1\right)}\left(x\right)=d\psi^{\left(0\right)}\left(x\right)/dx$
then asymptotically $\psi^{\left(1\right)}\left(x\right)\sim\frac{1}{x}$
so the function $x\psi^{\left(1\right)}\left(x-1\right)\sim x/\left(x-1\right)$
and it converges to $1$ from above. Therefore, for any $\epsilon_{2}>0$
there is a $x'>0$ such that for all $x>x'$
\[
\psi^{\left(0\right)}\left(x+1\right)+x\psi^{\left(1\right)}\left(x-1\right)\leq\ln x+1+\epsilon_{2}.
\]
Conveniently choosing $\epsilon_{2}=\frac{\pi^{2}}{2}\left(2-\epsilon_{1}\right)-1$
and solving for $x'$ results in $x'\approx1.722$. Thus, for all
$g\geq2>x'$ we have 
\begin{align*}
S_{\min} & \geq g-\frac{2}{\pi^{2}}\left(\ln g+1+\epsilon_{2}\right)-\epsilon_{1}=g-\frac{2}{\pi^{2}}\ln g-2\\
 & \geq\gamma-\frac{2}{\pi^{2}}\ln\gamma-3
\end{align*}
where the last inequality follows from $g=\left\lfloor \gamma\right\rfloor \geq\gamma-1$
and $\ln g\leq\ln\gamma$.

Recalling that $\tr\left[\Pi^{\left(n\right)}\Pi^{\left(n'\right)}\right]\geq S\left(\alpha\right)\geq S_{\min}$
and $\gamma=w/k=w^{2}/d$, we return to the Eq. \eqref{eq: lam_agree in terms of P^n P^n'}
and get the result

\begin{align*}
\left\langle \boldsymbol{p}_{\textrm{agree}}\right\rangle  & =\frac{1}{d}\sum_{n,n'=0}^{k-1}\tr\left[\Pi^{\left(n\right)}\Pi^{\left(n'\right)}\right]\geq\frac{k^{2}}{d}S_{\min}\geq\frac{1}{w^{2}/d}\left[w^{2}/d-\frac{2}{\pi^{2}}\ln\left(w^{2}/d\right)-3\right]\\
 & =1-\frac{2}{\pi^{2}}\frac{\ln\left(w^{2}/d\right)+3\pi^{2}/2}{w^{2}/d}.
\end{align*}

\end{document}